\documentclass[aps,pra,longbibliography,onecolumn,superscriptaddress,nofootinbib]
{revtex4-2}
\pdfoutput=1

\makeatletter
\def\l@subsubsection#1#2{}
\makeatother

\usepackage[toc,title]{appendix}

\usepackage{amsmath}
\usepackage{amssymb}
\usepackage{amsthm}
\usepackage{cases} 
\usepackage{mathtools} 
\usepackage{bbold}
\usepackage{braket}
\usepackage{enumitem}
\usepackage{bm}
\usepackage{comment}
\usepackage[normalem]{ulem}
\usepackage{thm-restate}

\usepackage{chngcntr}
\usepackage{apptools}
\usepackage{graphicx}
\usepackage{bbm}

\usepackage{tikz}  

\usepackage{footmisc} 

\usepackage{thm-restate} 

\usepackage{mathrsfs} 
\usepackage{calligra}
\usepackage{mathalpha}

\allowdisplaybreaks

\usepackage[hidelinks,breaklinks]{hyperref}
\hypersetup{pdfauthor={M. P. Woods},pdftitle={name}}

\usepackage{thmtools,nameref,cleveref}

\declaretheorem[name=Theorem,
refname={theorem,theorems},
Refname={Theorem,Theorems}]{theorem}
\declaretheorem[name=Proposition,
refname={proposition,propositions},
Refname={Proposition,Propositions}]{prop}
\declaretheorem[name=Lemma,
refname={lemma,lemmas},
Refname={Lemma,Lemmas},
numberlike=prop]{lemma}

\theoremstyle{remark}

\newtheorem*{remark*}{Remark}

\numberwithin{equation}{section}

\AtAppendix{\numberwithin{prop}{section}}
\AtAppendix{\numberwithin{lemma}{section}}
\AtAppendix{\numberwithin{theorem}{section}}
\AtAppendix{\numberwithin{definition}{section}}
\AtAppendix{\numberwithin{corollary}{section}}

\usepackage{color}

\newcommand\mpwS[1]{{\let\helpcmd\sout\parhelp#1\par\relax\relax} }
\long\def\parhelp#1\par#2\relax{%
	\helpcmd{#1}\ifx\relax#2\else\par\parhelp#2\relax\fi%
}

\newcommand{\mpwh}[1]{{\unskip}}

\newcommand{\mref}[1]{\mbox{\cref{#1}}}

\DeclareRobustCommand{\MyInlineMatrix}{%
	\scalebox{0.7}{%
		$\begin{array}{|c|c|c|}
			\hline
			0 & 0 & 0 \\
			\hline
		\end{array}$%
	}%
}

\newcommand{\ketbra}[2]{\ket{#1}\!\!\bra{#2}}
\newcommand{\proj}[1]{\ketbra{#1}{#1}}
\newcommand{\tr}{\textup{tr}}

\newcommand{\nnp}{{\mathbbm{N}_{>0}}}
\newcommand{\nnz}{{\mathbbm{N}_{\geq 0} }}
\newcommand{\rr}{{\mathbbm{R}}}
\newcommand{\rrp}{{\mathbbm{R}_{>0}}}
\newcommand{\rrz}{{\mathbbm{R}_{\geq 0} }}

\newcommand{\hh}{{\mathbbm{H}}}
\newcommand{\cc}{{\mathbbm{C}}}
\newcommand{\zz}{{\mathbbm{Z}}}
\newcommand{\zzp}{{\mathbbm{Z}_{>0}}}
\newcommand{\me}{\mathrm{e}}
\newcommand{\mi}{\mathrm{i}}
\newcommand{\dd}{\mathrm{d}}
\newcommand{\idd}[1]{1}%
\newcommand{\id}{{\mathbbm{1}}}

\newcommand{\thh}{{^\text{th}}} 
\newcommand{\stt}{{^\text{st}}} 
\newcommand{\ndd}{{^\text{nd}}} 
\newcommand{\rdd}{{^\text{rd}}} 

\newcommand{\hash}{\scriptscriptstyle{\#}}

\newcommand{\Sy}{\textup{S}}
\newcommand{\Cl}{\textup{C}}
\newcommand{\Com}{\textup{Com}}
\newcommand{\lo}{\textup{L}}
\newcommand{\Sym}{\textup{[S$\backslash$L]}}
\newcommand{\M}{\textup{M}}
\newcommand{\cont}{\textup{Cntl}}
\newcommand{\W}{\textup{W}}
\newcommand{\Pu}{\textup{P}}
\newcommand{\m}{\mathbb{m}}

\newcommand{\Reg}{\textup{R}}

\usepackage{bm}

\newcommand{\lb}{\bm{(\!\!\!(}}
\newcommand{\rb}{\bm{)\!\!\!)}}
\newcommand{\lsb}{\bm{[\!\![\!\![\!\![}}
\newcommand{\rsb}{\bm{]\!\!]\!\!]\!\!]}}
\newcommand{\brho}{\bm{\rho}}

\newcommand{\U}{U}

\newcommand{\off}{\textup{off}}
\newcommand{\on}{\textup{on}}
\newcommand{\zero}{{\bf 0}}
\newcommand{\SupPolyDecay}{{\textup{SupPolyDecay}_{\displaystyle {\,\bar{\varepsilon}}}\,}}

\newcommand{\app}{\text{appendix}}

\newcommand{\doc}{\text{manuscript}}  

\newcommand{\bo}{\mathcal{O}}

\newcommand{\Biggg}[1]{\scalebox{1.0}{$\left#1\rule{0pt}{0.7cm}\right.$}}
\newcommand{\Bigggg}[1]{\scalebox{1.1}{$\left#1\rule{0pt}{0.9cm}\right.$}}

\newcommand{\nocontentsline}[3]{}
\newcommand{\tocless}[2]{\bgroup\let\addcontentsline=\nocontentsline#1{#2}\egroup}

\begin{document}

\title{Quantum Frequential Computing: {\it a quadratic runtime advantage for all computations}}

\begin{abstract}
An enduring challenge in computer science is reducing the runtime required to solve computational problems. Quantum computing has attracted significant attention due to its potential to deliver asymptotically faster solutions to certain problems compared to the best-known classical algorithms. 
This advantage is enabled by the quantum mechanical nature of the logical degrees of freedom. To date, it was unknown if permitting other parts of the computer to be quantum mechanical, rather than semi-classical, could yield additional runtime speed-ups as a function of resource utilization (e.g., power consumption or cooling requirements).

In this work, we prove that when the control mechanisms associated with gate implementation are optimal quantum mechanical states, a  quadratic runtime speedup (with respect to power consumption or cooling rate) is achievable for any algorithm, relative to optimal classical or semi-classical control schemes. Moreover, we demonstrate that only a small fraction of the computer’s architecture needs to employ optimal quantum control states to realize this advantage, thereby significantly simplifying the design of future systems.

We call this new device a \emph{quantum frequential computer}, since the quantum speedup arises from an increase in gate frequency. In current state-of-the-art designs, gate frequency is often limited by the coupling strength between components. Notably, our approach achieves the speedup without requiring an increase in coupling strength.

\end{abstract}

\author{Mischa P. Woods}
\affiliation{ENS Lyon, Inria, France}
\affiliation{University Grenoble Alpes, Inria, Grenoble, France}
\maketitle

\tableofcontents
\section{Introduction}

\subsection{Background and challenge}


Reducing the runtime required to solve computational problems is the end goal of many different research fields in computer science. For example, better algorithms, improved system architectures, or faster computer clock and gate speeds, can all lead to a reduction in the runtime.

Historically, all three areas of research have led to significant runtime gains.  For example, algorithms for improved matrix multiplication~\cite{Fawzi2022} have led to a speedup in a wide number of computational problems.  Architectural methods such as instruction pipe-lining and out-of-order execution have also led to speedups~\cite{1410064,Gonzlez2018}. 
Likewise, the clock speed of computers has increased dramatically since their inception. In 1938 the first fully mechanical analogue computer, the Z1, operated at one Hz. Later electronic computers further increased this frequency with decade-on-decade speedups until around 2003, when progress stalled in the GHz range. Since then, there has been a significant slowdown due to a breakdown in Dennard's scaling, which predicts how transistor speed increases when chips follow Moore's law~\cite{1050511}. The breakdown was caused by changes in transistor geometry required for the continuation of Moore's-law scaling.


Similarly, improvements in classical algorithms can only take us so far. The development of quantum computing was motivated precisely because it was realised that quantum mechanics allowed for better scaling than the best-known (or optimal) classical algorithm in a select group of problems. Grover's search algorithm~\cite{Grover1996}, and Shor's factoring algorithm~\cite{Shor1997} are some of the most well-known examples. This said, the number of problems with significant algorithmic quantum speed-up is, to date, quite restrictive~\cite{ScottFewAlgs,Hoefler2023}. Additionally, some of the key fields which quantum computing is set to revolutionise|materials research, chemistry and drug design, to name a few~\cite{Lloyd1996,Wang2008,Reiher2017,PRXQuantum.2.017001}|are becoming increasingly superfluous due to the advent of AI (as the AlphaFold program may have already achieved for protein folding~\cite{2023AlfaNat,Bertoline2023,Abramson2024}).

In addition to the hurdle of decreasing runtime, there is also the desire to consume less power. Computation is estimated to consume over 3\%  of global power generation~\cite{Somavat2011,Van2014}
and the world's still-exploding boom in AI is predicted to drive that number up significantly|and fast~\cite{AIBoom}. 

\subsection{Approach taken in this~\doc{} and overview of results}
The above-mentioned runtime speedups originating from using quantum algorithms require a quantum rather than classical logical register. 
\begin{center}
	\begin{figure}[h!]
		\includegraphics[scale=0.38]{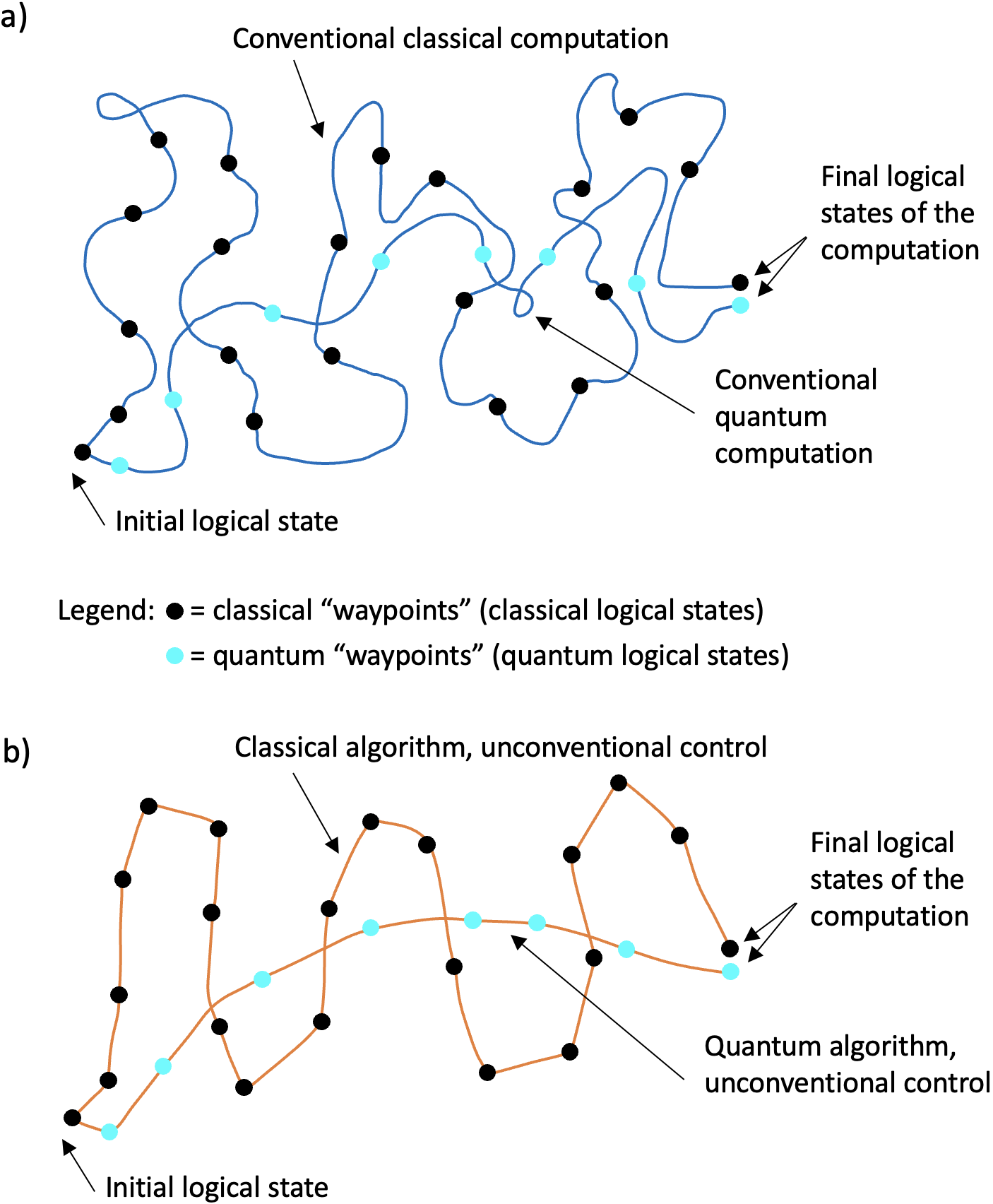}
	\caption{
	{\bf a)} Pictorial illustration of a conventional classical and quantum computer: the physical implementation of an algorithm realises a trajectory through physical state space of the information-bearing degrees of freedom (d.o.f.) and their control d.o.f.: This trajectory (dark blue line) passes through a sequence of ``waypoints'' (logical states) corresponding to the sequence of states dictated by the algorithm when compiled in machine code, i.e. the logical states generated by the logical gate set. Path-length corresponds to the algorithm's runtime. The difference between a classical and a quantum computer is that in the former the waypoints can only be classical logical states while in the latter they can be quantum. It is this ``quantization of waypoints'' along the trajectory which gives quantum computers their runtime advantage (for select problems), and why there is much excitement surrounding them. But this is only half the story: it is not \emph{only} the total number of required waypoints which determines an algorithm's runtime, but also the time required to move the logical state between waypoints. The latter is determined by the nature of the gate control (e.g. ``control pulses''). In quantum and classical computing, the gate control is considered to be a classical or semi-classical state. For a given energy in the control, this restricts the available trajectories of the logical state between two consecutive waypoints.\\ FIG. 1. 	{\bf b)}  The same two sequences of waypoints as fig. 1 a) but now allowing for the control to be a quantum state: for the same energy in the gate control, the permissible trajectories between two consecutive waypoints is much larger and the path-length (now in orange) between consecutive waypoints, i.e. the time required to implement each gate, for the same given energy in the control, can be drastically reduced. This holds true regardless of whether the waypoints are purely classical or  quantum. This constitutes a new type of computer, since the quantum runtime speedup holds for quantum and classical algorithms alike. We call it a \emph{quantum frequential computer}.}\label{fig:motivation}
\end{figure}
\end{center}
Here we explore the possibility of using quantum properties in a different way in order to achieve a runtime speedup. Namely, rather than using quantum properties to reduce the gate count itself, we will aim to reduce the runtime by reducing the time required to apply each gate.

It is well-known that one can increase the gate frequency by increasing both the energy/power required to run the computer and the coupling strength between the control and logical bits/qubits. However, what  had not been explored, is whether quantum effects in the control itself can speed up the application of gates, without an increase in power nor coupling strength. Here we prove that there is such a quantum advantage by using quantum squeezed states for the control of the application of logical gates. Moreover, we show that one can achieve a quadratic increase in frequency as a function of  power consumed without increasing the interaction strength and within a completely autonomous model.  The root cause of this phenomenon, as we will see, is that squeezed states effectively couple more strongly than their semi-classical counterparts. A key feature of this speed-up is that, unlike conventional quantum algorithmic runtime speed-ups, our quantum speed-up applies to all algorithms|quantum and classical ones alike. In~\cref{fig:motivation} we compare the classical and conventional quantum computing paradigms to this new type of quantum computer paradigm.


\subsection{Organisation of this~\doc}

Here is a high-level overview of what is achieved in each section of the main text. Note that~\Cref{sec:Classical and quantum upper limits,sec:gearing} are concerned with Hamiltonian dynamical models while~\Cref{sec:thm3 main text,sec:thm4 main text up} develop dynamical semigroup models. As a general rule throughout the main text, we will start by demanding the most basic feature required for our results and steadily introduce more complexity as it is required. At each step, the added complexity will be motivated by a shortcoming highlighted in the previous result.

\begin{itemize}
	\item \textbf{\Cref{sec:Classical and quantum upper limits}:}
	\begin{itemize}
		\item Derives upper limits for gate frequency in both quantum and semi-classical systems that meet basic computational requirements.
		\item These two upper bounds exhibit a quadratic separation as a function of initial energy.
		\item This motivates the definition of a quantum frequential computer.
	\end{itemize}
	
	\item \textbf{\Cref{sec:quantum advantage}:}
	\begin{itemize}
		\item Specializes the model from~\Cref{sec:Classical and quantum upper limits} to ensure it meets additional key properties required for computation.
		\item Demonstrates that within this specialized model, computers can saturate both bounds from~\Cref{sec:Classical and quantum upper limits}.
		\item Establishes the theoretical existence of optimal classical and optimal quantum frequential computers.
	\end{itemize}
	
	\item \textbf{~\Cref{sec:gearing}:}
	\begin{itemize}
		\item Introduces an additional architectural feature: an internal data bus for the computer with its own dedicated control.
		\item This enhancement permits computation to be carried out over multiple cycles while requiring only semi-classical control for the bus—even in an optimal quantum frequential computer.
	\end{itemize}
	
	\item \textbf{~\Cref{sec:thm3 main text}:}
	\begin{itemize}
		\item Reformulates the quantum frequential computer from~\Cref{sec:gearing} within a dynamical semigroup framework allowing for computation in a nonequilibrium steady-state.
		\item This formulation allows for the notion of power consumption and heat dissipation.
		\item Shows that a quantum frequential computer can run indefinitely with a gate frequency proportional to its power consumption.
	\end{itemize}
	
	\item \textbf{~\Cref{sec:thm4 main text up}:}
	\begin{itemize}
		\item Establishes two upper bounds on gate frequency as a function of power consumption per cycle:
		\begin{enumerate}
			\item In the case of a quantum frequential computer: Gate frequency is upper-bounded by the power consumed per cycle, confirming that the linear scaling found in~\Cref{sec:thm3 main text} is optimal.
			\item In the case of semi-classical control: Gate frequency is upper-bounded by the square root of power consumed per cycle.
		\end{enumerate}
		\item Taken together, these results prove a quantum quadratic advantage for computation runtime versus power consumption.
	\end{itemize}
	
	\item \textbf{\Cref{sec:Discussion,sec:Conclusion}:}
	\begin{itemize}
		\item Conclude with a discussion, relation of the framework to previous frameworks, and final remarks.
	\end{itemize}
\end{itemize}

\section{Classical and quantum upper limits to computation as a function of energy}\label{sec:Classical and quantum upper limits}

\subsection{General Hamiltonian model and theorem} 
In this section, we derive upper bounds the gate frequency as a function of mean initial energy. 
The computer's dynamics implements a total of $N_g\in\nnp$ gates sequentially on $\lo$ during a runtime $T_0\in\rrp$, with the  $j\thh$ gate at time  $t_j:= j\, T_0 / N_g$. For simplicity, we assume that the gate frequency $f=N_g/T_0$ is independent of which gates are being implemented and in the following, study the scenario that the computer is implementing a classical algorithm. As such, the computational logical space $\mathcal{H}_\lo$ passes through a sequence of orthogonal states at times $t_1, t_2$, $\ldots,$ $t_{N_g}$. A reasonable quantum computer should also be able to implement said algorithms, and allowing for the possibility that some gates in the sequence do not map the logical register between orthogonal states, while reasonable, would not change significantly the results derived in this section. (We comment on why this is the case later in~\Cref{sec:quantum speed limits}.)

Our strategy for deriving generic bounds will be to see how accurately one can deduce the elapsed time $t\in[0,T_0]$ by measuring the logical state of the computer mid computation. We will then use results from the field of metrology to obtain  frequency-energy relations.
In our model, the computer's dynamics is governed by a time-independent Hamiltonian $H_\Com$ evolving unitarily from its  initial state $\rho_\Com$ during a time interval $[0,T_0]$, 
\begin{align}
	\rho_\Com(t):=\me^{-\mi (t/T_0) H_\Com T_0} \rho_\Com \me^{\mi (t/T_0) H_\Com T_0},\quad \text{ such that  }\quad \tr[\rho_\Com(t_j)\rho_\Com(t_k)]=0 \text{ if } j\neq k.\label{eq:uni dynamics of com}
\end{align}
It is multipartite and includes the computational logical register on $\mathcal{H}_\lo$. (The details of the other systems are not relevant for our theorem, but it is natural to assume that it also includes all the other subsystems required to run the computer, such as  control systems, memory registers and batteries required). It may even include a thermal bath and hence thermodynamic effects.

For simplicity, and w.l.o.g., we define the ground state energy to be zero (Hamiltonians without a ground state are excluded by definition). The mean energy of the system is thus
\begin{align}
	E:= \tr[H_\Com \rho_\Com ].\label{eq:mean energy def}
\end{align}
Since we are working in so-called natural units (i.e. ``units where $\hbar=1$''), $E$ has units of inverse time.

Our notion of semi-classical states is that of non-squeezed states, and we define the set $ \mathcal{C}_\Com^\textup{class.}$ as the set of all  such states of the computer. (See~\Cref{sec:semi classical states def for Hamiltonian upper bounds} for formal definition.) As one should expect, it is a strict subset of $\mathcal{S}(\mathcal{H}_\Com)$, the set of density matrices on $\mathcal{H}_\Com$.

\begin{restatable}[Upper bounds on computational speed]{theorem}{thmUpperBoundsComp}\label{thm:upperboundsEnergy}
Consider any Hamiltonian $H_\Com$ and initial state $\rho_\Com\in\mathcal{S}(\mathcal{H}_\Com)$ with finite initial mean energy $E$.
The gate frequency $f:= N_g/T_0$, satisfies the following upper bounds:\\[0.2cm]
\noindent Case 1)
\begin{align}
	f \leq C_\textit{SQL} \sqrt{\frac{E}{T_0} }  
	, \label{eq:SQL frequncy}
\end{align}
if the initial state of the computer is semi-classical, namely if $\rho_\Com\in \mathcal{C}_\Com^\textup{class.}$. \\[0.2cm]
\noindent Case 2)
\begin{align}
f \leq C_\textit{HL} E. \label{eq:HL frequncy}
\end{align}
 $C_\textit{SQL}>0$ and $C_\textit{HL}>0$ are dimensionless constants. 
\end{restatable}
The proof can be found in~\cref{app:upper bounds proofs}. It uses tools from quantum metrology, and in the language of this discipline, bounds~\labelcref{eq:SQL frequncy,eq:HL frequncy} can be thought of as a  the ``Standard Quantum Limit'' and the ``Heisenberg Limit''  respectively.  The proof is as follows: We measure the logical space in the computational basis at an unknown elapsed time between starting the computation at $t=0$ and it finishing at  runtime $T_0$. We then choose an estimator function, which estimates the unknown elapsed time (given our measurement outcome). The  Cram{\'e}r-Rao bound and its generalizations~\cite{Giovannetti2012,Maccone2020squeezingmetrology}, then put upper bounds on how precisely we can use this estimate to work out the time. It allows one to generate bounds on the gate frequency as a function of the available mean initial energy.

While~\cref{eq:SQL frequncy} is, to the best of knowledge of the author, completely unknown,~\cref{eq:HL frequncy} follows directly from the so-called quantum speed limit. We discuss the relation of this~\doc{} as a whole to quantum speed limits in~\Cref{sec:quantum speed limits}.

Note that if $\rho_\Com$ is a probabilistic ensemble over semi-classical (non-squeezed) states in $\mathcal{C}_\Com^\textup{class.}$, then bound~\labelcref{eq:SQL frequncy} applies to each state in the ensemble. Moreover, from~\cref{eq:mean energy def} it follows that $E$ is simply the ensemble average of the mean energies of each state in the ensemble. Therefore, as one should intuitively expect, taking ensemble averages of semi-classical states cannot help to increase the frequency beyond that of the highest frequency achievable by the optimal state of the enable.

\subsection{$T_0$ dependency}\label{sec:T0 dependency}

Let us now discuss the $T_0$ dependency. Observe that~\cref{eq:SQL frequncy} has explicit $T_0$ dependency while~\cref{eq:HL frequncy} does not. Recall that, by definition, $H_\Com$ is such that it implements gates sequentially at times $(t_j)_{j=1}^{N_g}$. Therefore, $H_\Com$ (and thus $E$), depend on $T_0$.  However, observe that we can eliminate all explicit $T_0$ dependency from the equation of motion,~\cref{eq:uni dynamics of com}, by writing said equation in terms of dimensionless time $t/T_0\in[0,1]$, and the dimensionless  Hamiltonian, $\bar H_\Com$ defined via $H_\Com=: \bar H_\Com/ T_0$ with $\bar H_\Com$ independent of $T_0$.\footnote{\label{foot:first ref}While one could consider $\bar H_\Com$ being dependent on the runtime $T_0$, such dependencies would be ``artificial'', in the sense that if there exists $ \bar H_\Com(T_0)/T_0$, which is a solution to~\cref{eq:uni dynamics of com} for runtime $T_0>0$, then $H_\Com(x)/T_0$ for any parameter $x>0$ and runtime $T_0>0$, is also a solution to~\cref{eq:uni dynamics of com}. Similarly, we assume the initial state $\rho_\Com$ to be $T_0$-independent because any such dependency would be ``artificial''.}
 It thus follows that $ T_0 E$ is $T_0$-independent. Thus writing the r.h.s. of~\cref{eq:SQL frequncy,eq:HL frequncy} in the form $\sqrt{E/ T_0} =\sqrt{T_0 E}/ T_0$ and $E= T_0 E/ T_0$ respectively, we observe that both  upper bounds have the same inverse dependency on $T_0$.   For any dynamics given by a generator (such as~\cref{eq:uni dynamics of com}), this scaling with $T_0$ is well-known. Physically, it corresponds to re-scaling time, i.e. ``speeding up everything in equal measure''.  In~\Cref{sec:nostrong interactions needed} we discuss why it is impractical to speed-up computation arbitrarily by simply decreasing $T_0$.  In a similar vein, one could work in units where $T_0=1$ analogously to how we are currently working in units where $\hbar=1$ (note the  absence of $\hbar$ from~\cref{eq:uni dynamics of com}). However, to avoid confusion, we refrain from doing so in this~\doc.  As will be made clear, all Hamiltonians in this~\doc{} which are constructed to saturate upper bounds, will be of this form, i.e. their sole $T_0$ dependency will be that of a multiplicative factor of $1/T_0$. 
 
\subsection{Some questions and motivation for~\Cref{sec:quantum advantage}}

The bounds from~\Cref{thm:upperboundsEnergy} naturally lead to a new question: is the linear and square-root scaling of~\cref{eq:SQL frequncy,eq:HL frequncy} actually achievable in the context of a model capable of universal (quantum or classical) computation? 

It is known from~\cite{Margolus1998} that the linear scaling of~\cref{eq:HL frequncy} is achievable, but it is unknown if it is achievable in the context of computation. 
Indeed, the abound~\labelcref{eq:HL frequncy} only demands that the logical space passes through a sequence of orthogonal states. It could, in principle, be saturable only for systems which cannot perform universal computation, such as a quantum system oscillating in a fixed basis.  At a minimum, one needs to be able to apply any sequence of gates to the logical space where said sequence is read from a memory. As regards to~\cref{eq:SQL frequncy}, it is not known if the bound is saturable at all, let alone in the context of universal computation. 

 We will answer both these questions in this~\doc{} going into substantial detail, starting in~\Cref{sec:quantum advantage}. Moreover, in anticipation of a positive answer to both these questions, let us introduce the following nomenclature.

\subsection{Concept of a quantum frequential computer}

We call a computer a \emph{quantum frequential computer} when it uses distinctly quantum effects in the control (such as squeezing) to a achieve a gate-frequency scaling beyond what is possible with merely semi-classical control, namely super-square-root scaling. The name reflects the fact that the high-frequencies would not be achievable without quantum effects in the control. The nomenclature should not be taken to be a mathematical definition, and is intentionally  a little vague. This is in analogy with the nomenclature \emph{quantum computer}, where the exact boundary between classical and quantum computer is also a bit vague, and where a formal definition would not be helpful.

There are two main reasons not to simply refer to it as a quantum computer. On the one hand, the conventional notion of a quantum computer does not require quantum effects in the control. Second, even a computer  solely capable of computing classical algorithms can have a quantum-enabled quadratic speed-up in its gate frequency  if its control is quantum. However, it would be very confusing to refer to such a computer as a quantum computer  even though its full description would require quantum theory.

Extending on the latter point, it is convenient to put quantum frequential computers into two classes: a \emph{type-1 quantum frequential computer} is solely capable of universal classical computation, while a \emph{type-2 quantum frequential computer} is capable of universal quantum computation.

Since we are dealing with Hamiltonian dynamical models, we thus far do not have a definition of power. As will be seen later in a more complete model, power is required for the computation to run forever in a steady-state. In such a scenario, the super-square-root scaling will be with respect to power consumed, not the mean energy in the system. These definitions will suffice in this~\doc{} since we will only discuss the asymptotic scaling behaviour, but more generally one can define a quantum frequential computer as any computer with quantum states whose gate-frequency-to-energy/power-consumption relation is unobtainable via classical or standard-quantum-limited states.

\section{Existence of optimal quantum frequential computers}\label{sec:quantum advantage}

\subsection{Specialised Hamiltonian model with control space, logical space and  gate-sequence memory}

We have seen how quantum metrology imposes upper bounds. We will now prove optimal type-1 and type-2 quantum frequential computers exist. We start by introducing a specialised model capable of universal classical and quantum computation. To achieve this, we will specialise the Hamiltonian generating the dynamics in~\Cref{thm:upperboundsEnergy} to one with more structure. We will model explicitly the relevant 3 subsystems: 1) the physical or logical space on which the gate sequence is implemented; labelled $\lo$. (Since we will not consider error correction here, there is no requirement to distinguish between the logical and physical spaces.) 2) the memory space $\M_0$ which stores the to-be-implemented gate sequence $(\m_1,\m_2,\ldots,\m_{N_g})$, where each element belongs to an alphabet $\mathcal{G}$.\footnote{To avoid ambiguities stemming from notation, we assume $\mathcal{G}$ to not contain purely numeric symbols nor symbols of the form $t_x$ where $x$ is any symbol.} 3) the control space $\Cl$, which will encode the degrees of freedom of the system which implements the gate sequence.

For every element $\m\in\mathcal{G}$, there is a corresponding unitary $U(\m)$ on $\lo$. The set of these $\{ U(\m) \,| \, \m\in\mathcal{G}\}=:\mathcal{U}_\mathcal{G}$ is the gate set and is of finite cardinality but otherwise  arbitrary. As such, the model can accommodate both universal classical and universal quantum computing. (Note that for simplicity, we have referred to both $\mathcal{G}$ and $\mathcal{U}_\mathcal{G}$ as ``gate set''. The former refers to their symbolic representation, while the latter to the maps themselves. It will be clear from the context and notation which one we are referring to.)

We assume the memory, $\M_0$, is formed by $N_g\in\nnp$ local memory cells, i.e. $\mathcal{C}_{\M_0}=\mathcal{C}_\textup{Cell}^{\otimes N_g}$, where each cell is of the dimension of the gate set: $\textup{Dim}(\mathcal{H}_\textup{Cell})=\big|\mathcal{G}\big|$. Its initial state, $\ket{0}_{\M_0}$, belongs to the set representing the memory states corresponding to all possible gate sequences the computer can implement:\footnote{We omit tensor products with the identity and tensor product symbols between kets and bras by convention. Unless stated otherwise, we also refrain from using the common partial trace convention, i.e. $O_{\textup{B}}:\neq \tr_\textup{A} O_{\textup{AB}}$ for  operators $O_{\textup{AB}}$ over a bipartite system.}
\begin{align}
\mathcal{C}_{\M_0}:= \big\{  \ket{\m_1}_{\M_{0,1}}\ket{\m_2}_{\M_{0,2}}\ket{\m_3}_{\M_{0,3}}\ldots \ket{\m_{N_g}}_{\M_{0,N_g}} \,\, \big|\,\,  \m_1,\m_2,\ldots, \m_{N_g}\in\mathcal{G} \big\}  \label{eq:def: classical M0}
\end{align}
where $\{  \ket{\m_l}_{\M_{0,l}}\}_{\m_l\in\mathcal{G}}$ forms an orthonormal basis for the $l\thh$ memory cell. These states are classical in the sense that there is no coherence in the above basis.  See~\cref{fig:them1} for a diagrammatic illustration of the setup.
\begin{center}
	\begin{figure}[h!]
		\includegraphics[scale=0.18]{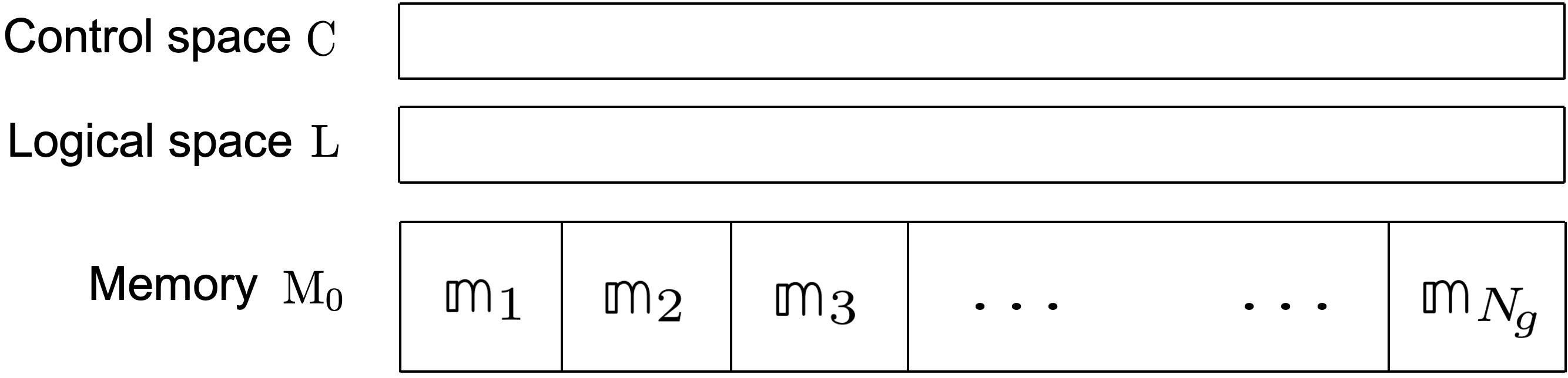}
		\caption{Diagram of the different systems involved: the memory $\M_0$, the logical space $\lo$, and control $\Cl$. The logical space $\lo$,  is further divided into sub-register spaces, as in a conventional quantum or classical computation register (this substructure only enters indirectly via the to-be-defined-later values of $\tilde d({\m})$, $\m\in\mathcal{G}$).}\label{fig:them1}
	\end{figure}
\end{center}
Let $t_j:= j T_0/N_g$, $(j=0,1,2,\ldots,N_g)$ be the time necessary to implement the first $j$ gates on $\lo$, so that $T_0>0$ is the total time required for the computation (i.e. to implement all $N_g$ gates sequentially). Likewise, we denote the state on $\lo$ at time $t_j$ by
\begin{align}
	\ket{t_j}_\lo:= U(\m_j) U(\m_{j-1}) \ldots U(\m_1) \ket{0}_\lo,
\end{align}
where the initial state, $\ket{0}_\lo\in\mathcal{P}(\mathcal{H}_\lo)$, with $\mathcal{P}(\mathcal{H}_\lo)$ the set of normalised pure states on $\mathcal{H}_\lo$.
The corresponding gate frequency is 
\begin{align}
f:=\frac{N_g}{T_0}.   \label{eq:frequncy}
\end{align}
We do not impose any further structure on the control state on $\Cl$ at time $t_j$ other than it being a tensor product state with the rest of the computer. As such, in an ideal scenario, the total state of the computer at time $t_j$ (for $j=0,1,2,\ldots, N_g$) is
\begin{align}
	\ket{0}_{\M_0}\ket{t_j}_\lo \ket{t_j}_\Cl.
\end{align}  
The question we will want  to answer is with what frequency $f$ can the gates be implemented as a function of available energy. To do so we introduce a family of Hamiltonians over $\lo,\M_0$  and $\Cl$  and then ask how well Hamiltonians from this family can mimic the dynamics of the states of the computer we have introduced at times $\{t_j\}_j$. 

These time-independent Hamiltonians have a particular structure, namely
\begin{align}
	H_{{\M_0}\lo\Cl}= H_\Cl +   \sum_{l=1}^{N_g} I_{{\M_0}\lo} ^{(l)}\otimes I_\Cl^{(l)} ,\label{eq:main complete ham}
\end{align}
where we use subscripts ${\M_0}$, $\lo$, $\Cl$ to indicate which subsystems the individual terms act upon. $H_{{\M_0}\lo\Cl}$ is self-adjoint and has a ground state energy of zero.  
While the Hamiltonians $	H_{{\M_0}\lo\Cl}$ 
can depend on $\mathcal{G}$, they cannot depend on the initial memory state $\ket{0}_{\M_0}\in \mathcal{C}_{\M_0}$. This is because we want the gate sequence we wish to implement to be encoded in $\ket{0}_{\M_0}$ and \emph{not} the Hamiltonian.  Since the state of the memory here will not change during the computation, one could in principle have considered a model where the gate sequence were encoded directly into the Hamiltonian. However, in later sections we will consider models where the gate sequence is updated in real time during the computation. In this case, the Hamiltonian would be time dependent. Since we are investigating the fundamental frequencies at which gates can be applied, we need to consider time independent Hamiltonians and as such one should encode the gate sequence information $(\m_1,\m_2,\ldots)$ into states, rather than into the Hamiltonian. We also impose the constraint that $ T_0 	H_{{\M_0}\lo\Cl}$ is $T_0$-independent, for the reasons discussed in~\Cref{sec:T0 dependency}.


The control system $\Cl$ requires a non-trivial free Hamiltonian $H_\Cl$ since the control state is permitted to freely evolve. Conversely, the initial state of the memory, $\ket{0}_{\M_0}$, does not evolve|hence the absence of  a term of the form $H_{\M_0}$. It does however require  interaction terms with the physical space, $\{I_{{\M_0}\lo}^{(l)}\}_l$ in order to be read.  Similarly, there is no free Hamiltonian on the logical space, since the only evolution comes from the application of gates in $\mathcal{U}_\mathcal{G}$. The $l\thh$ interaction term, $I_{\M_0\lo}^{(l)}$, only acts non-trivially on $\lo$ and $\M_{0,l}$. As such, we can associate the presence of $I_{\M_0\lo}^{(l)}$ with the application of $U(\m_l)$ on $\lo$.

The Hamiltonians $H_\Cl$ have a discrete spectrum which forms a basis $\{\ket{E_n}_\Cl\}_n$. The terms $\{I_\Cl^{(l)}\}_{l=1}^{N_g}$ are diagonal in the discrete Fourier transform basis generated from $\{\ket{E_n}_\Cl\}_n$. 

As before, we define the energy for this model as 
\begin{align}
	 E_0:=\tr[\rho^0_{{\M_0}\lo\Cl}H_{{\M_0}\lo\Cl}], \label{eq:def of power one clock only case}
\end{align}
where $\rho^0_{{\M_0}\lo\Cl}$ is the density matrix for the initial state (denoted $\ket{0}_{\M_0}\ket{0}_\lo \ket{0}_\Cl$) and is $T_0$-independent for the reasons discussed in~\Cref{sec:Classical and quantum upper limits} (i.e. so that $T_0 E_0$ is $T_0$-independent).
For the following theorems, it is convenient to introduce the set of classical states of the control $\mathcal{C}_\Cl^\textup{clas.}$: this is the class of non-squeezed states with respect to the free control Hamiltonian $H_\Cl$ and its conjugate Hermitian operator $t_\Cl$ \big(which is diagonal in the same basis as operators $\{I_\Cl^{(l)}\}_{l=1}^{N_g}$\,\big). In other words, the control states in $\mathcal{C}_\Cl^\textup{clas.}$  have  equal uncertainty with respect to both operators (up to normalization and vanishing corrections in the large $E_0$ limit). See~\Cref{sec:alternative def of squeezeing on C} for their full definitions. 

For the following theorem we introduce some notation: For $\m\in\mathcal{G}$, $\tilde d(\m)$  is the number of distinct eigenvalues of $U(\m)$  (which we assume to have point spectrum). 
We use $T(\ket{A},\ket{B})$ to denote the trace distance between two normalised kets $\ket{A}$, $\ket{B}$. As such, in the following theorem the term 	$T\big{(}\me^{-\mi t_j H_{{\M_0}\lo\Cl} } \ket{0}_{\M_0} \ket{0}_\lo \ket{0}_\Cl ,\,\, \ket{0}_{\M_0}\ket{t_j}_\lo \ket{t_j}_\Cl \big{)} $ can be understood colloquially as the error in running the computer up to time $t_j$.

Throughout the main text, we use the notation
\begin{align}\label{eq:def:SupPolyDecay}
	\SupPolyDecay(\cdot): \rrp\to\rrp
\end{align}
to describe a function which decays faster than any polynomial.  Specifically, the two defining properties of the class of functions it belongs to are:   for every real polynomial ${poly}(\cdot)$, there exists $\bar\varepsilon'>0$ such that  for all   $\bar\varepsilon\in(0,\bar\varepsilon')$
\begin{itemize}
\item [1)] \emph{Super polynomial decay}: 
\begin{align}
	\lim_{x\to\infty} poly(x)\, \SupPolyDecay(x) =0.
\end{align}
\item [2)] \emph{Increasing decay with decreasing $\bar\varepsilon$}: 
\begin{align}
	\lim_{x\to\infty} \frac{\SupPolyDecay(x)}{{{\textup{SupPolyDecay}_{\displaystyle {\,\bar{\varepsilon}'}}\,}}(x)}=0.
\end{align}
\end{itemize}
 For fixed input $x$, the functions  $\SupPolyDecay(x)$ may diverge in the limit $\bar\varepsilon\to0^+$; hence the need to define them in terms of the parameter $\bar\varepsilon$. For simplicity of expression, and with a slight abuse of notation, we will use the same notation $\SupPolyDecay(\cdot)$ when referring to different functions in said class.

The functions $\SupPolyDecay(\cdot)$ are independent of the elements in $\{\tilde d(\m)\}_{\m\in\mathcal{G}}$ (which will be defined later).

\subsection{Theorem statement and motivation for~\Cref{sec:gearing}}
 
\begin{restatable}[Optimal conventional and quantum frequential computers exist]{theorem}{thmPowerFrequency}\label{thm:comptuer with fixed memory}
For all gate sets $\mathcal{U}_\mathcal{G}$, initial memory states $\ket{0}_{\M_0}\in\mathcal{C}_{\M_0}$ and initial logical states $\ket{0}_\lo\in\mathcal{P}(\mathcal{H}_\lo)$,  there exists triplets $\{\ket{t_j}_\Cl\}_{j=0}^{N_g}$, $N_g$, $H_{\M_0\lo\Cl}
$ 
parametrised by the energy $E_0>0$ and a dimensionless parameter $\bar\epsilon$, such that for all $j=1,2,\ldots, N_g$ and fixed $\bar\varepsilon>0$ the large-$E_0$ scaling is as follows 
\begin{align}\label{eq:thm fixed memory 1}
T	\Big(\me^{-\mi t_j H_{{\M_0}\lo\Cl} } \ket{0}_{\M_0} \ket{0}_\lo \ket{0}_\Cl ,\, \, \ket{0}_{\M_0}\ket{t_j}_\lo \ket{t_j}_\Cl \Big{)} \leq  \left(\sum_{k=1}^j  \tilde d(\m_k) \right) \SupPolyDecay(E_0),
\end{align}
for the following two cases:\\[0.2cm]
\noindent Case 1)
\begin{align}\label{eq:thm fixed memory 2}
	f= \frac{1}{T_0}\left( {T_0} E_0 \right)^{1/2-\bar\varepsilon}+ \delta f, \qquad |\delta f| \leq \frac{1}{T_0} 
	+\SupPolyDecay(E_0)
\end{align}
and $\big(\,\ket{t_j}_\Cl\,\big)_{j=0}^{N_g}\in\mathcal{C}_\Cl^\textup{clas.}$.\\[0.2cm]
Case 2)\\
\begin{align}\label{eq:thm fixed memory 3}
	f= \frac{1}{T_0}\left( T_0 E_0 \right)^{1-\bar\varepsilon}+ \delta f', \qquad |\delta f'| \leq \frac{1}{T_0} 
	+\SupPolyDecay(E_0)
\end{align}  
\end{restatable}
The proof can be found in~\Cref{sec: proof of 1st quantum clock thorem}. It is by construction. In the proof, the parametrization of triplets $ \{\ket{t_j}_\Cl\}_{j=0}^{N_g}$, $N_g$, $H_{\M_0\lo\Cl}$ is defined explicitly in terms of $\bar\varepsilon$, while its parametrization in terms of $E_0$ is implicitly defined via~\cref{eq:def of power one clock only case}.

Thus, up to an error in trace distance which decays faster than any polynomial,~\Cref{thm:comptuer with fixed memory} shows the following.  Case 1)~\Cref{thm:comptuer with fixed memory}: There exists both conventional quantum and classical computers which saturate the optimal semi-classical control scaling (i.e that saturate the scaling in~\cref{eq:SQL frequncy}). Case 2)~\Cref{thm:comptuer with fixed memory}:  There exists optimal quantum frequential computers, either of type 1 or 2, (i.e that  saturate the scaling in~\cref{eq:HL frequncy}).

 The coefficients $\{\tilde d(\m_k)\}_k$ only depend on the number of qubits in $\lo$ which the gate $U(\m_k)$ acts non-trivially on. E.g. if $U(\m_k)$  is a 2-bit/qubit gate, then $\tilde d(\m_k)\leq 4$. It is important that the error only grows with the number of qubits the gates act upon, rather than the total dimension $d_\lo$ of the computation space $\lo$, which could be orders of magnitude larger. Note that the units in~\cref{eq:thm fixed memory 2,eq:thm fixed memory 3} are those of frequency as required since in this~\doc{} we are using units such that $\hbar=1$.
 

 From~\Cref{thm:comptuer with fixed memory} it also follows how $N_g$ scales with $E_0$.  In case 1) $N_g=T_0 f\sim \sqrt{T_0 E_0}$ as $E_0\to\infty$, while in case 2) $N_g=T_0 f\sim T_0 E_0$ as $E_0\to\infty$. This might be undesirable from a physical standpoint, since one may wish to increase the total number of gates $N_g$ to be implemented at a give frequency $f$, without needing to increase the frequency itself. The reason behind this is because in our construction, every interaction term $I_{\M_0\lo}^{(l)}\otimes I_\Cl^{(l)}$ is only ``used'' once, to implement one gate (since the number of interaction terms and number of gates are equal). Such a setup would be highly wasteful from an engineering perspective. The underlying reason can be traced back to the fact that each interaction term only reads one memory cell.
 
 One way to circumvent this shortcoming, would be to partition the gate sequence $\m_1,\m_2, \m_3, \ldots $ one wishes to implement into sequences of length $N_g$, run the computer with the $1\stt$ partition in the initial memory state $\ket{0}_{\M_0}$ and then reset  the control state to its initial state $\ket{0}_\Cl$ at time $t_{N_g}$ and $\ket{0}_{\M_0}$ to the state which encodes the next partition. Then run the computer again. This however would require external control and thus unforeseeable costs, including, potentially, a drop in frequency and or more energy. 

\section{Quantum frequential computers only require a semi-classical internal bus}\label{sec:gearing}

\subsection{Motivation}
In this section we will show how to extend the setup from the previous section to overcome the shortcoming discussed  in the previous two paragraphs. Moreover, we will show that even when we are operating in the quantum frequential computer regime, we only require additional classical resources to overcome these aforementioned difficulties. In other words, these additional classical resources in total will only require the same initial mean energy  as the quantum ones, thus their addition will not change how the gate frequency of the quantum frequential computer scales with total initial mean energy.  

To achieve this we will add to the setup an additional memory and control system which together can be thought of as an ``internal bus'' which will refresh the memory cells in $\M_0$ at an appropriate rate. We will show that an optimal quantum frequential computer (either of type 1 or 2), only requires an internal bus whose control state is semi-classical, and requires the same amount of initial energy as the quantum control system on $\Cl$. In light of~\Cref{thm:comptuer with fixed memory} and the upper bounds from~\Cref{sec:Classical and quantum upper limits}, this may seem a priori surprising, since this additional control system has to perform the same number $N_g$ of unitaries as the control on $\Cl$ in the same time interval $T_0$. We will give an intuitive explanation of why this is the case later in~\Cref{sec:why  bus of a quantum frequential computer can be classical}. 

\subsection{Specialised Hamiltonian model with control space, logical space, gate-sequence memory, and internal data bus}
Let us now introduce the additional computer architecture  associated with this new setup. As with~\Cref{sec:quantum advantage}, we will define the exact states of the computer at times $t_j$ just after the $j\thh$ gate has been applied,  and later see how fast (as a function of initial mean energy) they can be reached under the dynamics of a time-independent Hamiltonian up to a small error.

We consider the setup in which the  total number of logical gates implemented sequentially is $N_G=L N_g$ within a total runtime of the computer of $(L +1)T_0$, $L\in\nnp$. The gate frequency is $f=N_G/LT_0=N_g/T_0$ (since in the initial time interval $[0,T_0]$ the gates applied are not computational logic gates in $\mathcal{U}_\mathcal{G}$).  We denote by $t_{j,l}$ the time at which the $(j+l N_g)\thh$ gate is applied: $t_{j,l}:=t_{j+l N_g}=t_k+lT_0 $ ($j=0,1,2,\ldots, N_g; \, l=0,1,2,\ldots, L$), $L\in\nnz$. The computer goes through $L+1$ ``cycles'' (called $0\thh$ cycle, $1\stt$ cycle, $2\ndd$ cycle, \ldots) during the computation by which it is meant that the state on $\Cl$ is periodic with a period $T_0$, so that 
\begin{align}
	\ket{t_{j,l}}_\Cl=\ket{t_{j,0}}_\Cl,\label{cond:eq:cyclicity of cl state}
\end{align}
for all $l=0,1,2,\ldots, L$ and $j=0,1,\ldots, N_g$. As such, $T_0$ is called the cycle time.\footnote{In the model of~\Cref{sec:quantum advantage} it was the runtime because the computer only ran for one cycle.}

With the exception of the $0\thh$ cycle, the computer will run analogously to that of the previous section during each cycle, but with the memory $\M_0$ refreshed so that it implements a new gate sequence on every cycle.  The quantum frequential computer will implement a logical gate sequence with elements in $\mathcal{U}_\mathcal{G}$, according to the sequence  $(\m_{1,1}$, $\m_{1,2}$,$\ldots$, $\m_{1,N_g}$, $\m_{2,1}$,$\m_{2,2}$,$\ldots,$ $\m_{2,N_g}$, $\m_{3,1}$, $\ldots$, $\m_{3,N_g}$, $\ldots,$ $\m_{L,1},$ $\m_{L,2}$,$\ldots$, $\m_{L,N_g})$ with elements in $\mathcal{G}$. The logical space will remain unchanged during the $0\thh$ cycle: $\ket{0}_\lo=\ket{t_{k,0}}_\lo$, $k=1,2,3,\ldots,N_g$, and the gate sequence is fully implemented on it  over the subsequent cycles: For $k=1,2,\ldots, N_g$; $l=1,2,\ldots, L$ 
\begin{align}
	\ket{t_{k,l}}_\lo:= U(\m_{l,k}) U(\m_{l,k-1})U(\m_{l,k-2})\ldots U(\m_{2,2}) U(\m_{2,1}) U(\m_{1,N_g})    \ldots U(\m_{1,3}) U(\m_{1,2}) U(\m_{1,1}) \ket{0}_\lo.  \label{eq:gate implementation clock 2}
\end{align}
The system responsible for refreshing the memory is what we call the internal bus.  It has $N_g$ bus lanes, where the $l\thh$ lane is responsible for refreshing the $l\thh$ memory cell on $\M_0$. All bus lanes cannot be turned on at once, but instead need to be turned on in a staggered manner over the $0\thh$ cycle. This requires a switch space $\W$ consisting in $N_g$ switches|one per lane. Each switch state can be in an ``on'' or ``off'' state at time $t_{k,l}$ : 
\begin{align}
	\ket{t_{k,l}}_\W:= \ket{t_{k,l}}_{\W_1} \ket{t_{k,l}}_{\W_2}\ldots \ket{t_{k,l}}_{\W_{N_g}}, \qquad \ket{t_{k,l}}_{\W_j} \in \mathcal{C}_{\W_j}:= \{  \ket{\off}_{\W_j}, \ket{\on}_{\W_j} \},
\end{align}
where $\ket{\off}_{\W_j}$, $\ket{\on}_{\W_j}$  form an orthonormal basis for the Hilbert space $\mathcal{H}_{\W_j}$. The switches are initiated to the all-off state: $\ket{0}_\W:= \ket{\off}_{\W_1} \ket{\off}_{\W_2}\ldots \ket{\off}_{\W_{N_g}}$.
Similarly, the memory cells on $\M_0$ are initialised to a state $\zero\notin\mathcal{G}$,
\begin{align}
	\ket{0}_{\M_0}:= \ket{\zero}_{\M_{0,1}}\ket{\zero}_{\M_{0,2}}\ldots \ket{\zero}_{\M_{0,N_g}},  \label{eq:def:initial logical memory state}
\end{align}
where $	\ket{0}_{\M_0}$ is orthogonal to all states in $\mathcal{C}_{\M_0}$. At later times, some of the memory cells in $\M_0$ will have changed to take on values in $\mathcal{G}$, and as such the elements of  $\big\{\ket{t_{j,l}}_{\M_0}\big\}_{j,l}$ belong to $\mathcal{C}_{\M_0}':={\mathcal{C}_{\textup{Cell}}'}^{\!\!\!\!\!\otimes N_g}$,  $\mathcal{C}_{\textup{Cell}}'=\big\{\ket{\m} \big\}_{\m\in\mathcal{G}\cup\{\zero\}}$. We want to describe a scenario where the state $\ket{\zero}_{\M_{0,k}}$ informs the control on $\Cl$ to turn the $k\thh$ switch from the off state, $\ket{\off}_{\W_k}$, to the on state $\ket{\on}_{\W_k}$. Moreover, we will now assume that the control system on $\Cl$ reads memory cell ${\M_{0,k}}$ and implements the corresponding gate on $\lo$ or $\W$ over a time interval $t\in[t_{k-1,l}, t_{k,l}]$. 
 Thus all switches are turned on by the end of the $0\thh$  cycle and are not turned off again during the subsequent cycles. Explicitly, and defining $\ket{t_{k,0}}_{\W}$ as the state of the switch at time $t_{k,0}$:
\begin{align}
		\ket{t_{k,0}}_{\W}:= \left[  	\ket{\on}_{\W_1}  	\ket{\on}_{\W_2}  \ldots 	\ket{\on}_{\W_{k}} 	\ket{\off}_{\W_{k+1}}   \ket{\off}_{\W_{k+2}}\ket{\off}_{\W_{k+3}}\ldots \right]_{N_g}, \label{eq:switch 1}
\end{align}
for $k=1,2,\ldots, N_g$ and where the notation $[\,\cdot\,]_{N_g}$ indicates that only the first $N_g$ kets in the sequence are kept. In later cycles, the switches stay in their on position. I.e., for $k=0,1,2,\ldots, N_g$; $r=1,2,\ldots, L$
\begin{align}
	\ket{t_{N_g,0}}_{\W}=	\ket{t_{k,r}}_{\W}  := 	\ket{\on}_{\W_1}  	\ket{\on}_{\W_2}  \ldots 	\ket{\on}_{\W_{k}} 	\ket{\on}_{\W_{N_g}}. \label{eq:switch 2}
\end{align}
In the $2\ndd$ cycle, the aim is for the control on $\Cl$ of the quantum frequential computer to implement the first $N_g$ gates in the gate sequence. For this, we need the internal bus, with control space $\Cl_2$, to exchange the state on $\M_0$ in~\cref{eq:def:initial logical memory state} for one containing the gate sequence  for the first $N_g$ logical gates, namely $\m_{1,1}, \m_{1,2},\m_{1,3},\ldots, \m_{1,N_g}$. Similarly, in the $l\thh$ cycle, the control on $\Cl$ should implement the gates corresponding to the sequence $\m_{l,1}, \m_{l,2},\m_{l,3},\ldots, \m_{l,N_g}$.

Each of the $N_g$ internal bus lanes control their own memory block. In particular, the initial state of the $l\thh$ memory block, $\ket{0}_{\M_{\hash,l}}$, is a tensor-product state encoding a partition of the to-be-implemented gate sequence and thus belongs to the set:
\begin{align}
\mathcal{C}_{\M_{\hash,l}}:=  \big\{ \ket{\zero}_{\M_{0,l}}  \ket{\m_{1,l}}_{\M_{1,l}} \ket{\m_{2,l}}_{\M_{2,l}} \ldots \ket{\m_{L,l}}_{\M_{L,l}}\, \big{|} \, \m_{1,l}, \m_{2,l} , \ldots , \m_{L,l}\in \mathcal{G}\, \big\},
\end{align}
$l=1,2,\ldots, N_g$.  Thus collectively, the initial state of the entire memory, $\ket{0}_\M:= 	\ket{0}_{\M_{\hash,1}} 	\ket{0}_{\M_{\hash,2}} 	\ket{0}_{\M_{\hash,3}} \ldots 	\ket{0}_{\M_{\hash,N_g}}\in\mathcal{C}_{\M_{\hash,1}} \otimes \mathcal{C}_{\M_{\hash,2}}\otimes \ldots \otimes \mathcal{C}_{\M_{\hash,N_g}}=: \mathcal{C}_\M$ encodes the entire logical gate sequence which is to be implemented on $\lo$. See~\cref{fig:thm2} for an illustration of the subsystems involved in their initial-state configuration.

While the memory $\M_0$ is initially in a product state over the distinct memory cells, since the control on $\Cl_2$ will be operating at a much lower frequency than that of the logical gate implementation, $f=N_gL/(T_0L)=N_g/T_0$, it need not maintain its  product-state nature at times $\{t_{k,l}\}_{k,l}$, ($k\neq 0$ and $l\neq 0$). Intuitively, one may think of the state at these times as containing some memory elements  $\{\m_{k,l}\}_{k,l}$ which are in the process of being read/written to by the bus. Consequently, while initially the memory and control on $\Cl_2$ will be in product-state form, $\ket{0}_\M\ket{0}_{\Cl_2}$, they will be described by a non-product state a later times: $\{\ket{t_{j,l}}_{\M\Cl_2}\}_{j,l}$. 
Nonetheless, it is important that the memory state being read during the interval $[t_{k-1,l}, t_{k,l}]$ by the control on $\Cl$ is in the appropriate state so that the gate sequence~\cref{eq:gate implementation clock 2} is implemented. Denoting $\ket{t}_{\M_{l,k}}$ the state of the cell $\M_{l,k}$ at time $t$; we thus require,
\begin{align}
	\ket{t}_{\M_{0,k}}=\ket{\m_{l,k}}_{\M_{0,k}},  \label{eq:cell restriction}
\end{align} 
for all 
$t\in [t_{k-1 ,l}, t_{k,l}]$, and $l=1,2,3,\ldots,L$; $k=1,2,\ldots, N_g$.  Note that this condition also necessitates that the state of the memory cell on $\M_{0,k}$ at time $t\in [t_{k-1,l}, t_{k,l}]$ is a product state with the other systems.  Moreover, it allows one to define the bus frequency $f_\textup{bus}$ as the frequency at which each bus lane re-freshens the memory $\M_0$ with new information about the logical gate sequence. From~\cref{eq:cell restriction} we see that this frequency, (i.e. the inverse time between updates of the memory cell on $\M_{0,k}$) is bus-lane independent (i.e. $k$-independent), and given by
\begin{align}
	f_\textup{bus}= \frac{1}{T_0}.
\end{align}

\begin{center}
	\begin{figure}[h!]
		\centering
		\includegraphics[width=0.5\textwidth]
		{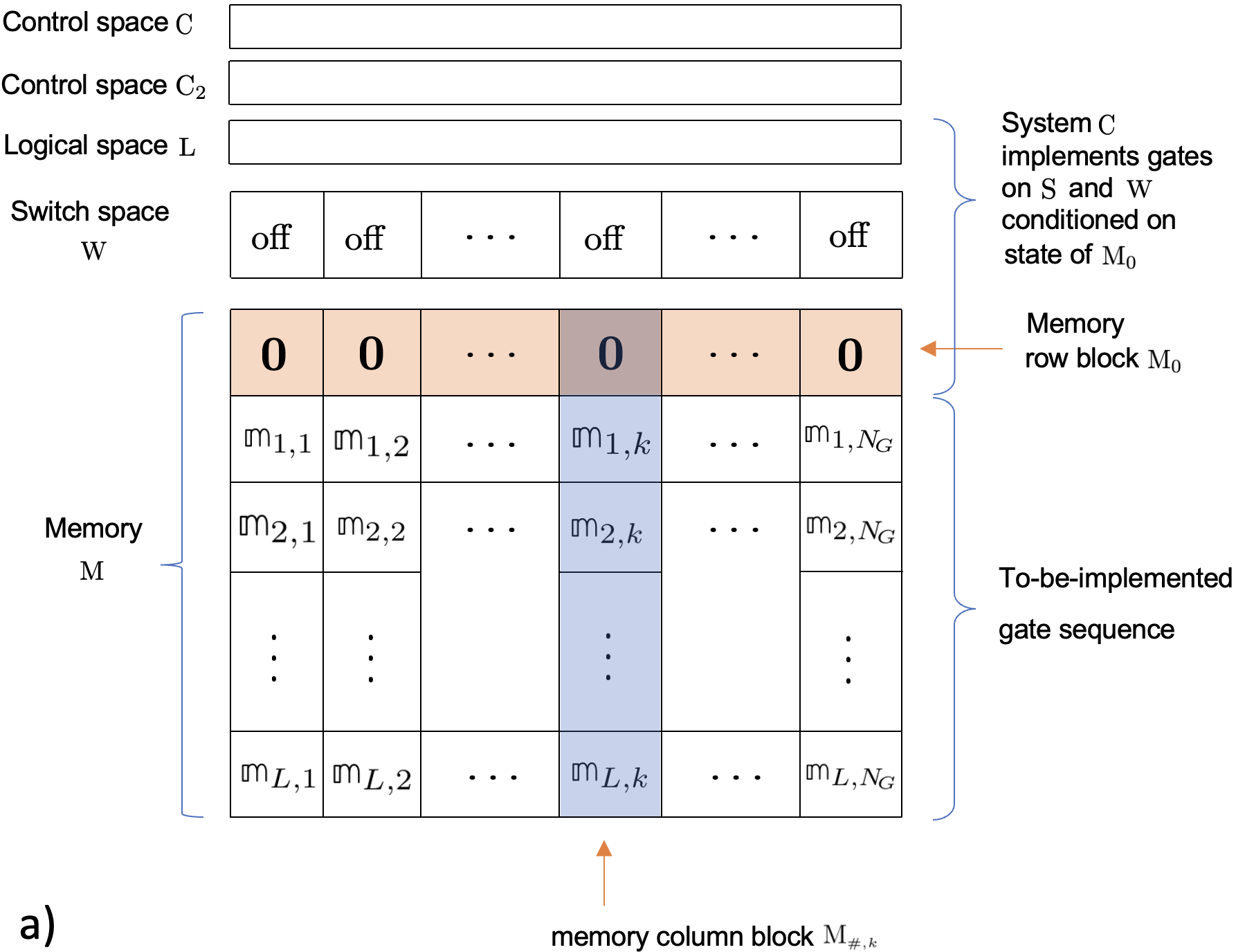}
		\hfill
		\includegraphics[width=0.5\textwidth]{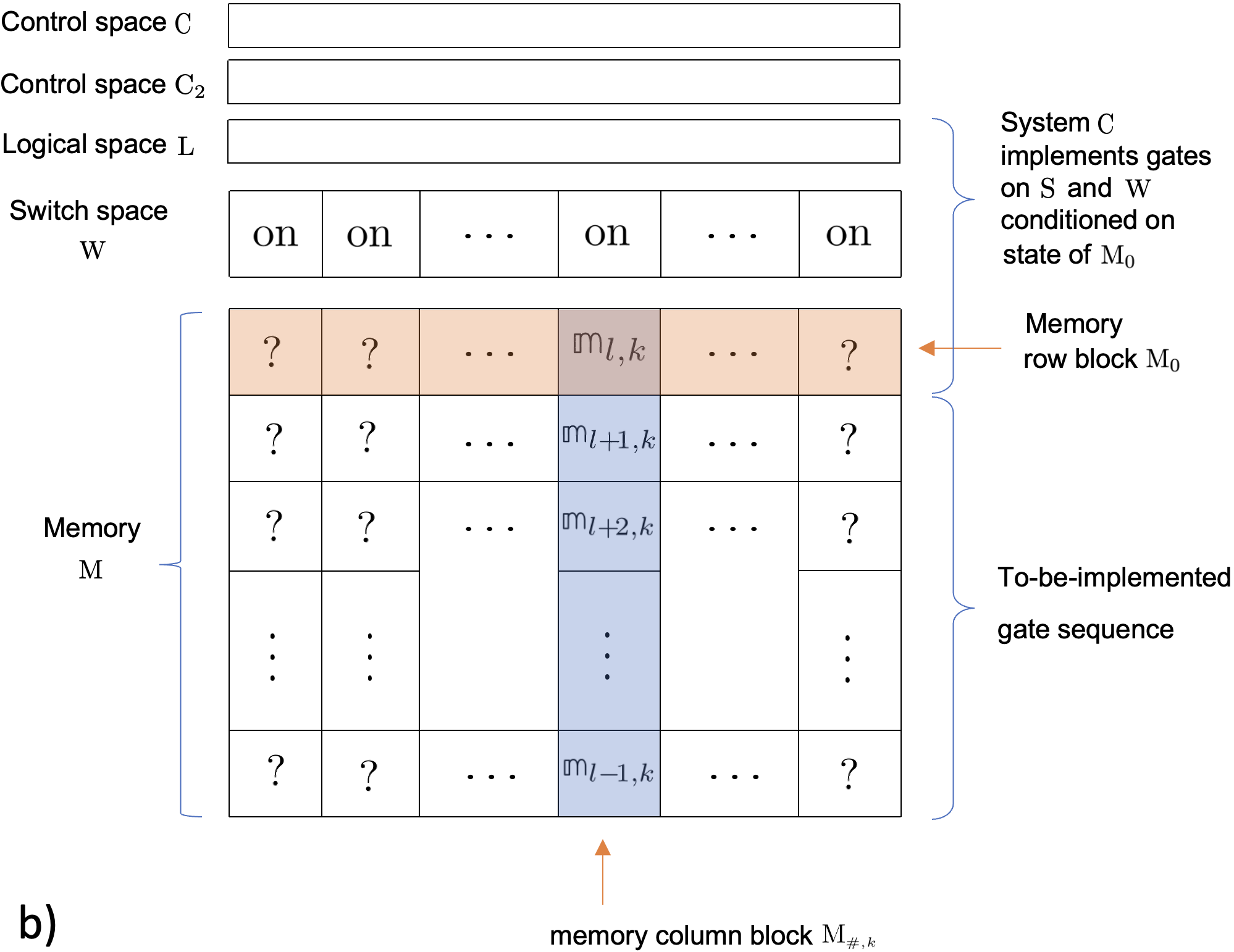}
		\caption{{\bf a)} Schematic of the computer's architecture. The memory $\M$ and switches $\W$ are shown in their initial product states.  
			The cells which the $k\thh$ bus lane can read and write to are highlighted in blue while the cells in $\M_0$ on which logical gates $U(\m)$, $\m\in\mathcal{G}$ are controlled on is depicted in orange. The $l\thh$  bus lane can copy and write to memory cells in memory block $\M_{\hash,l}$ and can be turned on or off via changing the state of the switch in $\mathcal{C}_{\W_l}$ (which is located directly above the bus lane in the figure).	During the $0\thh$ cycle, the $\zero$ instructions in $\M_0$ are read,  informing the control on $\Cl$ to turn the switches from off to on in a staggered fashion|No other change is present in $\M$, and no computation is performed. 
			In the subsequent $l\thh$ cycle, logical gates corresponding to the sequence $\m_{l,1}$, $\m_{l,2}$, $\ldots$, $\m_{l,N_g}$ is implemented hence performing computation.\\
			FIG. 3. {\bf  b)} A snapshot of the computer's memory at time $t\in [t_{k-1,l}, t_{k,l}]$, $l>0$ (subscript $x$ in $\m_{x,k}$ is mod $L$ and $\m_{0,k}:=\zero$.). Memory cell $\M_{0,k}$ stores the appropriate gate sequence information for the gate being applied during time $[t_{k-1,l}, t_{k,l}]$, while other memory cells in other bus lanes (denoted by a question mark) are not necessarily in definite states in $ \mathcal{C}_{\M_{\hash,l}}$ (the information is in transit between memory cells via the bus in preparation for later times).}
		\label{fig:thm2}
	\end{figure}
\end{center} 

As in~\Cref{sec:gearing}, we wish to find a time-independent Hamiltonian whose induced dynamics reproduces the above-described states at times $\{t_{j,l}\}_{l,j}$ up to a small error, and at a particular gate frequency and mean energy. The locality structure of our new Hamiltonian is 
\begin{align}
H_{\M\W\lo\Cl\Cl_2}&:= H_{{\M_0}\W\lo\Cl} + H_{\M\W\Cl_2},\quad H_{\M\W\Cl_2} := H_{\Cl_2} +   \sum_{l=1}^{N_g} I_{\M\W} ^{(l)}\otimes I_{\Cl_2}^{(l)},\quad H_{{\M_0}\W\lo\Cl}:= H_\Cl+ \sum_{l=1}^{N_g} I_{\M_0 \W\lo}^{(l)}\otimes I_{\Cl}^{(l)},\label{eq:main thm2 complete ham} 
\end{align}
where $I_{\M\W}^{(l)}$ only acts non-trivially on memory cells and switch states on $\M_{\hash,l}\W_l$ i.e.|the $l\thh$ bus lane and its corresponding on/off switch. 
The Hamiltonian  $H_{{\M_0}\W\lo\Cl}$ is identical to those described in~\cref{eq:main complete ham} up to the interaction terms $\{I_{\M_0\W\lo}^{(l)}\}$ having additional support on the switch space---this is so that it can implement the turning on/off of the switches in the $1\stt$ cycle. The total Hamiltonian $H_{\M\W\lo\Cl\Cl_2}$ is self-adjoint and has a ground state energy of zero. 
As with Hamiltonians~\cref{eq:main complete ham}, the Hamiltonians in~\cref{eq:main thm2 complete ham} can depend on $\mathcal{G}$ but are independent of the initial memory state in $\mathcal{C}_\M$ and are such that $T_0 H_{\M\W\lo\Cl\Cl_2}$ is $T_0$-independent. Recall sections~\Cref{sec:Classical and quantum upper limits} and~\Cref{sec:quantum advantage} for motivation. 

Analogously to the pair $H_\Cl$, $\{I_\Cl^{(l)}\}_l$, the terms $\{I_{\Cl_2}^{(l)}\}_l$ are diagonal in the discrete Fourier Transform basis generated from the energy-eigenbasis of $H_{\Cl_2}$.

As before, the energy is defined as  the total initial state average energy 
\begin{align}
	E_0':= \tr[\rho^0_{\M\W\lo\Cl\Cl_2}H_{\M\W\lo\Cl\Cl_2} ],  \label{eq:def:P'}
\end{align}
where $\rho^0_{\M\W\lo\Cl\Cl_2}$ is the density matrix for state $\ket{0}_\M\ket{0}_\W\ket{0}_\lo\ket {0}_\Cl\ket{0}_{\Cl_2}$ and is $T_0$-independent by definition for the reasons discussed in~\Cref{sec:Classical and quantum upper limits}  (i.e. so that $T_0 E_0'$ is $T_0$-independent).
Analogously to how we defined the set of semi-classical states for the control on $\Cl$, we need to define the set of classical states for the control of the bus on $\Cl_2$ as the set of non-squeezed states: the set $\mathcal{C}_{\Cl_2}^\textup{clas.}$ is the set of minimum-uncertainty states which share the same standard deviation with respect to $H_{\Cl_2}$ and a canonically conjugate operator $t_{\Cl_2}$ (up to normalization and vanishing  corrections in the large $E_0$ limit). See~\Cref{sec:alternative def of squeezeing on C2} for their full definitions.  

\subsection{Some clarifications}
 Importantly, while we plan to implement a total of $L N_g$ logical gates sequentially, the Hamiltonian only has $2 N_g$ interaction terms. Since $L$ will be large, this means we are ``re-using'' each interaction term many times during the computation which is far more efficient from an engineering perspective.

Note that the terms $H_{{\M_0}\W\lo\Cl}$ and $H_{\M\W\Cl_2}$ do not commute in general since they both act non-trivially on memory cells $\M_0$ and switches $\W$. Physical speaking, this is because $H_{{\M_0}\W\lo\Cl}$ generates the dynamics for performing gates on $\lo$ and $\W$ (controlled on the state of $\M_0$),  while $H_{\M\W\Cl_2}$ generates the dynamics to write to $\M_0$ the necessary memory cells from $\M_{\hash,1},\M_{\hash,2},\ldots,\M_{\hash,L}$; (controlling these operations on the state of $\W$, i.e. controlled on whether the switches are on or off).

One may wonder why we included switch bits and did not simply use a simpler setup where the bus lanes are always on. The reason for including a switch state on $\W_j$ (for turning on/off the control of the memory block on $\M_{\hash,j}$), is because of initial-condition requirements. In particular, in our construction, we need to turn on the memory blocks $\{\M_{\hash,j}\}_j$ sequentially to avoid malfunction. It is expected (we will not prove this here) that if the states on $\Cl_2$ where quantum, that these switch bits would not be required and all bus lanes could always be on. However, they only add a relatively small overhead to the computational architecture and, as we will see, will permit us to use  far fewer quantum resources to obtain the same performance|see~\Cref{sec:why  bus of a quantum frequential computer can be classical} for a longer discussion.


\subsection{Theorem statement and motivation for~\Cref{sec:thm3 main text}}
In the following, we also extend the definition of $\tilde d(\m)$ for $\m\notin\mathcal{G}$:  $\tilde d(\zero)=2$.  What is more, $\SupPolyDecay(E_0')$ is independent of the elements in $\{\tilde{d} (\m)\}_{\m\in\mathcal{G}\cup\{\zero\}}$.

\begin{restatable}[Optimal quantum frequential computers only require a semi-classical internal bus]{theorem}{thmtwoclocks}\label{thm:contrl with two clocks}
For all gate sets $\mathcal{U}_\mathcal{G}$, initial memory states $\ket{0}_{\M}\in\mathcal{C}_{\M}$ and initial logical states $\ket{0}_\lo\in\mathcal{P}(\mathcal{H}_\lo)$,  there exists $\ket{0}_\Cl$, $\ket{0}_{\Cl_2}$, \\$\{\ket{t_{j,l}}_\Cl,\ket{t_{j,l}}_{\M\Cl_2}\}_{j=1,2,\ldots,N_g;\, l=0,1,\ldots, L }$, $N_g$, $H_{\M_0\lo\Cl}$ parametrised by the energy $E_0'>0$ and a dimensionless parameter $\bar\varepsilon$ (where elements $\ket{t_{j,l}}_\Cl$, $\ket{t_{j,l}}_{\M\Cl_2}$ satisfy~\cref{cond:eq:cyclicity of cl state,eq:cell restriction} respectively), such that for all  $j=1,2,3, \ldots, N_g$; $l=0,1,2,\ldots, L$ and fixed $\bar\varepsilon>0$, the large-$E_0'$ scaling is as follows
\begin{align}\label{eq:thrm2 main eq}
	&	T\Big( \me^{-\mi t_{j,l} H_{\M\W\lo\Cl\Cl_2}} \ket{0}_\M\ket{0}_{\Cl_2}\ket{0}_\W\ket{0}_\lo\ket {0}_\Cl ,\,\, \ket{t_{j,l}}_{\M\Cl_2}\ket{t_{j,l}}_\W\ket{t_{j,l}}_\lo \ket{t_{j,l} }_\Cl \Big) \\
	& 	\leq   \left(\sum_{r=0}^l   \sum_{k=1}^j  \tilde d(\m_{r,k}) \right) \SupPolyDecay(E_0')
	, \label{eq:summation thm2}
\end{align}
where 
$\proj{0}_{\Cl_2}$, $\tr_\M\big[\proj{t_{j,l}}_{\M\Cl_2}\big]  \in \mathcal{C}_{\Cl_2}^\textup{clas.}$,   and
\begin{align}\label{eq:thm fixed memory 4}
	f= \frac{1}{T_0}\left( T_0 E_0' \right)^{1-\bar\varepsilon}+ \delta f'', \qquad |\delta f''| \leq \frac{1}{T_0} +\SupPolyDecay(E_0')
\end{align}
\end{restatable}

Thus, up to an error in trace distance which decays faster than any polynomial, we have that an optimal quantum frequential computer (either of type 1 or 2) can run over many cycles while only requiring a semi-classical bus.

The dependency on the cycle time $T_0$, is analogous to that of the previous sections|recall discussion in~\Cref{sec:T0 dependency}. As regards to the dependency on $L$,  since $E_0'$ is $L$-independent,  $L$ can be increased without changing the gate frequency $f$ (this is what one should expect, since $L$ is the total number of cycles the computer runs through). Meanwhile, since $ \left(\sum_{r=0}^l   \sum_{k=1}^j  \tilde d(\m_{r,k})\right)$ scales, at most, linearly in $L$ (since $\mathcal{G}$ is a finite set), the error, characterized by the r.h.s. of~\cref{eq:thrm2 main eq}, only increases linearly with $L$ while it decreases faster than any polynomial in $E_0'$.  As such, the quantum frequential computer can run over many cycles before errors become intolerable. More precisely, by setting $\bar\varepsilon$ small enough, the number of cycles can increase arbitrarily fast as a function of energy: $L\sim{E_0'}^{1/\sqrt{2 \bar\varepsilon}}$, and the r.h.s. of~\cref{eq:thrm2 main eq} still converges to zero for large $E_0'$. However, in said limit, the initial state also tends to infinite energy, and after $L+1$ cycles, would be degraded and would require renewing. This renewal would itself cost resources such as energy, among other things. As such, while this limit is in-principle physical, $E_0'$ in this case should not be considered as capturing the total cost. The principle goal of the next section is to remedy this.

\section{Nonequilibrium steady-state dynamics,  power consumption and heat dissipation}
\label{sec:thm3 main text}
\subsection{Motivation and the concept of self-oscillation}
\Cref{thm:comptuer with fixed memory,thm:contrl with two clocks} illustrate how computation can be formulated in a Hamiltonian dynamics picture with finite energy. However, we have seen that the errors in the computation build up over time, and at some point would cause a malfunction. At least in theory this is a priori not a problem: since in practice all algorithms one runs terminate in finite time, one could simply reset the oscillators to their initial state at the end of the computation. One could then simply quantify the costs associated with the reset process. However, this is not how conventional computers work; indeed, they work in a nonequilibrium steady-state configuration where the computer's clock is a self-oscillator~\cite{Jenkins2013}. The advantages of this is that self-oscillators automatically stabilise themselves leading to the above-mentioned nonequilibrium steady-state configuration where computation can (in principle) run indefinitely.
 
The physics of classical and semi-classical self-oscillators is well studied~\cite{Jenkins2013,RevModPhys.94.045005,Strasberg2021,Wchtler2019,Culhane2022,2307.09122}, but the case of an oscillator in a non-classical state implementing gates, such as in the case of a quantum frequential computer in unknown. An important question is whether a quantum frequential computer can also run in a nonequilibrium steady-state where the oscillators are stabilised.

Since nonequilibrium steady-states can be formulated as open quantum systems, energy can be exchanged between the system and its environment.  In this section, we will show how the need for power consumption arises naturally in this nonequilibrium stead-state setting. We will then show how, in the case of a quantum frequential computer, the linear relationship between gate frequency and energy of the Hamiltonian formulations in previous sections, is replaced with a linear relationship between gate frequency and power consumption.  (The fact that this linear relationship is optimal, and that the quadratic quantum advantage shown in previous sections is maintained, will be proven in~\Cref{sec:thm4 main text up}.)

Since, on cycle average, the mean energy of the internal state of the computer remains constant, the power flowing into the computer must be dissipated back to the environment, so the role of the environment can be thought of as an entropy sink| similarly to how quantum error correction can be understood as a special kind of heat engine removing entropy from the logical space~\cite{Steane2003}. 

\begin{center}
	\begin{figure}[h!]
			\includegraphics[scale=0.3]{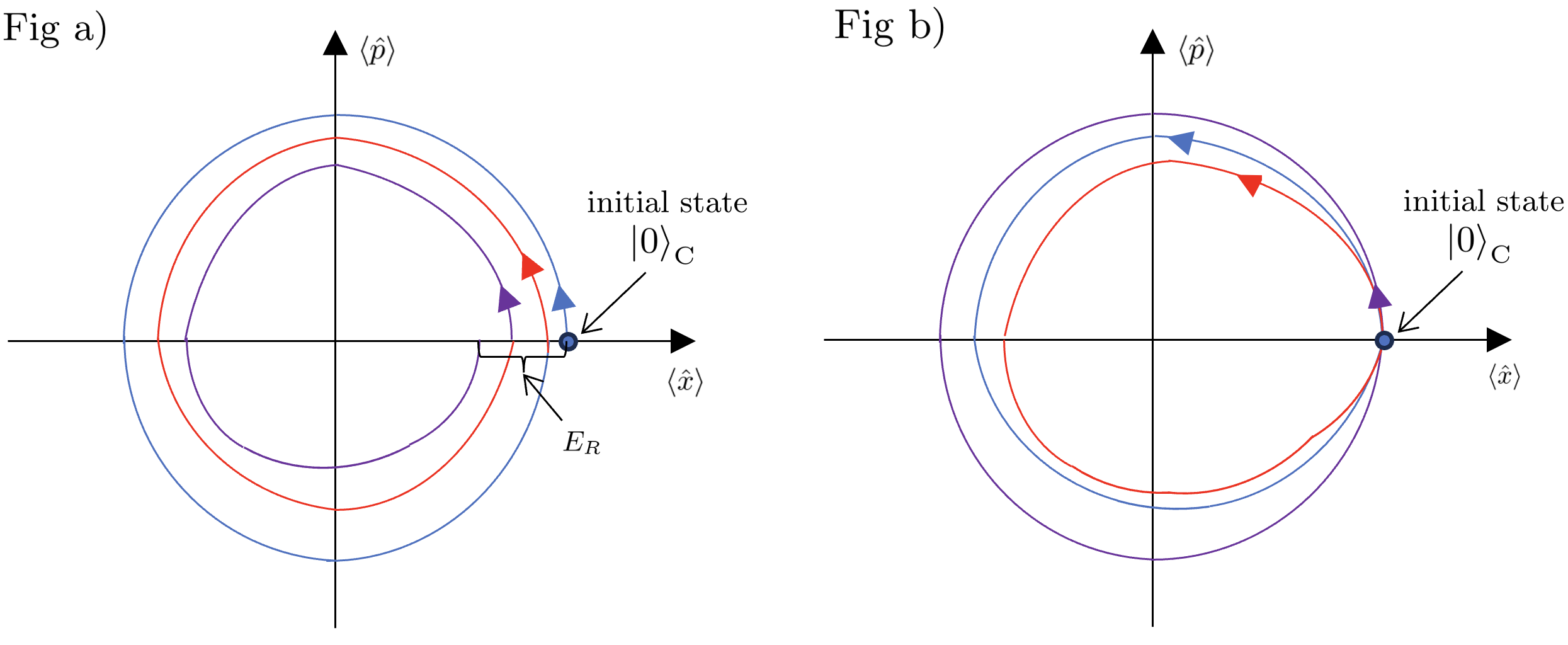}
		\caption{Qualitative illustration of representative dynamics of the oscillator on $\Cl$ in quadrature space over 3 cycles ($1\stt$ cycle in blue, $2\ndd$ in red, $3\rdd$ in purple). Arrows indicate direction of dynamics over time. \\
		FIG. 4. \!{\bf a)} State of the oscillator from~\Cref{sec:quantum advantage}: At the end of each cycle (of duration $T_0$) the oscillator gets close to the state it started in when initiating the cycle. In this way, small errors accumulate over many cycles ultimately leading to an intolerable error $E_R$. While we see from~\Cref{thm:contrl with two clocks} that said errors decrease rapidly with increasing power, for any given power, they eventually become large after sufficiently many cycles.  \\ 
		FIG. 4. {\bf b)} State of the oscillator from~\Cref{sec:thm3 main text}: With high probability, the oscillator is renewed to its initial state at the end of each cycle, leading to a quantum frequential computer in a nonequilibrium steady-state. Observe that the dynamics during each cycle are not identical, due to differing perturbations caused by the implementation of different gate sequences in each cycle. As such, this renewal cannot be unique: it must map many states to the same initial state, correcting for these small errors towards the end of each cycle to prevent them becoming large over many cycles. This many-to-one nature is a form of self-error correction. And, as we will see, it requires work to be done on it to run|similarly to Landauer's erasure principle or quantum error correction.}\label{fig:thm3}  
	\end{figure}
\end{center}

\subsection{Specialised self-oscillating model (with control space, logical space, gate-sequence memory, and internal data bus)}
We now describe the mathematical framework for this section. We use an open quantum system to stabilise the oscillator on $\Cl$ which is driving the computation. We will also have a bus playing the same role as in~\Cref{sec:gearing}. This is mainly because we wish the computation to be performed over many cycles. However, since we have already proved in~\Cref{thm:contrl with two clocks} that the bus can operate at the semi-classical limit, even in the case of an optimal quantum frequential computer, we will not need to model it explicitly this time around. This is because the physics of stable classical oscillators is well understood. In particular, a class of so-called self-oscillators are stable under small perturbations and thus when used for the bus in our protocol, would be stabilised. However, since these stabilization protocols are for classical oscillators, it is not clear if our quantum oscillator on $\Cl$ can also be autonomously stabilized while also implementing logical gates, and if so, at what cost. 

We model the dynamics of the computer and its interaction with the environment via a dynamical semigroup.

 Before we explain each term of said dynamical semigroup, let us pause for a moment to garner intuition from the Hamiltonian model of a quantum frequential computer in~\Cref{sec:gearing} about the properties said dynamical semigroup should posses: in the Hamiltonian model, the state of the oscillator on $\Cl$ at times $t_{0,l}=l T_0$ is close to (but not equal to) its initial state, $\ket{0}_\Cl$. These deviations grow with increasing $l$; see~\cref{fig:thm3}~a) for an illustration. The origin of these deviations can be understood as a consequence of back-action on the control due to the implementation of gates, since when the interaction terms responsible for gate implementation, $\{I_{\M_0 \W\lo}^{(l)}\otimes I_{\Cl}^{(l)}\}_l$, are removed, the resulting dynamics of the control on $\Cl$ is \emph{exactly} periodic, i.e. it is returned to its initial state $\ket{0}_\Cl$ at times $t_{0,l}=l T_0$, $l\in\nnp$.
 
 We therefore want the dynamical semigroup to have the property that it maps the state of the oscillator on $\Cl$ to its initial state periodically at the end of each time interval $[0,T_0]$; see~\cref{fig:thm3}~b). Since these perturbations are small, this only consists in a small correction per cycle orchestrated via the coupling to the environment. (Note that the coupling to the environment, and hence the dynamical semigroup, are necessary due to the irreversible nature of said interaction: Hamiltonian dynamics cannot perform it due to its reservable nature.) Another important subtilty, is that since the irreversible correction to the control is mediated though the interaction with the environment, it will be stochastic in nature. This implies that the time at which the $l\thh$ renewal occurs is a random variable. So long as it's mean is that of the cycle time $T_0$, and its standard deviation very small (in comparison with the time between the last gate of the cycle and the first gate of the $(l+1)\thh$ cycle), this is not a troublesome situation and is in line with how computers function. Another reason to want the renewal time to coincidence to a good approximation with the cycle time $T_0$, is that  the purpose of a self-oscillator is to correct for \emph{small} errors in each oscillation which would otherwise build up over time---not to significantly alter its dynamics in each cycle.
 
  We will prove that said deviations are indeed very small  in the theorem later in this section and discuss the physical significance of this point later in~\Cref{sec:stochasticity of correction}. In keeping with the convention from~\Cref{sec:gearing}, we label the cycles as zeroth, first, second,  third, etc.
 
 To avoid the action of this stabilisation mechanism inadvertently corrupting the application of the last gate in each cycle, we will refrain from applying said gate in each cycle. We will therefore only apply $N_g-1$ gates per cycle and will only require $N_g-1$ bus lanes (memory bloque $M_{\hash,N_g}$ is removed).\footnote{As we will see, this change will not affect the asymptotic behaviour since $N_g-1\sim N_g$ for large $N_g$.} Furthermore, since we will not model the bus explicitly, we can use a Hamiltonian of the form~\cref{eq:main complete ham} but without the interaction for the last gate. Recall that this Hamiltonian is bloque-diagonal in the memory-basis of the register, $\mathcal{C}_{\M_0}$ (see~\cref{eq:def: classical M0}). Therefore, since the dissipative part of the dynamical semigroup will not couple to the register, the memory will not evolve under the dynamics we are modelling explicitly.

The time-independent generator of dynamics is thus of the form
\begin{align}
	\mathcal{L}_{\M_0\lo\W\Cl}(\cdot)=-\mi \big[H'_{{\M_0\lo\W\Cl}}, \,\cdot \,\big] + \mathcal{D}_\Cl(\cdot), \label{generator def}
\end{align}
where $\mathcal{D}_\Cl$ is a dynamical semigroup dissipater on $\Cl$ and $T_0 \,  	\mathcal{L}_{\M_0\lo\W\Cl}(\cdot)$ is $T_0$-independent by definition for analogous reasons to those discussed for the Hamiltonian case  in~\Cref{sec:Classical and quantum upper limits}. For our purposes, it can be further decomposed as $\mathcal{D}_\Cl= \mathcal{D}_\Cl^\textup{re}+ \mathcal{D}_\Cl^\textup{no re}$, where $\mathcal{D}_\Cl^\textup{re}$ is a norm non-preserving and completely positive  map called the renewal operator. It generates the renewal process: it maps all input states to one unique output state $\proj{0}_\Cl$ (up to a normalization which depends on the input and determines the probability of said process occurring). This irreversible (many-to-one) aspect is crucial for stability since the state of the control towards the end of each cycle depends on the gate sequence implemented in said cycle. And since this differs for each cycle, so does the control state. As such, it is important that all said states are dynamically and autonomously mapped back to the same state to complete the cycle of the oscillator exactly so that it is stable under the perturbations caused by gate implementations. The restoring channel  is also called a ``restoring force'' in the literature on classical self-oscillators. The probability with which this renewal process occurs is however input-state dependent|this is also crucial, since it is important that the stabilisation events occur with overwhelming probability at the end of each cycle, and not at some other time earlier on in the cycle. This is the mechanism via which we will accomplish this.  The other term $\mathcal{D}_\Cl^\textup{no re}$ is required to ensure that $\mathcal{D}_\Cl$ is a valid dissipater (i.e.~\cref{generator def} is a Lindbladian).

The Hamiltonian part is 
\begin{align}
	H'_{{\M_0}\lo\W\Cl}= H_\Cl +   \sum_{l=1}^{N_g-1} I_{{\M_0}\lo\W} ^{(l)}\otimes I_\Cl^{(l)} ,\label{eq:main complete ham 2}
\end{align}
and bloque-diagonal in the $\mathcal{C}_{\M_0}$ basis.

The initial state of the computer we will consider in the theorem in this section, denoted $\rho_{\M_0\lo\W\Cl}(0)$, is any pure state of the form $\ket{0}_\lo\ket{0}_\W\ket{0}_{\M_0}\ket{0}_\Cl$\footnote{Although see~\Cref{Pure vs. mixed states} for generalisations to mixed states.} and independent of $T_0$, as in the previous sections. Similarly to the Hamiltonians in the previous sections, we demand that the generator of dynamics~\cref{generator def} is independent of the gate sequence $(\m_{l,k})_{l,k}$ encoded into the initial state $\ket{0}_\M$. The reason for imposing this constraint is the same as that explained in~\Cref{sec:quantum advantage}.

\subsection{Partial unravelling of the dynamical semigroup}\label{sec:PartialUnravelling}
Let us denote by $\rho_{\M_0\lo\W\Cl}(\tau)$ the state of the computer after evolving it for a time $\tau>0$ according to the generator of dynamics,~\cref{eq:main complete ham 2}. This ensemble at time $\tau$ comprises of two distinct energy flows: energy flowing into the computer and heat dissipated out to the environment. It is common to perform a partial unravelling of the dynamics in order to appropriately distinguish and measure these two important quantities separately (see e.g.~\cite{Jacobs2010,Landi2024}). With this objective in mind, we start by observing that this state can be decomposed into a mixed state, where the states of the mixture form a semi-infinite sequence. The sequence can be indexed by $l\in\nnz$, where $l$ is the number of times the renewal operator has acted on it during the time interval $[0,\tau]$, i.e. by how many times it has been ``renewed''. What is more, a classical signal is generated every time the renewal channel is autonomously applied. If said classical information is observed continuously on-the-fly,  one would observe no change followed by a renewal event occurring, followed by no change, and then a another renewal event occurring, etc. (See~\cref{sec:classical bit for renewal prcess} for details.) In order to understand the computational model, it is convenient to assume this classical information is observed. In the language of dynamical semigroups, this corresponds to a partial unravelling. It is a mathematically equivalent formulation of the dynamics (since one can reconstruct $\rho_{\M_0\lo\W\Cl}(\tau)$ by computing the probabilities of the classical events occurring at different times). 

To this end, let us denote by $\rho_{\M_0\lo\W\Cl}\lb t+\tau_l| \tau_l\rb$ the state at time $t+\tau_l>\tau_l$, which results from the following conditioning:   no renewal event occurring in the interval $[t+\tau_l, \tau_l)$, the $l\thh$ renewal event occurring at time $\tau_l$, the $(l-1)\thh$ renewal occurring at time $\tau_{l-1}<\tau_l$,  etc, and  the $1\stt$ renewal occurring at time $\tau_1>0$. The state $\rho_{\M_0\lo\W\Cl}\lb t+\tau_l| \tau_l\rb$ is unnormalised, its trace corresponds to the conditional probability of said conditioned sequence of events occurring.

Meanwhile, we denote by $P(t+\tau_l,+1|\tau_l)$ the probability associated with this state being renewed one more time at time $t+\tau_l$. (We use the convention $\tau_0:=0$, such that $P(t,+1|\tau_0)$ is the probability associated with the $1\stt$ renewal.)

\subsection{Idealised dynamics we wish to approximately implement}
In the case of~\Cref{thm:comptuer with fixed memory,thm:contrl with two clocks}, we defined idealised dynamics which we wished to approximate via Hamiltonian dynamics. Now we make analogous definitions but for the dynamical semigroup case.

Since we are not modelling the bus explicitly, the only requirement is that it can perform its function of updating the memory on $\M_0$ analogously to how it did in~\Cref{sec:gearing}. We can write this condition in terms of $t$ as
\begin{align}
	\ket{t}_{\M_{0,k}}=\ket{\m_{l,k}}_{\M_{0,k}},  \label{eq:cell restriction thm3}
\end{align} 
for all 
 $t\in [t_{k-1}, t_{k})$,  $k=1,2,\ldots, N_g-1$. Since we will be operating in a steady-state,~\cref{eq:cell restriction thm3} must hold for all $l\in\nnz$ which implies we are either assuming that the memory on $\M$ is unbounded or that the gate sequence is periodic (this is a merely mathematically convenient assumption for the obvious reasons).
 
 As in~\Cref{sec:quantum advantage,sec:gearing}, it is useful to introduce states which correspond to the state of the quantum frequential computer during its computational runtime under the hypothetical assumption that no errors occurred (i.e., among other things, all gates are implemented with zero error). In our current setup, this corresponds to states at times $\{t_j+\tau_l\}_{j,l}$ (given that the $l\thh$ renewal occurred at time $\tau_l$), where the dynamics of the $l\thh$ cycle has not introduced any errors. The utility of introducing such states is that we can see how close the actual dynamics is to said states, thus quantifying errors. We denote these states by $\{  \rho_{\M_0 \lo\W\Pu}{[t_j+\tau_l|\tau_l]} \otimes \rho_\Cl{[ t_j+\tau_l| \tau_l ]} \}_{j=0}^{N_g}$ with purifications $\{  \ket{[t_j+\tau_l|\tau_l]}_{\M_0 \lo\W\Pu} \ket{[ t_j+\tau_l| \tau_l ]}_\Cl \}_{j=0}^{N_g}$. Here $\Pu$ is a fictitious purifying system to allow us to work with pure states for simplicity of notation. It is useful for these states to only capture the idealised dynamics \emph{between} each renewal. As such, they are defined to be equal to the actual dynamics just after each renewal event occurs:
 \begin{align}
 	 \ket{[\tau_l|\tau_l]}_{\M_0 \lo\W\Pu} \ket{[\tau_l| \tau_l ]}_\Cl := 	\ket{\rho(\tau_l | \tau_l)}_{\M_0\lo\W\Cl\Pu}=	\ket{\rho(\tau_l | \tau_l)}_{\M_0\lo\W\Pu} \ket{0}_\Cl,
 \end{align}
 where $\tr_\Pu[\proj{\rho(\tau_l | \tau_l)}_{\M_0\lo\W\Cl\Pu}]= \rho_{\M_0\lo\W\Cl}{(\tau_l|\tau_l)}$ and the last equality is due to the fact that the renewal process maps the state on $\Cl$ to $\ket{0}_\Cl$ at time $\tau_l$. 
The idealised states $\ket{[t_k+\tau_l|\tau_l]}_{\M_0 \lo\W\Pu}$ update via the exact application of the logical gates $\{ U(\m_{l,j}) \}_{j=1}^k$ on $\lo$ or $\W$ and a set of local unitaries on the memory $\{U_{\M_0\Pu}^{(l,j)}\}_{j=1}^k$ to guarantee the fulfilment of~\cref{eq:cell restriction thm3}:
\begin{align}
	\ket{[t_k+\tau_l|\tau_l]}_{\M_0 \lo\W\Pu}= U_{\M_0\Pu}^{(l,k)}  \ldots U_{\M_0\Pu}^{(l,2)} U_{\M_0\Pu}^{(l,1)}    U(\m_{l,k})\ldots U(\m_{l,2})U(\m_{l,1}) 	\ket{\rho(\tau_l | \tau_l)}_{\M_0\lo\W\Pu}. \label{eq:idealised states condition thm3}
\end{align}
 Since we assume this for all $l\in\nnz$, this assumption is implicitly assuming that the oscillator on $\Cl_2$ has been stabilized, otherwise small errors would add up over many cycles leading to the impossibility to implement~\cref{eq:cell restriction thm3} to high precision. The exact nature of the state of $\M_0\backslash \M_{0,k}$ over the time interval $t\in [t_{k-1}, t_{k}]$ is not so relevant since the dynamics will have close to zero support on it. For concreteness, we will assume zero knowledge of this state i.e. that it is in a maximally mixed state.  For the idealised states of the control, we require  $\{\ket{[ t_j+\tau_l| \tau_l ]}_\Cl:=  \ket{ t_j}_\Cl \}_{j,l}$, where $\{\ket{ t_j}_\Cl\}_j$ can be any sequence of states. I.e. that the dynamics after each renewal is independent of which cycle $l$ the oscillator is passing through.

\subsection{Self-oscillator power consumption}
As previously discussed, our model for computation where the control is a self-oscillator, requires the autonomous implementation of an irreversible process. Said irreversible process requires work to implement it. Here we will characterise the work done on the computer in order to achieve its self-oscillating behaviour.

In order to best understand the power consumption of the computer, it is convenient to examine the actual events which unfold during each run of the computer. (Recall that the dynamics under the generator of dynamics,~\cref{generator def}, only captures the dynamics averaged out over many runs, and by taking into account classically available information, the actual observed dynamics is given by a so-called partial unravelling of the dynamics as discussed in~\Cref{sec:PartialUnravelling}.) To this end,  it is convenient to analyse the dynamics between no renewal and renewal events.

The dynamics immediately after a renewal event is, to a good approximation, Hamiltonian dynamics under $H'_{\M_0\lo\W\Cl}$ since initially the dissipater term plays, to a good approximation, no role (we quantify this later in~\Cref{eq:calculation of isentropic regime and epsilon H for the pecialised theorem}). We call this reversible and entropy-conserving time interval the \emph{isentropic time interval}. It is only towards the end of the cycle where the dissipative coupling to the environment takes effect, that the dynamics enters an irreversible time interval. The irreversibility of dynamics in this later time interval is due to the renewal operator,  $\mathcal{D}^\textup{re}_\Cl(\cdot)$, inducing irreversible dynamics on $\Cl$, and results in non-isentropic dynamics on $\M_0\lo\W\Cl$. As beautifully illustrated by Landauer in the context of irreversible computation, this irreversibility has an work cost associated with its implementation. Since in our case, the implementation is performed autonomously, and the computational system $\M_0\lo\W$ itself does not loose nor gain energy during the process, the energy associated with this work must come from the environment. This energy flowing into the computer from the environment is by definition a power source, since power is exactly that|a low entropy source of energy which flows into the computer from the environment, allowing work to be performed, and the computer to run in a non-equilibrium steady-state configuration. See~\cref{fig:power} for an illustration and~\Cref{sec:Power consumption derivation} for technical details.

\begin{center}
	\begin{figure}[h!]
		\includegraphics[scale=0.34]{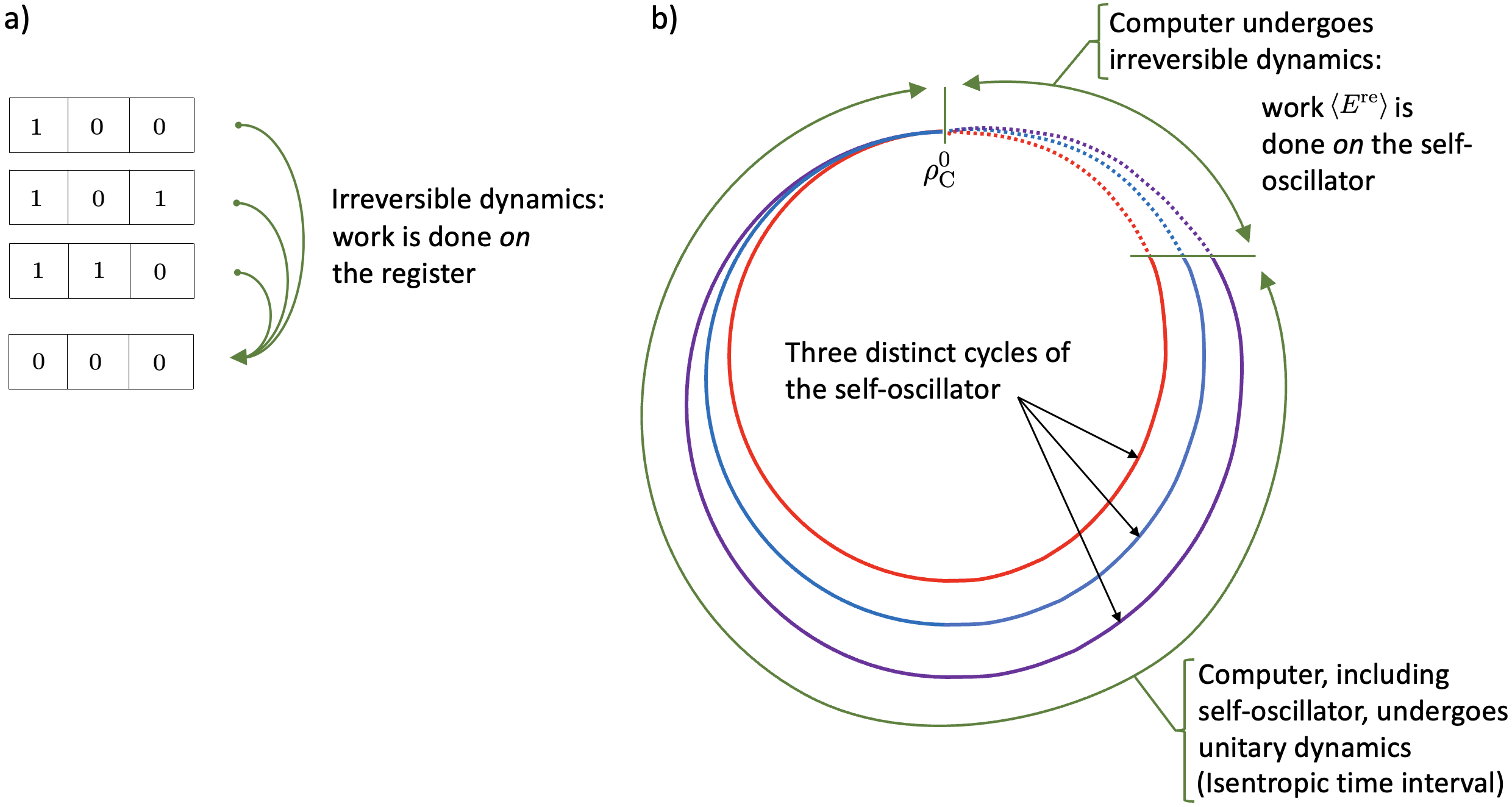}
		\caption{Comparison between Landauer erasure and self-oscillation: two distinct physical processes with the commonality of requiring work to be done \emph{on} the system by the environment to function due to the irreversibility of the process (the logical register in Landauer erasure, the self-oscillator in the computer). 
		\\ FIG. 5. {\bf a)} The Landauer erasure process: an irreversible (many-to-one) channel is applied to three different register states to map them all to a unique register state \MyInlineMatrix, thus erasing any information they contained.  
		\\FIG. 5. \!{\bf b)} the state of the self-oscillator (state on $\Cl$) during one cycle: at the end of the cycle's isentropic time interval, the state of the oscillator is a mixture over unknown perturbations to its orbit. At this stage, the self-oscillator enters an irreversible (non-isentropic) time interval, where work is done on the oscillator by the environment in order to restore it to its lower-entropy initial state, $\rho_\Cl^0$. This work $\braket{E^\textup{re}}$ is the power consumed by the oscillator per cycle.}\label{fig:power}  
	\end{figure}
\end{center}

 We can use standard theory~\cite{Rivas2012,Weiss2008} to quantify this power consumed per cycle, $P^\textup{in}$,  due to the cost of implementing this irreversible process. One finds on cycle average
\begin{align}\label{eq:def: enery re}
P^\textup{in} := \frac{\braket{E^\text{re}}}{T_0}	:= \frac{1}{T_0}  \int_{0}^{\infty} \dd s \, \tr[H_{\M_0\lo\W\Cl}'\,	\mathcal{D}_{\Cl}^\text{re}(\rho_{\M_0\lo\W\Cl}\lb s+\tau_{l}|\tau_{l} \rb)],
\end{align}
where the integrand can be shown to be non-negative for all $s\in[0,\infty)$, implying that it corresponds to energy flowing into the computer from the environment. Note that  the power consumption per cycle, $P^\textup{in}$, depends on the cycle itself (i.e.  r.h.s. of~\cref{eq:def: enery re} is $l$-dependent). Since the gates one can implement in any one cycle are the same, one may imagine that the power consumed per cycle is independent of this. In the proof of~\Cref{thm:heat dissipationapp}, we show that it is indeed $l$-independent to an exceedingly good approximation, hence the omission of the $l$-dependency on the l.h.s. of~\cref{eq:def: enery re}.
 
On cycle average the state of the computer on $\M_0\lo\W\Cl$ does not gain nor lose energy, hence the power $P^\textup{in}$ needs to be dissipated back out to the environment. During each cycle, the remaining energy flows come from the time interval in which the renewal channel is not acting. On cycle overage, we find this to be
\begin{align}
P^\textup{diss}	:=  \frac{\braket{E^\text{diss}}}{T_0}:=- \frac{1}{T_0}  \int_{0}^{\infty} \dd t  \, P( {t+\tau_{l}}  ,+1|\tau_{l})  \int_0^{t} \dd s\,   \tr[H_{\M_0\lo\W\Cl}'\,	\mathcal{D}_{\Cl}^\text{no re}(\rho_{\M_0\lo\W\Cl}\lb s+\tau_{l}|\tau_{l}\rb )]. \label{eq:S cycle new} 
\end{align}
The minus sign is by convention, so that $\braket{E^\text{diss}}$ is the energy flowing \emph{out} of the computer between renewals,  the dissipated power per cycle, $P^\textup{diss}$, is positive. We show in~\Cref{lem:upper bound on P diss} after the proof of~\Cref{thm:heat dissipation}, that (up to a vanishingly small additive contribution)
\begin{align}
	P^\textup{diss}	 \leq P^\textup{in},
\end{align}
namely the amount of dissipated energy per cycle cannot be more than the power consumed per cycle. This is as expected, because otherwise the internal state of the computer would slowly lose power and eventually stop functioning.

Note that the flows of energy out of  and into the environment in~\cref{eq:def: enery re,eq:S cycle new} are not inconsistent with the existence of the previously-mentioned isentropic intervals: in said intervals the integrands are well approximated by zero, and only contribute significantly to the integrals in the irreversible intervals.

\subsection{Theorem statement and explanation}
We are now ready to state our existence theorem for optimal nonequilibrium steady-state quantum frequential computers. In the below theorem, since we have removed  the $N_g{}\thh$ bus lane, the memory states $\{\m_{l,N_g}\}_l$ do not exist, and we make the association $\tilde d(\m_{l,N_g})=1$.  

\begin{restatable}[Nonequilibrium steady-state quantum frequential computers exist]{theorem}{thmStadyState}\label{thm:heat dissipation}
For all gate sets $\mathcal{U}_\mathcal{G}$, initial gate sequences $(\m_{l,k})_{l,k}$ with elements in $\mathcal{G}$, and initial logical states $\ket{0}_\lo\in\mathcal{P}(\mathcal{H}_\lo)$,  there exists $\ket{0}_\Cl$, $\{\ket{t_{j}+\tau_l|\tau_l}_\Cl \}_{j=1,2,\ldots,N_g; l\in\nnz }$, $N_g$, $\mathcal{L}_{\M_0\lo\W\Cl}$ parametrised by the power $P^\textup{in}>0$ and a dimensionless parameter $\bar\varepsilon>0$
, such that for all  $j=1,2,3, \ldots, N_g$; $l\in\nnz$ and fixed $\bar\varepsilon>0$, the following large-$P^\textup{in}$ scaling hold simultaneously
	\item [1)]    Given that $l\in\nnz$ renewals occurred in the time interval $[0,\tau_l]$, the probability that the next renewal  (namely $\tau_{l+1}$), occurs in the interval $[\tau_l+T_0-t_1, \tau_l+T_0]$ is:
	\begin{align}
		\int_{T_0-t_1}^{T_0} \dd t\,  {P}(t+\tau_l,+1|\tau_l) = 1 - \varepsilon_r,\qquad 0<\varepsilon_r \leq  \left(\sum_{k=1}^{N_g}  \tilde d(\m_{l,k})\right) \SupPolyDecay(P^\textup{in})
		, \label{eq:thm3 prob dist}
	\end{align}
	\item [2)]  The deviations in the state between renewals are small:  For $j=1,2,\ldots,N_g$, 
	\begin{align}
		T\Big( \rho_{\M_0\lo\W\Cl}(t_{j}+\tau_l|\tau_l) ,\,\, \rho_{\M_0 \lo\W\Pu}{[t_j+\tau_l|\tau_l]} \otimes \rho_\Cl{[ t_j+\tau_l| \tau_l ]} \Big)\leq    \left(\sum_{k=1}^{j}  \tilde d(\m_{l,k})\right) \SupPolyDecay(P^\textup{in}) 
		,   \label{eq:dist condition thmor}
	\end{align}
	\item [3)]    The gate frequency has the asymptomatically optimal scaling in terms of power:
	\begin{align}\label{eq:thm fixed memory 5}
		f= \frac{1}{T_0}\left( {T_0^2}\, {P^\textup{in}} \right)^{1-\bar\varepsilon}+ \delta f', \qquad |\delta f'| \leq \frac{1}{T_0} +\SupPolyDecay(P^\textup{in}) 
	\end{align} 
\end{restatable}

Item 1) demonstrates that, up to a vanishing error, the renewal is occurring exactly when we want it to: in the interval just before the cycle ends. Moreover, recall that $t_1=T_0/N_g\sim 1/(T_0 P^\textup{in})$ so the length of interval of integration $[\tau_l+T_0-t_1, \tau_l+T_0]$ approaches zero at a rate inversely proportional to the power. Thus item 1) implies that $\tau_l$ converges to $l T_0$ so that cycle time and renewal times coincide|in line with previous discussions.

Item 2) shows that the actual dynamics is close in trace distance to the desired dynamics, and that said errors are not accumulative. 
Moreover, the only dependency on $l$ in the r.h.s. of the inequalities in Items 1) and 2) is through the summation over $\tilde{d}(\m_{l,k})$. However, since $\mathcal{G}$ is a finite set, the sum $\sum_{k=1}^{N_g} \tilde{d} (\m_{l,k}) $ is upper bounded by a function that is linear in $T_0^2 P^\text{in}$ and independent of $l$ and $\bar\varepsilon$. This $l$-independence of the r.h.s. of the inequalities is important because it means that the errors are independent of how many cycles the quantum frequential computer has been through. This is in contrast  to the summation $ \left(\sum_{r=0}^l   \sum_{k=1}^j  \tilde d(\m_{r,k}) \right)$ from~\cref{eq:summation thm2} in~\Cref{thm:contrl with two clocks}, which grows  approximately linearly in $l$, the total number of cycles the computer has been through.

This said, while errors on the control are not accumulative, logical errors due to small errors in gate implementation will still persist.   These are small and can be corrected using conventional error correction as elaborated on in~\Cref{sec:Discussion}. In other words, the effective channel on the logical space implementing logical gates is of constant, small, error, i.e. the effective error does not increase with the number of cycles. Note that this is the standard assumption in noisy gate channels models in computation: the error in implementing a particular gate|say CNOT|is independent of how may gates have been applied before it.

Item 3) proves that the quantum frequential computing regime is achievable (later we will prove that this scaling is optimal and only the square-root scaling is achievable for semi-classical states).

A few important generic observations: First note that, the bus width $N_g-1$ is give by $T_0 f-1$ and thus grows effectively linearly in the power [analogously to the scaling with mean initial energy in~\Cref{thm:contrl with two clocks}].

Also note that the dimensionless quantity $T_0^2 P^\text{in}$ is $T_0$-independent (proven in~\Cref{Dependency on T0}). As such, the frequency is directly proportional to $1/T_0$ analogously to all previous theorems|recall~\Cref{sec:T0 dependency}.

\subsection{Total power consumption}
Since we have modelled the classical oscillator on $\Cl_2$ driving the bus implicitly via the assumption~\cref{eq:cell restriction thm3}, the power consumption, $P^\textup{in}$, of the oscillator on $\Cl$ controlling the implementation of logical gates is not the only power consumption of the quantum frequential computer.  
The oscillator on $\Cl_2$ controlling the bus has the same cycle time and mean energy as the state on $\Cl$ (recall~\Cref{sec:gearing}). The main difference is that, contrary to the state on $\Cl$, it is semi-classical (i.e. a non-squeezed state). Squeezing states generally often requires energy and as such the power required to stabilise the oscillator on $\Cl_2$ should be less than or equal to that required to stabilise the oscillator on $\Cl_2$. As such, at most, the power requirements for the oscillator on $\Cl_2$ should be directly proportional to those of $\Cl$, namely $P^\textup{in}$. Therefore, the total power requirements $P_\text{tot}$ should scale as $P_\text{tot}\sim P^\textup{in}$ and hence the quantum advantage of quantum frequential computers should be maintained when total power considerations are taken into account. Note also that these total-power-to-gate-frequency relationships are in line with the results obtained in~\Cref{thm:contrl with two clocks}, where both oscillators (the one on $\Cl$ and $\Cl_2$) where modelled explicitly and the total energy of both oscillators taken into account. 

These results are also abstract-theoretic and assume no dissipation from heating through, e.g., resistive processes. If one used a superconducting circuit design to implement a quantum frequential computer, then these costs should be minimal. Moreover, they should not affect the scaling of the frequency-power relation and thus should not compromise the quantum advantage. This is because power consumption originating from resistive dissipation is directly proportional to the input power itself, and thus should only contribute to the total power consumption by adding a term of the form $C_\text{resistive}P^\textup{in}$, for some constant $C_\text{resistive}$. Consequently, the relation $P_\text{tot}\sim P^\textup{in}$ should still hold in the presence of resistive elements. However, one should ensure that such dissipation does not introduce noise that converts the squeezed control states into non-squeezed states.


\section{Nonequilibrium steady-state dynamics,  power consumption and heat dissipation: upper bounds}
\label{sec:thm4 main text up}
\subsection{Motivation}
In the previous section, we have proven that a quantum frequential computer based on a self-oscillator can function indefinitely, so long as it consumes power at a rate proportional to its gate frequency. What's more, the proof is by construction and the state of the self-oscillator is squeezed. In this section, we prove that not only is the scaling tight, but that the quadratic quantum advantage remains, namely if we restrict the state of the oscillator to purely semi-classical states (i.e. non-squeezed states), then the gate frequency is upper bounded by square-root of the power consumed.

\subsection{General self-oscillating model}
For this task, we need to define a very general model of a self-oscillator controlling the application of gates and functioning in a nonequilibrium steady-state. We briefly describe it here leaving full details to~\Cref{sec:generic self-oscillator dynamics,Upper bounds on gate frequency for self-oscillators}. We consider a generic system denoted $\Sy$ which interacts with the oscillator on $\Cl$.  Similarly to~\cref{generator def}, the generator of dynamics is given by a time-independent Lindbladian of the form $\mathcal{L}_{\Sy\Cl}(\cdot)=-\mi \big[H_{{\Sy\Cl}}, \,\cdot \,\big] + \mathcal{D}_\Cl(\cdot)$, where  $\mathcal{D}_\Cl$ is any dissipater  of the form $\mathcal{D}_\Cl= \mathcal{D}_\Cl^\textup{re}+ \mathcal{D}_\Cl^\textup{no re}$, where $\mathcal{D}_\Cl^\textup{re}$ is responsible for the restoring force, which corrects for errors induced on the oscillator on $\Cl$ during each cycle. It is any completely positive and trace non-increasing map whose output|the renewed state of the self-oscillator|is denoted $\rho^0_\Cl$ and allowed to be any state, i.e. $\rho^0_\Cl\in\mathcal{S}(\mathcal{H}_\Cl)$. The Hamiltonian $H_{{\Sy\Cl}}$ is any Hamiltonian with ground state zero and of bounded (but arbitrarily large) dimension.

We will now derive upper bounds on the gate frequency as a function of power. For this, we need to introduce two quantities. The $1\stt$ quantifies how the state of the self-oscillator on $\Cl$ interacts with the system on $\Sy$ directly after the renewal operation occurs. The rationale is that  the job of the renewal is to correct errors in the state of the control and if the state it is renewed to, namely $\rho^0_\Cl$,  were to have non-zero overlap with interaction terms coupling $\Sy$ to $\Cl$, then these interactions would be occurring during the reset process itself, thus leading to additional errors. What's more, the derivation of the dynamical semigroup requires the separation of timescales, and this may not hold when these different processes overlap. We quantify the magnitude of said interactions via the dimensionless quantity $\varepsilon^0_\textup{H}\geq 0$.  We call it the \emph{instantaneous initial-cycle-state parameter}. Its full technical definition can be found in~\Cref{sec:Quantities pertaining to the quality of computer}. It is zero when the overlap of $\rho^0_\Cl$ with said interaction terms is zero, and small when the overlap is small. Note that, importantly, $\rho^0_\Cl$ will still evolve under $H_{\Sy\Cl}$, just that initially, when $\varepsilon^0_\textup{H}$ is very small, it will do so under $H_\Cl$ (i.e, the term in $H_{\Sy\Cl}$ which  acts trivially on $\Sy$). The evolution under $H_{\Sy\Cl}$  is vital, because during the isentropic time interval, the state of the control needs to interact strongly with $\Sy$ in order to implement the required gate sequence on it.

We assume that $\Sy$ contains a subsystem $\lo$ which is the logical state of the computation where gates are applied sequentially at frequency $f$, and map the logical states to orthogonal states, up to some gate implementation error $\varepsilon^\textup{gate}\geq 0$.\footnote{We could assume more generally that the states are not always mapped to orthogonal states, but this would add significant technical complication and would not affect the overall scaling laws we aim to derive|recall discussion in~\Cref{sec:Classical and quantum upper limits}.}

Finally, since the gates which need to be applied in the $l\thh$ cycle should only be applied during the isentropic time interval, it is important that said interval is quantified. In reality, there is no sudden change between an exactly isentropic time interval and a non-isentropic (irreversible) time interval. In contrast, in each cycle, there is a progressive increase in irreversibility starting from exactly isentropic dynamics after the renewal process. We will thus introduce a cut-off denoted $t_{\max,l}>0$ which quantifies how large this approximate time interval is in the $l\thh$ cycle, and  refer to $[\tau_l, t_{\max,l} + \tau_l)$ as the $l\thh$ isentropic interval. We leave the quantifying definition of $t_{\max,l}$ to~\Cref{Upper bounds on gate frequency for self-oscillators}. This said, the increase in irreversibility need not be gradual, and can be, to an extremely good approximation, zero for the vast majority of the cycle (as we will prove via example in~\Cref{eq:calculation of isentropic regime and epsilon H for the pecialised theorem}).

Similarly to previous sections, we denote the set of semi-classical states for this system in the $l\thh$ cycle $\mathcal{C}^{\text{clas.,\,}l}_{\Sy\Cl}$ (defined in~\cref{def:semi calssical states self-oscillator}). We denote the power consumed per $l\thh$ cycle by $P_l^\textup{in}$. These are defined analogously to in the previous section. This said, unlike before, we refrain from omitting the  subscripts $l$ here since they may now  differ between cycles. (The explicit definition is written in~\cref{sec:Power consumption derivation} for completeness.)

\subsection{Theorem statement}
The following theorem proves upper bounds on the frequency as a function of power for when there are no constraints on the state of the control and when it is constrained to be  semi-classical (not squeezed). The two upper bounds differ by a quadratic factor.

	\begin{restatable}[Upper bounds on gate frequency in the nonequilibrium steady-state regime]{theorem}{thmStadyStateUpperbounds}\label{thm:heat dissipation upper bounds}
Consider the dynamics during
the $l\thh$ isentropic time interval, $[\tau_l, \tau_l+t_{\max,l})$, for $l\in\nnz$.  The gate frequency $f$ during said interval is upper bounded for all $\varepsilon^\textup{gate}\in[0,1/2]$ as follows:
\\[0.2cm]
\noindent Case 1)
\begin{align}
		f \leq 	\frac{ \sqrt{T_0^2 P^\textup{in}_l +\varepsilon^0_\textup{H}}}{(\lambda+1)c_0 T_0},
\end{align} 
if $\rho_{\Sy\Cl}(\tau_l|\tau_l)\in \mathcal{C}^{\text{clas.,\,}l}_{\Sy\Cl}.$\\[0.2cm]
\noindent Case 2)
\begin{align}
			f \leq 	\frac{T_0^2 P^\textup{in}_l +\varepsilon^0_\textup{H}}{(\lambda+1)\kappa T_0}.\quad
\end{align}
Here $\lambda>0, \kappa>0, c_0>0$ are numerical constants. 
\end{restatable} See~\Cref{sec:secnical proofs for self-oscillators} for a proof, and~\Cref{sec:generalisation to more general semi-classical states} for a generalisation of the theorem  to more general sets of semi-classical states.
Thus, taken together with the bound from~\Cref{thm:heat dissipation},  it proves that quantum frequential computers operating in a steady-state, have a quadratic runtime advantage over conventional computers with purely classical or semi-classical control.


\subsection{Bounding the isentropic time interval and the  instantaneous initial-cycle-state parameter error for~\Cref{thm:heat dissipation}}\label{eq:calculation of isentropic regime and epsilon H for the pecialised theorem}
\Cref{thm:heat dissipation} shows that a linear  scaling of frequency and power is achievable for a quantum frequential computer.  The upper bound of~\Cref{thm:heat dissipation upper bounds}  would imply that said bound is optimal, if, in the context of~\Cref{thm:heat dissipation upper bounds}, two properties hold: 1) The isentropic time interval in the $l\thh$ cycle includes the time interval in which gates are being applied. 2)  the  instantaneous initial-cycle-state parameter error $\varepsilon^0_\textup{H}$ is small. In this section, we state two lemmas which prove 1) and 2) hold. 

The first lemma proving property 1) is:

\begin{restatable}{lemma}{boundingtmax}\label{lem:bounding tmax for the quantum frequantic comp in the proof}
	For all 
	fixed $\bar\varepsilon\in(0,1/6)$, the $l\thh$ isentropic time interval, $[\tau_l, t_{\max,l}+\tau_l)$,  for the quantum frequential computer for which the results of~\Cref{thm:heat dissipation} apply, has $t_{\max,l}$ bounded by  
	\begin{align}
		T_0 - t_1 \leq t_{\max,l} < T_0 , 	\label{eq:isentropic interval in lemma}
	\end{align}
	in the $P^\textup{in}$ $\to$ $\infty$ limit.
\end{restatable}
See~\Cref{sec:Calculation of the range of the self-oscillator isentropic regime  and} for a proof.
Some comments: the first important observation is that  the isentropic interval's length is almost the entire cycle time $T_0$ (recall that $t_1=1/f$ and so tends to zero as $P^\textup{in}$ $\to$ $\infty$.) This shows that it can be large. Moreover, the interval of uncertainty of where the $l\thh$ isentropic interval ends, $[T_0-t_1+\tau_l, T_0+\tau_l)$, is precisely the same interval in which the probability of the next renewal occurring is large (recall~\cref{eq:thm3 prob dist} in~\Cref{thm:heat dissipation}). This is exactly what one should expect, since the triggering of a renewal necessarily involves the control interacting with the environment and hence the exiting of the isentropic regime.
Furthermore, the lower bound  $T_0 - t_1 \leq t_{\max,l} $ guarantees that all the gates in the $l\thh$ cycle  are all performed in the isentropic regime (since  they are all performed in the time interval $[\tau_l, \tau_l+ T_0- t_1)$ by construction).

The second lemma proving property 2) is:
\begin{restatable}{lemma}{boundingInstantaneous}\label{lem:bounding instantaneous epsilon}
	For all fixed $\bar\varepsilon>0$,  the quantum frequential computer for which the results of~\Cref{thm:heat dissipation} apply, the instantaneous initial-cycle-state parameter, $\varepsilon_\textup{H}^0$, is bounded as follows
	\begin{align}
		\varepsilon_\textup{H}^0 \leq \SupPolyDecay(P^\textup{in}) \label{eq:main lemma eq for initantaneous para calc}
	\end{align}
	as $P^\textup{in}\to\infty$.
\end{restatable}

\section{Discussion}\label{sec:Discussion}
In this section, we discuss some important additional points. Cursory readers may wish to skip this section and go straight to~\Cref{sec:Conclusion}.

\subsection{Quantum advantage does not required stronger bit/qubit-control coupling}\label{sec:nostrong interactions needed} 
All bounds on gate frequency (both upper and lower) in this~\doc{} have the same dependency on $T_0$ (which is the runtime in the case of~\Cref{thm:upperboundsEnergy,thm:comptuer with fixed memory} or cycle time in the case of~\Cref{thm:contrl with two clocks,thm:heat dissipation,thm:heat dissipation upper bounds}). This is to say, they are all proportional to $1/T_0$. Furthermore, the generators of dynamics (Hamiltonian or Lindbladian, depending on the theorem), are directly proportional to $1/T_0$ also. As already discussed in~\Cref{sec:Classical and quantum upper limits}, this well-known dependency physically corresponds to re-scaling time itself. 

One may wonder why we cannot simply increase $1/T_0$ indefinitely to achieve ever greater gate frequencies. In practice, doing so is highly challenging and often represents the bottleneck itself as we now explain. Since the generator of dynamics has free oscillator terms and interaction terms coupling said oscillators to the qubits/bits, there are two separate things which must be increased in tandem via a multiplicative factor of $1/T_0$: 
\begin{itemize}
	\item [1)] The oscillator frequencies. This is usually not the bottleneck, as one can easily have oscillators run at much higher frequencies than those used in computation. (E.g., such as that in a quartz self-oscillator in a computer.)
	
	\item [2)] Increasing the interaction strengths. When $1/T_0$ is large, this corresponds to the so-called ``strong'' interaction regime. It is notoriously hard to achieve strong interaction regimes  in physical devices and often poses the  limiting factor  to how fast gates can be applied.
\end{itemize}

So in summary, while oscillator frequencies can easily be increased, often one finds that the ``gates cannot keep up'' because one  cannot increase the interaction strength in tandem with the increased frequency.

On the other-hand, the quadratic increase achieved via a quantum frequential computer is of a very different nature and does not suffer from this bottleneck.  The increase in frequency is achieved by a change in the state|not a change in the interaction strength in the generator of dynamics. It also does not require an increase in mean-energy of the state. For concreteness, we show these points in~\Cref{sec:trenghth indepedent} in the context of~\Cref{thm:comptuer with fixed memory}, although the observations made there hold more generally to all theorems in the main text.

\subsection{Stochastic nature of control renewal and the nature of the non-equilibrium steady state}\label{sec:stochasticity of correction}
We have seen that, while in the dynamical semigroup models of a quantum frequential computer the cycle time is not constant but rather a stochastic variable, it is nevertheless extremely well localised: the standard deviation of the cycle time decays superpolynomially with power (recall~\Cref{thm:heat dissipation}). Thus, any effects associated with the stochastic nature of the cycle time will be minimal. Furthermore, the application of gates in each cycle is conditioned on the start of that cycle, so gate errors arising from overlap between gates from different cycles are avoided. Finally, we have seen that a classical signal is generated at the start of each new cycle. Thus, classical or semi-classical controls external to the current setup (such as rounds of error correction on the logical space) could condition their own operation on this signal, thereby eliminating hypothetical timing issues.

\subsection{Error correction}\label{Error correction}  
When implementing solely classical algorithms on a quantum frequential computer, error correction of the logical space may not be necessary, since error correction for classical algorithms is notoriously easier than for quantum ones. For example, the Eastin–Knill no-go theorem~\cite{EastinKnill2019}, which prevents simple gate error correction strategies,  does not hold for classical algorithms.
  
In the case of implementing quantum algorithms however, the logical space of a quantum frequential computer will likely require error correction.  One may be concerned that this might be much harder for a quantum frequential computer compared with a conventional quantum or classical computer running at a much lower gate frequency. However, note that the theorems developed here show that errors per gate can decrease with power fast enough so that, in a fixed time window, while the total number of implemented gates is increasing, the total gate error is decreasing (superpolynomially quickly). This is because quantum frequential computers require minimal back-reaction on the control, which in turn, reduces the errors. 

 Even if the gate frequency is much higher than the time required to perform one round of error correction,  one can ``pause'' the computation for the required time needed to implement one round of error correction at regular intervals (say a fixed multiple of $T_0$). Thus error correction should only add a small multiplicative factor to the runtime.  Of course, this reasoning only takes into account  errors caused directly by the control itself, but not those from the environment. However, the rate of environmentally-induced errors should depend solely on the rate of background processes unrelated to the control itself (e.g. an incoming galactic gamma-ray). As such these errors should also be adequately correctable via the above scheme. 
 
 It is also important to note that the ``pausing'' of the quantum frequential computer mentioned above can be achieved completely autonomously already within the models presented here. To do so, one simply needs to include the identity gate $id$ in the gate set $\mathcal{G}$ and insert sequences ($id$, $id$, $\ldots$, $id$) of length $J$ in-between the memory cells encoding the algorithm at regular intervals. Since the identity gate acts trivially on $\lo$ this will ``pause'' the computation for a time $J t_1$ at regular intervals  allowing external intervention.
 While these error correction intervals can be predicted in advance of starting the computation and hence already interlaced with the memory cells containing the algorithm before starting the computation, one can also do it on the fly with only classical control since the memory cells are only read/written to by the bus at the bus frequency, $f_\text{bus}=1/T_0$, thus leading to ample time for updating.  Note also that this ``pausing'' mechanism can also be used at read out, if the readout mechanism is slower than the logical gate frequency.
 
 Finally, further research is required to understand if the control ideas developed in this \doc{} for quantum frequential computing, can also be applied to control the error correction. If so, it might be possible to speed-up the error correction process itself.
 
 \subsection{Pure vs. mixed states}\label{Pure vs. mixed states} Pure states are arguably an idealisation of mixed states. In this~\doc, in the main, we have used a pure-state formalism for simplicity of expression. Moreover, in the theorems of the main text, the pure system state on $\lo$ can be identified with the purification of a mixed state of a smaller system which we would now associate with the ``actual'' logical/physical system. In such a scenario, the upper bounds on the dynamics still hold when replacing the purifications with their mixed counterparts since the trace distance satisfies the data processing inequality and the operations performed on $\lo$ would act trivially on the ancillary purifying system.  For the classical register, we could replace it with a probabilistic mixture over the pure orthogonal register states. Since our bounds hold for every said pure state, they would also hold for the ensemble state. This would of course correspond to increase error in the computation due to uncertainty in the initial state of the memory|This setup would correspond to the computer implementing different algorithms according to some probability distribution over them. 
 
 \subsection{Intuitive explanation to why the internal bus of a quantum frequential computer can be classical}\label{sec:why  bus of a quantum frequential computer can be classical}

In the model of a quantum frequential computer of~\Cref{sec:gearing}, the oscillator on $\Cl$, which controls the application of logical gates, is quantum, whereas the oscillator on $\Cl_2$, responsible for updating the memory cells on $\M_0$ with gate-sequence information, is semi-classical. However, both oscillators perform the same number of unitary operations per cycle and consume the same power.

The crucial difference is the timing structure of these operations. For the quantum oscillator on $\Cl$, the time intervals during which logical gates are applied must not overlap, as these gates generally do not commute. In contrast, for the classical bus control system on $\Cl_2$, the unitary operations significantly overlap in time. This overlap poses no issue because the bus lanes apply commuting unitaries. The only constraint is that the duration of each unitary application on a given bus lane must not exceed the cycle time $T_0$, as longer intervals would lead to overlap and thus cause errors. It is precisely because these time intervals match the cycle time itself that updating memory cells in $\M_0$ for each gate requires careful timing control: this motivates the inclusion of memory cells starting at staggered intervals, implemented via the memory control system on $\W$.

The semi-classical state on $\Cl_2$, as used in the proof of the theorem in~\Cref{sec:gearing}, operates at the optimal classical limit (also known as the optimal standard quantum limit). We suspect this optimal classical limit is necessary, as noisier classical states would likely only degrade performance. However, this remains an open question.

Finally, we note that the notion of clearly defined "time windows" during which a specific unitary is applied is used purely for explanatory convenience. In the actual underlying model, gates are continuously applied; the concept of discrete time intervals serves as an excellent approximation of the true continuous dynamics.

\subsection{Oscillator synchronization} In the context of~\Cref{thm:heat dissipation} some synchronization of the two oscillators is required due to the small statistical fluctuations in cycle time originating from interactions with the environment. Synchronisation of two classical oscillators is routine and well understood~\cite{Coleman2021,Wang2018} but one may wonder if quantum resources are required to do this in the case of a quantum frequential computer, since one of the oscillators is quantum in nature. We envisage that even for quantum frequential computers, the physics of classical synchronization suffice. This is because the conditions on the registers (\cref{eq:cell restriction thm3,eq:cell restriction thm3}) are only necessary conditions, in practice the classical bus oscillator can write this information to the allocated memory cell before this time and update it after this time. It can do this at some constant fraction of the bus frequency $f_\textup{bus}$ with high probability. Indeed, this is actually the case for the classical oscillator on $\Cl_2$ in the case of the oscillator used in the proof of~\Cref{thm:contrl with two clocks} (as can be seen in~\Cref{sec:proof of 2 clock theomre in main text}). As such, the classical oscillator on $\Cl_2$ controlling the bus only needs to be in sync with the oscillator on $\Cl$ within a constant fraction of the frequency $f_\textup{bus}$, which is far less than the much faster gate frequency $f$. Therefore, in order to keep the oscillator on $\Cl_2$ sufficiently in sync  with that on $\Cl$, the oscillator on $\Cl$ only needs to generate a classically-detectable signal at the end of each cycle. This can either be done by measuring a classical bit on $\lo$ which $\Cl$ generates per cycle or via monitoring classically when the renewal process occurs. In~\Cref{sec:classical bit for renewal prcess} we show how the renewal process can easily generate this classical information in the setting of~\Cref{thm:heat dissipation}.

\subsection{Heat generation and cooling requirements} We have shown that an optimal quantum frequential computer operating in a nonequilibrium steady state is achievable in which the heat generated is proportional to the power consumed. This is not surprising, since heat generation is usually proportional to the power consumption in classical devices. Moreover, the cooling rate required to prevent a device from overheating is proportional to the rate at which heat is produced, so that a constant temperature can be maintained.  Since we have shown that an optimal quantum frequential computer has a quadratically higher frequency  as a function of power, it also has a quadratically higher frequency as a function of the required cooling.

In practice, there might be some heat generated when each gate is applied due to some noisy coupling with the environment, but since our results show that this coupling would not be fundamental, in principle, it could be made arbitrarily small, and thus not leading to a significant overall increase in heat generation.


\subsection{Irreversible computing}
While our quantum frequential computer is modelled with unitary gates (even in the case of implementing solely classical algorithms), the concept of non-reversible computing can still be formed within our framework. Ultimately, all physical processes, including gate applications, are reversible if one has access to the appropriate degrees of freedom (d.o.f.). So-called irreversible computing simply "appears" irreversible due to lack of access to these degrees of freedom, because some information associated with them is effectively lost to inaccessible parts of the environment. Thus, irreversibility emerges solely due to incomplete information about the system, rather than any fundamental non-reversible nature of the underlying physics.

If one were to take said viewpoint, the main caveat is that one would have an additional heat dissipation source: the entropy associated with the Landauer erasure cost of irreversible gate implementation.

\subsection{Nature of the logical space in a quantum frequential computer implementing solely classical algorithms} In this case, $\mathcal{G}$ only admits a classical gate set, and hence the logical space is always in a tensor-product state of logical zeros and ones after the application of each gate|as is to be expected in classical computation. However, we have not restricted the dynamics \text{during} the application of said gates.  Therefore, it is likely that said dynamics takes logical states momentarily  into superpositions of  logical zeros and ones only returns to a logical state of zeros and ones at the end of the gate application. In this sense, even when implementing purely classical algorithms, the logical space is ``quantum''. Note, that this requirement to use a quantum description of the gate implementation, is not confined to quantum frequential computers|even in the case of classical logic, it is actually nowadays common place in state-of-the-art transistors. See, e.g.~\cite{Rahman2003}.

Conveniently, conventional classical active error correction techniques still apply, even when the logical space is only classical between gate applications. This is an advantage of a quantum frequential computer which only implements classical algorithms over one which implements quantum algorithms since classical active error correction is easier than quantum error correction as discussed in~\cref{Error correction}.

\subsection{Semi-classical control states}
In this~\doc, we have used the term ``semi-classical states'' to refer to quantum systems on $\Cl$ or $\Cl_2$ with standard quantum limit properties (i.e. not squeezed). The motivation is that quantum theory is the best representation of the world that we have for non-relativistic physics, and as such the most meaningful. 
As with other standard quantum limit results in quantum metrology~\cite{Giovannetti2011,Degen2017,Pirandola2018}, it is expected that the results of this~\doc{} can be reproduced when the systems we have referred to as classical can be replaced by stochastic ones. Furthermore, oscillators with the same relevant properties as the conventional laser should suffice for usage as the classical oscillator systems we consider here|see~\Cref{sec:Outlook} for more details. 

\subsection{Coupling terms in the Hamiltonian}
We have explained that the state on $\Cl$ is non-classical in the case of a quantum frequential computer. It is worth remarking that the nature of the interaction terms $\{I_\Cl^{(l)}\}_l$ used for the Hamiltonians of the quantum frequential computers in this~\doc{} also appear to be critical. They have to be chosen in a way that they exert minimal back-reaction on the state of the oscillator in order not to degrade it too quickly. 

As detailed in~\Cref{sec:quantum advantage}, they are chosen to have a particular form, namely diagonal in the discrete Fourier transform basis of the eigenbasis of the free Hamiltonian $H_\Cl$ (a quantum harmonic oscillator). This is not a common basis for interaction terms to couple to. It is far more common for the coupling terms to be diagonal in the position basis, (i.e. a function of the position operator $\hat x$) or  sometimes the momentum basis (i.e. a function of the momentum operator $\hat p$).  An interesting question is whether such Hamiltonians (i.e. Hamiltonians of the form~\cref{eq:main complete ham} with $I^{(l)}_\Cl\mapsto I^{(l)\prime}_\Cl(\hat x)$,  $I^{(l)\prime}_\Cl(\cdot):\rr\to\rrz$) are capable of producing quantum frequential computers|We suspect not. Interaction terms of the form $\{I_\Cl^{(l)}\}_l$ can be constructed in physical settings, as shown theoretically in~\cite{Arman2022}.

\subsection{Other dynamical models for gate-based  computation}
In the early days of quantum computation when there were many doubts whether such a device could exist, even in principle, Paul Benioff made the observation that if quantum computers do exist (even in principle), they must obey the Schr\"odinger equation, and set out to prove such models existed in a series of papers~\cite{benioff1980computer,benioff1982quantum}. In~\cite{benioff1980computer} the author derives a Hamiltonian dynamical model which is local (in the sense of having a control field that only has local support at all times during the computation but suffers from the problem of not having a ground state, and hence requiring infinite energy). In~\cite{benioff1982quantum} the author rectifies this by introducing two Hamiltonian dynamical models (both of which require finite energy). The first is highly non-local, which is undesirable for computation from an engineering perspective, as the author opines. The second is local but has a time-dependent Hamiltonian. The time-dependency of the Hamiltonian is a problem for one wishing to understand autonomy, since it requires external control, the costs of which are thus not accounted for. This is in contrast to the work presented here, where we have presented a fully local model which can run forever in a non-equilibrium steady state of computation. The cost for doing so is that it cannot be a Hamiltonian model and needs power consumption to function. 

The first study looking into the thermodynamics of the evolution of computation was carried out by Feynman~\cite{FeynmanComp} following the works of Landauer and Bennett on Landauer erasure. A study of reversible classical computation leveraging a modern quantum thermodynamic approach was performed in~\cite{Frank2021}. A fully thermodynamically autonomous model was introduced in~\cite{LipkaBartosik2024}, and a study of power consumption and energy dissipation for the case of a field-effect transistor in~\cite{Gao2021}.

Another interesting approach is to allow for the possibility of a quantum or semi-classical clock for the timing of gates, while applying said gates via incoherent, dissipative operations and modelling the computation autonomously. This is the approach of~\cite{FlorienquantumComputer} which models computation via a dynamical semigroup and derives  bounds on gate fidelity of the application of gates. These bounds suggest much larger errors than the superpolynomial decay given by the theorems (i.e. bounds~\labelcref{eq:thm fixed memory 1,eq:summation thm2,eq:dist condition thmor}). The reason comes down to the assumptions made in the modelling: the authors use incoherent dissipative interactions to turn the gates on/off ``under load'', while in this \doc{} we have modelled said interactions coherently and not ``under load''. Incoherent modelling is noisier and creates more entropy. The only  incoherent interactions used in this \doc{} are those used to power the self-oscillator and not for the application of the gates themselves. This said, we expect that incoherent, dissipative  interactions could still be used to turn gates on/off without incurring the large errors of the authors, so long as the turning on/off is performed when the gate is not ``under load'', this is to say, when the gate is not being implemented. When a gate is not implemented ``under load'', there is a time delay between the act of switching on/off a particular interaction, and the actual implementation of the gate. An example of not ``under load'' interactions in this \doc{} are those of the bus updating the register: here the bus switches the register cells in $\M_0$ between states encoding for different gates from the gate set $\mathcal{G}$, but the gates themselves (implemented via the control on $\Cl$) interact with said cells at different times. Another caveat, is that unlike in the models corresponding to~\Cref{thm:contrl with two clocks,thm:heat dissipation}, their model does not reuse interaction terms responsible for applying gates, and hence the number required in their model is directly proportional to the total number of gates one wishes to apply sequentially. As such, in their model, it is impossible to perform computation is a steady state with finitely-many engineered interaction terms. This issue was overcome in this~\doc{} by the use of an internal bus and moving to a cyclic model.

The main question in this \doc{} has been:~\emph{Is there a quantum runtime advantage to be made if the control of the logical bits/qubits itself is quantum  as a function of the relevant resources?} None of the previous work touches on this question.

Finally, it is worth mentioning a completely different approach: one commonality between the work in this \doc and the above-cited works (and all conventional classical and quantum computers too, for that matter) is  that the gates are applied in a causally ordered fashion. Remarkably, in~\cite{Jordan2017} it is shown that arbitrarily high speeds of computation are achievable for a fixed amount of energy if this gate-causal-ordering assumption is removed. The downside is that it would require superpositions between the control and gates and it is also unknown if quantum error correction would be, even in principle, viable in this approach. And even if it is viable, at what energetic cost (since it would require a new paradigm of error correction).

\subsection{Quantum speed limits}\label{sec:quantum speed limits}
All quantum systems evolving under Hamiltonian dynamics with a ground state obey so-called quantum speed limits. These are lower bounds on the time required for a quantum state to become orthogonal to itself, expressed as a function of its mean and standard deviation in energy~\cite{Levitin1982,Margolus1998}. Here, we discuss their relation to the results presented in this paper, starting from the classical and quantum upper bounds~\labelcref{eq:HL frequncy,eq:SQL frequncy}, and subsequently addressing~\Cref{thm:comptuer with fixed memory,thm:contrl with two clocks}.

The optimal classical bound~\labelcref{eq:SQL frequncy} does not follow directly from existing knowledge about quantum speed limits. Indeed, until recently, it was widely believed that classical systems do not satisfy a speed limit at all~\cite{PhysRevLett.120.070401,PhysRevLett.120.070402}, nor was it clear whether a quantum advantage even exists. Both classical and quantum limits are required to demonstrate a genuine quantum advantage when quantum controls are used, even for classical algorithms. Without such a result, the concept of a quantum frequential computer would lack significance. Conversely, the upper quantum bound~\labelcref{eq:HL frequncy} readily follows from quantum speed limits. Nevertheless, deriving both bounds via a metrological approach clarifies their connection with metrology and the associated concepts of the standard quantum limit and Heisenberg limit. Furthermore, while the quantum bound~\labelcref{eq:HL frequncy} (or equivalently, the quantum speed limit) imposes an upper limit on computation speed, it does not imply actual achievability in computation, since a useful computer requires much more structure than merely traversing a sequence of orthogonal states during a gate sequence. \Cref{thm:comptuer with fixed memory}, on the other hand, demonstrates that this rate is achievable, at least in principle, for universal computation.

In~\cite{Margolus1998}, it was shown that the Salecker-Wigner-Peres clock model~\cite{PhysRev.109.571,Peres1980} saturates the quantum speed limits derived in~\cite{Levitin1982,Margolus1998}. Such clock states can be viewed as infinitely squeezed versions of the quasi-ideal clocks discussed in~\cite{WoodsAut,WoodsPRXQ}. The control states used for $\Cl$ in the proofs of~\Cref{thm:comptuer with fixed memory,thm:contrl with two clocks,thm:heat dissipation} correspond to quasi-ideal clock states with finite squeezing (the precise amount chosen to optimize performance). These quasi-ideal clock states maintain, up to small corrections, a constant squeezing relative to a fixed basis. However, this constancy breaks down when the initial squeezing surpasses a certain threshold. Beyond this threshold, squeezing oscillates in time, and states become anti-squeezed in the basis diagonalizing the interaction terms. Unfortunately, we suspect that employing Salecker-Wigner-Peres clock states as control would not yield a quantum frequential computer. Future research is needed to verify or refute this conjecture. Moreover, quasi-ideal clock states provide good approximations to canonically conjugate operators, whereas Salecker-Wigner-Peres clock states do not—see~\cite{WoodsAut,WoodsPRXQ} for details.

In the Salecker-Wigner-Peres clock model, the mean energy required to sequentially traverse $N$ orthogonal states is proportional to $N$ itself~\cite{Margolus1998}. This also holds true for the Hamiltonian employed in the proof of~\Cref{thm:comptuer with fixed memory}. Consequently, this necessitates an undesirable linear increase in gate frequency with the total number of implemented gates. This issue is resolved by~\Cref{thm:contrl with two clocks} (see the discussion following the theorem). Reframed in the language of quantum speed limits,~\Cref{thm:contrl with two clocks} constitutes a novel result, demonstrating that the optimal orthogonalization rate can persist significantly longer than previously known within a Hamiltonian framework.

In~\cite{Margolus1998}, it was remarked that quantum speed limits could restrict computation rates. In~\cite{Lloyd2000}, the mean energy of the initial state for a Hamiltonian system was termed the ``power’’ required to implement one gate. However, this notion is highly misleading: useful computation requires applying many gates, and it does not follow that the total initial energy needed for $N$ gates equals $N$ times the mean energy per gate, because energy can be recycled in the Hamiltonian picture considered in~\cite{Lloyd2000}. Indeed, in this~\doc, we explicitly demonstrate through examples that energy can be recycled precisely in this manner. This issue was discussed in the main text following~\Cref{thm:comptuer with fixed memory}, motivating subsequent sections. The discussion following~\Cref{thm:contrl with two clocks} specifically illustrates how the computer considered in its proof recycles energy.

For the upper bounds in~\Cref{thm:upperboundsEnergy}, we have considered only gate sequences mapping the logical space between orthogonal states. Of course, computers capable of executing quantum algorithms require additional functionality—such as gates that do not necessarily map between orthogonal states. We avoided this scenario to prevent added complexity, and because it would not qualitatively alter our main findings. Such considerations would affect only multiplicative constants and not the scaling of gate frequency with initial mean energy, a fact well-known from quantum speed-limit literature, see e.g.,~\cite{Giovannetti2003}.

Finally, while quantum speed limits have been generalized to dynamical semigroups (see e.g.,~\cite{PhysRevLett.110.050403,PhysRevLett.131.220201}), such generalizations are not useful for proving~\Cref{thm:heat dissipation,thm:heat dissipation upper bounds}, even though dynamical semigroup dynamics are considered there. This is for two reasons: 1) In~\Cref{thm:heat dissipation,thm:heat dissipation upper bounds}, we were not concerned with system energy, but rather with the power used to run the computer—this power is given by energy flow into the system, not by the system's mean energy itself. 2) Even if mean system energy were our concern, speed-limit bounds for dynamical semigroups (e.g.~\cite{PhysRevLett.110.050403,PhysRevLett.131.220201}) would not help, since those bounds express speed limits in terms of energy standard deviation rather than mean energy itself. Deriving such mean-energy bounds remains an open research problem.

\subsection{Technical tools used to derive the results in this~\doc} From a technical standpoint, the main tools for deriving the upper bounds from~\Cref{sec:Classical and quantum upper limits} came from Cram{\'e}r-Rao-bound-like arguments and~\cite{Maccone2020squeezingmetrology}. While~\Cref{thm:comptuer with fixed memory,thm:contrl with two clocks,thm:heat dissipation} use technical results derived across papers~\cite{WoodsAut,WoodsPRXQ,Resource} and new insights developed here.  Since the proofs are by construction (up to a constant in an exponentially decaying term which is via existence), the power of the polynomials $\text{poly}(\cdot)$ in~\Cref{thm:comptuer with fixed memory,thm:contrl with two clocks,thm:heat dissipation}, can be calculated exactly if desired. It is expected that they will be of low order.

\section{Conclusion}\label{sec:Conclusion}
\subsection{Summary} 
In this article, we derive bounds on the energy and power required to run conventional classical and quantum computations at a fixed gate frequency, assuming that gates are applied sequentially and that the control is described by semi-classical physics. We then extend this analysis to cases where the control of bits/qubits can be any quantum state, showing that—at least in theory—a quadratic quantum advantage in gate frequency (and thus runtime) is achievable as a function of the consumed energy or power. 

We term computers with such quantum control \emph{quantum frequential computers}. Furthermore, we demonstrate that a quantum frequential computer only needs semi-classical states to control its internal bus, thus highlighting that not all of the control needs to be quantum. These autonomous systems require energy consumption and dissipation as the resources underpinning their operation. Because reduced dissipation directly lowers the cooling requirements (given that the cooling rate is proportional to the rate of energy dissipation), these devices offer practical benefits.

One of the major advantages of this architecture is that it provides a quadratic runtime speed-up for all classical algorithms. For many computational problems—either because no quantum algorithm exists or none has yet been discovered—a quantum frequential computer may be the only viable method to achieve such an advantage.

As for the numerous quantum algorithms with relatively modest quantum speed-ups, quantum frequential computing provides an additional quadratic improvement on top of their existing algorithmic advantage. A prime example is Grover's search algorithm~\cite{Grover1996}, which already offers a quadratic algorithmic advantage. Implementing Grover's algorithm on a quantum frequential computer would thus yield an overall quartic runtime quantum advantage compared to the optimal fully classical implementation.

\subsection{Outlook}\label{sec:Outlook} 
An important practical criterion for quantum frequential computers is the power threshold at which they begin to exhibit quantum advantage. In this work, we have focused exclusively on the asymptotic behaviour, yet there is reason for optimism. Similar mathematical techniques used here for controlling $\Cl$ were previously applied in~\cite{Arman2022} to show that a quantum quadratic advantage in electron decay via spontaneous emission is achievable at a Hilbert space dimension of two. This result suggests that quantum frequential computers could operate effectively even in low-power regimes. Moreover, while our analysis is asymptotic, we have observed that the additive errors decay extremely rapidly—superpolynomially.

Conventional quantum computers currently operate at much lower gate frequencies than state-of-the-art classical computers. This slow performance is one reason why some researchers believe that quantum algorithms with only quadratic or quartic speedups may not demonstrate a practical advantage in the near term~\cite{Hoefler2023,Babbush2021}. Quantum frequential computing, by contrast, has the potential to boost implementation speeds significantly.

Another notable advantage of quantum frequential computing is the potential for higher fidelity gate operations. This benefit arises from two factors: lower power consumption and dissipation, and reduced back-reaction during gate implementation.

Looking ahead, a major challenge is identifying suitable physical systems for realizing quantum frequential computers. Interestingly, recent work has revealed that conventional lasers operate only up to the classical coherence length limit (the Schawlow–Townes limit), even though a quantum limit exists beyond this threshold~\cite{Wiseman2021}. The principles governing both standard quantum-limited and Heisenberg-limited lasers apply more broadly to other oscillating systems~\cite{Loughlin2023}. We envision constructing the oscillator in a quantum frequential computer using Heisenberg-limited oscillators, while its internal bus could be driven by an optimal standard quantum-limited oscillator, such as a conventional laser. However, coupling the control to the qubits may require innovative engineering solutions to minimise back-reaction. Proposals for building these oscillators have already been advanced~\cite{Wiseman2021,DavidNatComms2021}, and they could play an important role in realizing quantum frequential computers.

Finally, while quantum computers with practical advantages over state-of-the-art classical systems are not yet available, other quantum technologies based on squeezed states have already surpassed their optimal classical counterparts~\cite{PhysRevX.13.041021}. Given that squeezed control states are a key component of quantum frequential computers, these successes provide encouraging real-world evidence supporting our approach.

\begin{acknowledgments}
We acknowledge useful discussions with \'Alvaro Alhambra, Andrew Doherty, Christopher Chubb, Daniel Stilck Fran\c{c}a, Mark Mitchison, Omar Fawzi.
	
This work benefited from a government grant managed by the Agence Nationale de la Recherche under the Plan France 2030 with the reference ANR-22-PETQ-0006.
\end{acknowledgments}
\appendix

\appendix
\addcontentsline{toc}{section}{Appendices}

\section{Semi-classical states}\label{sec:non squeezed states def}
In this article, we deal with semi-classical states in some of the theorems. Here we define these. 

\subsection{Generic definition of semi-classical (non-squeezed states)}\label{sec sub sub:Generic definition of semi-classical}
By semi-classical, we mean non-squeezed states and use the criterion introduced in~\cite{Trifonov1994} specialised to our setup. This class of states is a generalization of the characterization of coherent and squeezed coherent states.

In the context of metrology, the criterion is defined with respect to a Hamiltonian $H$ generating dynamics and an observable $M$ which is used to measure the system~\cite{Maccone2020squeezingmetrology}. A state $\ket{\psi}$ evolving unitarily under Hamiltonian $H$ is said to be non-squeezed with respect to observable $M$ if it is an eigenvector of the operator 
\begin{align}
	L(\lambda):= \lambda \bar M +\mi \bar H, \qquad 	\lambda \in\cc \label{eq:L def}
\end{align}
for $\lambda=1$, where $\bar M:=M/T_0$, $\bar H:= H\, T_0$ are dimensionless time and Hamiltonian operators respectively. This defines the pure semi-classical states. 

This removal of dimensions when defining non-squeezed states is common practice: the criteria is always with respect to dimensionless quantities as above, it is meaningless to say that two quantities are/aren't equal if they have different units. E.g. when we say that a state is non-squeezed with respect to position and momentum, we need dimensionless measures of position and momentum in order to make the comparison. One defines dimensionless quantities by multiplying said quantities by the relevant length, time, and mass scales in the problem to find dimensionless representations.

It is useful to extend the above set of semi-classical states by considering all eigenstates with $|\lambda|=1$ of the operator $L(\lambda)\otimes \id_\Pu$, for all finite dimensional identity operators $\id_\Pu$.  Importantly, any element $\ket{\rho_j}$, of said set of pure states, has the property that under the evolution of the Hamiltonian $\bar H\otimes \id_\Pu$, the measurement statistics of $\bar M\otimes \id_\Pu$ are the same as the measurement statistics of $\bar M$ when evolving the state $ \rho_j=\tr_\Pu\big[\proj{\rho_j}\big]$  under $\bar H$.

Let us denote this set of semi-classical states more formally:  
\begin{align}
	\mathcal{C}(\bar M,\bar H):=\Big\{  \rho_j  \,\big{|}\,   \rho_j = \tr_\Pu\big[\proj{\rho_j}\big]   \Big\}.\label{def:set semi classical states}
\end{align}
We call said set the set of semi-classical states for the operator $L(\lambda)= \lambda \bar M +\mi \bar H$. 

On an intuitive level, it is natural to also define probabilistic mixtures of the states in $\mathcal{C}(\bar M,\bar H)$ as semi-classical states. This is not necessary when we want to prove that a semi-classical limit is obtainable since the above set will suffice. In the case of the upper bounds on what is achievable with semi-classical states, to avoid making the theorems more technical, we do not include such states directly. We simply comment on how it generalises or prove a more technical version of the theorem with its inclusion in an appendix.

\subsection{The Robertson-Schr\"odinger uncertainty relation and the $L(\lambda)$ operator}\label{sec:RS UR}
In this section, we give some context to the semi-classical states in $	\mathcal{C}(\bar M,\bar H)$. We start by  considering the eigenstates of the operator $L(\lambda)=\lambda A+\mi B$, where $A$, $B$ are any two dimensionless observables. (See~\cite{Trifonov1994} for proof and more details, properties  and examples)

The Robertson-Schr\"odinger uncertainty relation is
\begin{align}
	\sigma_A^2 \sigma_B^2 \geq \frac{1}{4} \left( \braket{C}^2 +4 \,\sigma_{AB}^2\right), \quad C:= -\mi [A, B],\label{eq:RS UR}
\end{align}
where 	
\begin{align}
 \sigma_{AB}:= \frac{1}{2}   \braket{AB + BA}-\braket{A}\braket{B},\quad  \sigma_D := \sqrt{\braket{D^2} - \braket{D}^2}, \quad D\in\{A,B\}.
\end{align}
When $ \sigma_{AB}=0$, the above Robertson-Schr\"odinger uncertainty relation reduces to the so-called Heisenberg  uncertainty relation. The inequality holds for all pure states.

As proved in~\cite{Trifonov1994}, all eigenstates of any operator of $L(\lambda)=\lambda A+\mi B$, (where $A$, $B$ are any two observables), are Robertson-Schr\"odinger minimum uncertainty states for any $\lambda\in\cc$. This is to say, they satisfy~\cref{eq:RS UR} with equality.  Moreover, in the special case $|\lambda|=1$, they are of equal uncertainty in $A$ and $B$, namely  
\begin{align}
	\sigma_A = \sigma_B.
\end{align}

Since operators $A\otimes \id_\Pu$, and $B\otimes \id_\Pu$ are themselves observables, the above reasoning also applies to the operator  $L(\lambda)\otimes \id_\Pu$. It thus readily follows that the states in $\mathcal{C}(\bar M,\bar H)$ in~\cref{def:set semi classical states} are also Robertson-Schr\"odinger minimum uncertainty states with  equal uncertainty in $A$ and $B$. The remaining states in $\mathcal{C}(\bar M,\bar H)$ are convex mixtures of said minimum and equal uncertainty  states.

\subsection{Semi-classical states for~\Cref{thm:upperboundsEnergy,thm:comptuer with fixed memory,thm:contrl with two clocks,thm:heat dissipation upper bounds}}
In the following subsections, we will specialise the definition of the sets of semi-classical states to the theorems which use them. This is because in~\Cref{thm:upperboundsEnergy,thm:heat dissipation upper bounds} upper bounds are generated which hold for all states in the sets, and as such,  we wish to define them more generally so that the bounds hold under more general conditions. Meanwhile, in~\Cref{thm:comptuer with fixed memory,thm:contrl with two clocks}, we show that certain limits hold for a state in the set of semi-classical states. As such we want these sets of semi-classical states to be more restrictive, so as to not include states which one may argue are not ``fully'' semi-classical. Hence these distinctions allow for stronger results.

\subsubsection{Semi-classical states in~\Cref{thm:upperboundsEnergy}}\label{sec:semi classical states def for Hamiltonian upper bounds}
In this section, we define the set of semi-classical states considered in~\Cref{thm:upperboundsEnergy}. It is a specialization of the framework from~\Cref{sec sub sub:Generic definition of semi-classical}.  Recall that in~\Cref{thm:upperboundsEnergy}, we assume that the computer is implementing a classical algorithm with its logical state passing through a sequence of orthogonal states corresponding to  the implemented gate sequence, namely $\big(\ket{l}_\lo\big)_{l=1}^{N_g}$ at times $\big(t_l\big)_{l=1}^{N_g}$ respectively, where $t_l=l \,T_0/N_g$.  We then measure the state of the logical state during the computation to deduce the time, via the  observable
\begin{align}
	M=\sum_{l=0}^{N_g-1} t_l  \proj{l}_\lo. \label{eq:M for thm 1}
\end{align}
We denote by $\mathcal{C}^{\text{clas.}}_{\Com}$ the set of semi-classical initial states of  the computer. It is defined as follows:  A state $\rho_\Com$ is in $\mathcal{C}^{\text{clas.}}_{\Com}$ if there exists two consecutive pairs of states $\rho_\Com(t_l),\rho_\Com(t_{l+1})$ ($t_l, t_{l+1}\in 0,1,\ldots, N_g$) which are both in the set $\mathcal{C}(\bar M,\bar H_\Com)$, with $M$ given by~\mref{eq:M for thm 1}.  Here the states $\rho_\Com(t_l),\rho_\Com(t_{l+1})$ are defined via
\begin{align}
	\rho_{\Com}(t):= \me^{-\mi t H_\Com}  	\rho_{\Com}  \me^{\mi t H_\Com}  .
\end{align}

Note that with this definition, there could be some states in the sequence which are quantum (i.e. squeezed).  This is well motivated because we want to use this set to prove an upper bound on what is possible with classical states, thus allowing for the possibility of only some of the states to be semi-classical only makes our results stronger.

\subsubsection{Semi-classical states in~\Cref{thm:comptuer with fixed memory}} 
\label{sec:alternative def of squeezeing on C}

At times $\{t_j\}_{j=0}^{N_g}$, the state of the memory of the computer described~\Cref{sec:quantum advantage} is always in a classical state in $\mathcal{C}_{\M_0}$. Likewise, in the case of a gate set $\mathcal{G}$ which can only implement classical algorithms, the logical space of the computer  $\lo$ is also in a classical state at said times. Thus the only component which may be in a quantum state is the state of the control on $\Cl$ itself. Moreover, even in the case where  $\mathcal{G}$  permits the application of quantum algorithms, the logical gate frequency $f$ is independent of the state of $\lo$; thus while the states of $\lo$ may be quantum at times $\{t_j\}_{j=0}^{N_g}$, this should not lead to a quantum advantage in frequency.  In this case (like with the classical algorithm case), the only source of quantumness which may lead to a better scaling of $f$ with energy $E$, is the state of the control at times $\{t_j\}_{j=0}^{N_g}$. This setup allows us to see if a semi-classical control state can achieve the scaling of the optimal classical bound in~\Cref{thm:upperboundsEnergy}. As such, here we introduce a special set of semi-classical states on $\Cl$. 

Here the relevant ones are the Hamiltonian of the control, $H_\Cl$, and the basis which diagonalises the interaction terms $\{I_\Cl^{(l)}\}_{l=1}^{N_g}$. The Hamiltonian $H_\Cl$ we use in the proof of~\Cref{thm:comptuer with fixed memory} is $H_\Cl=\sum_{n=0}^{d-1}\frac{2\pi}{T_0} n \proj{E_n}_\Cl$, while the basis which diagonalises the terms $\{I_\Cl^{(l)}\}_l$ is the discrete Fourier transform basis of the orthonormal basis $\{\ket{E_n}_\Cl\}_n$, namely
		$\Big\{	\ket{\theta_k}=\frac{1}{\sqrt{d}} \sum_{n=0}^{d-1} \me^{-\mi 2\pi n k/d}  \ket{E_n}\Big \}_{k=0}^{d-1}.$
One thus defines the ``time observable'' by $t_\Cl:= \sum_{k=0}^{d-1} 2\pi k \,T_0 \proj{\theta_k}_\Cl$.\footnote{See~\cite{WoodsAut} for more insights into $t_\Cl$ and $H_\Cl$.} Defining the operator $L_\Cl:=\lambda t_\Cl+\mi H_\Cl$, we say that a pure state $\ket{\psi}_\Cl$ on $\Cl$ is semi-classical (non-squeezed) if it is an eigenstate of $L_\Cl$ for which $|\lambda|=1$ up to vanishingly small additive corrections in the large $d$ limit, i.e. if  $L_\Cl \ket{\psi(d)}_\Cl= E(\psi) \ket{\psi(d)}_\Cl +\ket{\epsilon(d)}_\Cl$ for $|\lambda|=1$, $E(\psi)\in\cc$ and where $\| \ket{\epsilon(d)}_\Cl \|_2\to0$ as $d\to\infty$.  We denote the set of such states by $\mathcal{C}_\Cl^\textup{clas.}$. In~\Cref{thm:comptuer with fixed memory}, we show that all states of the control on $\Cl_2$ are in $\mathcal{C}_\Cl^\textup{clas.}$ at times $\{t_j\}_{j=0}^{N_g}$.

While this set is defined asymptotically for large $d$, this quantity increases with energy $E_0$, so $d\to\infty$ as $E_0\to\infty$ and so the set is well-defined in the context of~\Cref{thm:comptuer with fixed memory}. Furthermore, since we want to use this set in a theorem which shows what is achievable with a set of states which does not include quantum ones, we want to make sure that \emph{every} relevant state is semi-classical (not just a subset of said relevant states). For this reason, we do not to include other non-squeezed states in this set (e.g. certain types of mixed states).


Note that one can alternatively use $H_{\M_0 \lo\Cl}$ and $\id_{\M_0\lo}\otimes t_\Cl$ in the above  definition of $L_\Cl$  (instead of $H_\Cl$ and $t_\Cl$ respectively) without modification to~\Cref{thm:comptuer with fixed memory} which uses this definition of semi-classical control states. Recall that $H_{\M_0 \lo\Cl}$ defined in~\cref{eq:main complete ham} describes the dynamics of the entire computer, while $H_\Cl$ is the free dynamics term of the control. This is because the relevant approximate eigenstates of $L_\Cl$ which is used in the proof of the Theorem, are orthogonal to all terms in $H_{\M_0 \lo\Cl} - H_{\Cl}$ up to additive contributions which vanish exponentially fast with increasing $d$. 


\subsubsection{Semi-classical states in~\Cref{thm:contrl with two clocks}}\label{sec:alternative def of squeezeing on C2}
For the case of~\Cref{thm:contrl with two clocks}, in~\Cref{sec:gearing}, the classical states on $\Cl_2$ are initially (i.e. at time $t_{0,0}$) a pure product state with the rest of the computer. At later times (i.e. $t_{l,k}$, $l,k>0$), the state on $\Cl_2$ becomes  correlated with the state of the memory $\M$. 
For $t_{0,0}$, we define the semi-classical states analogously to the definition of the pure semi-classical state on $\Cl$ in~\Cref{thm:comptuer with fixed memory}, but on system $\Cl_2$ rather than $\Cl$. (Since the two have isomorphic Hilbert spaces, this is well defined).  From~\Cref{sec:RS UR}, we thus have that these states are both minimum and equal uncertainty states. For later times (i.e. $t_{l,k}$, $l,k>0$), we define states to belong to the set of semi-classical states if they are both minimum and equal uncertainty states with respect to the same two observables which was used for the initial state. 


More formally, one can define this set as follows. In~\cref{eq:main thm2 complete ham}, the Hamiltonian $H_{\Cl_2}$ is defined identically to $H_\Cl$ up to the change of Hilbert space $\Cl\to\Cl_2$. We can thus define the discrete  Fourier transform basis $\{ \ket{\theta_k}_{\Cl_2} \}_k$ analogously to above but for $\Cl_2$.  (The interaction terms $\{I_{\Cl_2}^{(l)}\}_{l=1}^{N_g}$ are also diagonal in this basis.) The operator $t_{\Cl_2}$ can thus be defined as $t_{\Cl_2}:= \sum_{k=0}^{d-1} 2\pi k T_0\proj{\theta_k}_{\Cl_2}$, and the semi-classical (non-squeezed) pure states on $\Cl_2$,  by solving $L_{\Cl_2} \ket{\psi(d)}_{\Cl_2}= E(\psi) \ket{\psi(d)}_{\Cl_2} +\ket{\epsilon(d)}_{\Cl_2}$ for $|\lambda|=1$, $E(\psi)\in\cc$ and where $\| \ket{\epsilon(d)}_{\Cl_2}\|_2=0$ as $d\to\infty$. We define $\rho_{\Cl_2}$ as semi-classical (non-squeezed) if there exists $\ket{\psi(d)}_{\Cl_2}$ such that 
$\Delta t_{\Cl_2}(\rho_{\Cl_2}):=\tr[t_{\Cl_2}^2 \rho_{\Cl_2}]- (\tr[t_{\Cl_2} \rho_{\Cl_2}])^2$ and $\Delta H_{\Cl_2}(\rho_{\Cl_2}):=\tr[H_{\Cl_2}^2 \rho_{\Cl_2}]- (\tr[H_{\Cl_2} \rho_{\Cl_2}])^2$ both satisfy 
\begin{align}
	\Delta t_{\Cl_2}\big(\rho_{\Cl_2}\big) \geq  	\Delta t_{\Cl_2}\big(\ket{\psi(d)}_{\Cl_2} \big), \qquad  \Delta H_{\Cl_2}(\rho_{\Cl_2}) \geq \Delta H_{\Cl_2}\big(\ket{\psi(d)}_{\Cl_2}\big),
\end{align}
where $	\Delta t_{\Cl_2}\big(\ket{\psi(d)}_{\Cl_2}\big) := {\mathstrut}_{\Cl_2}\!\!\bra{\psi(d)}   t_{\Cl_2}^2  \ket{\psi(d)}_{\Cl_2} - \left({\mathstrut}_{\Cl_2}\!\!\bra{\psi(d)}   t_{\Cl_2} \ket{\psi(d)}_{\Cl_2} \right)^2$ and $H_{\Cl_2}\big(\ket{\psi(d)}_{\Cl_2}\big):={\mathstrut}_{\Cl_2}\!\!\bra{\psi(d)}   H_{\Cl_2}^2  \ket{\psi(d)}_{\Cl_2}\\ - \left({\mathstrut}_{\Cl_2}\!\!\bra{\psi(d)}   H_{\Cl_2} \ket{\psi(d)}_{\Cl_2} \right)^2$. This is to say, the have uncertainties in the Hamiltonian and time operator which are above or equal to what is classically optimal.

We denote by $\mathcal{C}_{\Cl_2}^\textup{clas.}$, the set of semi-classical states  on $\Cl_2$ according the above definition. In~\Cref{thm:comptuer with fixed memory}, we show that the state of the control is in this set at all times $\{t_{j,l}\}_{j,l}$.

 Thus the Theorem uses a  strict notion  of semi-classicality, in the sense that it must apply at all relevant times, in contrast to a subset of times. This is important since~\Cref{thm:comptuer with fixed memory} characterises what is possible with solely classical resources. So while we could have included more semi-classical states in  the definition of this set, this would not be beneficial for our purposes.

\subsubsection{Semi-classical states in~\Cref{thm:heat dissipation upper bounds}}\label{def:semi calssical states self-oscillator}

Since the self-oscillator can run in a non-equilibrium steady-state, we define a semi-classical state condition which applies to the states in one cycle, with implications for the achievable frequency in said cycle. 

The relevant pair of observables are the Hamiltonian $H_{\Sy\Cl}$ which generates the unitary part of the dynamics, and the observable $M^{(j,l)}_{\lo} := 
\sum_k t_1 (1-\delta_{k,j})\proj{k,l}_\lo$ which distinguishes between being in state $\ket{j,l}_\lo$ and the other logical states of the computation. (Here $\{ \ket{k,l}_\lo\}_k$ is the basis in~\cref{eq:sequnce orthgonal states}.) Their dimensionless counterparts are
\begin{align}
	\bar M^{(j,l)}_{\lo}:=M^{(j,l)}_{\lo}/T_0,  \qquad \bar H_{\Sy\Cl}:= H_{\Sy\Cl} T_0.
\end{align}
We denote by $\mathcal{C}^{\text{clas.,\,}l}_{\Sy\Cl}$, the set of semi-classical self-oscillator states of the $l\thh$ cycle. It is defined as follows: 
A state $\rho_{\Sy\Cl}(\tau_l|\tau_l)$  is in  $\mathcal{C}^{\text{clas.,\,}l}_{\Sy\Cl}$, if there  exist a pair of states $ \rho_{\Sy\Cl}(t_j+\tau_l|\tau_l)$, $ \rho_{\Sy\Cl}(t_{j+1}+\tau_l|\tau_l)$ (with $t_j, t_{j+1} \in[0,t_{\max,l})$, $j \in \nnz$), such that  $ \rho_{\Sy\Cl}(t_j+\tau_l|\tau_l) \in \mathcal{C}(\bar M^{(j,l)}_{\lo},\bar H_{\Sy\Cl})$ and $\rho_{\Sy\Cl}(t_{j+1}+\tau_l|\tau_l)\in \mathcal{C}(\bar M^{(j+1,l)}_{\lo},\bar H_{\Sy\Cl})$ respectively.  The requirement $t_j, t_{j+1} \in[0,t_{\max,l})$ guarantees that the states belong to the isentropic time interval of the $l\thh$ cycle.
 
Recall that the states $\rho_{\Sy\Cl}(t+\tau_l|\tau_l)$ are the normalised states of the dynamics during the $l\thh$ cycle, namely $ \rho_{\Sy\Cl}(t+\tau_l|\tau_l)= \rho_{\Sy\Cl}\lb t+\tau_l|\tau_l\rb/\tr[\rho_{\Sy\Cl}\lb t+\tau_l|\tau_l\rb]$, where $ \rho_{\Sy\Cl}(t+\tau_l|\tau_l)$ is given by~\cref{eq:conditional states AC def}.

Note that analogously to the definition in~\Cref{sec:semi classical states def for Hamiltonian upper bounds}, the dynamics generated from states $\rho_{\Sy\Cl}(\tau_l|\tau_l)\in \mathcal{C}^{\text{clas.,\,}l}_{\Sy\Cl}$  may turn out to be squeezed at a later time $t_j+\tau_l > \tau_l$, but since we are interested in upper bounds, said set will suffice, since a set where all the states are semi-classical is a subset and thus our definition only imposes stricter conditions.



\section{Proof of~\Cref{thm:upperboundsEnergy}: upper bounds on quantum and classical computers}\label{app:upper bounds proofs}
\thmUpperBoundsComp*
\begin{proof}
	In the Fisher information approach to quantum metrology, the mean squared error of the signal \big(which in our case is $t\in[0,T_0)$\big) is given by~\cite{Tth2014}
	\begin{align}
		\braket{(t_\text{est}-t)^2}&= \sum_k
		P(\xi_k|t) (t_\text{est}(\xi_k)-t)^2 ,\qquad 	\sum_k P(\xi_k|t)=1
	\end{align}
	where $P(\xi_k|t)$ is the probability of predicting measurement outcome $\xi_k$ given that the signal takes on value $t$, and $t_\text{est}=t_\text{est}(\xi_k)\in\rr$ is our estimate for $t\in[0,T_0)$ which we make based on our measurement outcome $\xi_k$. By assumption, the logical space of our computer is initialised to $\ket{0}_\lo$ and then passes though a sequence $\big(\ket{l}_\lo\big)_{l=1}^{N_g}$ of states after the application of the first $l$ gates: the logical state is $\ket{l}_\lo$ for $t=t_l=l\,T_0/N_g$ for $l=0,1,2,\ldots, N_g-1$.  The conditional probability $P(\xi_k|t)$ can be written as $P(\xi_k|t) =\tr[M'(\xi_k) \rho_\Com(t)]$ where $\{M'(\xi_k)\}_\xi$ is a complete set of POVM elements and $\rho_\Com(t)$ the state of the computation (which includes the logical space as a subspace) at time~$t$.  
	
	For our estimate, we choose a projective POVM $\{  \proj{l}_\lo\}_l$ corresponding to the part of the logical state space the computation passes through. The corresponding observable is  given by
	\begin{align}
		M'= \sum_{l=0}^{N_g-1} t_\text{est}(\xi_l)  \proj{l}_\lo ,\label{eq:1st observable def}
	\end{align}
	\begin{align}
		t_\text{est}(\xi_l)=  \left( l+\Delta_0 \right)  \frac{T_0}{N_g}= t_l + \Delta_0   \frac{T_0}{N_g}.
	\end{align}
	Thus the root mean squared error in our measurement as a function of the signal $t$ at times $t=t_l$ is 
	\begin{align}
		\Delta t_\text{est}(t_l):=\sqrt{	\braket{(t_\text{est}-t_l)^2}} = \Delta_0 \frac{T_0}{N_g}, \label{eq:sigma est}
	\end{align} 
	for $l=0,1,\ldots, N_g-1$. We start by proving~\cref{eq:HL frequncy}. 
	From the appendix of~\cite{Giovannetti2012} it is stated that for any  signal $x\in\rr$ with a root mean squared error $\Delta X (x):=\sqrt{\braket{(X-x)^2}}$, for which there exist $x, x'$ such that:
	\begin{align}
		\begin{split}
			1)& \qquad   \Delta X(x) >0 , \qquad \Delta X(x')>0\\
			2)&\qquad |x-x'|=(\lambda+1) \big(\Delta X (x) + \Delta X (x') \big) 
		\end{split} \label{def: delta X form other paper}
	\end{align} 
	where $\lambda=4.64$, then 
	\begin{align}
		\Delta X(x) \geq \frac{\kappa}{\tr[H\rho_0] -E_0},\label{eq:delta x cond}
	\end{align}
	where $\kappa$ is a numerical constant $\kappa\approx 0.091$, $\rho_0$ the initial probe state,
	\begin{align}
		\rho^{}(x)=\me^{x\mi H} \rho_0^{} \me^{-x\mi H}\label{eq:unitary encoding generic}
	\end{align}
is the probe state for signal value $x$, and $E_0$ is the ground state of Hamiltonian $H$. (Here we have specialised to the case of a single copy of the probe state since this is sufficient for our purposes.)
	
	In our case, we can choose the two values of the signal to be $t=t_{l+1}$ and $t'=t_l$, such that $|t-t'|=T_0/N_g$, and $\Delta t_\textup{est}(t_l)+ \Delta t_\textup{est}(t_{l+1})= 2 \Delta_0 T_0/N_g$. Therefore, by choosing $\Delta_0= 1/\big(2(\lambda+1)\big)\approx 0.0887$ it follows that 
	\begin{align}
		\Delta t_\text{est}(t_l) \geq \frac{\kappa}{\tr[H_\Com\rho_0] -E^0_\Com},
	\end{align}
	where we have denoted the ground state of $H_\Com$ by $E^0_\Com$. Therefore, by recalling that $f=N_g/T_0$ with no restrictions of the probe state nor the Hamiltonian, from~\cref{eq:delta x cond,eq:sigma est} we arrive at~\cref{eq:HL frequncy}.
	
	We now move on to the proof of~\cref{eq:SQL frequncy}. Using techniques from quantum metrology, in~\cite{Maccone2020squeezingmetrology} the quantum advantage for squeezed states under similar unitary encoding scheme via a signal-independent Hamiltonian was investigated. Our choice of estimator above is such that their results also apply to it, after changing to dimensionless variables.
	In particular, we consider a dimensionless signal, defined by $t/T_0\in[0,1]$, and a corresponding dimensionless Hamiltonian $\bar H_\Com= H_\Com T_0$. Since the dynamics of $\rho_\Com(t)$ depends on their product, the dynamics can be expressed solely in terms of these dimensionless quantities.  The root mean squared of the dimensionless signal is thus
	\begin{align}
		\bar \Delta t_\textup{est} =\frac{\Delta t_\textup{est} }{T_0}=  \frac{ \Delta_0}{N_g}.
	\end{align}
	The authors of~\cite{Maccone2020squeezingmetrology} consider the same unitary encoding of a signal  described above (i.e.~\cref{eq:unitary encoding generic}), but now for the special case that said signal $\bar x\in\rr$ and Hamiltonian $\bar H$ are dimensionless. They show that when the probe state is not squeezed, and conditions~\cref{eq:delta x cond} are satisfied, then 
	\begin{align}
		\Delta X(x) \geq \frac{\kappa'}{\sqrt{\tr[\bar H\rho_0] -\bar E_0}},\label{eq:delta x cond 2}
	\end{align}
	where $\kappa'$ is a numerical constant.  We provide the definition of non-squeezed states (specialised to our setup) in~\Cref{sec:non squeezed states def}. Alternatively, see~\cite{Maccone2020squeezingmetrology} for details.
	
	
	In our case, we have for $l=0,1,2,\ldots, N_g-1$:
	\begin{align}
		\bar \Delta t_\textup{est} \geq \frac{c_0}{\sqrt{\tr[\bar H_\Com\rho_0] -\bar E^0_\Com}} 
		\label{eq:delta t test classical}
	\end{align}
	if $\rho_\Com(t_l)$ and $\rho_\Com(t_{l+1})$ [where $\rho_\Com(t):=\me^{\mi t H_\Com}\rho_\Com(0) \me^{-\mi  t H_\Com}$] are 
	non-squeezed, where the set of ``non squeezed'' states, denoted $\mathcal{C}_\Com^\textup{clas.}$, is defined in~\Cref{sec:semi classical states def for Hamiltonian upper bounds}. Here $c_0$  is a numerical constant from~\cite{Maccone2020squeezingmetrology}. Thus since $f=N_g/T_0$,~\cref{eq:SQL frequncy} follows from~\cref{eq:delta t test classical}.
	
	Note that we have slightly changed the definition of the observable in the squeezed state criteria (rather than $M'/T_0$ in~\cref{eq:1st observable def}, we have used $M/T_0$ in~\cref{eq:M for thm 1}). These two observables only differ by an additive constant factor times the identity, and the eigenvectors of a matrix are invariant under the addition of a scalar times the identity matrix. Therefore, these different observables yield the same non-squeezed-state criteria.
\end{proof}

\section{Proof of~\Cref{thm:comptuer with fixed memory}: Attaining the quantum limit}\label{sec: proof of 1st quantum clock thorem}
Before stating the proof of the main theorem, we prove several crucial lemmas. The proof of the main theorem is by construction. We will specialise definitions as we proceed and as becomes necessary to prove the desired results.

\subsection{Main structural technical lemma}
Let us introduce Hamiltonians of the form
\begin{align}
	H_{\lo\Cl}:=   H_{\Cl} +  \sum_{\substack{l=1}}^{N_g} I_\lo ^{(l)}\otimes I_\Cl^{(l)}. \label{eq:2 body ham def} 
\end{align}
 In the following lemma, we assume that the states $\{ \ket{t_j}_\lo, \ket{t_j}_\Cl \}_{j=0}^{N_g}$  are normalised and $H_{\lo\Cl}$ in~\cref{eq:2 body ham def} is finite-dimensional.\footnote{This last assumption is overkill; indeed under minimal assumptions it can be extended to the infinite-dimensional case. However, for our purposes this will suffice.}

\begin{lemma}\label{lem:quantum dynamics upper bounded by two terms} For $j=1,2,3, \ldots, N_g$
	\begin{align}\label{eq:first lem unitary quant eq1}
		\big\|  \me^{-\mi t_j H_{\lo\Cl}} \ket{0}_\lo\ket {0}_\Cl - \ket{t_{j}}_\lo \ket{t_j }_\Cl \big \|_2  \leq & \sum_{k=1}^j \big \| \ket{t_{k}}_\lo \ket{t_{k}}_\Cl - \me^{-\mi t_1 H_{\lo\Cl}}  \ket{t_{k-1}}_\lo \ket{t_{k-1} }_\Cl\big \|_2  \\ 
		\leq & \sum_{k=1}^j \bigg( \big \| \ket{t_k}_\lo \ket{t_{k}}_\Cl - \me^{-\mi t_1 H_{\lo\Cl}^{(k)}}  \ket{t_{k-1}}_\lo \ket{t_{k-1} }_\Cl\big \|_2  \\
		&+    t_1  \max_{x\in[0,t_1]} \big \| \bar{H}_{\lo\Cl}^{(k)} \me^{-\mi x H_{\lo\Cl}^{(k)}} \ket{t_{k-1}}_\lo\ket{t_{k-1}}_\Cl \big \|_2 \bigg ) ,
	\end{align}
	where 
	\begin{align}\label{eq:local Ham defs}
		H_{\lo\Cl}^{(k)}:= H_\Cl + I_\lo ^{(k)}\otimes I_\Cl^{(k)} ,    \qquad  \bar 	H_{\lo\Cl}^{(k)}:=   H_{\lo\Cl}- 	H_{\lo\Cl}^{(k)}= \sum_{\substack{l=1\\ l\neq k}}^{N_g} I_\lo ^{(l)}\otimes I_\Cl^{(l)}  .
	\end{align}
\end{lemma}
The Lemma is useful because it permits us to compute the error by only computing unitary evolution w.r.t. $H_{\lo\Cl}^{(k)}$ rather than $H_{\lo\Cl}$. While the latter does not factorise into a product between system and control, it is readily apparent that the former does. This is of great utility as we will see in later proofs.
\begin{proof}
	The first inequality in \cref{eq:first lem unitary quant eq1} is a direct consequence of \Cref{lem:unitari errs add linararly}. To see this, in \Cref{lem:unitari errs add linararly} we choose $\ket{\Phi_m}=\ket{t_{m}}_\lo\ket{t_m }_\Cl$ and $\Delta_m=\me^{-\mi t_1 H_{\lo\Cl}}$ and note $t_j=j t_1$.   We can now add and subtract an appropriately-chosen term and apply the triangle inequality to achieve
	\begin{align}\label{eq:first lemm 1st proof eq}
		&	\big\|  \me^{-\mi t_j H_{\lo\Cl}} \ket{0}_\lo\ket {0}_\Cl - \ket{t_j}_\lo \ket{t_j }_\Cl \big \|_2  \\&\leq \sum_{k=1}^j \big \| \ket{t_k}_\lo \ket{t_{k}}_\Cl - \me^{-\mi t_1 H_{\lo\Cl}}  \ket{t_{k-1}}_\lo \ket{t_{k-1} }_\Cl\big \|_2  \\ 
		&\leq  \sum_{k=1}^j \Big( \big \| \ket{t_k}_\lo \ket{t_{k}}_\Cl - \me^{-\mi t_1 H_{\lo\Cl}^{(k)}}  \ket{t_{k-1}}_\lo \ket{t_{k-1} }_\Cl\big \|_2 \\
		&\;\;\;\,+   \big \| \me^{-\mi t_1 H_{\lo\Cl}^{(k)}}  \ket{t_{k-1}}_\lo \ket{t_{k-1}}_\Cl - \me^{-\mi t_1 H_{\lo\Cl}}  \ket{t_{k-1}}_\lo \ket{t_{k-1} }_\Cl\big \|_2  \Big).
	\end{align}
	The first term after the inequality is the first term after the inequality in \cref{eq:first lem unitary quant eq1}. We focus on upper bounding the $2\ndd$ term in \cref{eq:first lemm 1st proof eq}.  For this, we start by applying \Cref{lem:unitari errs add linararly} again. This time, we make the association $\ket{\Phi_m}=\me^{-\mi \delta_n H_{\lo\Cl}^{(k)}} \ket{\Phi_{m-1}}= \me^{-\mi \delta_n (m-1) H_{\lo\Cl}^{(k)}} \ket{\Phi_{0}}$,  $\Delta_m= \me^{-\mi \delta_n H_{\lo\Cl} }$, with $\delta_n= t_1/n$.  Thus applying \Cref{lem:unitari errs add linararly} ,  we have for all $n \in \nnp$
	\begin{align}
		&	\big \| \me^{-\mi t_1 H_{\lo\Cl}^{(k)}}  \ket{t_{k-1}}_\lo \ket{t_{k-1}}_\Cl - \me^{-\mi t_1 H_{\lo\Cl}}  \ket{t_{k-1}}_\lo \ket{t_{k-1}}_\Cl\big \|_2 \label{eq:first of differene main lem}  \\
		&\leq \sum_{m=1}^n  \big\|  \me^{-\mi m \delta_n H_{\lo\Cl}^{(k)}}  \ket{t_{k-1}}_\lo \ket{t_{k-1}}_\Cl   - \me^{-\mi \delta_n H_{\lo\Cl}}  \me^{-\mi (m-1) \delta_n H_{\lo\Cl}^{(k)}}  \ket{t_{k-1}}_\lo \ket{t_{k-1}}_\Cl\big\|_2\\
		&= \sum_{m=1}^n  \big\|  \Big( \me^{-\mi \delta_n H_{\lo\Cl}^{(k)}}  -\me^{-\mi \delta_n H_{\lo\Cl}} \Big)  \me^{-\mi (m-1) \delta_n H_{\lo\Cl}^{(k)}}  \ket{t_{k-1}}_\lo \ket{t_{j-1}}_\Cl\big\|_2\\
		&\leq \sum_{m=1}^n \max_{x_n\in[0,t_1]}  \big\|  \Big( \me^{-\mi \delta_n H_{\lo\Cl}^{(k)}}  -\me^{-\mi \delta_n H_{\lo\Cl}} \Big)  \me^{-\mi x_n H_{\lo\Cl}^{(k)}}  \ket{t_{k-1}}_\lo \ket{t_{j-1}}_\Cl\big\|_2\\
		& = n  \max_{x_n\in[0,t_1]} \sqrt{ {}_\lo\!\bra{t_{k-1}} \!{}_\Cl\!\bra{t_{k-1}} \me^{\mi x_n H_{\lo\Cl}^{(k)}} A^\dag(\delta_n) A(\delta_n) \me^{-\mi x_n H_{\lo\Cl}^{(k)}} \ket{t_{k-1}}_\lo \ket{t_{k-1}}_\Cl  }, \label{eq:unitari errs add linararly 10}
	\end{align}
	where $A(y):=  \Big( \me^{-\mi y H_{\lo\Cl}^{(k)}}  -\me^{-\mi y H_{\lo\Cl}} \Big)$. Applying Taylor's remainder theorem to the real function $f:\rr\to\rr$
	\begin{align}
		f(y):= {}_\lo\!\bra{t_{k-1}} \!{}_\Cl\!\bra{t_{k-1}} \me^{\mi x_n H_{\lo\Cl}^{(k)}} A^\dag(y) A(y)  \me^{-\mi x_n H_{\lo\Cl}^{(k)}} \ket{t_{k-1}}_\lo \ket{t_{k-1}}_\Cl  \label{eq:f(y) def for lemma}
	\end{align} about the point $y=0$, we find
	\begin{align}
		f(y) = \,  {}_\lo\!\bra{t_{k-1}} \!{}_\Cl\!\bra{t_{k-1}}\big[ \dot{A}^\dag(y) \dot{A}(y)\big]_{y=0}  \ket{t_{k-1}}_\lo \ket{t_{k-1}}_\Cl  \,y^2 + R(y) \,y^3, \label{eq:fof y for lemma}
	\end{align}
	where dots represent derivatives w.r.t. $y$ and $R:\Re\to\Re$ is a remainder function satisfying $\lim_{y\to0}R(y)=0$. Calculating explicitly the derivatives $\dot{A}^\dag(y) \dot{A}(y)$ and plugging into \cref{eq:unitari errs add linararly 10} and recalling $\delta_n=t_1/n$ we find 
	\begin{align}
		&	\big \| \me^{-\mi t_1 H_{\lo\Cl}^{(k)}}  \ket{t_{k-1}}_\lo \ket{t_{k-1}}_\Cl - \me^{-\mi t_1 H_{\lo\Cl}}  \ket{t_{k-1}}_\lo \ket{t_{k-1}}_\Cl\big \|_2  \\
		& \leq     \max_{x_n\in[0,t_1]} \sqrt{{}_\lo\!\bra{t_{k-1}} \!{}_\Cl\!\bra{t_{k-1}} \me^{\mi x_n H_{\lo\Cl}^{(k)}} \Big(\bar H_{\lo\Cl}^{(k)}\Big)^2  \me^{-\mi x_n H_{\lo\Cl}^{(k)}} \ket{t_{k-1}}_\lo \ket{t_{k-1}}_\Cl  t_1^2+R(t_1/n) t_1^3/n  },
	\end{align}
	for all $n\in\nnp$. Therefore, taking the limit $n\to\infty$ and noting that the remainder term is uniformly bounded in $x_n\in[0,t_1]$, and that the only dependency on $n$ in the non-remainder term in~\cref{eq:fof y for lemma} is via its dependency on $x_n$, we find that
	\begin{align}
		&	\big \| \me^{-\mi t_1 H_{\lo\Cl}^{(k)}}  \ket{t_{k-1}}_\lo \ket{t_{k-1}}_\Cl - \me^{-\mi t_1 H_{\lo\Cl}}  \ket{t_{k-1}}_\lo \ket{t_{k-1}}_\Cl\big \|_2  \\
		&	\leq  \lim _{n\to\infty} 	  \max_{x_n\in[0,t_1]} \sqrt{{}_\lo\!\bra{t_{k-1}} \!{}_\Cl\!\bra{t_{k-1}} \me^{\mi x_n H_{\lo\Cl}^{(k)}} \Big(\bar H_{\lo\Cl}^{(k)}\Big)^2  \me^{-\mi x_n H_{\lo\Cl}^{(k)}} \ket{t_{k-1}}_\lo \ket{t_{k-1}}_\Cl  t_1^2+R(t_1/n) t_1^3/n  },\\
		&	\leq   \sqrt{\left(\lim _{n\to\infty} 	  \max_{x_n\in[0,t_1]} {}_\lo\!\bra{t_{k-1}} \!{}_\Cl\!\bra{t_{k-1}} \me^{\mi x_n H_{\lo\Cl}^{(k)}} \Big(\bar H_{\lo\Cl}^{(k)}\Big)^2  \me^{-\mi x_n H_{\lo\Cl}^{(k)}} \ket{t_{k-1}}_\lo \ket{t_{k-1}}_\Cl  t_1^2\right)+}\\
		&\qquad\overline{+ \left(\lim _{n\to\infty} 	  \max_{x_n\in[0,t_1]}R(t_1/n) \frac{t_1^3}{n} \right)},\\
		&	=   \max_{x\in[0,t_1]} t_1\sqrt{{}_\lo\!\bra{t_{k-1}} \!{}_\Cl\!\bra{t_{k-1}} \me^{\mi x H_{\lo\Cl}^{(k)}} \Big(\bar H_{\lo\Cl}^{(k)}\Big)^2  \me^{-\mi x H_{\lo\Cl}^{(k)}} \ket{t_{k-1}}_\lo \ket{t_{k-1}}_\Cl }.
	\end{align}
	Thus identifying the last line with the corresponding 2-norm, we conclude the proof.
\end{proof}

\subsection{Some additional definitions required for the lemmas of~\Cref{sec:Final technical lemmas}}
Here we specialise the form of the terms in Hamiltonian~\cref{eq:2 body ham def} and the states of the control on $\Cl$. There will still be some free parameters which will only be set at later stages as it becomes required in order to prove the desired results.

The free control term, $H_\Cl$, is chosen identically to that of~\cite{WoodsAut}, namely:
\begin{align}
	{H}_\Cl=\sum_{n=0}^{d-1} n \omega_0\proj{E_n},  \label{eq:H_C definition}
\end{align}
where $\{\ket{E_n}\}_{n=0}^{d-1}$ forms an orthonormal basis for the Hilbert space of the control, $\mathcal{H}_\Cl$. The frequency $\omega_0$ determines both the energy recurrence of the control when no interaction terms are present, $T_0=2 \pi / \omega_0>0$, as $\me^{-\mi \hat{H}_\Cl T_0}=\id_\Cl$.

The states $\{\ket{\theta_k} \}$  are the discrete Fourier Transform basis of the energy basis: For $k\in\zz$
\begin{align}
		\ket{\theta_k}=\frac{1}{\sqrt{d}} \sum_{n=0}^{d-1} \me^{-\mi 2\pi n k/d}  \ket{E_n}. 
		\label{eq:psi def}
\end{align} 
Note that any subset of $d$ consecutive terms forms an orthonormal basis for $\mathcal{H}_\Cl$. We will make use of this redundancy below when defining the stats on $\Cl$.

Let us now define the interaction terms: For $l=1,2,3,\ldots, N_g$ 
\begin{align}
	I_\Cl^{(l)}:=\frac{d}{T_0} \sum_{k\in\mathcal{S}_d (k_0)} I_{\Cl, d}^{(l)}(k) \proj{\theta_k}, \qquad I_{\Cl, d}^{(l)}(x) := \frac{2\pi}{d} \bar V_0\left(\frac{2\pi}{d} x\right) \Bigg{|}_{x_0=x_0^{(l)}} ,\label{eq:def I in terms of bar V0}
\end{align}
with $\mathcal{S}_d(k_0):= \{   k\, |\,  k\in\zz \text{ and } -d/2 \leq k_0-k < d/2 \}$
and where  $x_0^{(l)}:= 2\pi (l-1/2)/N_g$. The function $\bar V_0: \rr\to\rr$ is defined in~\cite{WoodsPRXQ}. It has $x_0\in\rr$ as a parameter in its definition, which is\footnote{Technically, $\bar V_0$ comes with the additional additive factor of $1/\delta d$ is its definition in~\cite{WoodsPRXQ}, but in the current application it is not required and has thus been neglected for simplicity by setting $\delta=1$ and mapping the $1/d$ additive factor to zero. (Importantly, it is readily seen by following the proofs in~\cite{WoodsPRXQ}, that the lemmas from~\cite{WoodsPRXQ} which we will require hold equally well when it is omitted in the definition up to minor modifications which we will highlight as the become come up in the~\doc.  We leave it as an exercise to re-do the derivations in~\cite{WoodsPRXQ}) under this minor modification. When we use a modified result in this~\doc, we will notify the reader of the modification.\label{f:main footnote about modified pot}} 
\begin{align}
\bar V_0(x) =  n A_0 \sum_{p=-\infty}^{+\infty} V_B (n(x-x_0+2\pi p)) ,\label{eq:vbar of x}
\end{align}
where $n>0$, $x_0\in\rr$ and $A_0$ is a normalization constant such that 
\begin{align}
\int_0^{2\pi} \bar V_0(x) dx =1,\label{eq:normalised pt}
\end{align}
and takes on the value  (see F168 in~\cite{WoodsPRXQ})
\begin{align}
A_0= \frac{1}{\int_{-\infty}^\infty \textup{d}x V_B(x) },\label{eq:A0 def}
\end{align}
where
\begin{align}
V_B(\cdot)= \textup{sinc}^{2N}(\cdot)=(\sin(\pi\,\cdot)/(\pi\, \cdot))^{2N}, \qquad N\in\nnp.   \label{eq:def VB}
\end{align}
Notice that $\bar V_0(\cdot)$ is $2\pi$ periodic and as such, the summation in $I_\Cl^{(l)}$ is independent of $k_0\in\rr$. We will later take advantage of this $k_0$ independency and also show how to parametrise $n$ in terms or $d$ to achieve our desired result. As we will see, $N$ on the other hand will be chosen such that it is $d$-independent. Therefore, from~\cref{eq:A0 def} it follows that $A_0$ will also be $d$-independent.

One can use the Weierstrass M test (see Theorem 7.10 in~\cite{rudin1976principles}), to show that the sum in~\cref{eq:vbar of x} converges uniformly. We thus have 
\begin{align}
\int_a^b \dd x	\bar V_0(x) =  n A_0 \sum_{p=-\infty}^{+\infty} \int_a^b \dd x\, V_B (n(x-x_0+2\pi p)),\qquad \forall a ,b,x_0\in\rr, n>0.
\end{align}
We will use this property extensively in proofs in this~\doc.

Let
\begin{align}
	\ket{t_{j}}_\lo := \me^{\mi I_\lo^{(j)} }   \ket{t_{j-1}}_\lo,  \quad j=1,2,\ldots, N_g, \label{eq:idealsed system dynamics}
\end{align}
where we assume w.l.o.g. that the spectrum of $I_\lo^{(j)}$ lies in the interval $(0,2\pi]$. Its eigenvalues are arbitrary so that states $\ket{t_j}_\lo$ and $\ket{t_{j-1}}_\lo$ can be related by any unitary transformation (we will later specify the spectrum in relation to the gate set $\mathcal{G}$). As mentioned in the main text, $\ket{0}_\lo$ is any pure state in $\mathcal{H}_\lo$.

Since the memory is fixed in this section, the definition of $\tilde d(\m_l)$ corresponds to the number of non-identical eigenvalues of $I_\lo^{(l)}$ and 
\begin{align} 
	t_j:= j T_0/ N_g,  
\end{align} for $j=1,2,3, \ldots, N_g$.

As for the states of the control, we use so-called quasi-ideal clock states coming from~\cite{WoodsAut}. In particular, we define $\ket{t_j}_\Cl:=	\ket{\Psi(t_jd/T_0)}_\Cl $ for $j=0,1,2,\ldots, N_g$, where for $t\in\rr$ , $T_0>0$, $d\in\nnp$, 
\begin{align}
	\ket{\Psi(td/T_0)}_\Cl :=  \sum_{k\in\mathcal{S}_d(t d/T_0)} \psi_\textup{nor}\big(td/T_0,k\big)\ket{\theta_k}, \label{eq:qusi idea no pot def}
\end{align}
where
\begin{align}
	\psi_\textup{nor}\big(k_0;k\big):= A_\textup{nor}\, \me^{-\frac{\pi}{\sigma^2} (k-k_0)^2} \me^{\mi 2 \pi n_0 (k-k_0)/d}.\label{eq:initial state}
\end{align}
with $\sigma>0$, $n_0\in(0,d-1)$. The amplitude $A_\textup{nor}$ is defined such that $	\ket{\Psi(td/T_0)}_\Cl $ is normalised. Its large-$d$ scaling is
\begin{align}
	|A_\textup{nor}|^2= \left(\frac{2}{\sigma^2}\right) +\epsilon(d),
\end{align}
where $\epsilon (d) \to 0$ as $d\to\infty$ (under reasonable assumptions about how $\sigma$ depends on $d$ which are satisfied in this~\doc; see{} \app~E in~\cite{WoodsAut} for details).

\subsection{Final technical lemmas}\label{sec:Final technical lemmas}

\begin{lemma}\label{lem:1st secondrary lem for unitary bound}
For all $k\in 1,2,3,\ldots, N_g$, the last terms in~\Cref{lem:quantum dynamics upper bounded by two terms} are upper bounded by 
\begin{align}
	& t_1  \max_{x\in[0,t_1]} \big \| \bar{H}_{\lo\Cl}^{(k)} \me^{-\mi x H_{\lo\Cl}^{(k)}} \ket{t_{k-1}}_\lo\ket{\Psi(t_{k-1}d/T_0)}_\Cl \big \|_2   \leq \\
	& \frac{ 	 \tilde d(\m_k) 4\pi^2 n A_0}{T_0} \Biggl( 3  A_\textup{nor} d \, N_g  \bigg(     \me^{-{\pi}{\frac{d^2}{\sigma^2 (4 N_g)^2} }} +
	\Big( \frac{2 N_g }{\pi^2 n} \Big) ^{2 N}  \,\bigg)\label{eq:1 line inequality ln1st tech lem} \\&
	\quad+  A_\textup{nor} \left(\pi^2-\frac{79}{9}\right) N_g  d \left(\frac{1}{2\pi n} \right)^{2N}    + N_g d \,     \left( 1 + \frac{\pi^2}{3}\right) \varepsilon_v(t_{1},d) \Biggr),\label{eq:2 line inequality ln1st tech lem}
\end{align}
where $n>0$, $N\in\nnp$ and  
\begin{align}
\begin{split}
		\varepsilon_v(t,d)=& |t| \frac{d}{T_0} \left[  \bo\left(\frac{\sigma^3}{\sigma d^{-\epsilon_5+1}}\right)^{1/2} +\bo\left(\frac{d^2}{\sigma^2}+2 C_0 n\right)   \right] \exp{\left(-\frac{\pi}{4} \frac{\alpha_0^2}{(1+d^{\epsilon_5}/\sigma)^2} d^{2\epsilon_5}\right)}\\
	&+  \bo\left( |t| \frac{d^2}{\sigma^2}+1  \right) \me^{-\frac{\pi}{4} \frac{d^2}{\sigma^2}} + \bo\left(\me^{-\frac{\pi}{2} \sigma^2}\right) \text{ as } d\to\infty,\quad (0,d)\ni \sigma\to\infty 
\end{split}\label{eq:var e v def}
\end{align}
where $C_0>0$, $\alpha_0>0$ , $\epsilon_5>0$ are fixed constants and 
\begin{align}
	A_\textup{nor} \leq\left(\frac{2}{\sigma^2}\right)^{1 / 4}+\sqrt{\frac{\bar{\epsilon}_1+\bar{\epsilon}_2}{\frac{\sigma}{\sqrt{2}}\left(\frac{\sigma}{\sqrt{2}}-\bar{\epsilon}_1-\bar{\epsilon}_2\right)}}, \label{eq:A nor up bound}
\end{align}
with 
\begin{align}
	\bar{\epsilon}_2:= \frac{\sigma}{\sqrt{2}} \frac{2 e^{-\frac{\pi \sigma^2}{2}}}{1-e^{-\pi \sigma^2}}, \quad  \bar{\epsilon}_1:= \frac{2 e^{-\frac{\pi d^2}{2 \sigma^2}}}{1-e^{-\frac{2 \pi d}{\sigma^2}}}.
\end{align}
\end{lemma}

\begin{proof} For all $x\in [0,t_1]$

\begin{align}
	&  \big \| \bar{H}_{\lo\Cl}^{(k)} \me^{-\mi x H_{\lo\Cl}^{(k)}} \ket{t_{k-1}}_\lo\ket{\Psi(t_{k-1}d/T_0)}_\Cl \big \|_2 \label{eq:send tem proof l1}\\ &  \leq
	\sum_{\substack{l=1\\ l\neq k}}^{N_g}
	\big \|  I_\lo ^{(l)}\otimes I_\Cl^{(l)}  \me^{-\mi x H_{\lo\Cl}^{(k)}} \ket{t_{k-1}}_\lo\ket{\Psi(t_{k-1}d/T_0)}_\Cl \big \|_2 \label{eq:send tem proof l2}	\\
	&  \leq
	\sum_{j=1}^{\tilde d(\m_k)  } 	 \sum_{\substack{l=1\\ l\neq k}}^{N_g}  \big|A_j^{(k-1)} \big| 
	\big \|  I_\lo ^{(l)}\otimes I_\Cl^{(l)}   \ket{\theta_j^{(k-1)}}_\lo \me^{-\mi x H_{\Cl}^{(k)}\big(\theta_j^{(k)}\big)}\ket{\Psi(t_{k-1}d/T_0)}_\Cl \big \|_2   \label{eq:send tem proof l3}	 \\
	& =
	\sum_{j=1}^{\tilde d(\m_k) } 	 \sum_{\substack{l=1\\ l\neq k}}^{N_g}  \big|A_j^{(k-1)} \big| 
	\big \|  I_\lo ^{(l)} \ket{\theta_j^{(k-1)}}_\lo\big\|_2 \big\|  I_\Cl^{(l)}   \me^{-\mi x H_{\Cl}^{(k)}\big(\theta_j^{(k)}\big)}\ket{\Psi(t_{k-1}d/T_0)}_\Cl \big \|_2 \label{eq:send tem proof l4}\\
	&\leq  	 \sum_{\substack{l=1\\ l\neq k}}^{N_g} 	 \sum_{j=1}^{\tilde d(\m_k) } 	  \big|A_j^{(k-1)} \big| 
	\big \|  I_\lo ^{(l)} \big\|_2  \max_{\vartheta\in[-\pi,\pi]} \big\|  I_\Cl^{(l)}   \me^{-\mi x H_{\Cl}^{(k)}(\vartheta)}\ket{\Psi(t_{k-1}d/T_0)}_\Cl \big \|_2  \label{eq:send tem proof l5}\\
	& \leq 2\pi	 \tilde d(\m_k)	 \sum_{\substack{l=1\\ l\neq k}}^{N_g} 
	\max_{\vartheta\in[-\pi,\pi]} \big\|   I_\Cl^{(l)}   \me^{-\mi x H_{\Cl}^{(k)}(\vartheta)}\ket{\Psi(t_{k-1}d/T_0)}_\Cl \big \|_2  \label{eq:send tem proof l6}
\end{align}
where in line \labelcref{eq:send tem proof l2} we have used~\cref{eq:local Ham defs} together with the triangle inequality. For line~\labelcref{eq:send tem proof l3}, we have used the decomposition 
\begin{align}
	\ket{t_{k-1}}_\lo= \sum_{j=1}^{\tilde d(\m_k) } A_j^{(k-1)} \ket{\theta_j^{(k)}}_\lo, \label{eq:interaction state decomposition}
\end{align}
where $\ket{\theta_j^{(k)}}_\lo$ belongs to the subspace spanned by the vectors of eigenvalue $\theta_j^{(k)}$ for the operator $I_\lo^{(k)}$. Thus $\tilde d(\m_k) \leq d_\lo^{(k)}$, where recall that $\tilde d(\m_k)$ is the number of non-identical eigenvalues of $I_\lo^{(k)}$.
We have also defined
\begin{align}
	H_{\Cl}^{(k)}\big(\gamma\big):= H_\Cl + \gamma\, I_\Cl^{(k)}, \label{eq:Reduced H clock}
\end{align}
$\gamma\in\rr$ (c.f. def. of $H_{\lo\Cl}^{(k)}$ in \cref{eq:local Ham defs}).  
For line~\labelcref{eq:send tem proof l4}, we have used  the fact that the 2-norm of a tensor product is the product of the 2-norms. In line~\labelcref{eq:send tem proof l5}, we have used $ \|  I_\lo ^{(l)} \|_2$ to denote the 2-norm-induced operator norm. In line~\labelcref{eq:send tem proof l6}, we have used the assumption $ \|  I_\lo ^{(l)} \|_2\leq 2\pi$ for all $l=1,2,\ldots, N_g$. 

To continue, we will need the particular form of the interaction terms introduced in~\cref{eq:def I in terms of bar V0} and we will need to recall Theorem IX.1 (\textit{Moving the clock through finite time with a potential}) from \cite{WoodsAut},  which states 
\begin{align}
	\me^{-\mi t (H_\Cl + \hat V_d)} \ket{\bar \Psi_\textup{nor}(k_0, \Delta)}_\Cl = \ket{\bar \Psi_\textup{nor}(k_0+t \, d / T_0, \Delta+t \,d / T_0)}_\Cl +\ket{\epsilon}_\Cl, \quad \| \ket{\epsilon}_\Cl\|_2 \leq \varepsilon_v(t,d),\label{eq:them IX}
\end{align}
where 
\begin{align}\label{eq:ket psi def}
	\ket{\bar \Psi_\textup{nor}(k_0, \Delta)}_\Cl :=  \sum_{k\in\mathcal{S}_d(k_0)}\me^{-\mi 	 \int_{k-\Delta}^k \textup{d}y V_d(y)} \psi_\textup{nor}\big(k_0;k\big)\ket{\theta_k}_\Cl,
\end{align}

\begin{align}
	\hat V_d:=\frac{d}{T_0} \sum_{k=0}^{d-1} V_d(k) \proj{\theta_k}_\Cl, \qquad	V_d(x) = \frac{2\pi}{d} \bar V_0 \left(\frac{2\pi}{d} x\right).
\end{align}
We can apply this theorem to approximate $\me^{-\mi x H_{\Cl}^{(k)}(\vartheta)}\ket{\Psi(t_{k-1}d/T_0)}_\Cl$ appearing in line~\labelcref{eq:send tem proof l6} by identifying $x$ with $t$, and $\hat V_d$ with $\vartheta I_\Cl^{(k)}$.   To do so, first note that by definition it follows that $\ket{\Psi_\textup{nor}(t_{k-1}d/T_0)}_\Cl =  \ket{\bar\Psi_\textup{nor}(t_{k-1} d/T_0, 0)}_\Cl$  (recall definition~\labelcref{eq:qusi idea no pot def}).

Continuing from line~\labelcref{eq:send tem proof l6}, but now maximizing over $x$ we thus find
\begin{align}
	&  \max_{x\in[0,t_1]} 2 \pi	 \tilde d(\m_k) 	 \sum_{\substack{l=1\\ l\neq k}}^{N_g} 
	\max_{\vartheta\in[0,2\pi]} \big\|   I_\Cl^{(l)}   \me^{-\mi x H_{\Cl}^{(k)}(\vartheta)}\ket{\Psi(t_{k-1}d/T_0})_\Cl \big \|_2\\
	&\leq  \max_{x\in[0,t_1]} 	2 \pi	 \tilde d(\m_k) 	 \sum_{\substack{l=1\\ l\neq k}}^{N_g} 
	\max_{\vartheta\in[0,2\pi]} \Big\|   \frac{d}{T_0} \sum_{q\in\mathcal{S}_d ([t_{k-1}+x]d/T_0)} I_{\Cl, d}^{(l)}(q) \proj{\theta_q}  \left(  \ket{\bar \Psi_\textup{nor}([t_{k-1}+x]d/T_0 , x d / T_0)}_\Cl +\ket{\varepsilon_C(x,d)} \right) \Big \|_1  \label{eql:1 for lem 3 second eq}
	\\& \leq  \max_{x\in[0,t_1]} 	\frac{ 2\pi	 \tilde d(\m_k) d}{T_0} \sum_{\substack{l=1\\ l\neq k}}^{N_g}  \max_{\vartheta\in[0,2\pi]}	 \sum_{q\in\mathcal{S}_d ([t_{k-1}+x]d/T_0)}  \Bigg(\,\Big|      I_{\Cl,d}^{(l)} (q) \psi_\textup{nor} \left([t_{k-1}+x]d/T_0,q\right) \me^{-\,i \vartheta \int_{q-xd/T_0}^q I_{\Cl,d}^{(k)}(y)\textup{d}y}    \Big| \label{eql:1 for lem 3 3rd eq}\\&\quad
	+ \Big |    I_{\Cl,d}^{(l)} (q) \braket{\theta_q| \varepsilon_C(x,d)}       \Big| \,\Bigg) \label{eql:1 for lem 3 3rd eq b}
	\\& \leq  	\frac{ 	 \tilde d(\m_k) 4\pi^2}{T_0} \Biggl(\max_{x\in[0,t_1]}  \Biggl[ \sum_{\substack{l=1\\ l\neq k}}^{N_g}  	 \sum_{q\in\mathcal{S}_d ([t_{k-1}+x]d/T_0)}   \Big|   n A_0 \sum_{p=-\infty}^{+\infty} V_B \left( n (2\pi q/d-x_0^{(l)}  -2\pi p )  \right) A_\textup{nor}\, \me^{-\frac{\pi}{\sigma^2} \left(q-[t_{k-1} +x]d/T_0\right)^2}  \Big| \Biggr]   \label{eql:1 for lem 4 3rd eq}\\
	&\quad+  \frac{ d}{2\pi}    \sum_{\substack{l=1\\ l\neq k}}^{N_g}  	 \sum_{q\in\mathcal{S}_d ([t_{k-1}+x]d/T_0)}    \frac{2\pi}{d} \left( \max_{y\in[0,2\pi]}  \bar V_0(y)\right) \varepsilon_v(t_{1},d) \Biggr) \label{eql:1 for lem 4 3rd eq b} 
	\\& \leq  	\frac{ 	 \tilde d(\m_k) 4\pi^2}{T_0} \Biggl( \Biggl[ \max_{x'\in[0,1]} \sum_{\substack{l=1\\ l\neq k}}^{N_g}  d	 \max_{q\in[-d/2+(k-1+x')d/N_g, \,\, d/2+(k-1+x')d/N_g ]}  \label{eql:1 for lem 4 3rd eq new}\\
	&\quad  \Big|   n A_0 \sum_{p=-\infty}^{+\infty} V_B \left( n (2\pi q/d-x_0^{(l)}  -2\pi p )  \right) A_\textup{nor}\, \me^{-\frac{\pi}{\sigma^2} \left(q-[t_{k-1} +x' t_1]d/T_0\right)^2}  \Big| \Biggr]  \\
	&\quad+  N_g d \,n A_0    \left( 1 + \sum_{p\in\zz\backslash\{0\}} \frac{1}{p^{2N}}\right) \varepsilon_v(t_{1},d) \Biggr) \label{eql:1 for lem 4 4rd eq b}
	\\& \leq \frac{ 	 \tilde d(\m_k) 4\pi^2  n A_0}{T_0} \Biggl(  \Bigg[A_\textup{nor}  \max_{x'\in[0,1]}  \sum_{\substack{l=1\\ l\neq k}}^{N_g}  d \max_{q'\in[-d/2,d/2]}     \sum_{p=-\infty}^{+\infty} V_B \Big( 2\pi n \big( q'/d +[k-l+x'-1/2]/N_g- p \big)  \Big)  \me^{-\frac{\pi}{\sigma^2} {q'}^2}  \Biggr]  \label{eql:1 for lem 4 4rd eq}\\
	&\quad   + N_g d \,     \left( 1 + \frac{\pi^2}{3}\right) \varepsilon_v(t_{1},d) \Biggr)  \label{eql:1 for lem 4 4rd eq b 2}
	\\& \leq 	\frac{ 	 \tilde d(\m_k) 4\pi^2 n  d A_0}{T_0} \Bigg( \Biggl[ A_\textup{nor}    \max_{x'\in[0,1]} \sum_{\substack{l=1\\ l\neq k}}^{N_g}  	\max_{q'\in[-d/2,d/2]}     \sum_{p\in\{0,\pm 1\}}   V_B \Big( 2\pi n \big( q'/d +[k-l+x'-1/2]/N_g -p\big)  \Big)\me^{-\frac{\pi}{\sigma^2} {q'}^2} \Biggr] \\&
	\quad+  \Biggl[A_\textup{nor}  \sum_{\substack{l=1\\ l\neq k}}^{N_g}  	 \max_{q'\in[-d/2,d/2]}   \sum_{p\in\zz\backslash \{0,\pm 1\}}\Big| 2\pi n \big( 1/2- p \big)  \Big|^{-2N}   \Biggr]  + N_g  \,     \left( 1 + \frac{\pi^2}{3}\right) \varepsilon_v(t_{1},d) \Biggr),  \label{eql:1 for lem 4 5rd eq b}
	\\& \leq 	\frac{ 	 \tilde d(\m_k) 4\pi^2 n d A_0}{T_0} \Bigg( \Biggl[ A_\textup{nor}    \max_{x'\in[0,1]} \sum_{\substack{l=1\\ l\neq k}}^{N_g}  	 \max_{q'\in[-d/2,d/2]}     \sum_{p\in\{0,\pm 1\}}    V_B \Big( 2\pi n \big( q'/d +[k-l+x'-1/2]/N_g -p\big)  \Big)\me^{-\frac{\pi}{\sigma^2} {q'}^2} \Biggr] \\&
	\quad+  \Biggl[A_\textup{nor}  N_g   \left(\frac{1}{2\pi n} \right)^{2N}  \sum_{p\in\zz\backslash \{0,\pm 1\}}\big|  1/2- p   \big|^{-2N}   \Biggr]  + N_g  \,     \left( 1 + \frac{\pi^2}{3}\right) \varepsilon_v(t_{1},d) \Biggr),  \label{eql:1 for lem 4 6rd eq b} 
	\\& \leq 	\frac{ 	 \tilde d(\m_k) 4\pi^2 n d A_0}{T_0} \Bigg( \Biggl[ A_\textup{nor}    \max_{x'\in[0,1]} \sum_{\substack{l=1\\ l\neq k}}^{N_g}  \max_{q'\in[-d/2,d/2]}    \sum_{p\in\{0,\pm 1\}}    V_B \Big( 2\pi n \big( q'/d +[k-l+x'-1/2]/N_g -p\big)  \Big)\me^{-\frac{\pi}{\sigma^2} {q'}^2} \Biggr] \\&
	\quad+  A_\textup{nor} \left(\pi^2-\frac{79}{9}\right) N_g   \left(\frac{1}{2\pi n} \right)^{2N}    + N_g \,     \left( 1 + \frac{\pi^2}{3}\right) \varepsilon_v(t_{1},d) \Biggr),  \label{eql:1 for lem 4 7rd eq b} 
\end{align} 
where in line~\labelcref{eql:1 for lem 3 second eq}, we have inserted the definition of $I^{(l)}_\Cl$ and chosen $k_0= [t_{k-1}+x]d/T_0$ in $\mathcal{S}_d(k_0)$, followed by applying~\cref{eq:them IX} and using the fact that the 1-norm upper bounds the 2-norm.   In line~\labelcref{eql:1 for lem 3 3rd eq}, we have used definition~\labelcref{eq:ket psi def} and the triangle inequality. In line~\labelcref{eql:1 for lem 4 3rd eq} we have first removed the maximization over $\vartheta$ since it is $\vartheta$-independent and then substituted in~\cref{eq:vbar of x}. In line~\labelcref{eql:1 for lem 4 3rd eq b} we have used the definition of $\varepsilon_v(\cdot,\cdot)$. In line~\labelcref{eql:1 for lem 4 3rd eq new}  we have defined $x' =x/ t_1\in[0,1]$ and used the fact that $q\in\mathcal{S}_d([t_{k-1}+x]d/T_0)$ is equivalent to  $q\in\{ \lceil -d/2+[k-1+x']d/N_g\rceil, \lceil -d/2+[k-1+x']d/N_g\rceil+1, \lceil -d/2+[k-1+x']d/N_g\rceil+2,\ldots, \lceil -d/2+[k-1+x']d/N_g\rceil+d-1\}$. Since $\lceil -d/2+[k-1+x']d/N_g\rceil+d-1=\lceil -d/2+d-1+[k-1+x']d/N_g\rceil\leq d/2+[k-1+x']d/N_g$, we have that $q$ takes on $d$ values in the interval $q\in[-d/2+[k-1+x']d/N_g, d/2+[k-1+x']d/N_g]$. 
In line \labelcref{eql:1 for lem 4 4rd eq} we have first made the change of variables $q'=q-[t_{k-1} +x' t_1]d/T_0$ so that $q'\in[-d/2,d/2]$, followed by substituting for $t_{k-1}$ and $x_0^{(l)}$. In line~\labelcref{eql:1 for lem 4 4rd eq b} we used definition~\labelcref{eq:vbar of x} to upper bound the maximization over $y$.   For line~\labelcref{eql:1 for lem 4 4rd eq b 2} we have used $ \sum_{p\in\zz\backslash\{0\}} \frac{1}{p^{2N}} \leq  \sum_{p\in\zz\backslash\{0\}} \frac{1}{p^{2}}=\pi^2/3$ for all $N\in \nnp$.

For line~\labelcref{eql:1 for lem 4 5rd eq b} have used the bound  $\sum_{p\in\zz\backslash\{0,1\}} V_B(2\pi n (x-p))\leq (2\pi n)^{-2N} \sum_{p\in\zz\backslash\{0\}}  (1/2-p)^{-2N}$ for all $x\in [-1/2,3/2]$, $n>0$, (which follows from the function's definition,~\cref{eq:vbar of x}), and the observation that $q'/d+[k-l+x'-1/2]/N_g \in [-1/2,3/2]$ for all $x'\in[0,1]$, $q'\in[-d/2,d/2]$ and $l,k\in {1,2,\ldots, N_g}$ s.t. $l\neq k$.

We will now continue with the proof. We start by focusing on the $p=0,1$ terms in the summation.  To bound these terms, we will separate the range of $q'\in[-d/2,d/2]$ into three intervals, two ``tail'' intervals and one ``centre'' interval and proceed to bound the central and tail intervals separately. In particular, observe that $q'/d$ is contained in the intervals $q'/d\in[-1/2,1/2] = [-1/2, -1/(4 N_g)]\cup [-1/(4 N_g),1/(4 N_g)]\cup [1/(4 N_g),1/2]$.  For the $p=0$ term in brackets in line~\labelcref{eql:1 for lem 4 4rd eq} we find
\begin{align}
	\max_{x'\in[0,1]}   A_\textup{nor}   \sum_{\substack{l=1\\ l\neq k}}^{N_g}  	 \sum_{q'\in[-d/2,d/2]}\sum_{p\in\{0,\pm 1\}} &  V_B \Big( 2\pi n \big( q'/d +[k-l+x'-1/2]/N_g-p\big)  \Big)  \me^{-\frac{\pi}{\sigma^2} {q'}^2}   \label{eql:1 for lem 4 5rd eq}\\
	\leq
	\max_{x'\in[0,1]}   A_\textup{nor}   \sum_{\substack{l=1\\ l\neq k}}^{N_g} \sum_{p\in\{0,\pm 1\}}   \Bigg(  
	&\max_{q''\in  [-1/2, -1/(4 N_g)]\cup [1/(4 N_g),1/2]}  
	V_B \Big( 2\pi n \big( q'' +[k-l+x'-1/2]/N_g-p\big)  \Big)  \me^{-{\pi}{\frac{d^2}{\sigma^2} {q''}^2}} \\
	&+\max_{q'''\in[-1/(4 N_g),1/(4 N_g)]} \Big( \frac{1}{2\pi^2 n | q''' +(k-l+x'-1/2)/N_g -p|} \Big) ^{2 N} \me^{-{\pi}{\frac{d^2}{\sigma^2} {q'''}^2}}   \Bigg)
	\label{eql:1 for lem 4 6rd eq}\\
	\leq
	\max_{x'\in[0,1]}  A_\textup{nor}  \sum_{\substack{l=1\\ l\neq k}}^{N_g}  \Bigg(   3\big(\max_{y\in\rr}&V_B (y)\big)\,  \me^{-{\pi}{\frac{d^2}{\sigma^2 (4 N_g)^2} }} +\max_{q''''\in[-1,1]}\sum_{p\in\{0,\pm 1\}} 
	\Big( \frac{N_g }{2\pi ^2n | q''''/4 +(k-l+x'-1/2)-p N_g|} \Big) ^{2 N}  \Bigg)\qquad \label{eq:new correction}\\
	\leq 3  A_\textup{nor}  N_g  \Bigg(     \me^{-{\pi}{\frac{d^2}{\sigma^2 (4 N_g)^2} }} +
	&\Big( \frac{2 N_g }{\pi^2 n} \Big) ^{2 N}  \Bigg), \label{eql:1 for lem 4 7rd eq}
\end{align}
Where in line~\labelcref{eql:1 for lem 4 6rd eq} we have used the bound $V_B(x)\leq (1/|\pi x|)^{2N}$ for the first term (which follows from the function's definition,~\cref{eq:vbar of x}).  In line~\labelcref{eql:1 for lem 4 7rd eq} we have used the bound $V_B(x)\leq 1$ $\forall x\in$ for the first term. For the second term, we have noted that for $p\in\{0,\pm 1\}$, $k,l\in1,2,\ldots, N_g$ s.t. $l\neq k$ we have $k-l-pN_g\in (-\infty,-1]\cup[1,+\infty)$. Therefore, $| q''''/4 +(k-l+x'-1/2)-p N_g|\geq 1/4$.  Observe that it is critical for this argument that the summation is restricted to $l\neq k$, since for $l=k$ the denominator takes on the value zero and the second term is infinite. 

Finally, the upper bound for~\cref{eq:A nor up bound} is derived in~\cite[Appendix E.1.1, pg 208]{WoodsAut}.
\end{proof}

\begin{lemma}\label{lem:2st secondrary lem for unitary bound}
For all $k\in 1,2,3,\ldots, N_g$, the first term in~\Cref{lem:quantum dynamics upper bounded by two terms} is upper bounded by 
\begin{align}
	& \big \| \ket{t_k}_\lo \ket{\Psi(t_{k}d/T_0)}_\Cl - \me^{-\mi t_1 H_{\lo\Cl}^{(k)}}  \ket{t_{k-1}}_\lo \ket{\Psi(t_{k-1}d/T_0) }_\Cl\big \|_2\\
	&\leq  \sqrt{2\varepsilon_v(t_{k-1},d)} +\pi \sqrt{2d  A_\textup{nor}} \left( \me^{-\frac{\pi}{16} \left(\frac{d}{\sigma N_g}\right)^2} +     4\pi n A_0\left(\left(\frac{2N_g}{\pi^2\,n} \right)^{2N} + \frac{\pi^2}{3}\left(\frac{1}{2\pi^2 n}\right)^{2N} \right) \right), \label{eq:line: 1 of upper bound of tech lmee 2}
\end{align}
where $\varepsilon_v$ and $A_\textup{nor}$ are defined in~\Cref{lem:1st secondrary lem for unitary bound} and $n>0$, $N\in\nnp$.
\begin{proof}
	\begin{align}
		&\big \| \ket{t_k}_\lo \ket{\Psi(t_{k}d/T_0)}_\Cl - \me^{-\mi t_1 H_{\lo\Cl}^{(k)}}  \ket{t_{k-1}}_\lo \ket{\Psi(t_{k-1}d/T_0) }_\Cl\big \|_2  \\
		& = \sqrt{2\bigg(1-\Re\Big[  \big(\bra{t_k}_\lo \bra{\Psi(t_{k}d/T_0)}_\Cl\big) \big( \me^{-\mi t_1 H_{\lo\Cl}^{(k)}}  \ket{t_{k-1}}_\lo \ket{\Psi(t_{k-1}d/T_0) }_\Cl\big)\Big]\bigg)}\label{eq:line 2 3rd proof}\\
		& = \sqrt{2\bigg(1-\Re\Big[\sum_{j_1,j_2=1}^{\tilde d(\m_k)  } A_{j_1}^{(k-1)*} A_{j_2}^{(k-1)}\me^{\mi \theta_{j_1}^{(k)}} \big(\bra{\theta_{j_1}^{(k)}}_\lo \bra{\Psi(t_{k}d/T_0)}_\Cl\big) \big( \me^{-\mi t_1 H_{\lo\Cl}^{(k)}}  \ket{\theta_{j_2}^{(k)}}_\lo \ket{\Psi(t_{k-1}d/T_0) }_\Cl\big)\Big]\bigg)} \label{eq:line 3 3rd proof}\\
		& = \sqrt{2\bigg(1-\sum_{j=1}^{\tilde d(\m_k) } |A_{j}^{(k-1)}|^2 \,\,\Re\Big[\me^{\mi \theta_{j}^{(k)}}  \bra{\Psi(t_{k}d/T_0)}_\Cl  \me^{-\mi t_1 H_{\Cl}^{(k)}(\theta_j^{(k)})}   \ket{\Psi(t_{k-1}d/T_0) }_\Cl\Big]\bigg)} \label{eq:line 4 3rd proof}\\
		& = \sqrt{2\bigg(1-\sum_{j=1}^{\tilde d(\m_k) } |A_{j}^{(k-1)}|^2 \,\,\Re\Big[  \sum_{l\in\mathcal{S}(t_k d/T_0)} \left|  \psi_\textup{nor}(t_k d/T_0;l) \right|^2 \me^{\mi \theta_{j}^{(k)} \big(1-\int_{l-t_1 d/T_0}^l I_{\Cl,d}^{(k)}(y)\textup{d}y\big)}            +\me^{\mi \theta_{j}^{(k)}}  \braket{\Psi(t_{k}d/T_0)| \varepsilon_C(t_{k-1},d)}_\Cl \Big] \bigg)} \label{eq:line 5 3rd proof}\\
		& = \sqrt{2}\sqrt{1-\sum_{j=1}^{\tilde d(\m_k) } |A_{j}^{(k-1)}|^2 \bigg(\,\Re\Big[  \sum_{l\in\mathcal{S}(t_k d/T_0)} \left|  \psi_\textup{nor}(t_k d/T_0;l) \right|^2 \me^{\mi \theta_{j}^{(k)} \big(1-\int_{l-t_1 d/T_0}^l I_{\Cl,d}^{(k)}(y)\textup{d}y\big)}\Big]      -\varepsilon_v(t_{k-1},d) \bigg)} \label{eq:line 6 3rd proof}\\
		& = \sqrt{2}\sqrt{1+\max_{\vartheta\in[0, 2\pi]}\bigg(-\Re\Big[  \sum_{l\in\mathcal{S}(t_k d/T_0)} \left|  \psi_\textup{nor}(t_k d/T_0;l) \right|^2 \me^{\mi \vartheta \big(1-\int_{l-t_1 d/T_0}^l I_{\Cl,d}^{(k)}(y)\textup{d}y\big)}\Big]\bigg)      +\varepsilon_v(t_{k-1},d) } \label{eq:line 7 3rd proof}\\
		& \leq \sqrt{2}\sqrt{\varepsilon_v(t_{k-1},d)+ \sum_{l\in\mathcal{S}(t_k d/T_0)} \left|  \psi_\textup{nor}(t_k d/T_0;l) \right|^2 4\pi^2\left(1-\int_{l-t_1 d/T_0}^l I_{\Cl,d}^{(k)}(y)\textup{d}y\right)^2   } \label{eq:line 8 3rd proof}\\
		& \leq \sqrt{2}\sqrt{\varepsilon_v(t_{k-1},d)+2d\pi^2\max_{q\in [-1/2,1/2]} \left|  \psi_\textup{nor}(t_k d/T_0;t_kd/T_0-d q) \right|^2 \left(1-\int_{ t_{-1/2} d/T_0-q d}^{t_{1/2}d/T_0-qd} I_{\Cl,d}^{(k)}(y'+x_0^{(k)}d/(2\pi))\textup{d}y'\right)^{\!\!2}   } \label{eq:line 9 3rd proof}\\
		& \leq \sqrt{2}\sqrt{\varepsilon_v(t_{k-1},d)+2d\pi^2\max_{q\in [-1/2, -1/(4 N_g)]\cup [1/(4 N_g),1/2]}   A_\textup{nor}\, \me^{-\frac{\pi d^2}{\sigma^2} q^2} \left(1-\frac{2\pi}{d}\int_{-d/(2N_g)-q d}^{d/(2N_g)-qd}\! \!\bar V_0 \left(2\pi y'/d+x_0^{(k)}\right)\!\!\Bigg{|}_{x_0=x_0^{(k)}}\!\!\!\!\!\!\!\!\!\!  \textup{d}y'\right)^{\!\!\!2}    } \label{eq:line 10 3rd proof}\\
		&\,\quad\qquad \overline{+\,2d\pi^2\max_{q\in [-1/(4 N_g),1/(4 N_g)]}   A_\textup{nor}\, \me^{-\frac{\pi d^2}{\sigma^2} q^2} \left(1-\frac{2\pi}{d}\int_{-d/(2N_g)-q d}^{d/(2N_g)-qd} \! \!\bar V_0 \left(2\pi y'/d+x_0^{(k)}\right)\!\!\Bigg{|}_{x_0=x_0^{(k)}}\!\!\!\!\!\!\!\!\!\!  \textup{d}y'\right)^{\!\!2}    }\nonumber\\
		& \leq \sqrt{2}\sqrt{\varepsilon_v(t_{k-1},d)+2d\pi^2   A_\textup{nor}\, \me^{-\frac{\pi}{8} \left(\frac{d}{\sigma N_g}\right)^2} +\,2d\pi^2\max_{q\in [-1,1]}   A_\textup{nor}\, \left(1-\frac{2\pi}{N_g}\int_{-1/2-q/4}^{1/2-q/4} \! \!\bar V_0 \left(2\pi y''/N_g+x_0^{(k)}\right)\!\!\Bigg{|}_{x_0=x_0^{(k)}}\!\!\!\!\!\!\!\!\!\!  \textup{d}y''\right)^{\!\!2}    } \label{eq:line 11 3rd proof}
	\end{align} 
	In line~\labelcref{eq:line 2 3rd proof} we have used $\Re[\cdot]$ to denote the real part. In line~\labelcref{eq:line 3 3rd proof}, we have used~\cref{eq:idealsed system dynamics} followed by~\cref{eq:interaction state decomposition}. In line~\labelcref{eq:line 4 3rd proof} we have used definition~\labelcref{eq:Reduced H clock}. In line~\labelcref{eq:line 5 3rd proof}, we have used the theorem displayed in~\cref{eq:them IX}, taking into account the definition $\ket{\Psi(\cdot)}_\Cl =  \ket{\bar\Psi_\textup{nor}(\cdot, 0)}_\Cl$. In line~\labelcref{eq:line 6 3rd proof} we have used the fact that for all $c\in\cc$, $|\Re[c]|\leq |c|$ and that $\| \ket{\varepsilon_C(t_{k-1},d)}_\Cl\|_2\leq \varepsilon_v(t_{k-1},d)$. In line~\labelcref{eq:line 7 3rd proof} we have used that $\sum_{j=1}^{\tilde d(\m_k) } |A_{j}^{(k-1)}|^2=1$ due to state normalization. In line~\labelcref{eq:line 8 3rd proof} we have taken the real part of the term in square brackets and used $-\cos(\theta)\leq \theta^2-1$ for all $\theta\in\rr$.  ln line~\labelcref{eq:line 9 3rd proof} we have made the change of variable $y'=y-d x_0^{(k)}/(2\pi)$ which shifts $I_{\Cl,d}^{(k)}(y)$ to be centred at zero. We have also defined $\tilde q:= t_k/T_0-l/d$ and noted that $l\in\mathcal{S}(t_kd/T_0)$ implies $-1/2\leq \tilde q <1/2$ and finally upper bounded the summation for a maximization over the set $[-1/2,1/2]$. In line~\labelcref{eq:line 10 3rd proof} we have substituted in the definition of the functions, followed by upper bounding the maximization over $q\in[-1/2,1/2]$ as the sum of maximizations over the sub-intervals $[-1/2, -1/(4 N_g)]\cup[1/(4 N_g)],1/2]$ and $[-1/(4 N_g),1/(4 N_g)]$. 
	In line~\labelcref{eq:line 11 3rd proof}, for the first term, we have used that $\bar V_0$ is a non-negative, $2\pi$-periodic  function integrated over an interval less than $2\pi$. For the second term we performed a change of variable. 
	
	We will now derive an alternative expression for the term in brackets in line~\labelcref{eq:line 11 3rd proof} before continuing. From~\cref{eq:normalised pt} it follows:
	\begin{align}
		1=	\int_{-\pi}^{\pi} \bar V_0\left(x+x_0^{(k)}\right) \!\!\Bigg{|}_{x_0=x_0^{(k)}}\!\!\!\!\!\!\!\!\!\!  dx\,\, &=  \frac{2\pi}{N_g} \left(   \int_{-N_g/2}^{q/4-1/2}+ \int_{q/4-1/2}^{q/4+1/2} + \int_{q/4+1/2}^{N_g/2}     \right) \bar V_0\left(\frac{2\pi}{N_g} y+x_0^{(k)}\right) \!\!\Bigg{|}_{x_0=x_0^{(k)}}\!\!\!\!\!\!\!\!\!\! \textup{d}y \\
		&= \frac{2\pi}{N_g} \left(   \int_{-q/4+1/2}^{N_g/2}+ \int_{q/4-1/2}^{q/4+1/2} + \int_{q/4+1/2}^{N_g/2}     \right) \bar V_0\left(\frac{2\pi}{N_g} y+x_0^{(k)}\right) \!\!\Bigg{|}_{x_0=x_0^{(k)}}\!\!\!\!\!\!\!\!\!  \textup{d}y\,, \label{eq:normalised pt 2}
	\end{align}
	where we have used the property $\bar V_0\left(\frac{2\pi}{N_g} y+x_0^{(k)}\right)=\bar V_0\left(-\frac{2\pi}{N_g} y+x_0^{(k)}\right)$ which follows form~\cref{eq:vbar of x}. Therefore,
	\begin{align}
		&\max_{q\in [-1,1]}  \left(1-\frac{2\pi}{N_g}\int_{-1/2-q/4}^{1/2-q/4} \! \!\bar V_0 \left(2\pi y''/N_g+x_0^{(k)}\right)\!\!\Bigg{|}_{x_0=x_0^{(k)}}\!\!\!\!\!\!\!\!\!\!  \textup{d}y''\right)^{\!\!2} \label{eq:line:1 of many}\\
		&\leq    \max_{q\in [-1,1]} \left( \frac{4\pi}{N_g}\int_{q/4+1/2}^{N_g/2}     \bar V_0\left(\frac{2\pi}{N_g} y+x_0^{(k)}\right) \!\!\Bigg{|}_{x_0=x_0^{(k)}}\!\!\!\!\!\!\!\!\!  \textup{d}y\,\, \right)^{\!\!2} \label{eq:q change line 1}\\
		&=   \left( \frac{4\pi}{N_g}\int_{1/4}^{N_g/2}  \bar V_0\left(\frac{2\pi}{N_g} y+x_0^{(k)}\right) \!\!\Bigg{|}_{x_0=x_0^{(k)}}\!\!\!\!\!\!\!\!\!  \textup{d}y\,\, \right)^{\!\!2} \label{eq:q change line 2}\\
		&=    \left( \frac{4\pi n A_0}{N_g}  \right)^2 \left( \int_{1/4}^{N_g/2}     \left(\frac{N_g}{2n\pi^2 y} \right)^{2N}  \textup{d}y+ \int_{1/4}^{N_g/2}     \left(\frac{N_g}{2n\pi^2 } \frac{1}{N_g-y} \right)^{2N}  \textup{d}y + \!\!\!\sum_{p\in\zz\backslash\{0,-1\}} \int_{1/4}^{N_g/2}   \!\!  \left(2n\pi^2\left(\frac{y}{N_g} +p\right)\right)^{-2N}  \!\!\!\textup{d}y \right)^{\!\!2} \label{eq:q change line 3}\\
		&\leq    \left( \frac{4\pi n A_0}{N_g}  \right)^2 \left(   \left(\frac{2N_g}{\pi^2\,n} \right)^{2N}\frac{N_g}{2} + \left(\frac{N_g}{\pi^2\,n} \right)^{2N}\frac{N_g}{2} + \frac{N_g}{2} \left(\frac{1}{2\pi^2 n}\right)^{2N}  \sum_{p\in\zz\backslash\{0,-1\}}    \left(\min_{y\in [1/4,N_g/2] }\left|\frac{y}{N_g} +p\right|\right)^{-2N}   \right)^{\!\!2}\label{eq:q change line 4}\\
		&\leq   \left( \frac{4\pi n A_0}{N_g}  \right)^2 \left(  2 \left(\frac{2N_g}{\pi^2\,n} \right)^{2N}\frac{N_g}{2} + \frac{N_g}{2} \left(\frac{1}{2\pi^2 n}\right)^{2N}  \left( \sum_{p\in\nnp}    p^{-2N}  +  \sum_{p\in\nnp}    \left|\frac{1}{2} -p-1\right|^{-2N}  \right)  \right)^{\!\!2}\label{eq:q change line 5}\\
		&\leq     \left( 2\pi n A_0  \right)^2 \left(  2 \left(\frac{2N_g}{\pi^2\,n} \right)^{2N} + \left(\frac{2\pi^2}{3}-5\right)\left(\frac{1}{2\pi^2 n}\right)^{2N}  \right)^{\!\!2}\label{eq:q change line 6}\\
		& \leq  \left( 4\pi n A_0  \right)^2 \left(   \left(\frac{2N_g}{\pi^2\,n} \right)^{2N} + \frac{\pi^2}{3}\left(\frac{1}{2\pi^2 n}\right)^{2N}  \right)^{\!\!2}\label{eq:q change line 7}
	\end{align}
	where in  line~\labelcref{eq:q change line 1} we have taking into account the integrals with intervals of integration  $[-q/4+1/2, N_g/2]$ and $[q/4+1/2, N_g/2]$, map to one another under the transformation  $q\to-q$.  In line~\labelcref{eq:q change line 2} we have used the non-negativity of $ \bar V_0$.  In line~\labelcref{eq:q change line 3} we have substituted for $\bar V_0$ using~\cref{eq:vbar of x} and used the bound $V_B(x)\leq (\pi x)^{-2N}$ for all $x\in\rr$. We have also exchanged the limits of summation and integration. This is justified via the Weierstrass M-test (see Theorem 7.10 in~\cite{rudin1976principles}).  In line~\labelcref{eq:q change line 6}, we have used the fact that the expression is upper bounded by the smallest value of $N$, i.e. one.
	
	Thus using $\sqrt{a+b}\leq \sqrt{a}+\sqrt{b}$ for all $a,b>0$ and plugging~\cref{eq:q change line 7} into~\cref{eq:line 11 3rd proof}, we finalise the proof.
\end{proof}
\end{lemma}

So far in the~\app{} we have considered a Hamiltonian of the form~\cref{eq:2 body ham def}  while we are interested in ones of the form~\cref{eq:main complete ham} since these are the Hamiltonians appearing in~\Cref{thm:comptuer with fixed memory}.  We now introduce a Hamiltonian of the form~\cref{eq:main complete ham} and relate it to the Hamiltonian appearing in~\Cref{lem:quantum dynamics upper bounded by two terms}

\begin{align}
H_{{\M_0}\lo\Cl}= H_\Cl +   \sum_{\substack{l=1}}^{N_g} I_{{\M_0}\lo} ^{(l)}\otimes I_\Cl^{(l)} ,\label{eq:main mapped complete ham specialised}
\end{align}
where $\{I_{{\M_0}\lo} ^{(l)}\}_l$ are defined by 
\begin{align}
I_{{\M_0}\lo} ^{(l)}:= \sum_{\m\in \mathcal{G}} \proj{\m}_{\M_{0,l}} \otimes I_\lo^{(l,\m)}, \label{eq:interaction terms with subscripts}
\end{align}
$l\in 1,2,\ldots, N_g$, and where $\ket{\m}_{\M_{0,l}}$ is the memory state of cell $\M_{0,l}$ corresponding to gate $\m\in\mathcal{G}$, where recall $\mathcal G$ is the gate set.  Each term $I_\lo^{(l,\m)}$ is defined as follows: it has spectrum which lies in the interval $(0, 2\pi]$ and the unitary $U(\m)=\me^{\mi I_\lo^{(l,\m)}}\in\mathcal{U}_\mathcal{G}$ is the representation of gate $\m\in\mathcal{G}$ on $\lo$. While $I_\lo^{(l,\m)}$  is $l$-independent, we keep the label to distinguish it from the term $I_\lo^{(l)}$. 

What is more, given our previous definitions, it is readily apparent that $H_\Cl\geq 0$ and $\{ I_\Cl^{(l)}\geq 0 \}_{l=1}^{N_g}$. As such $H_{\M_0\lo\Cl}\geq 0$. Therefore if $H_{\M_0\lo\Cl} > 0$ we allow for an additional vector $\ket{\textup{ground}}_{\M_0\lo\Cl}$ in the Hilbert space of $\M_0\lo\Cl$ which is orthogonal to the terms in~\cref{eq:main mapped complete ham specialised}. This is the ground state of $H_{\M_0\lo\Cl}$. As such we always have  that the ground state of $	H_{\M_0\lo\Cl}$ has zero energy. This is purely for convenience since later it will allow us to calculate the mean energy of state $\rho_{\M_0\lo\Cl}$ by simply taking its trace with $H_{\M_0\lo\Cl}$. Since none of the states nor operators discussed in this~\doc{} have support on $\ket{\textup{ground}}_{\M_0\lo\Cl}$, we neglect its mention for now on for simplicity.

\subsection{Proof of the theorem}

Finally we are in a stage to prove~\Cref{thm:comptuer with fixed memory}. 
Below we state a more explicit version of it, which is what we will prove. In particular, we state more explicitly the form which the functions $\SupPolyDecay(\cdot)$ take (One can readily check that they belong to said function class). For this, we introduce 2 new functions:
$g(\bar\varepsilon)>0$ is a $E_0$-independent function of $\bar\varepsilon>0$, while $\textup{poly}(E_0)$ is an $\bar\varepsilon$-independent polynomial in $E_0$.  Both are independent of the elements in $\{\tilde d(\m)\}_{\m\in\mathcal{G}}$.

\begingroup
\renewcommand{\thetheorem}{\ref{thm:comptuer with fixed memory}}  

\begin{theorem}[\text{[}More explicit than main-text version.\text{]} Optimal conventional and quantum frequential computers exist]
	\label{thm:comptuer with fixed memoryapp}
	For all gate sets $\mathcal{U}_\mathcal{G}$, initial memory states $\ket{0}_{\M_0}\in\mathcal{C}_{\M_0}$ and initial logical states $\ket{0}_\lo\in\mathcal{P}(\mathcal{H}_\lo)$,  there exists triplets $\{\ket{t_j}_\Cl\}_{j=0}^{N_g}$, $N_g$, $H_{\M_0\lo\Cl}
	$ 
	parametrised by the energy $E_0>0$ and a dimensionless parameter $\bar\epsilon$, such that for all $j=1,2,\ldots, N_g$ and fixed $\bar\varepsilon\in(0,1/6)$ the large-$E_0$ scaling is as follows 
	\begin{align}\label{eq:thm fixed memory 1app}
		T	\Big(\me^{-\mi t_j H_{{\M_0}\lo\Cl} } \ket{0}_{\M_0} \ket{0}_\lo \ket{0}_\Cl ,\, \, \ket{0}_{\M_0}\ket{t_j}_\lo \ket{t_j}_\Cl \Big{)} \leq  \left(\sum_{k=1}^j  \tilde d(\m_k) \right) g( \bar\varepsilon) \,\textup{poly}(E_0) \, E_0^{-1/\sqrt{\bar\varepsilon}}, 
	\end{align}
	for the following two cases:\\[0.2cm]
	\noindent Case 1)
	\begin{align}\label{eq:thm fixed memory 2app}
		f= \frac{1}{T_0}\left( {T_0} E_0 \right)^{1/2-\bar\varepsilon}+ \delta f, \qquad |\delta f| \leq \frac{1}{T_0} +\bo\left(\textup{poly}(E_0) E_0^{-1/\sqrt{\bar\varepsilon}}\right) \text{ as } E_0\to\infty,
	\end{align}
	and $\big(\,\ket{t_j}_\Cl\,\big)_{j=0}^{N_g}\in\mathcal{C}_\Cl^\textup{clas.}$.\\[0.2cm]
	\noindent Case 2)\\
	\begin{align}\label{eq:thm fixed memory 3app}
		f= \frac{1}{T_0}\left( T_0 E_0 \right)^{1-\bar\varepsilon}+ \delta f', \qquad |\delta f'| \leq \frac{1}{T_0} +\bo\left(\textup{poly}(E_0) E_0^{-1/\sqrt{\bar\varepsilon}}\right) \text{ as } E_0\to\infty.
	\end{align}  
\end{theorem}

\endgroup

\begin{proof} \label{proof:thm one clock only}
We will first show that this trivially reduces to a problem not involving the memory states. Then we will prove the scaling in~\cref{eq:thm fixed memory 1app} via~\Cref{lem:quantum dynamics upper bounded by two terms,lem:2st secondrary lem for unitary bound,lem:1st secondrary lem for unitary bound}, but as a function of dimension rather than initial energy. We will then proceed to bound the dimension as a function of initial energy, and frequency as a function of initial energy.

From~\cref{eq:main mapped complete ham specialised}, we observer that  $H_{\M_0 \lo\Cl}$ is block-diagonal in the basis of the memory, $\big\{  \ket{\m_1}_{\M_{0,1}} \ket{\m_2}_{\M_{0,2}}$ $\ldots$ $ \ket{\m_{N_g}}_{\M_{0,N_g}} \big\}_{\m_l\in\mathcal{G}}$. Since on the l.h.s. of~\cref{eq:main mapped complete ham specialised}, the kets are a tensor product of an element of this set, the sole effect of commuting the memory state with the exponentiated Hamiltonian, is the mapping of $H_{{\M_0}\lo\Cl}$ in~\cref{eq:main mapped complete ham specialised} to
\begin{align}
	H_\Cl +   \sum_{\substack{l=1}}^{N_g} I_{\lo}^{(l,\m_l)}\otimes I_\Cl^{(l)}. \label{eq:mapped interaction ham}
\end{align}
In the remainder of this proof, we will work with Hamiltonians of this form but using the shorthand notation   $I_{\lo}^{(l)}$ rather than $I_{\lo}^{(l,\m_l)}$. Since our proof considers operators $I_{\lo}^{(l)}$ which can implement arbitrary gates, it can implement any corresponding gate $\m_l$, from any gate set.

Plugging the bounds from~\Cref{lem:2st secondrary lem for unitary bound,lem:1st secondrary lem for unitary bound} into the bounds from~\Cref{lem:quantum dynamics upper bounded by two terms}, and simplifying the resultant expression we find, for $j=1,2,3, \ldots, N_g$
\begin{align}\label{eq:first theorem unitary quant eq1}
	&	\big\|  \me^{-\mi t_j H_{\lo\Cl}} \ket{0}_\lo\ket {\Psi(0)}_\Cl - \ket{t_j}_\lo \ket{\Psi(t_j d/T_0)}_\Cl \big \|_2\\  
	\leq & \sum_{k=1}^j \bigg( \big \| \ket{t_k}_\lo \ket{\Psi(t_{k}d/T_0)}_\Cl - \me^{-\mi t_1 H_{\lo\Cl}^{(k)}}  \ket{t_{k-1}}_\lo \ket{\Psi(t_{k-1}d/T_0) }_\Cl\big \|_2  \label{eq:line2 of first theorem unitary quant eq1}\\
	&+    t_1  \max_{x\in[0,t_1]} \big \| \bar{H}_{\lo\Cl}^{(k)} \me^{-\mi x H_{\lo\Cl}^{(k)}} \ket{t_{k-1}}_\lo\ket{\Psi(t_{k-1}d/T_0)}_\Cl \big \|_2 \bigg ) \label{eq:line3 of first theorem unitary quant eq1}\\
	=&\sum_{k=1}^j  \left[   \sqrt{2\varepsilon_v(t_{k-1},d)} +\pi \sqrt{2d  A_\textup{nor}} \left( \me^{-\frac{\pi}{16} \left(\frac{d}{\sigma N_g}\right)^2} +     4\pi n A_0\left(\left(\frac{2N_g}{\pi^2\,n} \right)^{2N} + \frac{\pi^2}{3}\left(\frac{1}{2\pi^2 n}\right)^{2N} \right) \right) \right]+\\
	&\sum_{k=1}^j  \Bigg[   
	\frac{ 	 \tilde d(\m_k) 4\pi^2 n A_0}{T_0} \Biggl( 3  A_\textup{nor} d \, N_g  \bigg(     \me^{-{\pi}{\frac{d^2}{\sigma^2 (4 N_g)^2} }} +
	\Big( \frac{2 N_g }{\pi^2 n} \Big) ^{2 N}  \,\bigg) \\&
	\quad+  A_\textup{nor} \left(\pi^2-\frac{79}{9}\right) N_g  d \left(\frac{1}{2\pi n} \right)^{2N}    + N_g d \,     \left( 1 + \frac{\pi^2}{3}\right) \varepsilon_v(t_{k-1},d) \Biggr) \Bigg]\\
	\leq &   \sqrt{j}(j-1) \frac{T_0}{N_g}\sqrt{2\varepsilon_v(1,d)} +j \pi \sqrt{2d  A_\textup{nor}} \left( \me^{-\frac{\pi}{16} \left(\frac{d}{\sigma N_g}\right)^2} +     4\pi n A_0\left(\left(\frac{2N_g}{\pi^2\,n} \right)^{2N} + \frac{\pi^2}{3}\left(\frac{1}{2\pi^2 n}\right)^{2N} \right) \right) \label{eq:first theorem unitary quant eq1 line 3}\\
	&+\left(\sum_{k=1}^j  \tilde d(\m_k) \right) \Bigg[   
	\frac{  4\pi^2 n A_0}{T_0} \Biggl( 3  A_\textup{nor} d \, N_g  \bigg(     \me^{-{\pi}{\frac{d^2}{\sigma^2 (4 N_g)^2} }} +
	\Big( \frac{2 N_g }{\pi^2 n} \Big) ^{2 N}  \,\bigg)\label{eq:first theorem unitary quant eq1 line 5} \\&
	+  A_\textup{nor} \left(\pi^2-\frac{79}{9}\right) N_g  d \left(\frac{1}{2\pi n} \right)^{2N}    +  d \,     \left( 1 + \frac{\pi^2}{3}\right) T_0 (j-1)\, \varepsilon_v(1,d) \Biggr) \Bigg],  \label{eq:first theorem unitary quant eq1 line 4}
\end{align}
where in lines~\labelcref{eq:first theorem unitary quant eq1 line 3,eq:first theorem unitary quant eq1 line 4}, we have used the definition of $\varepsilon_v(t,d)$ from~\cref{eq:var e v def}. Now observe that in order for  lines~\labelcref{eq:first theorem unitary quant eq1 line 3,eq:first theorem unitary quant eq1 line 4} to be small, we need $d/(\sigma N_g)$ to tend to infinity as $d\to\infty$, and $[2 N_g/(\pi^2 n)]^N$ to tend to zero as $d\to\infty$ sufficiently quickly.
We start by recalling the definitions of $N$ and $n$ used in~\cite{WoodsPRXQ} (see eqs. F18, F220 
in~\cite{WoodsPRXQ}):
\begin{align}
	N&=\left\lceil \frac{3-4\epsilon_5-\epsilon_9}{2(\epsilon_7-\epsilon_5)} \right\rceil \geq  \frac{3-4\epsilon_5-\epsilon_9}{2(\epsilon_7-\epsilon_5)} , \label{eq:N def}\\
	n &	= \frac{\ln(\pi\alpha_0 \sigma^2)}{2\pi C_0 \alpha_0 \kappa} \frac{d^{1-\epsilon_5}}{\sigma},\label{eq:n equation}
\end{align}
\footnote{Note that~\cref{eq:N def,eq:n equation} differ slightly from the definitions cited from~\cite{WoodsPRXQ}. Namely an $\epsilon_8$ (which was introduced via eq. F232 in~\cite{WoodsPRXQ}) has been omitted in the r.h.s. of ~\cref{eq:N def,eq:n equation} and $n \delta$ has been replaced with $n$. These modifications are due to the modification of~\cref{eq:vbar of x}. See footnote~\ref{f:main footnote about modified pot} for explanation.}
where $\kappa=0.792$, and $C_0(N)$ is solely a function of $N$, (e.g. independent of $d$, $\sigma$, and  $N_g$; see Lemma 28 in~\cite{WoodsPRXQ}) and where $\alpha_0$ is related to the initial mean energy parameter $n_0$ (recall~\cref{eq:initial state}) via
\begin{align}
	\alpha_0=1-\left| 1-n_0\left(\frac{2}{d-1}\right)\right|  \in (0,1],  \label{eq: alph0 def for clock 1}
\end{align}
and is uniformly bounded from below since we assume that $n_0=\tilde n_0  (d-1)$, $\tilde n_0\in(0,1)$ a fixed numerical constant (i.e. independent of $d$). 
The coefficients $\epsilon_5,\epsilon_7, \epsilon_9,\eta$ can be chosen to be any $d$-independent constants satisfying the relations
\begin{align}
	&0< \epsilon_5< \epsilon_6=\frac{\ln \sigma}{\ln d}<1,\label{1st epsion eqs}\\
	& 0< \epsilon_7 < \eta/2,\\
	&\!\! \epsilon_5< \epsilon_7,\\
	& 0<\epsilon_9<\eta,\\
	& 0 < 3-4 \epsilon_5-\epsilon_9,\\
	& \frac{4}{\sigma} < d^{\eta/2} \leq \frac{d}{\sigma}. \label{last epsion eqs}
\end{align}
Let us choose 
\begin{align}
	\epsilon_5&=\eta\bar \varepsilon, \quad \epsilon_7=2\eta \bar\varepsilon, \quad
	\epsilon_9=\eta/2, \quad  \sigma=d^{\eta/2} = d^{\epsilon_6}\label{eq:epsilon defs}
\end{align}
We observe that for this choice of constants satisfies~eqs.~\eqref{1st epsion eqs} to \eqref{last epsion eqs} for all $0<\bar \varepsilon<1/6$ and  $0<\eta\leq 1$. It now follows 
\begin{align}
	\left(\frac{2 N_g}{\pi^2 n}\right)^{2N} \leq \left( \frac{4 C_0(N) \alpha_0 \kappa}{\pi\ln(\pi \alpha_0)}\right) ^{\!2N}d^{[\eta ( \bar\varepsilon+1/2)-\varepsilon_g]3/(\eta \varepsilon)},
\end{align}
where we have defined 
\begin{align}
	N_g= \lfloor d^{1-\varepsilon_g} \rfloor\leq d^{1-\varepsilon_g}, \quad \varepsilon_g>0,   \label{eq:def: Ng in terms of ep g}
\end{align}
and were it follows from the definitions and properties of $C_0$, that the prefactor 
\begin{align}
	\left( \frac{4 C_0(N) \alpha_0 \kappa}{\pi\ln(\pi \alpha_0)}\right) ^{\!2N}
\end{align}
is independent of $d$. (It is however dependent on $\bar \varepsilon$ and $\eta$ and might diverge if we were to take a limit in which either or both tend to zero. This is why they are fixed and $d$-independent by definition.) Now choose
\begin{align}
	\varepsilon_g= \eta (\bar\varepsilon+1/2)+\eta\sqrt{\bar\varepsilon}  \label{eq: varepsion_g as a fuction of eta ep}
\end{align}
thus resulting in the bound
\begin{align}
	\left(\frac{2 N_g}{\pi^2 n}\right)^{2N} \leq  \left( \frac{4 C_0(N) \alpha_0 \kappa}{\pi\ln(\pi \alpha_0)}\right) ^{\!2N}\, d^{-3/ \sqrt{\bar\varepsilon}}. \label{eq:upper bound for frac Ng n to N}
\end{align}
Since $\eta$ and $\bar\varepsilon$ are $d$-independent by definition, so is $\varepsilon_g$. Therefore, for sufficiently large $d$, we have $N_g>1$ and both terms $(1/(2\pi^2 n))^{2N}$, $(1/(2\pi n))^{2N}$ appearing in~\cref{eq:first theorem unitary quant eq1 line 3,eq:first theorem unitary quant eq1 line 4} are upper bounded by~\cref{eq:upper bound for frac Ng n to N}.
For $(d/(\sigma N_g))^2$ we find from the above definitions
\begin{align}
	\left(\frac{d}{\sigma N_g}\right)^2 =  \frac{d^{2-\eta}}{\left( \lfloor d^{1- \eta (\bar\varepsilon+1/2)-\eta\sqrt{\bar\varepsilon}} \rfloor\right)^2} \sim d^{2\eta(\bar\varepsilon+\sqrt{\bar\varepsilon})},\label{eq:bound of d sigma Ng decay term}
\end{align}
as $d\to \infty$.

 Finally, from~\cref{eq:var e v def} and~\cref{eq:epsilon defs}, it follows that $\varepsilon_v(1, d)$ decays faster than any polynomial in $d$. 
Therefore, taking into account that $N_g\leq d$, and $n\leq d$ for sufficiently large $d$, and taking into account the upper bound on $A_\textup{nor}$ (see~\cref{eq:A nor up bound}), and that $A_0$ is solely a function of $N$ (see~\cref{eq:A0 def}), it follows from ~\cref{eq:first theorem unitary quant eq1 line 4} that
\begin{align}\label{eq:first theorem unitary quant eq3}
	&	\big\|  \me^{-\mi t_j H_{\lo\Cl}} \ket{0}_\lo\ket {\Psi(0)}_\Cl - \ket{t_j}_\lo \ket{\Psi(t_j d/T_0)}_\Cl \big \|_2  
	\leq   \left(\sum_{k=1}^j  \tilde d(\m_k) \right) f(\eta, \bar\varepsilon) \, \textup{poly}(d) \, d^{-3/\sqrt{\bar\varepsilon}}, 
\end{align}
for all $\eta\in(0,1]$, $\bar\varepsilon\in(0,1/6)$, $d\in\nnp$. The function  $f\geq 0$ is independent of $d$ and the elements $\{\tilde d(\m_k) \}_{k=1}^j$. Meanwhile, $\textup{poly}(d)\geq 0$ is independent of, $\eta$, $\bar\varepsilon$, and the elements $\{\tilde d(\m_k) \}_{k=1}^j$. Note that while the degree of the polynomial $\textup{poly}(d) $ is easily deducible from line~\labelcref{eq:first theorem unitary quant eq1 line 4}, it is not important for our purposes.  We can lower bound the difference in kets appearing on the l.h.s. of~\cref{eq:first theorem unitary quant eq3} in terms of trace distance rather (than 2-norm) as per~\Cref{thm:comptuer with fixed memory},  via the  use of~\Cref{lem:trace dist}.

We now turn our attention to calculating the mean energy of the initial state. Recall the discussion below~\cref{eq:main mapped complete ham specialised}: the Hamiltonian $H_{\M_0\lo\Cl}$ has a ground-state energy of zero. As such the mean energy of the initial state is
\begin{align}
	{}_{\M_0}\!\!\bra{\vec \m} {}_\lo\!\bra{0} {}_\Cl\!\bra{0}    H_{{\M_0}\lo\Cl}  \ket{\vec \m}_{\M_0} \ket{0}_\lo \ket{0}_\Cl &=  {}_\Cl\!\bra {\Psi(0)} H_\Cl \ket {\Psi(0)}_\Cl+\sum_{l=1}^{N_g} {}_\lo\!\bra{0} I_\lo^{(l,\m_l)}\ket{0}_\lo  \,  {}_\Cl\!\bra{\Psi(0)} I_\Cl^{(l)}\ket{\Psi(0)}_\Cl .\label{eq:mean energy initial state 1st eq}
\end{align}
We will now show that the terms $\{ {}_\Cl\!\bra{\Psi(0)} I_\Cl^{(l)}\ket{\Psi(0)}_\Cl \}_{l=1}^{N_g}$ are zero in the large $d$ limit. For all $l=1,2,\ldots, N_g$ we find
\begin{align}
	\left|{}_\Cl\!\bra{\Psi(0)} I_\Cl^{(l)}\ket{\Psi(0)}_\Cl \right |  \leq & \frac{2\pi}{T_0} \sum_{k\in\mathcal{S}{(k_0)}} \left|\bar V_0 \left(\frac{2\pi }{d}k\right){\bigg |}_{x_0=x_0^{(l)}}\right| |\braket{\theta_k|\Psi(0)}_\Cl|^2 \label{eq:mean energy line 2}\\
	\leq &\frac{2\pi}{T_0}A_\textup{nor} n A_0 \sum_{k\in\mathcal{S}{(0)}}  \sum_{p=-\infty}^\infty V_B\Big(n 2\pi (k N_g/d -l+1/2+ N_g p)/N_g\Big) \, \me^{-2\pi \left(\frac{k}{\sigma}\right)^2}   \label{eq:mean energy line 3}\\
	\leq &  \frac{2\pi}{T_0}A_\textup{nor} n A_0 \bigg(d \max_{k'\in[-1/2,1/2]} \sum_{p\in\{0,\pm 1\}} V_B\Big(n 2\pi (k' N_g -l+1/2+N_g p)/N_g\Big) \, \me^{-2\pi \left(\frac{k'}{\sigma}\right)^2} \\
	&\quad  + \sum_{k\in\mathcal{S}(0)}  \sum_{p\in\zz\backslash\{0,\pm 1\}}  \Big|n 2\pi^2 (k N_g/d -l+1/2+ N_g p)/N_g\Big|^{-2N} \, \me^{-2\pi \left(\frac{k}{\sigma}\right)^2}\bigg)  \label{eq:mean energy line6} \\
	\leq &  \frac{2\pi}{T_0}A_\textup{nor} n A_0 \bigg[2d \, \me^{-2\pi \left(\frac{d}{4 N_g \sigma}\right)^2}+d \max_{k'\in[-1/(4N_g),1/(4N_g)]} \sum_{p\in\{0,\pm 1\}}\!\!\!\!\! V_B\Big(\frac{n 2\pi}{N_g} (k' N_g -l+1/2+N_g p)\Big)  \label{eq:mean energy line 7}\\
	&  + d(n 2\pi^2 )^{-2N} \!\!\!\max_{\substack{k\in\mathcal{S}(0) \\ l\in\{1,2,\ldots,N_g\} }}  \sum_{p=2}^\infty  \left(\Big|\frac kd -(l-1/2)/N_g+ p\Big|^{-2N}+\Big|\frac kd -(l-1/2)/N_g- p\Big|^{-2N}\right) \bigg]  \label{eq:mean energy line 8} \\
	\leq &  \frac{2\pi}{T_0}A_\textup{nor} n A_0 \bigg[2d \, \me^{-2\pi \left(\frac{d}{4 N_g \sigma}\right)^2}+d \max_{q'\in[-1/4,1/4]} \sum_{p\in\{0,\pm 1\}}  \Bigg|\frac{n 2\pi^2}{N_g} (q -l+1/2+N_g p)\Bigg|^{-2 N}  \label{eq:mean energy line 9} \\
	&  + d(n 2\pi^2 )^{-2N} \, \sum_{p=2}^\infty  \left(\Big|-\frac 12 -1+ p\Big|^{-2N}+\Big|\frac 12- p\Big|^{-2N}\right) \bigg]  \label{eq:mean energy line 10} \\
	\leq &  \frac{2\pi}{T_0}A_\textup{nor} n A_0 \bigg[2d \, \me^{-2\pi \left(\frac{d}{4 N_g \sigma}\right)^2}+3 d\,   \Bigg|\frac{n 2\pi^2}{4 N_g} \Bigg|^{-2 N}  \label{eq:mean energy line 11}\\
	&  + d(n 2\pi^2 )^{-2N}   \sum_{p=2}^\infty  \left(\Big|-\frac 12 -1+ p\Big|^{-2N}+\Big|\frac 12- p\Big|^{-2N}\right) \bigg]  \label{eq:mean energy line 12} \\
	\leq &  \frac{2\pi}{T_0}A_\textup{nor} n A_0 d\Bigg[2 \, \me^{-2\pi \left(\frac{d}{4 N_g \sigma}\right)^2}+3    \Bigg(\frac{2 N_g} {\pi^2 n}\Bigg)^{2 N}   +  \pi^2 \left(\frac{1}{n 2\pi^2} \right)^{2N}  \Bigg],  \label{eq:mean energy line 14} 
\end{align}
where in line~\labelcref{eq:mean energy line 2} we have used~\cref{eq:def I in terms of bar V0}. In line~\labelcref{eq:mean energy line 3} we have set $k_0=0$ (recall that $I_\Cl^{(l)}$ is independent of this parameter due to the periodicity of the summand). This choice means we can easily calculate the overlaps $\braket{\theta_k|\Psi(0)}_\Cl$. In line~\labelcref{eq:mean energy line 7}  we have used the fact that $k'=k/d\in[-1/2,1/2]= [-1/2,-1/(4 N_g)]\cup [-1/(4N_g),1/(4N_g)]\cup[1/(4N_g),1/2]$ and the bound $V_B(x)\leq 1$ for all $x\in\rr$.  In line~\labelcref{eq:mean energy line 10} we have used the fact that $k/d-(l-1/2)/N_g\in[-1/2-1,1/2]$ for all $k\in\mathcal{S}_d(0)$, $l\in\{1,2,\ldots, N_g\}$, $N_g\in\nnp$. In line~\labelcref{eq:mean energy line 11} we have used the fact that $-l+1/2+N_g p$ is always a half integer. In line~\labelcref{eq:mean energy line 14}  we have used $N\in\nnp$.

Therefore, using~\cref{eq:upper bound for frac Ng n to N} and the similar lines of reasoning to those used just after this equation, we conclude
\begin{align}
	\Bigg| \sum_{l=1}^{N_g} {}_\lo\!\bra{0} I_\lo^{(l,\m_l)}\ket{0}_\lo  \,  {}_\Cl\!\bra{\Psi(0)} I_\Cl^{(l)}\ket{\Psi(0)}_\Cl  \Bigg| &\leq 	2 \pi d \max_{l=1,2,\ldots, N_g} \Bigg|  {}_\Cl\!\bra{\Psi(0)} I_\Cl^{(l)}\ket{\Psi(0)}_\Cl  \Bigg| \label{eq:firs line energy eq} \\
	&\leq 
	f'(\eta, \bar\varepsilon) \, \textup{poly}'(d) \, d^{-3/\sqrt{\bar\varepsilon}},   \label{eq:sum interaction temts small}
\end{align}
for all $\eta\in(0,1]$, $\bar\varepsilon\in(0,1/6)$, $d\in\nnp$. In line~\labelcref{eq:firs line energy eq} we used the fact that the spectrum of $I_\lo^{(l,\m_l)}$ is in the interval $(0,2\pi]$.  The function  $f'\geq 0$ is independent of $d$  and $\textup{poly}'(d)\geq 0$ is independent of $\eta$, and $\bar\varepsilon$.

In~\cite{WoodsAut}, the mean energy of the initial state $ \ket {\Psi(0)}_\Cl$ for the Hamiltonian $H_\Cl$ was calculated. It was found that 
\begin{align}
	0<	{}_\Cl\!\bra {\Psi(0)} H_\Cl \ket {\Psi(0)}_\Cl =  \frac{2\pi}{T_0} n_0 +\epsilon_E,   \label{eq:free mean energy}
\end{align}
where 
\begin{align}
	|\epsilon_E| \leq  \textup{poly}''(d) \,\me^{-\frac{\pi}{4} \left(\frac{d}{\sigma}\right)^2}.
\end{align}
Recall that we assume in this manuscript that $n_0=\tilde n_0 (d-1)$, with $\tilde n_0\in(0,1)$ and $d$-independent. It is easily verified that $\textup{poly}''(d)$ is independent of the parameters $\eta$, $\bar\varepsilon$ introduced here. Since $\sigma=d^{\eta/2}$,  $\eta\in(0,1]$ here, we have using~\cref{eq:mean energy initial state 1st eq} that 
\begin{align}
	E_0:=	 {}_{\M_0}\!\!\bra{\vec \m} {}_\lo\!\bra{0} {}_\Cl\!\bra{0}    H_{{\M_0}\lo\Cl}  \ket{\vec \m}_{\M_0} \ket{0}_\lo \ket{0}_\Cl &=   \frac{2\pi}{T_0} \tilde n_0 (d-1)+ \delta E',   \label{eq:E0 clcu}
\end{align}
where 
\begin{align}
	| \delta E'| \leq 	f''(\eta, \bar\varepsilon) \, \textup{poly}'''(d) \, d^{-3/\sqrt{\bar\varepsilon}},  \label{eq:uper bound dleta E}
\end{align}
with  $f''\geq 0$  independent of $d$  and $\textup{poly}'''(d)\geq 0$  independent of $\eta$, and $\bar\varepsilon$. Recall that $\tilde n_0\in(0,1)$ a fixed numerical constant (i.e. independent of $d$).  Therefore, 
\begin{align}
	d= \frac{T_0}{2\pi \tilde n_0} E_0 +\delta d, \qquad \delta d :=1-  \frac{T_0}{2\pi \tilde n_0}  \delta E'.   \label{eq:d as a fuction of P}
\end{align}
This provides a parametrization of $d$ in terms of $E_0$. We have to be cautions, because $\delta E'$ depends on $\eta$, $\bar\varepsilon$ and $d$, and thus if we plug this relation into a function which depends on $d$ but not on $\eta$, or $\bar\varepsilon$, we will not obtain a function which depends on $E_0$ but not on $\eta$, or  $\bar\varepsilon$. Nevertheless, 
note that since 
$\delta d$ converges to $1$ in the large $d$ limit for all $\eta\in(0,1]$ and $\bar\varepsilon\in(0,1/6)$, we have that 
for all  $\bar\varepsilon\in(0,1/6)$, here exists $E_{00}(\eta, \bar\varepsilon)>0$ such that for all  $\eta\in(0,1]$, $\bar\varepsilon\in(0,1/6)$, and $E_0\geq E_{00}(\eta, \bar\varepsilon)$, the following holds
\begin{align}
	f(\eta, \bar\varepsilon) \, \textup{poly}(d) \, d^{-3/\sqrt{\bar\varepsilon}} \leq f(\eta, \bar\varepsilon) \left( \frac{T_0^2}{2\pi \tilde n_0}\right)^{-3/\sqrt{\bar\varepsilon}}\, \textup{poly}''''(E_0) \, E_0^{-3/\sqrt{\bar\varepsilon}}, \label{eq:changing from d to P}
\end{align}
where $\textup{poly}''''(E_0)>0$ for all $E_0\geq 0$, is a polynomial which is independent of $\eta$ and $\bar\varepsilon$.  Now let $R$ denoted the ratio between the l.h.s and the r.h.s. of~\cref{eq:changing from d to P} and let $h>0$ be defined by $h:=\max_{E_0\in[0, E_{00}(\eta,\bar\varepsilon)]} R(E_0)$, where the l.h.s. of~\cref{eq:changing from d to P} is written as a function of $E_0$ rather than $d$ by virtue of~\cref{eq:d as a fuction of P}. Clearly $h$  depends on $\eta$ and $\bar\varepsilon$ but not $E_0$. Therefore there exists a function $g(\eta,\bar\varepsilon)>0$ which is $E_0$-independent such that 
\begin{align}
	f(\eta, \bar\varepsilon) \, \textup{poly}(d) \, d^{-3/\sqrt{\bar\varepsilon}} \leq g(\eta, \bar\varepsilon) \,\textup{poly}''''(E_0) \, E_0^{-3/\sqrt{\bar\varepsilon}}, \label{eq:changing from d to P 2}
\end{align}
holds for all $E_0\geq 0$ obeying~\cref{eq:d as a fuction of P}. Plugging this relation into the r.h.s. of~\cref{eq:first theorem unitary quant eq3}, we obtain~\cref{eq:thm fixed memory 1app} after a re-labelling of  $\textup{poly}''''(E_0)$ by $\textup{poly}(E_0)$, and defining $g(\bar\varepsilon)=g(\eta(\bar\varepsilon),\bar\varepsilon)$, where the parametrization $\eta(\bar\varepsilon)$ is chosen differently in cases 1) and 2) below.\\

We now turn our attention to calculating the gate frequency as a function of the mean energy $E_0$. Using~\cref{eq:def: Ng in terms of ep g}
\begin{align}
	f=\frac{N_g}{T_0}= \frac{\lfloor d^{1-\varepsilon_g} \rfloor}{T_0}=  \frac{\left(\frac{T_0}{2\pi \tilde n_0  }E_0+\delta d  \right)^{1-\varepsilon_g}+\delta_1}{T_0}=  \frac{1}{T_0}{\left(\frac{T_0}{2\pi \tilde n_0  }E_0 \right)^{1-\varepsilon_g}}+ \delta f, \qquad |\delta f| \leq   \Biggg|  \frac{\left(\frac{T_0}{2\pi \tilde n_0  }E_0 \right)^{-\varepsilon_g}}{T_0}\delta d +\frac{ \delta_1}{T_0}  \Biggg|, \label{eq:frequnsy case}
\end{align}
where $\delta_1:= - d^{1-\varepsilon_g} +\lfloor d^{1-\varepsilon_g} \rfloor\in(-1,0]$.
For case 1) of~\Cref{thm:comptuer with fixed memory}, we choose $\eta$ such that $\varepsilon_g=1/2+\bar\varepsilon$, which from ~\cref{eq: varepsion_g as a fuction of eta ep} gives us
\begin{align}
	\eta (\bar\varepsilon)=\frac{1/2+\bar\varepsilon}{1/2+\bar\varepsilon+\sqrt{\bar\varepsilon}}  \label{eq:eta barepsilon classical}
\end{align}
which satisfies the condition $\eta\in(0,1]$ for all $\bar\varepsilon\in(0,1/6)$. 
 To finalise case 1) of the proof, we need to show that $\ket{t_j}_\Cl\in\mathcal{C}_\Cl^\textup{clas.}$. Recalling the definition of $\mathcal{C}_\Cl^\textup{clas.}$ from~\Cref{sec:alternative def of squeezeing on C}, this consists in showing that $\ket{t_j}_\Cl$ is an eigenstate of $L_\Cl=\lambda t_\Cl/T_0+\mi H_\Cl T_0$ with $|\lambda|=1$ up to an additive vanishing term in the large $d$ limit. Let $x$ and $p$ denoted canonically conjugate position and momentum operators. In~\cite{WoodsAut}, it was shown that the action of $t_\Cl$ and $H_\Cl$ on the state  $\ket {\Psi(t d/T_0)}_\Cl$ of width $\sigma$ , is equal to the action of  $\hat x$ and $\hat p$ on a Gaussian wave-packet state of the same width (up to additive exponentially small correction terms). In~\cite[4.1 Position measurements]{Maccone2020squeezingmetrology} (also see~\cite{Trifonov1994}) the authors show that said Gaussian wave-packet states are eigenstates of $\lambda \hat x+\mi \hat p$, for $\lambda=1$ when the width of the Gaussian wave-packet is of equal width in both the position and momentum basis. These are the states   $\ket {\Psi(t d/T_0)}_\Cl$ we have used in the classical case, as is manifest by the fact that they also have width $\sigma$ which is equal in both these conjugate bases (i.e. $\sigma= d^{1/\eta}=\sqrt{d}$).
 

For case 2),  we choose $\eta$ such that $\varepsilon_g=\bar\varepsilon$ which gives, using~\cref{eq: varepsion_g as a fuction of eta ep},
\begin{align}
	\eta({\bar\varepsilon})=\frac{2 \bar\varepsilon }{2 \bar\varepsilon +2 \sqrt{\bar\varepsilon }+1} \label{eq:eta as func of ep}
\end{align}
and tends to zero from above as $\bar\varepsilon$ tends to zero from above.
We can therefore make $\varepsilon_g$ arbitrarily small by choosing $\bar\varepsilon>0$ sufficiently small.
Now, from~\cref{eq:frequnsy case} it follows
\begin{align}\label{eq:thm fixed memory 3 app}
	f= \frac{1}{T_0}\left( \frac{T_0}{2\pi \tilde n_0} E_0 \right)^{1-\bar\varepsilon}+ \delta f', \qquad |\delta f'| \leq \frac{1}{T_0} +\bo\left(\textup{poly}(E_0) E_0^{-1/\sqrt{\bar\varepsilon}}\right) \text{ as } E_0\to\infty.
\end{align} 
To achieve~\cref{eq:thm fixed memory 3}, we choose $\tilde n_0=1/(2\pi)$. This is consistent with the parameter regime of $\tilde{n}_0$ [see text below~\cref{eq: alph0 def for clock 1}].  Note that we cannot make $f$ arbitrarily large by choosing $\tilde{n}_0$ arbitrarily small because in said limit $E_0$ also becomes arbitrarily small [recall~\cref{eq:E0 clcu}].
\end{proof}

\subsection{Generic known useful technical lemmas}

The following two lemmas are rather trivial but crucial for this work.
\begin{lemma}[Unitary errors add linearly: Lemma C.0.2. in~\cite{WoodsAut} or~\cite{Nielsen2012}]\label{lem:unitari errs add linararly}
Let $\{ \ket{ \Phi_m } \}_{m=1}^n$ be a set of states\footnote{In Lemma C.0.2. in~\cite{WoodsAut} these states were normalised by definition. Here we remove this assumption since it is not necessary and we will use it in the proof of~\Cref{thm:heat dissipation} in the case of sub-normalized states.} in a 2-normed vector space satisfying $\|  \ket{ \Phi_m } - \Delta_m \ket{ \Phi_{m-1} }  \|_2 \leq \varepsilon_m$, $\| \Delta_m\|_2\leq 1$.  Then 
\begin{align}
	\|  \ket{\Phi_n}- \Delta_n \Delta_{n-1} \ldots \Delta_1 \ket{\Phi_0} \| _2  \leq \sum_{m=1}^n \varepsilon_m.
\end{align}
\end{lemma}

\begin{proof}
By induction. The theorem is true by definition for $n=1$, and if the theorem is true for all $n$ up to $k$, then for $n=k+1$,
\begin{align}
	\|\left|\Phi_{k+1}\right\rangle-\Delta_{k+1} \Delta_k \ldots \Delta_1\left|\Phi_0\right\rangle \|_2 & =\|\left|\Phi_{k+1}\right\rangle-\Delta_{k+1}\left|\Phi_k\right\rangle+\Delta_{k+1}\left(\left|\Phi_k\right\rangle-\Delta_k \ldots \Delta_1\left|\Phi_0\right\rangle\right) \|_2 \\
	& \leq \|\left|\Phi_{k+1}\right\rangle-\Delta_{k+1}\left|\Phi_k\right\rangle\left\|_2+\right\| \Delta_{k+1}\left(\left|\Phi_k\right\rangle-\Delta_k \ldots \Delta_1\left|\Phi_0\right\rangle\right) \|_2 \\
	& \leq \|\left|\Phi_{k+1}\right\rangle-\Delta_{k+1}\left|\Phi_k\right\rangle\left\|_2+\right\| \Delta_{k+1}\left\|_2\right\|\left(\left|\Phi_k\right\rangle-\Delta_k \ldots \Delta_1\left|\Phi_0\right\rangle\right) \|_2 \label{eq:line inter uni ad line}\\
	& \leq \|\left|\Phi_{k+1}\right\rangle-\Delta_{k+1}\left|\Phi_k\right\rangle\left\|_2+\right\|\left|\Phi_k\right\rangle-\Delta_k \ldots \Delta_1\left|\Phi_0\right\rangle \|_2 \\
	& =\epsilon_{k+1}+\sum_{m=1}^k \epsilon_m=\sum_{m=1}^{k+1} \epsilon_m
\end{align}
where we used the Minkowski vector norm inequality and the equivalence between the induced $l_2$ operator norm and the property $\left\|\Delta_m\right\|_2 \leq 1$ in line~\labelcref{eq:line inter uni ad line}.
\end{proof}

\begin{lemma}[Upper bounding trace distance by Euclidean distance for pure states]\label{lem:trace dist}
	The trace distance and Euclidean distance between two normalised pure states $\ket{A}$, $\ket{B}$ is
\begin{align}
		T\big(\ket{A}, \ket{B}\big)=\sqrt{1-|\braket{A|B}|^2}\,,\qquad \| \ket{A}-\ket{B}\|_2= \sqrt{1-\Re\big[\!\braket{A|B}\!\big]},
\end{align}
respectively. They are related by
\begin{align}
	T\big(\ket{A}, \ket{B}\big)\leq  \sqrt{2}\, \| \ket{A}-\ket{B}\|_2 .
\end{align}
\begin{proof}
	\begin{align}
		&\sqrt{1-|\braket{A|B}|^2}= \sqrt{1-\Re(\braket{A|B})^2-\Im(\braket{A|B})^2}\leq \sqrt{1-\Re(\braket{A|B})^2}= \| \ket{A}-\ket{B}\|_2 \sqrt{1+\| \ket{A}-\ket{B}\|_2^2/4}\\
		&\leq \| \ket{A}-\ket{B}\|_2 \sqrt{2} 
	\end{align}
\end{proof}
\end{lemma}

\section{Interaction strengths in the semi-classical and fully quantum  control-state  cases}\label{sec:trenghth indepedent}
Here we show that the interaction strength is not stronger in the case that achieves the linear scaling of gate frequency $f$ with mean energy $E_0$, compared with the semi-classical state limit where it depends on $\sqrt{E_0}$. For a longer discussion on the importance of these observations, see~\Cref{sec:nostrong interactions needed}. We will then compare this with the interaction strength ubiquitously found in nature, namely a linear (in the creation and annihilation operators) interaction. We will show that we do not require a stronger interaction term for our results. For concreteness,  we will focus on the interactions used in the proof of~\Cref{thm:comptuer with fixed memoryapp}, but one can readily check that we have used the same interaction terms in the later theorems, and hence the conclusions we arrive at here hold generally.

To start with, recall that the total interaction responsible for implementing the $l\thh$ gate is $I^{(l)}_{\M_0\lo} \otimes I^{(l)}_\Cl$.  We will characterise its ``strength'' by upper bounding how its norm grows with $d$ (the Hilbert space dimension of $\mathcal{H}_\Cl$). This is the relevant parameter, since we have seen that its is proportional (up to decaying additive terms) to the total mean energy of the total system, $E_0$ (recall~\cref{eq:E0 clcu}).  We will also multiply said interaction terms by $T_0$, since (as we have seen),  the quantum advantage of the quantum frequential computer is independent of the value of $T_0>0$.

The term $I^{(l)}_{\M_0\lo}$ solely depends on the $l\thh$ gate to be implemented (namely $U(\m_l)$), and (as we have seen) it is independent of $d$ (and $E_0$) in both cases (i.e. in both the semi-classical states and squeezed state case,~\cref{eq:thm fixed memory 1app,eq:thm fixed memory 2} respectively). It thus constitutes a constant factor when computing the strength of $I^{(l)}_{\M_0\lo} \otimes I^{(l)}_\Cl$, and we can neglect it. We thus focus on the terms $ I^{(l)}_\Cl$:  From~\cref{eq:def I in terms of bar V0}, we have 
Hilbert–Schmidt norm
\begin{align}
	\|  I^{(l)}_\Cl \|_2 = \frac{\sqrt{2\pi d}}{T_0} \sqrt{\sum_{k\in\mathcal{S}_d(k_0)} \frac{2\pi}{d}   {\bar V_0}^2\left(\frac{2\pi}{d} k\right) \Bigg{|}_{x_0=x_0^{(l)}}}.
\end{align}
Therefore, in the large $d$ limit, 
\begin{align}
	\frac{T_0}{\sqrt{2\pi d}}	\|  I^{(l)}_\Cl \|_2 &= \sqrt{ \sum_{k=0}^{d-1} \frac{2\pi}{d} {\bar V_0}^2\left(\frac{2\pi}{d} k+ \frac{2\pi}{d} \lceil -d/2+ k_0 \rceil  \right) \Bigg{|}_{x_0=x_0^{(l)}}} =  \sqrt{\int_{0}^{2\pi} dx\, {\bar V_0}^2\left(x+ \frac{2\pi}{d} \lceil -d/2+ k_0 \rceil  \right) \Bigg{|}_{x_0=x_0^{(l)}}}\label{eq:magniture 1} \\
	&= \sqrt{  \int_{0}^{2\pi} dx\, {\bar V_0}^2\left(x  \right) \bigg{|}_{x_0=x_0^{(l)}}}= \| \bar V_0 \|_2  \leq  \| \bar V_0 \|_1 = 1, \label{eq:magniture 2} 
\end{align}
where we have defined
\begin{align}
	\| \bar V_0 \|_p:= \left( \int_{0}^{2\pi} dx\, {\bar V_0}^p\left(x\right) \Big{|}_{x_0=x_0^{(l)}}\right) ^{1/p},
\end{align}
and where in line~\labelcref{eq:magniture 1} we have used the Riemannian integral, while in line~\labelcref{eq:magniture 2} we have used the fact that the smooth function $\bar V_0$ is $2\pi$ periodic and integrating periodic functions over their period  is invariant under a shift of the function's argument. For the last equality, we have used fact that the 1-norm upper bounds the 2-norm and recalled that $\bar V_0\geq 0$ and used~\cref{eq:normalised pt} (which holds independently of the value of real parameter $x_0$).  

Thus we see that the interaction strength grows at most linearly in $\sqrt{d}$ regardless of if we are in the case of semi-classical control states only or any state of the control. Moreover, we can compare this growth with that of the free evolution of the control, $H_\Cl$. Using definition~\labelcref{eq:H_C definition}, we have, in the large $d$ limit 
\begin{align}
	\frac{T_0}{\sqrt{2\pi d}}	\|  H_\Cl \|_2 &= \sqrt{\frac{2\pi}{d} \sum_{n=0}^{d-1} n^2} =  \sqrt{\frac{\pi }{3}} \sqrt{2 d^2-3 d+1}=  \sqrt{\frac{2\pi}{3}}d. \label{eq:2-norm free ham}
\end{align}
Thus the interaction term grows more slowly in Hilbert–Schmidt norm than the interaction term responsible for implementing the $l\thh$ gate. 

Finally, recall that the term in the total Hamiltonian responsible for free evolution of the control, namely $H_\Cl$, is that of the $1\stt$ $d$ levels of a quantum harmonic oscillator with frequency $\omega_0=2\pi/T_0$.  We can alternatively think of this Hamiltonian on $\Cl$ as being that of the quantum harmonic oscillator and restricting the control state to states without support on energy levels greater than $\omega_0 d$ (or, more generally,  it could have support on higher levels but with said support is such tiny magnitude that it has negligible effect on the dynamics).   In this scenario, a common interaction term found in many systems is that of a linear coupling in creation-annihilation operators to the systems the quantum harmonic oscillator interacts with. If we assume this, a meaningful comparison of interaction strength is to consider the Hilbert-Schmidt norm of the creation or annihilation operator when truncated to the 1st $d$ levels in energy, which we denote $a_d^\dag$, $a_d$ respectively. Thus using the expression  $a^\dag=\sum_{j=1}^\infty \sqrt{j} \ketbra{E_{j+1}}{E_{j}}$, (where $\ket{E_{j}}$ is the $j\thh$ energy level of the harmonic oscillator), we find in  the large $d$ limit
\begin{align}
	\| a_d^\dag \|_2 = 	\| a_d \|_2 =  \sqrt{\sum_{j=1}^{d-1} j}=  \frac{\sqrt{(d-1) d}}{\sqrt{2}}= \frac{d}{2}.\label{eq:2-norm creastion annihilation}
\end{align}
Thus we see from~\cref{eq:2-norm free ham,eq:2-norm creastion annihilation} that the Hilbert–Schmidt norm of the free Hamiltonian, $\|  H_\Cl \|_2$, grows faster with $d$ than that of creation-annihilation couplings,	$\| a_d^\dag \|_2$,	$\| a_d \|_2$   by a multiplicative factor of $\sqrt{d}$.  However, the Hilbert–Schmidt norm interaction terms $I^{(l)}_{\M_0\lo} \otimes I^{(l)}_\Cl$ grow slower than $\|  H_\Cl \|_2$ by at least a multiplicative factor of $d$ and slower than 	$\| a_d^\dag \|_2$ or $	\| a_d \|_2$ by a at least a $\sqrt{d}$ multiplicative factor.

\section{Proof of~\Cref{thm:contrl with two clocks}: Attaining the quantum limit with only a semi-classical bus}\label{sec:proof of 2 clock theomre in main text}
The proof will rely heavily on material from~\Cref{sec: proof of 1st quantum clock thorem}.  We first prove a theorem which allows us to decouple the errors originating from the control of the gates and the control of the memory. 

For this we will need to introduce a few definitions (eqs.~\eqref{eq:reproduce hams} are reproduced from the main text for convenience).
\begin{align}
	H_{\M\W\lo\Cl\Cl_2}&:= H_{{\M_0}\W\lo\Cl} + H_{\M\W\Cl_2},\quad H_{\M\W\Cl_2} := H_{\Cl_2} +   \sum_{l=1}^{N_g} I_{\M\W} ^{(l)}\otimes I_{\Cl_2}^{(l)},\quad H_{{\M_0}\W\lo\Cl}:= H_\Cl+ \sum_{l=1}^{N_g} I_{\M_0 \W\lo}^{(l)}\otimes I_{\Cl}^{(l)},\label{eq:reproduce hams} 
\end{align}
\begin{align}
H_{\M\W\lo\Cl\Cl_2}^{(k)}&:= H_{{\M_0}\W\lo\Cl}^{(k)} + \bar H_{\M\W\Cl_2}^{(k)},\label{eq:total ham with k dependency}
\end{align}
$k=1,2,\ldots, N_g$, where we define $H_{{\M_0}\W\lo\Cl}^{(k)} $ via
\begin{align}
H_{{\M_0}\W\lo\Cl}^{(k)} &:= H_\Cl +    I_{{\M_0}\W\lo} ^{(k)}\otimes I_\Cl^{(k)} ,  \\
I_{{\M_0}\W\lo} ^{(k)} &:=  I_{{\M_0}\lo} ^{(k)} + \proj{\zero}_{\M_{0,k}}\otimes I^{(k)}_{\W_k},\label{eq:int meme sys swit}
\end{align}
with $ I_{{\M_0}\lo} ^{(k)}$ defined in~\cref{eq:interaction terms with subscripts} and we restrict the  spectrum of $I_{\W_k}^{(k)}$ to lie in the interval $(0,2\pi]$. The term $	\bar H_{\M\W\Cl_2}^{(k)}$ is defined via
\begin{align}\label{eq:main thm2 incomple ham}
\bar H_{\M\W\Cl_2}^{(k)}:= H_{\M\W\Cl_2} - I_{\M\W}^{(k)}\otimes I_{\Cl_2}^{(k)}= H_{\Cl_2}+ \sum_{\substack{l=1  \\ l\neq k}}^{N_g} I_{\M\W} ^{(l)}\otimes I_{\Cl_2}^{(l)},
\end{align}
where $H_{\M\W\Cl_2}$ is defined in~\cref{eq:main thm2 complete ham} to be $H_{\M\W\Cl_2} = H_{\Cl_2} +   \sum_{l=1}^{N_g} I_{\M\W} ^{(l)}\otimes I_{\Cl_2}^{(l)}$.  Here we additionally define the structure
\begin{align}
I_{\M\W}^{(l)}:= I_\M^{(l)}\otimes \proj{\on}_{\W_l}. \label{eq:Int ram on switch def}
\end{align}
We will define the terms $\{I_\M^{(l)}, I_{\Cl_2}^{(l)}\}_l$ and $H_{\Cl_2}$ in the proofs when the need for their definitions arises. The terms $\{I_{\M_0\lo}^{(l)}, I_{\Cl}^{(l)}\}_l$ and $H_{\Cl}$ have been defined in~\Cref{sec: proof of 1st quantum clock thorem}.

\subsection{Main technical lemma: A decoupling of error contributions}

The following lemma permits one to decouple the errors resulting from the application of the gates under Hamiltonian $H_{\M_0\W\lo\Cl}$ and the errors due to shuttling memory around via Hamiltonian $H_{\M_0\W\Cl_2}$ (Indeed, in the following, lines~\labelcref{line:CPU bus decoupling lemma l2,line:CPU bus decoupling lemma l4} correspond to errors associated to dynamics under $H_{\M_0 \W\lo\Cl}$ while lines~\labelcref{line:CPU bus decoupling lemma l3,line:CPU bus decoupling lemma l5}, to  errors associated to dynamics under $H_{\M\W\Cl_2}$).

Unless stated otherwise, the lemmas and theorems in this section hold for all states $\ket{0}_\lo\in \mathcal{C}_\Cl^\textup{clas.}$. The states $\{\ket{t_{k,r}}_{\M\Cl_2} \}_{k,r}$ are assumed to obey~\cref{eq:cell restriction} throughout. They will later be further specialised in~\Cref{sec:description of clock 2} and all proceeding lemmas and theorems will apply to these specialised versions. Likewise the states of the control on $\Cl$, $\{\ket{t_{k,r}}_{\M\Cl_2} \}_{k,r}$, will be general at first and then specialised. The same is true for the Hamiltonians. The only generic assumption is that they act on a finite-dimensional Hilbert-space and are Hermitian (the former assumption could be easily relaxed for a lot of the lemmas in this section, but for our purposes this will be irrelevant and thus we have this assumption for simplicity). 

The states $\{ \ket{t_{r,k}}_{\W} \}_{r,k}$ and $\{ \ket{t_{k,l}}_\lo\}_{k,l}$  are given by~\cref{eq:switch 1,eq:switch 2,eq:gate implementation clock 2} respectively throughout. All lemmas and theorems hold for any gate set $\mathcal{U}_\mathcal{G}$. 

\begin{lemma} [Quantum-control-and-bus error decoupling]\label{lem:bus cpu decoupling} For $j=1,2,3, \ldots, N_g$ and $l=0,1,2,\ldots, L$ we have
\begin{align}\label{eq:first lem unitary quant eq1 version 2}
	&	\big\|  \me^{-\mi t_{j,l} H_{\M\W\lo\Cl\Cl_2}} \ket{0}_\M\ket{0}_\W\ket{0}_\lo \ket{0}_{\Cl} \ket {0}_{\Cl_2} - \ket{t_{j,l}}_{\M\Cl_2}\ket{t_{j,l}}_\W\ket{t_{j,l}}_\lo \ket{t_{j,l}}_\Cl \big \|_2 \\ 
	&\leq  \sum_{r=0}^l \sum_{k=1}^j \Bigg( \big \| \ket{t_{k,r}}_{\W_k} \ket{t_{k,r}}_\lo \ket{t_{k,r}}_\Cl 
	- \me^{-\mi t_1 H_{\W\lo\Cl}^{(k,\m_{r,k})}} \ket{t_{k-1,r}}_{\W_k}\ket{t_{k-1,r}}_\lo \ket{t_{k-1,r} }_\Cl\big \|_2  \label{line:CPU bus decoupling lemma l2}  \\ 
	&\quad +
	\| \ket{t_{k,r}}_{\bar\M_{0,k} \Cl_2} \ket{t_{k,r}}_{\bar\W_k}
	- \me^{-\mi t_1 \bar H_{\M\W\Cl_2}^{(k)}} \ket{t_{k-1,r}}_{\bar\M_{0,k} \Cl_2}\ket{t_{k-1,r}}_{\bar \W_k} \big \|_2  \label{line:CPU bus decoupling lemma l3} 
	\\
	&\quad+ t_1  \max_{x\in[0,t_1]} \max_{\,\,\,\big\{\m_l\in\mathcal{G}\cup\{\zero\}\big\}_{l=1}^{N_g}} \big \| \bar{H}_{\W\lo\Cl}^{(k,\vec \m)} \,\me^{-\mi x H_{\W\lo\Cl}^{(k, \m_k)}} \ket{t_{k-1,r}}_\W\ket{t_{k-1,r}}_\lo \ket{t_{k-1,r} }_\Cl \big \|_2 \label{line:CPU bus decoupling lemma l4} \\
	&\quad+  t_1  \max_{y\in[0,t_1]}   \big \| \left(I_{\M\W}^{(k)}\otimes I_{\Cl_2}^{(k)}\right) \me^{-\mi y \bar H_{\M\W\Cl_2}^{(k)}} \ket{t_{k-1,r}}_{\M \Cl_2} \ket{t_{k-1,r}}_\W \big \|_2  \Bigg)  \label{line:CPU bus decoupling lemma l5}
\end{align}
where 
\begin{align}
	H_{\W\lo\Cl}^{(k, \m_k)}&:= {}_{\M_{0,\hash}}\!\!\braket{\vec \m |   		H_{{\M_0}\W\lo\Cl}^{(k)}  |\vec \m } _{\M_{0,\hash}} =H_\Cl + I_\lo^{(k,\m_k)}\otimes I_\Cl^{(k)} + I_\W^{(k,\m_k)} \otimes I_\Cl^{(k)} \\
	H_{\W\lo\Cl}^{(k, \m_{r,k})}&:= H_{\W\lo\Cl}^{(k, \m_k)}{\big |}_{\m_k\,\mapsto\, \m_{r,k}}\\
	\bar H_{\W\lo\Cl}^{(k,\vec \m)}&:=  {}_{\M_{0,\hash}}\!\!\braket{\vec \m |   	\left( H_{{\M_0} \W\lo\Cl}- 	H_{{\M_0}\W\lo\Cl}^{(k)} 	\right)  |\vec \m } _{\M_{0,\hash}}   =\sum_{\substack{q=1 \\ q\neq k}}^{N_g}\left( I_{\lo}^{(q,\m_q)}\otimes I_\Cl^{(q)}+ I_{\W}^{(q,\m_q)}\otimes I_\Cl^{(q)}\right),\label{eq:def bar H WSC}
\end{align}
where $\vec \m=(\m_1,\m_2,\ldots, \m_{N_g})$,  $\ket{\vec \m}_{\M_{0,\hash}}:=\ket{\m_1}_{\M_{0,1}}\ket{\m_2}_{\M_{0,2}}\ldots \ket{\m_{N_g}}_{\M_{0,{N_g}}}$ with $\m_l\in\mathcal{G}\cup\{\zero\}$  and  recall $I_\lo^{(q,\m_q)}$ is defined in~\cref{eq:interaction terms with subscripts} for $\m_l\in\mathcal{G}$. For $\m_l=\zero$ we define $I_\lo^{(k,\zero)}:=\hat 0$ (with $\hat 0$ the zero operator) and $I_\W^{(k,\zero)}:= I_{\W_k}^{(k)}$ and  $I_{\W_k}^{(k,\m_{k})}:=\hat 0$ for all $\m_l\in\mathcal{G}$. 
\end{lemma}
Note that while in the above 4 lines on the r.h.s. of the inequality, the kets belong to different tensor-product subspaces, these subspaces are consistent with the spaces upon which the distinct Hamiltonians act non-trivially upon.
\begin{proof}\label{lem:sepration of control contributions of two clocks} For $j=1,2,3, \ldots, N_g $, $l=0,1,2,\ldots,L$
\begin{align}\label{eq:first lem unitary eq 1}
	&	\big\|  \me^{-\mi t_{j,l} H_{\M\W\lo\Cl\Cl_2}} \ket{0}_\M\ket{0}_\W\ket{0}_\lo\ket {0}_\Cl\ket {0}_{\Cl_2} - \ket{t_{j,l}}_{\M\Cl_2}\ket{t_{j,l}}_\W\ket{t_{j,l}}_\lo \ket{t_{j,l} }_\Cl \big \|_2 \\
	&\leq \sum_{r=0}^l \sum_{k=1}^j  \big \| \ket{t_{k,r}}_{\M\Cl_2}\ket{t_{k,r}}_\W\ket{t_{k,r}}_\lo \ket{t_{k,r}}_\Cl \label{eq:line 1 first lem unitary eq 1}  \\
	&\qquad\qquad\qquad\qquad\qquad\qquad
	- \me^{-\mi t_1 H_{\M\W\lo\Cl\Cl_2}}  \ket{t_{k-1,r}}_{\M\Cl_2}\ket{t_{k-1,r}}_\W\ket{t_{k-1,r}}_\lo \ket{t_{k-1,r}}_\Cl \big \|_2 \\ 
	&\leq  \sum_{r=0}^l  \sum_{k=1}^j  \bigg(   \big \| \ket{t_{k,r}}_{\M\Cl_2}\ket{t_{k,r}}_\W\ket{t_{k,r}}_\lo \ket{t_{k,r} }_\Cl  \\
	&\qquad\qquad\qquad\qquad\qquad\qquad\qquad 
	- \me^{-\mi t_1 H_{\M\W\lo\Cl\Cl_2}^{(k)}}  \ket{t_{k-1,r}}_{\M\Cl_2}\ket{t_{k-1,r}}_\W\ket{t_{k-1,r}}_\lo \ket{t_{k-1,r} }_\Cl\big \|_2   \\
	&\quad+    t_1  \max_{x\in[0,t_1]} \big \| \bar{H}_{\M\W\lo\Cl\Cl_2}^{(k)} \me^{-\mi x H_{\M\W\lo\Cl\Cl_2}^{(k)}} \ket{t_{k-1,r}}_{\M\Cl_2}\ket{t_{k-1,r}}_\W\ket{t_{k-1,r}}_\lo \ket{t_{k-1,r} }_\Cl  \big \|_2 \bigg ) ,\label{eq:line 2 first lem unitary eq 1}
\end{align}
where $H_{\M\W\lo\Cl\Cl_2}^{(k)}$ is given by~\cref{eq:total ham with k dependency} and $\bar H_{\M\W\lo\Cl\Cl_2}^{(k)}$ by
\begin{align}\label{eq:local Ham defs 2 clocks}
	\bar H_{\M\W\lo\Cl\Cl_2}^{(k)} :=H_{\M\W \lo \Cl\Cl_2}-H_{\M\W\lo\Cl\Cl_2}^{(k)}=H_{\M_0\W \lo\Cl}-H_{\M_0\W \lo\Cl}^{(k)} + I_{\M\W} ^{(k)}\otimes I_{\Cl_2}^{(k)} =\bar H_{\M_0\W\lo\Cl}^{(k)}+ I_{\M\W} ^{(k)}\otimes I_{\Cl_2}^{(k)},
\end{align}
where in the last line we substituted using~\cref{eq:main mapped complete ham specialised}. Lines~\labelcref{eq:line 1 first lem unitary eq 1} to~\labelcref{eq:line 2 first lem unitary eq 1} follow analogously to the proof of~\Cref{lem:quantum dynamics upper bounded by two terms}.  Now recall that $I_{\M\W}^{(l)}$ acts trivially on memory cells $\M_{0,1},\M_{0,2},\ldots, \M_{0,l-1},\M_{0,l+1},\ldots,\M_{0,N_g}$ and switches $\W_{1},\W_{2},\ldots, \W_{l-1},\W_{l+1},\ldots,\W_{N_g}$ (see~\cref{eq:main thm2 complete ham}). Therefore, it follows from the definitions of  $H_{{\M_0}\W\lo\Cl}^{(k)}$ and $\bar H_{\M\W\Cl_2}^{(k)}$ that these terms commute since they only act non-trivially on different Hilbert spaces. Therefore $\me^{-\mi t_1 H_{\M\W\lo\Cl\Cl_2}^{(k)}}=  \me^{-\mi t_1 H_{{\M_0}\W\lo\Cl}^{(k)}} \me^{-\mi t_1 \bar H_{\M\W\Cl_2}^{(k)}}$. Now recall the following identities for operators $O_\textup{A}$, $O_\textup{B}$ and kets $\ket{A}_\textup{A}$, $\ket{A'}_\textup{A}$ and   $\ket{B}_\textup{B}$, $\ket{B'}_\textup{B}$, on Hilbert spaces $\mathcal{H}_\textup{A}$ and $\mathcal{H}_\textup{B}$ respectively:  $\|\ket{A}_\textup{A}\ket{B}_\textup{B}-\ket{A'}_\textup{A}\ket{B'}_\textup{B}\|_2\leq \| \ket{A}_\textup{A}-\ket{A'}_\textup{A}\|_2+\| \ket{B}_\textup{B}-\ket{B'}_\textup{B}\|_2$, which follows from the triangle inequality, and   $\|(O_\textup{A}\otimes O_\textup{B}) \ket{A}_\textup{A}\ket{B}_\textup{B}\|_2= \| O_\textup{A} \ket{A}_\textup{A}\|_2 \, \| O_\textup{B} \ket{B}_\textup{B}\|_2$. \\
For $j=1,2,3, \ldots, N_g$, $l\in 0,1,2,\ldots, L$
\begin{align}
	&	\big\|  \me^{-\mi t_{j,l} H_{\M\W\lo\Cl\Cl_2}} \ket{0}_\M\ket{0}_\W\ket{0}_\lo\ket {\Psi(0)}_\Cl\ket {\Psi(0)}_{\Cl_2} - \ket{t_{j,l}}_{\M\Cl_2}\ket{t_{j,l}}_\W\ket{t_{j,l}}_\lo \ket{\Psi(t_{j,l} d/T_0)}_\Cl \big \|_2 \\
	&\leq  \sum_{r=0}^l \sum_{k=1}^j \Bigg(\\
	&\quad+ \big \| \ket{t_{k,r}}_{\M_{0,k}}\ket{t_{k,r}}_{\W_k}\ket{t_{k,r}}_\lo \ket{\Psi(t_{k,r} d/T_0)}_\Cl 
	- \me^{-\mi t_1 H_{\M_0\W\lo\Cl}^{(k)}}  \ket{t_{k-1,r}}_{\M_{0,k}}\!\ket{t_{k-1,r}}_{\W_k}\!\ket{t_{k-1,r}}_\lo \ket{\Psi(t_{k-1}d/T_0) }_\Cl\big \|_2 \label{line:1st summer 2 clocks}\\ 
	&\quad +
	\big\| \ket{t_{k,r}}_{\bar\M_{0,k} \Cl_2}\ket{t_{k,r}}_{\bar \W_k} 
	- \me^{-\mi t_1 \bar H_{\M\W\Cl_2}^{(k)}}  \ket{t_{k-1,r}}_{\bar\M_{0,k} \Cl_2} \ket{t_{k-1,r}}_{\bar\W_k}\big \|_2 
	\\
	&\quad+ t_1  \max_{x\in[0,t_1]} \bigg(  \big \| \bar{H}_{{\M_0}\W\lo\Cl}^{(k)} \,\me^{-\mi x H_{{\M_0}\W\lo\Cl}^{(k)}}\ket{t_{k-1,r}}_\lo \ket{\Psi(t_{k-1,r}d/T_0) }_\Cl \me^{-\mi x \bar H_{\M\W\Cl_2}^{(k)}} \ket{t_{k-1,r}}_{\M\Cl_2} \ket{t_{k-1,r}}_\W  \big \|_2\label{eq:3rf line 2 clocsk 1sr proofg}\\
	&\quad+    \big \| \left(I_{\M\W}^{(k)}\otimes I_{\Cl_2}^{(k)}\right) \me^{-\mi x \bar H_{\M\W\Cl_2}^{(k)}} \ket{t_{k-1,r}}_{\M\Cl_2}\ket{t_{k-1,r}}_\W \big \|_2 \bigg)\Bigg) ,\label{eq:line fors lemma 2 clocsk}
\end{align}
where in line~\labelcref{eq:line fors lemma 2 clocsk} we have used the fact that $H_{\M_0 \W\lo\Cl}^{(k)}$ is block-diagonal in the $\{\ket{\m}_{\M_{0,k}}\}_{\m\in{\mathcal{G}}\cup\{\zero\}}$ basis.  	To complete the proof for line~\labelcref{line:1st summer 2 clocks}, we first recall~\cref{eq:cell restriction} (which asserts that $\ket{t_{k-1,r}}_{\M_{0,k}}=  \ket{t_{k,r}}_{\M_{0,k}}=\ket{\m_{k,r}}_{\M_{0,k}}$). Now taking into account the block diagonality of $H_{\M_0\W \lo \Cl}^{(k)}$, we conclude line~\labelcref{line:CPU bus decoupling lemma l2}.	 For line~\labelcref{eq:3rf line 2 clocsk 1sr proofg}, first note that $\me^{-\mi x \bar H_{\M\W\Cl_2}^{(k)}} \ket{t_{k-1,r}}_{\M\Cl_2}\ket{t_{k-1,r}}_\W=   \ket{t_{k-1,r}}_\W \me^{-\mi x \bar H_{\M\Cl_2}^{(k,\vec w)}} \ket{t_{k-1,r}}_{\M\Cl_2}$ where $\bar H_{\M\Cl_2}^{(k,\vec w)}$  $:=$  ${}_\W\!\braket{ t_{k-1,r} |  \bar H_{\M\W\Cl_2}^{(k)}| t_{k-1,r}   }_\W$, since $H_{\M\Cl_2}^{(k)}$ is block-diagonal in the basis of the switches: $\{  \ket{\on}_{\W_l},\ket{\off}_{\W_l} \}_{l=1}^{N_g}$. Second, we expand the normalised vector $\me^{-\mi x \bar H_{\M\Cl_2}^{(k,\vec w)}} \ket{t_{k-1,r}}_{\M\Cl_2}$ in the orthonormal basis of the memory $\big\{ \Pi_{k,l} \ket{\m_{k,l}}_{\M_{k,l}}\, | \, \m_{k,l}\in\mathcal{G} \text{ or } \m_{0,l}=\zero\big\}$ and an arbitrary orthonormal basis for the state of $\Cl_2$. Third, we now note that $\bar{H}_{{\M_0}\W\lo\Cl}^{(k)} \,\me^{-\mi x H_{{\M_0}\W\lo\Cl}^{(k)}}$ acts trivially on $\Cl_2$ and is block-diagonal in the above orthonormal basis of the memory. By applying the definition of the two-norm, line~\labelcref{line:CPU bus decoupling lemma l4}  follows. 

\end{proof}

\subsection{Lemmas bounding error contributions from control on $\Cl$}
We will now state and prove a lemma which will bound lines~\labelcref{line:CPU bus decoupling lemma l2,line:CPU bus decoupling lemma l4}  in~\cref{eq:first lem unitary quant eq1 version 2}. We use the same specialised control states as in~\Cref{sec: proof of 1st quantum clock thorem}, namely $\{ \ket{t_{k,r}}_\Cl=	\ket{\Psi(t_{k,r}d/T_0)}_\Cl  \}_{k,r}$ where $	\ket{\Psi(td/T_0)}_\Cl $ is defined in~\cref{eq:qusi idea no pot def}. These states satisfy the cyclicity condition,~\cref{cond:eq:cyclicity of cl state}, as shown in the proof to the following lemma.

\begin{lemma}[Bound on quantum-control-like terms]\label{eq:first lem unitary quant eq1 version 4} There exists parametrizations of the  control states $\{ \ket{\Psi(t_{k,r} d/T_0)}_\Cl\}_{k,r}$ (\,$k=0,1,2,\ldots, N_g$; $r=0,1,2,\ldots, L$) and Hamiltonian $H_{{\M_0}\lo\Cl}^{(k)}$ in terms of  $\bar\varepsilon$ such that the following holds for all $\bar\varepsilon\in(0,1/6)$ and $j=1,2,3, \ldots, N_g$;  $l\in0,1,\ldots,L$, 
\begin{align}
	& \sum_{r=0}^l \sum_{k=1}^j \bigg( \big \| \ket{t_{k,r}}_{\W_k}\ket{t_{k,r}}_\lo \ket{\Psi(t_{k,r} d/T_0)}_\Cl 
	- \me^{-\mi t_1 H_{\W\lo\Cl}^{(k,\m_{r,k})}} \ket{t_{k-1,r}}_{\W_k}\ket{t_{k-1,r}}_\lo \ket{\Psi(t_{k-1,r}d/T_0) }_\Cl\big \|_2 \label{eq:1 line c2}  \\ 
	&\quad\qquad\qquad+ t_1  \max_{x\in[0,t_1]} \max_{\,\,\,\{\m_l\in\mathcal{G}\cup\{\zero\}\}_{l=1}^{N_g}} \big \| \bar{H}_{\W\lo\Cl}^{(k,\vec \m)} \,\me^{-\mi x H_{\W\lo\Cl}^{(k, \m_k)}}\ket{t_{k-1,r}}_\W \ket{t_{k-1,r}}_\lo \ket{\Psi(t_{k-1,r}d/T_0) }_\Cl \big \|_2 \bigg)\label{eq:2 line c2} \\
	& \leq   \left(\sum_{r=0}^l   \sum_{k=1}^j  \tilde d(\m_{r,k}) \right) h( \bar\varepsilon)  \, \textup{poly}(d) \, d^{-3/\sqrt{\bar\varepsilon}},  
\end{align}
and
\begin{align}
	N_g= \lfloor d^{1-\bar\varepsilon} \rfloor,
\end{align}
where the function  $h\geq 0$ is independent of $d$ and the elements $\{\tilde d(\m_{r,k}) \}_{k=1}^j$. Meanwhile, $\textup{poly}(d)\geq 0$ is independent of $\bar\varepsilon$ and the elements $\{\tilde d(\m_{r,k})  \}_{k=1}^j$.

\end{lemma}

\begin{proof}
In~\Cref{thm:comptuer with fixed memory}, the clock on $\Cl$ ran over one oscillation of the oscillator over a total time $T_0=2\pi/\omega_0$. In the current setup, we are running the computer over multiple $l$ runs of the oscillator. The proof consists in relating terms from the $l\thh$ run of the oscillator to the $1\stt$ run, and then using the results from~\Cref{sec: proof of 1st quantum clock thorem} to bound them.  

We start by recalling that $t_{l,r}=t_{l+r N_g}=t_l +r T_0$. Thus, recalling the definition of $\ket{\Psi(t d/T_0) }_\Cl$ in~\cref{eq:qusi idea no pot def}, we find  
\begin{align}
	\mathcal{S}_d(t_{l,r}d/T_0) &= \{ k-rd\, | \,k\in\zz \text{ and } -d/2\leq t_ld/T_0-k<d/2 \},\\
	\psi_\textup{nor}(t_{l,r}d/T_0)&= A_\textup{nor} \,\me^{-\frac{\pi}{\sigma^2} \left([k-rd]-t_ld/T_0\right)^2} \me^{-\mi 2\pi ([k-rd] -t_ld/T_0)/d},
\end{align}
from which it follows
\begin{align}
	\ket{\Psi(t_{l,r} d/T_0) }_\Cl= \ket{\Psi(t_{l} d/T_0) }_\Cl  \label{eq:clock cyclcicity}
\end{align}
 for all $j=1,2,\ldots, N_g$, and $l\in\nnz$.

A state of the logical computational space $\ket{t_{p,r}}_\lo$ is generated by  applying $p$ gates to the state of the logical computational space after $r N_g$ gates to it, i.e. $\ket{t_{p,r}}_\lo=\Pi_{l=1}^p U(\m_{r,l}) \ket{t_{0,r}}_\lo$, for $r=1,2,\ldots, L$. 
For $r=0$, recall that we are applying unitaries to the switches to turn them on sequentially, rather than applying gates to the computation: at times $t_{p,0}$, we apply $U(\zero)\ket{\off}_{\W_p}=\ket{\on}_{\W_p}$. Let us start by evaluating the terms in~\cref{eq:1 line c2,eq:2 line c2} for $r=0$ and $r>0$ separately:
\begin{align}
	&\big \| \ket{t_{k,r}}_{\W_k}\ket{t_{k,r}}_\lo \ket{\Psi(t_{k,r} d/T_0)}_\Cl 
	- \me^{-\mi t_1 H_{\W\lo\Cl}^{(k,\m_{r,k})}} \ket{t_{k-1,r}}_{\W_k}\ket{t_{k-1,r}}_\lo \ket{\Psi(t_{k-1,r}d/T_0) }_\Cl\big \|_2 \\
	&  =
	\begin{cases}
		\big \| \ket{\on}_{\W_k}\ket{\Psi(t_k d/T_0)}_\Cl 
		- \me^{-\mi t_1 H_{\W\Cl}^{(k)}} \ket{\off}_{\W_k} \ket{\Psi(t_{k-1}d/T_0) }_\Cl\big \|_2  &\mbox{ if } r=0\\
		\big \| \ket{t_{k,r}}_\lo\ket{\Psi(t_k d/T_0)}_\Cl 
		- \me^{-\mi t_1 H_{\lo\Cl}^{(k,\m_k)}} \ket{t_{k-1,r}}_\lo\ket{\Psi(t_{k-1}d/T_0) }_\Cl\big \|_2   &\mbox{ if } r=1,2,\ldots, L
	\end{cases}\label{eq:1 cases switch sys}
\end{align}
where
\begin{align}
	H_{\W\Cl}^{(k)}:= & H_\Cl+ I_\W^{(k)}\otimes I_\Cl^{(k)} ,\\ 
	H_{\lo\Cl}^{(k,\m_k)}:= & H_\Cl+ I_\lo^{(k,\m_k)}\otimes I_\Cl^{(k)} , \quad \m_k\in\mathcal{G} 
\end{align}
Similarly,
\begin{align}
	& \big \| \bar{H}_{\W\lo\Cl}^{(k,\vec \m)} \,\me^{-\mi x H_{\W\lo\Cl}^{(k, \m_k)}}\ket{t_{k-1,r}}_\W \ket{t_{k-1,r}}_\lo \ket{\Psi(t_{k-1,r}d/T_0) }_\Cl \big \|_2 \\
	&=
	\begin{cases}
		\big \| \bar{H}_{\W\lo\Cl}^{(k,\vec \m)}  \,\me^{-\mi x H_{\W\Cl}^{(k)}} \ket{t_{k-1,r}}_\W \ket{t_{k-1,r}}_\lo \ket{\Psi(t_{k-1}d/T_0) }_\Cl \big \|_2      &\mbox { if } \m_k=\zero\\
		\big \| \bar{H}_{\W\lo\Cl}^{(k,\vec \m)}  \,\me^{-\mi x H_{\lo\Cl}^{(k,\m_k)}} \ket{t_{k-1,r}}_\W \ket{t_{k-1,r}}_\lo  \ket{\Psi(t_{k-1}d/T_0) }_\Cl \big \|_2      &\mbox { if }  \m_k\in\mathcal{G},
	\end{cases}\label{eq:2 cases switch sys}
\end{align}	
where we can write $\bar{H}_{\W\lo\Cl}^{(k,\vec \m)} $ in the form $\bar{H}_{\W\lo\Cl}^{(k,\vec \m)} =  \sum_{\substack{q=1 \\ q\neq k}}^{N_g}\left( I_{\mathrm{\gamma_1}(\m_q)}^{(q,\m_q)}\otimes I_\Cl^{(q)}\right)$, where $\mathrm{\gamma_1}(\m_q)=\lo$ if $\m_q\in\mathcal{G}$, and $\mathrm{\gamma_1}(\m_q)=\W$ if $\m_q=\zero.$  (This follows from its definition;~\cref{eq:def bar H WSC}.)
In~\Cref{sec: proof of 1st quantum clock thorem}, we dealt with states of this form in~\cref{eq:1 cases switch sys,eq:2 cases switch sys} (i.e. in~\Cref{lem:1st secondrary lem for unitary bound,lem:2st secondrary lem for unitary bound} respectively).

Using the above identities we can now apply the proof of~\Cref{thm:comptuer with fixed memory} from lines~\labelcref{eq:line2 of first theorem unitary quant eq1,eq:line3 of first theorem unitary quant eq1} onwards. Importantly, the only difference is that we have $H_{\W \Cl}^{(k)}$ and $\bar H_{\W\lo \Cl}^{(k,\vec\m)}$ (or $H_{\lo \Cl}^{(k)}$ and $\bar H_{\W\lo \Cl}^{(k,\vec \m)}$), rather than $H_{\lo \Cl}^{(k)}$ and $\bar H_{\lo \Cl}^{(k)}$. These only differ by the fact that the latter have interaction terms $I_{\lo}^{(l)}$ while the former has terms $I_{\lo}^{(l,\m_l)}$ or $I_{\W}^{(l)}$. Recall that $I_{\lo}^{(l)}$ is a generic term responsible of implementing any unitary on $\lo$, while $I_{\lo}^{(l,\m_l)}$ is responsible for implementing the gate $\m_l\in\mathcal{G}$ on $\lo$, and $I_{\W}^{(l)}$ is responsible for implementing the gate $\zero$ on $\W$. Since the proof of~\Cref{thm:comptuer with fixed memory} was for all generic terms $I_{\lo}^{(l)}$, it also applies in the case at hand (once one identifies the subsystems $\W$ and $\lo$ in the current proof with $\lo$ in the original proof) and the maximization over gates, $ \max_{\,\,\,\{\m_l\in\mathcal{G}\cup\{\zero\}\}_{l=1}^{N_g}}$, vanishes.

Thus from~\cref{eq:first theorem unitary quant eq3}, we conclude that there exists an initial clock state $\ket{\Psi{(0)}}_\Cl$ (the same clock state used in case 2) of~\Cref{thm:comptuer with fixed memory}) such that 
\begin{align}
	& \sum_{k=1}^j \bigg( \big \| \ket{t_{k,r}}_{\W_k}\ket{t_{k,r}}_\lo \ket{\Psi(t_{k,r} d/T_0)}_\Cl 
	- \me^{-\mi t_1 H_{\W\lo\Cl}^{(k,\m_{r,k})}} \ket{t_{k-1,r}}_{\W_k}\ket{t_{k-1,r}}_\lo \ket{\Psi(t_{k-1,r}d/T_0) }_\Cl\big \|_2 \label{eq:1 line c222}  \\ 
	&\quad\qquad\qquad+ t_1  \max_{x\in[0,t_1]} \max_{\,\,\,\{\m_l\in\mathcal{G}\cup\{\zero\}\}_{l=1}^{N_g}} \big \| \bar{H}_{\W\lo\Cl}^{(k,\vec \m)} \,\me^{-\mi x H_{\W\lo\Cl}^{(k, \m_k)}}\ket{t_{k-1,r}}_\W \ket{t_{k-1,r}}_\lo \ket{\Psi(t_{k-1,r}d/T_0) }_\Cl \big \|_2 \bigg)\\
	&\leq \left(\sum_{k=1}^j  \tilde d(\m_{r,k}) \right) f(\eta, \bar\varepsilon)  \, \textup{poly}(d) \, d^{-3/\sqrt{\bar\varepsilon}}, 
	 \label{eq:rhs new one in theonre 2}
\end{align}
where the parameters are defined as per~\cref{eq:first theorem unitary quant eq3}; the only difference between line~\labelcref{eq:rhs new one in theonre 2} and the r.h.s. of~\cref{eq:first theorem unitary quant eq3} is the replacement of $\tilde d(\m_{k})$ with $\tilde d(\m_{r,k})$. Thus we have that~\cref{eq:rhs new one in theonre 2} holds for all $\eta\in(0,1]$, $\bar\varepsilon\in(0,1/6)$, $d\in\nnp$. The function  $f\geq 0$ is independent of $d$ and the elements $\{\tilde d(\m_{r,k})  \}_{k=1}^j$. Meanwhile, $\textup{poly}(d)\geq 0$ is independent of, $\eta$, $\bar\varepsilon$, and the elements $\{\tilde d(\m_{r,k})  \}_{k=1}^j$.

To finalise the proof, we simply need to specialise to case 2) by choosing $\eta$ such that $\varepsilon_g=\bar\varepsilon$, which provides that parametrization of  $\eta$ as a function of $\eta$ according to~\cref{eq:eta as func of ep}. We then define $h(\bar\varepsilon):=f(\bar\varepsilon,\bar\varepsilon)$ and recall the definition of $N_g$ in terms of $\varepsilon_g$ in~\cref{eq:def: Ng in terms of ep g}.
\end{proof}

\subsection{Description of the control on $\Cl_2$}\label{sec:description of clock 2}

In order to bound the terms associated with the dynamics of Hamiltonian $H_{\M\W \Cl_2}$, in~\cref{eq:first lem unitary quant eq1 version 2} (namely lines~\labelcref{line:CPU bus decoupling lemma l3,line:CPU bus decoupling lemma l5}), we have to specialise further the Hamiltonian. We do this here. In particular, the free Hamiltonian of the clock $H_{\Cl_2}$ will be a copy the free Hamiltonian of the clock controlling the gates, $H_\Cl$, i.e. $H_{\Cl_2}=\sum_{n =0}^{d-1} \omega_0 n \proj{E_n}_{\Cl_2}$, where $\omega_0=2\pi/T_0$ and $\{ \ket{E_n}\}_{n}$ is an orthonormal basis for the Hilbert space associated with $\Cl_2$. Also similar to before, the clock interaction terms, $\{I_\M^{(l)}\}_{l=1}^{N_g}$, will be chosen to be diagonal in the discrete Fourier Transform basis associated with $\{\ket{E_n}_\Cl\}_n$, namely
\begin{align}
I_{\Cl_2}^{(l)}:= \frac{d}{T_0} \sum_{k\in\mathcal{S}_d(k_0)} I_{{\Cl_2}, d}^{(l)}(k)\proj{\theta_k}_{\Cl_2},\qquad  I_{{\Cl_2}, d}^{(l)}(x):= \frac{2\pi}{d} \bar V_0\left(\frac{2\pi}{d} x\right)\Bigg{|}_{{x_0}=x_0^{\prime (l)}},   \label{eq:I C2 def}
\end{align}
where $I_{{\Cl_2}, d}(k)$ is chosen such that $	I_{\Cl_2}^{(l)}$ is independent of $k_0 \in\rr$ (as was the case with the interaction terms for the clock on $\Cl$.) and
\begin{align}
\begin{split}
	\bar V_0(x) &=  n_2 A_0 \sum_{p=-\infty}^{+\infty} V_B (n_2(x-x_0+2\pi p)),\\
	V_B(\cdot)&= \textup{sinc}^{2N_2}(\cdot)=(\sin(\pi\,\cdot)/(\pi\, \cdot))^{2N_2}, \qquad N_2\in\nnp.  
\end{split} ,\label{eq:vbar of x 2}
\end{align}
which is the same as before (c.f.~\cref{eq:vbar of x}), except for exchanging the free parameters $n>0$ for $n_2>0$ and $N$ for $N_2$. This is necessary since we will parametrise $n_2$ with $d$ differently for the interactions terms of $\Cl$ ($A_0$  above is the same as in~\cref{eq:vbar of x}) but with $N$ replaced with $N_2$). Likewise, the parameter $x_0$ (appearing in the definition of $\bar V_0(\cdot)$), will not take on the same value as for the clock on $\Cl$, moreover, it will take on a value $x_0^{\prime (l)}$ which will be chosen differently in this case. This is necessary, in order to avoid a ``read-write issue'': we cannot write to a memory cell which is simultaneously read without incurring a large error. To avoid such issues we want the unitary on the memory corresponding to the interaction $I_{\Cl_2}^{(l)}$ to be performed out of phase in time with the average time at which the $I_{\Cl_2}^{(l)}$ logical gate is performed. Since $x_0^{(l)}$ encodes this time for the $l^\text{th}$ logical gate as angle in cycle $R$, we choose
\begin{align}
x_0^{\prime (l)}= x_0^{(l)} + \pi = \frac{2\pi}{N_g} \left(l-\frac{1}{2}\right) +\pi.  \label{eq:x prime 0}
\end{align}
The justification of this choice (beyond physical intuition) will become apparent in the subsequent proofs. Finally, the parameters  $n$ and $N$ in the definition of $\bar V(\cdot)$ will be denoted by $n_2$ and $N_2$ respectively, in order to distinguish them from those coming from the interaction terms $\{I_\Cl^{(l)}\}_{l=1}^{N_g}$. 

Now that we have justified the actual timing, we should choose the interaction terms  $\{ I_{\M\W}^{(l)} \}_{l=1}^{N_g}$ in $H_{\M\W\Cl_2}$ appropriately. First recall the initial configuration of the memory: 
The total memory consists in $N_g (L+1)$ cells, arranged in a grid. Each cell stores in an orthogonal state a gate to be implemented. The initial memory state is thus of the form $\ket{\vec \m}_\M:=\ket{\vec \m_0}_{\M_0}\ket{\vec \m_1}_{\M_1}\ldots \ket{\vec \m_L}_{\M_L}$ with $\ket{\vec \m_j}_{\M_j}:=\ket{\m_{j,1}}_{\M_{j,1}}\ket{\m_{j,2}}_{\M_{j,2}}\ldots \ket{\m_{j,N_g}}_{\M_{j,N_g}}$. The clock on $\Cl$ can only read sequentially the memory cells in the first row, namely $\ket{\vec \m_0}_{\M_0}$. Therefore the clock on $\Cl_2$ is responsible for updating the memory cells on $\ket{\vec \m_0}_{\M_0}$ so that the entire gate sequence can be implemented. A consistent way to achieve is to have $\{ I_\M^{(l)} \}_{l=1}^{N_g}$ in $H_{\M\Cl_2}$ satisfy the following conditions (for all $l=1,2,\ldots, N_g$)
\begin{align}
\begin{split}
	&	\me^{-\mi I_{\M}^{(l)}} \ket{\m_{0,l}}_{\M_{0,l}} \ket{\m_{1,l}}_{\M_{1,l}} \ket{\m_{2,l}}_{\M_{2,l}} \ldots \ket{\m_{L-1,l}}_{\M_{L-1,l}}  \ket{\m_{L,l}}_{\M_{L,l}}\\
	&\;\,\quad =  \ket{\m_{1,l}}_{\M_{0,l}} \ket{\m_{2,l}}_{\M_{1,l}} \ket{\m_{3,l}}_{\M_{2,l}} \ldots  \ket{\m_{L,l}}_{\M_{L-1,l}} \ket{\m_{0,l}}_{\M_{L,l}},\\ 
	&	\me^{-\mi I_{\M}^{(l)}}  \ket{\m_{1,l}}_{\M_{0,l}} \ket{\m_{2,l}}_{\M_{1,l}} \ket{\m_{3,l}}_{\M_{2,l}} \ldots  \ket{\m_{L-1,l}}_{\M_{L-2,l}}\ket{\m_{L,l}}_{\M_{L-1,l}} \ket{\m_{0,l}}_{\M_{L,l}}\\
	&\;\,\quad =  \ket{\m_{2,l}}_{\M_{0,l}} \ket{\m_{3,l}}_{\M_{1,l}} \ket{\m_{4,l}}_{\M_{2,l}} \ldots  \ket{\m_{L,l}}_{\M_{L-2,l}} \ket{\m_{0,l}}_{\M_{L-1,l}}\ket{\m_{1,l}}_{\M_{L,l}} ,\\
	&	\me^{-\mi I_{\M}^{(l)}}  \ket{\m_{2,l}}_{\M_{0,l}} \ket{\m_{3,l}}_{\M_{1,l}} \ket{\m_{4,l}}_{\M_{2,l}} \ldots \ket{\m_{L-1,l}}_{\M_{L-3,l}} \ket{\m_{L,l}}_{\M_{L-2,l}} \ket{\m_{0,l}}_{\M_{L-1,l}}\ket{\m_{1,l}}_{\M_{L,l}} \\
	&\;\,\quad = \ket{\m_{3,l}}_{\M_{0,l}} \ket{\m_{4,l}}_{\M_{1,l}} \ket{\m_{5,l}}_{\M_{2,l}} \ldots  \ket{\m_{L,l}}_{\M_{L-3,l}}\ket{\m_{0,l}}_{\M_{L-2,l}} \ket{\m_{1,l}}_{\M_{L-1,l}}\ket{\m_{2,l}}_{\M_{L,l}} ,\\
	&\qquad\vdots\\
	& 	\me^{-\mi I_{\M}^{(l)}} \ket{\m_{L-1,l}}_{\M_{0,l}} \ket{\m_{L,l}}_{\M_{1,l}} \ket{\m_{0,l}}_{\M_{2,l}} \ldots \ket{\m_{L-1,l}}_{\M_{L-1,l}}  \ket{\m_{L-2,l}}_{\M_{L,l}} \\
	&\;\,\quad =\ket{\m_{L,l}}_{\M_{0,l}} \ket{\m_{0,l}}_{\M_{1,l}} \ket{\m_{1,l}}_{\M_{2,l}} \ldots \ket{\m_{L-2,l}}_{\M_{L-1,l}}  \ket{\m_{L-1,l}}_{\M_{L,l}}, 
\end{split}	\label{eqLlone I int 1}
\end{align}
for all $\m_{j,l}\in\mathcal{G}$, $j\in 1,\ldots,L$, such that $\m_{0,l}=\zero$ for all $l=1,2,3,\ldots, N_g$. Thus since the action of  $\me^{-\mi I_{\M}^{(l)}}$ maps from one orthonormal basis set to another, it is guaranteed that a unitary representation of $\me^{-\mi I_{\M}^{(l)}} $ exists. What is more, from its construction it readily follows that it is independent of the initial memory state $\ket{0}_\M\in\mathcal{C}_\M$ which encodes the to-be-implemented  gate sequence.

Note that the memory states $\zero$ play an additional role here beyond their primary role of controlling the turning on of the switches. Namely, they ensure that the states in~\cref{eqLlone I int 1} before and after the application of $	\me^{-\mi I_{\M}^{(l)}} $ are indeed orthogonal. To see this, imagine the fictitious scenario in which $\m_{0,l}\neq \zero$, $\m_{0,l} \in\mathcal{G}$, $\zero\notin\mathcal{G}$ and all initial memory cells in $\M$ have the same value $m\in\mathcal{G}$. This is highly degenerate, and the states in~\cref{eqLlone I int 1} before and after the application of $	\me^{-\mi I_{\M}^{(l)}} $ are indistinguishable and hence cannot be orthogonal. Since however, $\zero\notin\mathcal{G}$ by definition, even if the initial memory state is highly degenerate in the sense which has just been described, the states before and after the application of $\me^{-\mi I_{\M}^{(l)}} $ are always mutually orthogonal.

The definition of $\me^{-\mi I_{\M}^{(l)}} $ [\cref{eqLlone I int 1}] is also consistent with our prior assumption on $I_\M^{(l)}$, namely that it acts non-trivially only on memory block $\M_{\hash,l}$.  What is more, if $\me^{-\mi I_{\M}^{(l)}} $ is applied once in the time intervals $(t_{l,0}, t_{l,1})$,\ldots, $(t_{l,L-1}, t_{l,L})$, then condition~\cref{eq:cell restriction} is met. We will see that our idealised solution---the one we wish to approximate---will be consistent with this.

Let the set of pairs $\{ ( \Omega_{i_l},\, \ket{\Omega_{i_l}}_{\M_{\hash,l}}) \}_{i_l}$ be the eigenvalues and vectors of $I_{\M}^{(l)}$ respectively:
\begin{align}
I_\M^{(l)} \ket{\Omega_{i_l}}_{\M_{\hash,l}} = 	\Omega_{i_l} \ket{\Omega_{i_l}}_{\M_{\hash,l}},\qquad  	\Omega_{i_l}\in (0,2\pi], 
\label{eq:spect IM}
\end{align}
$l=1,2,\ldots, N_g$. 

We can now specify the idealised state on $\M\Cl_2$ at times $t_{k,r}$.\footnote{In this {}\doc, we use  notation $\prod_{i=1}^n [a_i]:= a_1 a_2...a_n$  and where the multiplication operation takes preference over summation, i.e.  $\prod_{l=1}^n\left[\sum_{i_l} h(i_l)\right]f(i_1,i_2,\ldots i_n):=\sum_{i_1} h(i_1)    \sum_{i_2} h(i_2)\ldots \sum_{i_n} h(i_n)   f(i_1,i_2,\ldots i_n)= \sum_{i_1,i_2,\ldots i_n} h(i_1)  h(i_2)\ldots  h(i_n)   f(i_1,i_2,\ldots i_n)$   }  For $k=1,2,3,\ldots, N_g-1$; $r=0,1,\ldots L$,
\begin{align}
\ket{t_{k,r}}_{\M\Cl_2}&:=  \ket{t_{k,r}}_{\M\backslash\{\M_{\hash,k}, \M_{\hash,k+1}\}\Cl_2}  \ket{t_{k,r}}_{ \M_{\hash,k}} \ket{t_{k,r}}_{ \M_{\hash,k+1}}, \label{eq:main line of def M clock 2 state}\\
\ket{t_{k,r}}_{ \M_{\hash,p}}&:= \me^{-\mi r I_{\M}^{(p)} } \ket{0}_{ \M_{\hash,p}},\qquad p\in\{k,k+1\}\\
\ket{t_{k,r}}_{\M\backslash\{\M_{\hash,k}, \M_{\hash,k+1}\}\Cl_2}&:= \prod_{\substack{l=1  \\ l\not\in \{k,k+1\}}}^{N_g}  \left[\sum_{ i_l} \,\,{\vphantom{\big[{k}\big]}}_{\M_{\hash,l}}\!\!\braket{\Omega_{ i_l}| 0}_{\M_{\hash,l}} \ket{\Omega_{ i_l}}_{\M_{\hash,l}} \right]	\ket{\Psi_{  \bar i_k, \bar i_{k+1} }(t_{k,r})}_{\Cl_2},
\end{align}
where $	\ket{\Psi_{  \bar i_k, \bar i_{k+1} }(t_{k,r})}_{\Cl_2}$ is a function of indices in the set $\{i_1,i_2,\ldots, i_{N_g}\} \backslash \{  i_k, i_{k+1} \}$ [defined below in~\cref{eq:psi muti index def}] and where recall the definition of the short-hand notation: $\ket{0}_{\M_{\hash,p}}:=$ $ \ket{\m_{0,p}}_{\M_{0,p}}$ $\ket{\m_{1,p}}_{\M_{1,p}}$ $\ket{\m_{2,p}}_{\M_{2,p}} \ldots$ $ \ket{\m_{L-1,p}}_{\M_{L-1,p}}$  $\ket{\m_{L,p}}_{\M_{L,p}}$ $\ket{0}_{\M_{\cont,p}}$. 
For the boundary values i.e. $k=0$; $r=0,1,\ldots, L$ we define the state to have periodic boundary conditions except for the case $k=r=0$ since their is no time before $t_{0,0}=0$. Namely, we define
\begin{align}
\ket{t_{0,r}}_{\M\Cl_2}&:=  \ket{t_{0,r}}_{\M\backslash\{\M_{\hash,1} \M_{\hash,N_g}  \}\Cl_2}  \ket{t_{0,r}}_{ \M_{\hash,1}}  \ket{t_{0,r-1}}_{ \M_{\hash,N_g}} ,\\
 \ket{t_{0,r}}_{\M\backslash\{\M_{\hash,1} \M_{\hash,N_g}  \}\Cl_2} &:= \prod_{\substack{l=1 \\ l\not\in \{1,N_g\}}}^{N_g} \left[ \sum_{ i_l} \,\,{\vphantom{\big[{k}\big]}}_{\M_{\hash,l}}\!\!\braket{\Omega_{ i_l}| 0}_{\M_{\hash,l}} \ket{\Omega_{ i_l}}_{\M_{\hash,l}} \right]	\ket{\Psi_{  \bar i_1, \bar i_{N_g} }(t_{k,r})}_{\Cl_2},
\\
 \ket{t_{0,r}}_{ \M_{\hash,1}} &:= \me^{-\mi r I^{(1)}_\M} \ket{0}_{\M_{\hash,1}},\qquad  \ket{t_{0,r-1}}_{ \M_{\hash,N_g}} :=
\begin{cases}
	 \ket{0}_{ \M_{\hash,N_g}}  &\mbox{ for } r=0 \vspace{0.2cm}
	\\   \me^{-\mi (r-1) I^{(1)}_\M} \ket{0}_{\M_{\hash,1}} &\mbox{ for } r=1, 2,\ldots,L.
\end{cases}
\end{align}

The state on $\Cl_2$ is defined (in slightly more generality) as 
\begin{align}
\ket{\Psi_{ \bar i_{n_1}, \bar i_{n_2}, \ldots, \bar i_{n_s} }(t)}_{\Cl_2}:=  \sum_{q\in\mathcal{S}_d(t d/T_0)} \prod_{\substack{l=1  \\ l\not\in \{n_1,n_2,\ldots,n_s\}}}^{N_g} \left[\me^{-\mi \,\Omega_{ i_l}^{(l)} \left(\int_{q-\Theta_l(t) d/T_0}^{q} \textup{d}y\, I_{\Cl_2,d}^{(l)} (y)\right) }\right] \psi_\textup{nor}^{(2)}(td/T_0,q) \ket{\theta_q}_{\Cl_2}\label{eq:psi muti index def}
\end{align}
for $n_1, n_2,\ldots ,n_s\in \{  1,2,\ldots ,N_g\}$,  and where 
\begin{align}
\Theta_l(t):=
\begin{cases}
	0 \mbox{ if } t\leq t_l \\
	t-t_l \mbox{ if }  t> t_l.
\end{cases}	
\end{align}
It readily satisfies the initial condition at $t_{0,0}=0$, namely $\ket{t_{0,0}}_{\M\Cl_2}= ( \Pi_{l,k} \ket{\m_{l,k}}_{\M_{l,k}} ) \ket{\Psi(0)}_{\Cl_2}$.  What is more, one can easily verify that this idealised memory state at different times $t_{k,r}$ satisfies the necessary condition in~\cref{eq:cell restriction}.   $\psi_\textup{nor}^{(2)}(\cdot,\cdot)$ is defined analogously to $\psi_\textup{nor}(\cdot,\cdot)$ for the state on $\Cl$ (see~\cref{eq:initial state}). We denote these normalised amplitudes with a superscript simply not to confuse them with the state on $\Cl$ which has a different standard deviation:
\begin{align}
\psi_\textup{nor}^{(2)}\big(k_0;k\big):= A_\textup{nor}^{(2)}\, \me^{-\frac{\pi}{\sigma_2^2} (k-k_0)^2} \me^{\mi 2 \pi n_{0,2} (k-k_0)/d}.\label{eq:initial state cl2}
\end{align}
Here the terms are defined the same as in case 1) of section~\Cref{sec: proof of 1st quantum clock thorem}. In particular, this means that they depend on the dimension\footnote{We have chosen the Hilbert space dimensions of both clocks to be the same, namely $d$.} $d$ as follows
\begin{align}
\sigma_2=\sqrt{d},\qquad n_{0,2}=\tilde n_{0,2}\, (d-1), \label{eq:defs of sigma and n0 for clock 2}
\end{align}
where $\tilde n_{0,2}\in(0,1)$ is $d$-independent by definition. Analogous to the definition of $\alpha_0$ in~\cref{eq: alph0 def for clock 1}, we define $\alpha_{0,2}$ for the clock on $\Cl$ to be
\begin{align}
\alpha_{0,2}=1-\left| 1-n_{0,2}\left(\frac{2}{d-1}\right)\right|  \in (0,1].  \label{eq: alph0 def for clock 2}
\end{align}
Analogously, to $|A_\textup{nor}|$,  $	|A_\textup{nor}^{(2)}|$ satisfies the upper bound
\begin{align}
	|A_\textup{nor}^{(2)}|^2= \left(\frac{2}{\sigma_2^2}\right) +\epsilon(d), \text{ as } d\to\infty,
\end{align}
where $\epsilon (d) \to 0$ as $d\to\infty$ (see{} \app~E in~\cite{WoodsAut} for details).

Note that  if the $l^\text{th}$ factor $ \me^{-\mi \,\Omega_{ i_l}^{(l)} \int_{q-\Theta_l(t)  d/T_0}^{q} \textup{d}y\, I_{\Cl_2,d}^{(l)} (y) } $ in $	\ket{\Psi_{\bar i_{n_1}, \bar i_{n_2},\ldots, \bar i_{n_s}  }(t)}_{\Cl_2}$  is equal to $ \me^{-\mi \,\Omega_{ i_l}^{(l)}  r } $, at time $t_{k,r}$, then state $\ket{t_{k,r}}_{ \M\Cl_{2}}$ becomes a bi-partite product state between the state on $\M_{\hash,l}$ and the rest of $\M\Cl_2$. What is more the state on $\M_{\hash,l}$ is equal to $ \me^{-\mi r I_{\M}^{(l)} } \ket{0}_{ \M_{\hash,l}}$, which in turn is a product state over memory cells, by virtue of eqs.~(\ref{eqLlone I int 1}).  This insight will be required in the following two lemmas in the next sections.

\subsection{Generalization of Theorem 9.1 in~\cite{WoodsAut}}
In the proof of~\Cref{thm:comptuer with fixed memory} we used the main theorem from~\cite{WoodsAut} (Theorem 9.1) several times. In order to continue in the proof of~\Cref{thm:contrl with two clocks}, we now need to prove a generalization of Theorem 9.1. In order to keep notation close to that used in~\cite{WoodsAut} and to avoid conflicts of notation used in other sections of this {}\doc, the definitions and notation used in this section do not apply to other sections of this {}\doc.

We start with a definition of a new clock state:
\begin{align}
\ket{\bar \Psi_\textup{nor} (k_0,\vec \Delta) } := \sum_{k\in\mathcal{S}_d(k_0)} \me^{-\mi \,\sum_{j=1}^D \int_{k-\Delta_j}^k \dd x \, V_d^{(j)}(x)   }  \psi_\textup{nor} (k_0,k) \ket{\theta_k},
\end{align}
with $\vec \Delta:= (\Delta_1, \Delta_2,\ldots, \Delta_D)$ $\in\rr^D$, $D\in\nnp$, $k_0\in\rr$ and $\psi_\textup{nor} (\cdot,\cdot)$ defined in~\cref{eq:initial state}. The functions $\{V_d^{(j)}(\cdot)\}_j$ are all defined in the same way that $V_d(\cdot)$ is defined in~\cite{WoodsAut}, namely
\begin{align}
{V}_d^{(j)}(\cdot):=\frac{2\pi}{d}  V_0^{(j)}\left(\frac{2\pi}{d} (\cdot)\right),\qquad j=1,2,\ldots, D
\end{align}
where $V_0^{(j)}: \rr\rightarrow \rr\cup \hh^{-}$ (where $\hh^{-}:=\left\{a_0+\mathrm{i} b_0: a_0 \in \rr, b_0<0\right\}$ denotes the lower-half complex plane) is an infinitely differentiable function of period $2 \pi$.  

For the following theorem, let us define the following terms. We start with the interaction potentials,
\begin{align}
\hat{V}_d^{(j)}:=\frac{d}{T_0} \sum_{k\in\mathcal{S}_d(k_0)} V_d^{(j)}(k)\proj{\theta_k},
\end{align}
where  $k_0\in\rr$ and is $k_0$-independent due to the periodic nature of $V_0$ and the definition of $\mathcal{S}_d$ below~\ref{eq:initial state}.
The definition of $b$ from~\cite[eq. 83]{WoodsAut} is updated as follows.\footnote{Note that the definition of $b$ from~\cite{WoodsAut} has a maximization over the function itself (in addition to over its derivatives). In the proof, since the bound for $b$ is resultant from bound on a phase, it can readily be seen that the maximization of the function is at most $2\pi$ due to the invariance of a phase under modulo $2\pi$ arithmetic.} $b$ is any real number satisfying:
\begin{align}
b \geq \max\left\{ 2\pi\,,\,\,  \sup _{k \in \mathbb{N}^{+}}\left(\max _{x \in[0,2 \pi]}\left| \sum_{j=1}^D \frac{d^{k}}{dx^{k}} V_0^{(j)}(x)  \right| + \max _{y \in[0,2 \pi]}\left| \sum_{j=1}^D \frac{d^{k}}{dy^{k}} V_0^{(j)}\left(y-\frac{2\pi\Delta_j}{d}\right)\right|\,\right)^{\!\! 1 / (k+1)}  \right\}. \label{eq:b uper bound gen}
\end{align}
Observe that, due to the $2\pi$-periodicity of the functions $\{V_0^{(j)}\}_{j=1}^D$, the r.h.s.  is invariant under the mapping $\{ \Delta_j\}_{j=1}^D$ $\to$ $\{\Delta_j+a\}_{j=1}^D$, for all $a\in\rr$.
We further define the rate parameter as
$$
\bar{v}=\frac{\pi \alpha_0 \kappa}{\ln \left(\pi \alpha_0 \sigma^2\right)} b,
$$
where $\kappa=0.792$ and 
\begin{align}
\alpha_0 & =\left(\frac{2}{d-1}\right) \min \left\{n_0,(d-1)-n_0\right\}  =1-\left|1-n_0\left(\frac{2}{d-1}\right)\right| \in(0,1],  \label{eq:psi 2 def}
\end{align}
where recall $n_0\in\rr$ is a parameter which appears in the definition of $\psi_\textup{nor} (\cdot,\cdot)$~(\cref{eq:initial state}). (These definitions are analogous to eqs. 85 and 27 in~\cite{WoodsAut}.)

In the following Theorem, the free clock Hamiltonian $H_\Cl$, is defined analogously to in the rest of this manuscript (see~\cref{eq:H_C definition}). In this~\doc, we will only use the case $\{V_0^{(j)}: \rr \rightarrow \rr\}_{j=1}^D$ but we write the general case for generality and to be in-keeping with the original theorem.

\begin{theorem}[Moving the clock through finite time with a generalised potential]\label{thm:genralised pot} Let $k_0\in\rr, \vec\Delta \in \rr^D$, $D\in\nnp$, and $t \in \rr$ if $\{V_0^{(j)}: \rr \rightarrow \rr\}_{j=1}^D$ while $t \geq 0$ otherwise. Then the effect of the generator ${H}_\Cl+\sum_{j=1}^D \hat{V}_d^{(j)}$ for time $t$ on $\left|\bar{\Psi}_{\textup{nor}}\left(k_0, \vec \Delta\right)\right\rangle$ is approximated by
$$
\mathrm{e}^{-\mathrm{i} t\left({H}_\Cl+\sum_{j=1}^D\hat{V}_d^{(j)}\right)}\left|\bar{\Psi}_{\textup{nor}}\left(k_0, \vec \Delta\right)\right\rangle=\left|\bar{\Psi}_{\textup{nor}}\left(k_0+\frac{d}{T_0} t, \vec \Delta+\frac{d}{T_0} t\right)\right\rangle+|\epsilon\rangle, \quad \||\epsilon\rangle \|_2 \leq \varepsilon_v(t, d)
$$
where $\vec \Delta+\frac{d}{T_0} t:=(\Delta_1+\frac{d}{T_0} t, \Delta_2+\frac{d}{T_0} t,\ldots, \Delta_D+\frac{d}{T_0} t)$, and in the limits $d \rightarrow \infty,(0, d) \ni \sigma \rightarrow \infty$, we have that
$$
\varepsilon_v(t, d)= \begin{cases}|t| \frac{d}{T_0}\left(\mathcal{O}\left(\frac{d^{3 / 2}}{\bar{v}+1}\right)^{1 / 2}+\mathcal{O}(d)\right) \exp \left(-\frac{\pi}{4} \frac{\alpha_0^2}{(1+\bar{v})^2} d\right)+\mathcal{O}\left(\mathrm{e}^{-\frac{\pi}{2} d}\right) & \text { if } \sigma=\sqrt{d} \\ |t| \frac{d}{T_0}\left(\mathcal{O}\left(\frac{\sigma^3}{\bar{v} \sigma^2 / d+1}\right)^{1 / 2}+\mathcal{O}\left(\frac{d^2}{\sigma^2}\right)\right) \exp \left(-\frac{\pi}{4} \frac{\alpha_0^2}{\left(\frac{d}{\sigma^2}+\bar{v}\right)^2}\left(\frac{d}{\sigma}\right)^2\right) & \\ \quad+\mathcal{O}\left(|t| \frac{d^2}{\sigma^2}+1\right) \mathrm{e}^{-\frac{\pi}{4} \frac{d^2}{\sigma^2}}+\mathcal{O}\left(\mathrm{e}^{-\frac{\pi}{2} \sigma^2}\right) & \text { otherwise. }\end{cases}
$$
\end{theorem}
It can readily be seen that \Cref{thm:genralised pot} reduces to Theorem 9.1 in~\cite{WoodsAut}, in a number of special cases. For example, when $\Delta_1=\Delta_2=\ldots =\Delta_D$, and we identify $\hat V_d$ in~\cite{WoodsAut} with $\sum_{j=1}^D \hat V_d^{(j)}$ for any $D\in\nnp$.

\begin{proof}
The proof follows analogously to the proof of Theorem 9.1 in~\cite{WoodsAut}, (which used Lemmas IX.0.1.,  IX.0.2., IX.0.3., IX.0.4., IX.0.5., as input lemmas to the proof of Theorem 9.1.).  
One needs to exchange definition 80 in~\cite{WoodsAut}, namely
\begin{align}
	\Theta(\Delta;x)=\int_{x-\Delta}^x dy V_d(y),
\end{align} 
with the new definition
\begin{align}
	\Theta(\Delta;x):=\sum_{j=1}^D \Theta_j(\Delta_j;x),\qquad   \Theta_j(\Delta_j;x):=\int_{x-\Delta_j}^x dy V_d^{(j)}(y).
\end{align} 
After this change, it is easy to go through the above-mentioned Lemmas line-by-line and the proof of the theorem itself, to verify~\Cref{thm:genralised pot}.
\end{proof}

\subsection{Lemmas bounding error contributions from control on $\Cl_2$}
In this section we will state and prove a lemma which will bound lines~\labelcref{line:CPU bus decoupling lemma l3,line:CPU bus decoupling lemma l5}  in~\cref{eq:first lem unitary quant eq1 version 2}. However, before doing so, we will need to introduce the following definition and lemma.
\begin{align}
\bar H_{\M\Cl_2}^{[k,r]}:= H_{\Cl_2}+ \sum_{q\in{\bf W}(k,r)} I_\M^{(q)} \otimes I_{\Cl_2}^{(q)}, \label{eq:bat h [] ham def}
\end{align} 
with
\begin{align}
{\bf W}(k,r):= 
\begin{cases}
	\emptyset  &\mbox{ if } r=0 \text{ and } k=1.\\
	\{1,2,3,\ldots, k-1\}  &\mbox{ if } r=0  \text{ and }  k=2,3,4,\ldots,N_g.\\
	\{2,3,4, \ldots, N_g\}  &\mbox{ if } r=1,2,\ldots, L,   \text{ and }  k=1. \\
	\{1,2,3,\ldots, k-1, k+1, \ldots, N_g\}  &\mbox{ if } r=1,2,\ldots, L,   \text{ and }  k=2,3,4,\ldots,N_g.\\
\end{cases}
\end{align}
where $\emptyset$ is the empty set. 

In the following lemma, is a technical lemma which will be used in the proofs of~\Cref{lem:2 of explit lem for clock 2,lem:1 of explit lem for clock 2}.

\begin{lemma}[Bounding the dynamics on $\Cl_2$]\label{lem:boundic dynamic gen ep}
Consider the following definition.
\begin{align}
	 \ket{\epsilon (t)}_{\Cl_2} :=\,
	& \me^{-\mi t\, \!\left( \!H_{\Cl_2} + \sum_{l\in{\bf W}(k,r)}   \Omega_{i_l}^{(l)} I_{\Cl_2}^{(l)}    \right)}  \ket{\Psi_{ \bar i_k  }(t_{k-1,r})}_{\Cl_2} -\ket{\Psi_{ \bar i_k }(t_{k,r})}_{\Cl_2}, \label{eq:line 4 big product line 2}
\end{align}
where $k=1,2,\ldots, N_g$; $r=0,1,2,\ldots, L$; $t\in\rr$  and $\{\Omega_{i_l}^{(l)} \in(0,2\pi] \}_{i_l}$. Assuming $n_2$ to be monotonically increasing with $d$, we have the bound
\begin{align}
	\| \ket{\epsilon (t)}_{\Cl_2} \|_2 =: 
	\varepsilon_v^{(2)}(t, d) \leq 
	|t| \frac{d}{T_0}\left(\mathcal{O}\left(d^{1/2+\varepsilon_g}\right)+\mathcal{O}(d)\right) \exp \left(-\frac{1}{8 \pi^2 \kappa^2 C_0^2(N_2)}\frac{d}{N_g n_2^2}\right)+\mathcal{O}\left(\me^{-\frac{\pi}{2} d}\right),  \label{eq:main eq lemm c3}
\end{align}
as $d\to\infty$.
\end{lemma}
\begin{proof}
Given the definitions, we can readily apply~\Cref{thm:genralised pot} with $\vec \Delta=(t_{k-2} d/T_0, t_{k-3} d/T_0, t_{k-4} d/T_0, \ldots,  t_{1} d/T_0, 0)$ for $r=0$ and $\vec \Delta=(rd+t_{k-2} d/T_0, rd+t_{k-3} d/T_0, rd+t_{k-4} d/T_0,\ldots, rd,rd+  t_{-2} d/T_0, rd+  t_{-3} d/T_0, \ldots, rd+  t_{k-1-N_g} d/T_0)$ for $r=1,2,\ldots,L$. (This amounts to $\vec \Delta$ being of dimension $k-1=|{\bf W}(k,0)|$ and $N_g-1=|{\bf W}(k,r)|$ respectively). In both cases, the Hamiltonian $H_\Cl+\sum_{j=1}^D \hat V_d^{(j)}$ is chosen to be $H_{\Cl_2} + \sum_{q\in{\bf W}(k,r)}   \Omega_{i_q}^{(q)} I_{\Cl_2}^{(q)}$. 

Thus, out task is to bound~\cref{eq:b uper bound gen} for this setting. Thus the summation $\sum_{j=1}^D V_0^{(j)} (x)$ in~\cref{eq:b uper bound gen} amounts to 
\begin{align}
	\sum_{j=1}^D V_0^{(j)} (x)  = \sum_{q\in{\bf W}(k,r)} \Omega_{i_q}^{(q)} \bar V_0 (x)\Big{|}_{x_0=x_0^{\prime (q)}} .
\end{align}
where recall $\bar V_0 (\cdot)$ is given by~\cref{eq:vbar of x 2}.  Therefore for $k=1,2,3,\ldots$ we have 
\begin{align}
	\max_{x\in[0,2\pi]}	\left| \sum_{j=1}^D \frac{d^{k}}{dx^{k}} \bar V_0\left(x\right)\right| \leq  \pi  N_g   	\max_{x\in[0,2\pi]}	\left|\frac{d^{k}}{dx^{k}} V_0\left(x\right) \right|     \leq    \pi  N_g  n_2^{k+1}  C_0(N_2)^{k+1} 
\end{align}
where we have used~\cref{eq:spect IM} in the penultimate inequality and Lemma 28 from~\cite[page 29]{WoodsPRXQ} in the last inequality. Recall $C_0^{(2)}(N_2)$ is solely a function on $N_2$, i.e. independent of $k$, $d$ and $n_2$ (it is denoted $C_0(N)$ in Lemma 28 from~\cite[page 29]{WoodsPRXQ}. Here we use the superscript to distinguish it from $C_0(N)$ used for the control on $\Cl$ in the proof of~\Cref{thm:comptuer with fixed memory}).

For the second term in~\cref{eq:b uper bound gen} we have for $p=1,2,3,\ldots$
\begin{align}
	\left|\sum_{j=1}^D \frac{d^p}{dy^p} V_0^{(j)} \left(y-\frac{2\pi \Delta_j}{d}\right)\right| &= \left|\sum_{q\in{\bf W}(k,r)} \Omega_{i_q}^{(q)}  \frac{d^p}{dy^p} \bar V_0 \left(y-\frac{2\pi t_{k-1-q,r}}{T_0}\right)\!\bigg{|}_{x_0=x_0^{\prime (q)}}\right|\\
	&= \left|\sum_{q\in{\bf W}(k,r)} \Omega_{i_q}^{(q)} \frac{d^p}{dy^p} \bar V_0 \left(y-\frac{2\pi}{N_g}(k-1-q)-x_0^{\prime (q)}\right)\!\bigg{|}_{x_0=0}\right| \label{eq:lemmac3 line after sub delta}\\
	&= \left|\left(\sum_{q\in{\bf W}(k,r)} \Omega_{i_q}^{(q)}\right) \frac{d^p}{dy^p}\bar V_0 \left(y-\frac{2\pi}{N_g}(k-3/2)-\pi\right)\!\bigg{|}_{x_0=0}\right|\label{eq:lemmac3 line after sub delta 2}	\\
	&\leq  \pi N_g\left|  \frac{d^p}{dy^p}  \bar V_0 \left(y-\frac{2\pi}{N_g}(k-3/2)-\pi\right)\!\bigg{|}_{x_0=0} \right| \label{eq:lemmac3 line after sub delta 3}	
\end{align}
In line~\labelcref{eq:lemmac3 line after sub delta}, we have used the $2\pi$ periodicity of $\bar V_0$. In line~\labelcref{eq:lemmac3 line after sub delta 2} we have used definition~\labelcref{eq:x prime 0}. In line~\labelcref{eq:lemmac3 line after sub delta 3} we have used definition~\labelcref{eq:spect IM}.
Therefore, using Lemma 28 from~\cite[page 29]{WoodsPRXQ} we find 
\begin{align}
	\left|\sum_{j=1}^D \frac{d^p}{dy^p} V_0^{(j)} \left(y-\frac{2\pi \Delta_j}{d}\right)\right| & \leq \pi N_g n_2^{p+1} C_0(N_2)^{p+1}.
\end{align}

Therefore, from definition~\labelcref{eq:b uper bound gen} it follows that we can choose $b$ such that
\begin{align}
	b &\geq \max\left\{ 2\pi\,,\,\,  \sup _{p \in \mathbb{N}^{+}}\left( (2 \pi N_g)^{1/(p+1)} n_2  C_0(N_2)\right) \right\}. \label{eq:b uper bound gen 2}\\
	& = \max\left\{ 2\pi\,,\,\,  (2 \pi N_g)^{1/2} n_2  C_0(N_2)
	\right\}.
\end{align}
Therefore, applying~\Cref{thm:genralised pot}, we conclude~\cref{eq:main eq lemm c3}.
\end{proof}
We will now state and prove a lemma which will bound line~\labelcref{line:CPU bus decoupling lemma l3}  in~\cref{eq:first lem unitary quant eq1 version 2}
We will also require the definition of how the switch states change overtime. This was discussed in the main text but we reproduce it here in a more compact form for  ease of readability. For $l=0,1,2,\ldots, N_g$; $r=0,1,2,\ldots, L$:
\begin{align}
	\ket{t_{l,r}}_{\W_k}= \label{eq:def:switch states}
\begin{cases}
	\ket{\on}_{\W_k}  & \mbox{ if }   r=1,2,3,\ldots, L\\
	\ket{\on}_{\W_k}  & \mbox{ if }   r=0 \text{ and } l\geq k\\
	\ket{\off}_{\W_k}  & \mbox{ if }  r=0 \text{ and } l< k
\end{cases}
\end{align}

\begin{lemma} [Bound for $1\stt$ bus-related term]\label{lem:1 of explit lem for clock 2} Consider the control states  $\{\ket{t_{k,r}}_{\Cl_2}\}_{k,r}$ and Hamiltonian interaction terms described in~\Cref{sec:description of clock 2}. The following holds for all $N_2-2, N_g-2 \in\nnz$,  $n_2>0$, $k=1,2,3, \ldots, N_g$ and  $r\in 0,1,2,\ldots, L$.
\begin{align}
	& 		\| \ket{t_{k,r}}_{\bar\M_{0,k} \Cl_2} \ket{t_{k,r}}_{\bar\W_k}
	- \me^{-\mi t_1 \bar H_{\M\W\Cl_2}^{(k)}} \ket{t_{k-1,r}}_{\bar\M_{0,k} \Cl_2}\ket{t_{k-1,r}}_{\bar \W_k} \big \|_2  \\
	& \leq   	\bar\delta_{r,0} \,A_\textup{nor}^{(2)} 4\pi^2\sqrt{2 d\Biggg[ \frac{5}{(2N_2-1)^2\, 8^2} \left(\frac{\pi^2 n_2}{4}\right)^{-4N_2}     +4\pi^2 \, \me^{-2\pi d/8^2}  \Biggg{]}}  + 	\varepsilon_v^{(2)}(t_1, d)        ,\label{eq:memma c3}
\end{align}	
where
\begin{align}
	\bar\delta_{r,0}= 
	\begin{cases}
		0 &\mbox{ if } r=0\\
		1 &\mbox{ if } r=1,2,\ldots, L
	\end{cases}
\end{align}
and $\varepsilon_v^{(2)}(t_1, d)$ is upper bounded in~\Cref{lem:boundic dynamic gen ep}.
\end{lemma}
\begin{proof}
Due to boundary conditions which would require a change of notation, we only consider $k=2,3,\ldots, N_g$ here. The case $k=1$ follows analogously with slight changes in the notation [e.g.~\cref{eq:ononon ofofoof 1st} requires some modification as is readily apparent].
Note that by definition, $\bar H_{\M\W\Cl_2}^{(k)}$ acts trivially on all memory cells $\M_{0,k}$, $\M_{1,k}$, $\M_{2,k}$, \ldots, $\M_{L,k}$. What is more, by definition~\labelcref{eq:main line of def M clock 2 state}, $\ket{t_{k,r}}_{\bar\M_{\hash,k}}= \ket{t_{k-1,r}}_{\bar\M_{\hash,k}}$, where ${\bar\M_{\hash,k}}$ denotes all of $\M$ except ${\M_{\hash,k}}$. 
\begin{align}
	&\big\| \ket{t_{k,r}}_{\bar\M_{0,k} \Cl_2}  \ket{t_{k,r}}_{\bar\W_k}
	- \me^{-\mi t_1 \bar H_{\M\W\Cl_2}^{(k)}}  \ket{t_{k-1,r}}_{\bar\M_{0,k} \Cl_2}  \ket{t_{k-1,r}}_{\bar\W_k} \big \|_2 \\
	&=  \big \| \ket{t_{k,r}}_{\bar\M_{\hash,k} \Cl_2}  \ket{t_{k,r}}_{\bar\W_k}
	- \me^{-\mi t_1 \bar H_{\M\W\Cl_2}^{(k)}}  \ket{t_{k-1,r}}_{\bar\M_{\hash,k} \Cl_2} \ket{t_{k-1,r}}_{\bar\W_k} \big \|_2. \label{eq:convergene eq 2 clocks}
\end{align}	
What is more, $\bar H_{\M\W\Cl_2}^{(k)}$ is diagonal in the on/off basis of the switch and only has support on the on-switch states $\big\{ \ket{\on}_{\W_l} \big\}_{l=1}^{N_g}$. What is more, from the definition of	$\ket{t_{l,m}}_{\W}$[\cref{eq:def:switch states}] it follows
\begin{align}
	\ket{t_{k-1,r}}_{\W}=  
	\begin{cases}
		\ket{\on}_{\W_1}  	\ket{\on}_{\W_2}  \ldots 	\ket{\on}_{\W_{k-1}} 	\ket{\off}_{\W_k}   \ket{\off}_{\W_{k+1}}\ldots \ket{\off}_{\W_{N_g}}            &\mbox{ if } r=0\\
		\ket{\on}_{\W_1}  	\ket{\on}_{\W_2}  \ldots\ket{\on}_{\W_{N_g}}            &\mbox{ if } r=1,2,\ldots, N_g  \label{eq:ononon ofofoof 1st}
	\end{cases}\\
	\ket{t_{k,r}}_{\W}=  
	\begin{cases}
		\ket{\on}_{\W_1}  	\ket{\on}_{\W_2}  \ldots 	\ket{\on}_{\W_{k}} 	\ket{\off}_{\W_{k+1}}   \ket{\off}_{\W_{k+2}}\ldots \ket{\off}_{\W_{N_g}}            &\mbox{ if } r=0\\
		\ket{\on}_{\W_1}  	\ket{\on}_{\W_2}  \ldots\ket{\on}_{\W_{N_g}}            &\mbox{ if } r=1,2,\ldots, N_g
	\end{cases}
\end{align}
Therefore, 
\begin{align}
	&\big\| \ket{t_{k,r}}_{\bar\M_{0,k} \Cl_2}  \ket{t_{k,r}}_{\bar\W_k}
	- \me^{-\mi t_1 \bar H_{\M\W\Cl_2}^{(k)}}  \ket{t_{k-1,r}}_{\bar\M_{0,k} \Cl_2}  \ket{t_{k-1,r}}_{\bar\W_k} \big \|_2 \\
	&=  \big \| \ket{t_{k,r}}_{\bar\M_{\hash,k} \Cl_2} 
	- \me^{-\mi t_1 \bar H_{\M\Cl_2}^{[k,r]}}  \ket{t_{k-1,r}}_{\bar\M_{\hash,k} \Cl_2} \big \|_2, \label{eq:convergene eq 2 clocks 2}
\end{align}	
where $\bar H_{\M\Cl_2}^{[k,r]}$ is defined in~\cref{eq:bat h [] ham def}.
Let us now compute $ \me^{-\mi t_1 \bar H_{\M\Cl_2}^{[k,r]}}  \ket{t_{k-1,r}}_{\bar\M_{\hash,k} \Cl_2}$:
\begin{align}
	&\me^{-\mi t_1 \bar H_{\M\Cl_2}^{[k,r]}}  \ket{t_{k-1,r}}_{\bar\M_{\hash,k} \Cl_2} \label{eq:cal of ev t1 1}\\
	&= \me^{-\mi t_1 \bar H_{\M\Cl_2}^{[k,r]}}     \ket{t_{k-1,r}}_{\M\backslash\{\M_{\hash,k-1}  \M_{\hash,k} \}} \ket{t_{k-1,r}}_{ \M_{\hash,k-1}}\\
	& =\!\!\prod_{\substack{l=1  \\ l\not\in \{k,k-1\}}}^{N_g}  \left[\sum_{ i_l} \,\,{\vphantom{\big[{k}\big]}}_{\M_{\hash,l}}\!\!\braket{\Omega_{ i_l}| 0}_{\M_{\hash,l}} \ket{\Omega_{ i_l}}_{\M_{\hash,l}}\right] 	\me^{-\mi t_1 \!\left( \!H_{\Cl_2} + \sum_{q\in{\bf W}(k,r)}  \Omega_{i_q}^{(q)} I_{\Cl_2}^{(q)}    \right)}\!\!\ket{\Psi_{\bar i_{k-1}, \bar i_{k} }(t_{k-1,r})}_{\Cl_2}\! \ket{t_{k-1,r}}_{ \M_{\hash,k-1}}\\
	& =\,\prod_{\substack{l=1  \\ l\not= k}}^{N_g} \left[ \sum_{ i_l} \,\,{\vphantom{\big[{k}\big]}}_{\M_{\hash,l}}\!\!\braket{\Omega_{ i_l}| 0}_{\M_{\hash,l}} \ket{\Omega_{ i_l}}_{\M_{\hash,l}} \right]	\me^{-\mi t_1 \!\left( \!H_{\Cl_2} + \sum_{q\in{\bf W}(k,r)}   \Omega_{i_q}^{(q)} I_{\Cl_2}^{(q)}    \right)}  \ket{\Psi_{ \bar i_k }(t_{k-1,r})}_{\Cl_2} \label{eq:line 4 big product line}\\
	& =\,\prod_{\substack{l=1  \\ l\not= k}}^{N_g} \left[\sum_{ i_l} \,\,{\vphantom{\big[{k}\big]}}_{\M_{\hash,l}}\!\!\braket{\Omega_{ i_l}| 0}_{\M_{\hash,l}} \ket{\Omega_{ i_l}}_{\M_{\hash,l}} \right]	\left(\ket{\Psi_{\bar i_k   }(t_{k,r})}_{\Cl_2}+\ket{\epsilon}_{\Cl_2}\right)  \label{eq:line 5 big product line}\\
	& =\ket{t_{k,r}}_{\M\backslash\{\M_{\hash,k}, \M_{\hash,k+1}\}\Cl_2}  \ket{t_{k,r}}_{ \M_{\hash,k+1}}+ \ket{\varepsilon_{k+1,r}}_{\bar\M_{\hash,k}  \Cl_2}+\,\prod_{\substack{l=1  \\ l\not= k}}^{N_g} \left[\sum_{ i_l} \,\,{\vphantom{\big[{k}\big]}}_{\M_{\hash,l}}\!\!\braket{\Omega_{ i_l}| 0}_{\M_{\hash,l}} \ket{\Omega_{ i_l}}_{\M_{\hash,l}} \right]	\,\ket{\epsilon}_{\Cl_2}\label{eq:line 6 big product line}  \\
	& =\ket{t_{k,r}}_{\bar\M_{\hash,k} \Cl_2}+ \ket{\varepsilon_{k+1,r}}_{\bar\M_{\hash,k}  \Cl_2}+\,\prod_{\substack{l=1  \\ l\not= k}}^{N_g} \left[\sum_{ i_l} \,\,{\vphantom{\big[{k}\big]}}_{\M_{\hash,l}}\!\!\braket{\Omega_{ i_l}| 0}_{\M_{\hash,l}} \ket{\Omega_{ i_l}}_{\M_{\hash,l}} \right]\ket{\epsilon}_{\Cl_2}.\label{eq:line 7 big product line}
\end{align}	 
In line~\labelcref{eq:line 4 big product line}, we have used that $\ket{\Psi_{\bar i_{k-1}, \bar i_k  }(t_{k-1,r})}_{\Cl_2}=\me^{\mi r \Omega_{i_{k-1}}^{(k-1)}}\ket{\Psi_{\bar i_k  }(t_{k-1,r})}_{\Cl_2}$. In line~\labelcref{eq:line 5 big product line} we have applied~\Cref{lem:boundic dynamic gen ep}. In line~\labelcref{eq:line 6 big product line} we have defined $\ket{\varepsilon_{k+1,r}}_{\bar\M_{\hash,k}  \Cl_2}$ as the difference between this line and the previous line. 

Therefore, from~\cref{eq:convergene eq 2 clocks 2,eq:line 6 big product line} we conclude: 
\begin{align}
	\| \ket{t_{k,r}}_{\bar\M_{0,k} \Cl_2} \ket{t_{k,r}}_{\bar\W_k}
	- \me^{-\mi t_1 \bar H_{\M\W\Cl_2}^{(k)}} \ket{t_{k-1,r}}_{\bar\M_{0,k} \Cl_2}\ket{t_{k-1,r}}_{\bar \W_k} \big \|_2  \leq  \| \ket{\varepsilon_{k+1,r}}_{\bar\M_{\hash,k}  \Cl_2} \|_2 + \| \ket{\epsilon}_{\Cl_2}\|_2\\
	\leq  \| \ket{\varepsilon_{k+1,r}}_{\bar\M_{\hash,k}  \Cl_2} \|_2 + \varepsilon_v^{(2)}(t_1, d),
\end{align}
where $\varepsilon_v^{(2)}(t, d)$ is upper bounded in~\Cref{lem:boundic dynamic gen ep}. Let us now bound $\| \ket{\varepsilon_{k+1,r}}_{\bar\M_{\hash,k}  \Cl_2} \|_2 $:
\begin{align}
	&	\| \ket{\varepsilon_{k+1,r}}_{\bar\M_{\hash,k}  \Cl_2} \|_2 ^2\,/2\\
	=& 1- \Re\left\{\left(\prod_{\substack{l=1  \\ l\not= k}}^{N_g} \left[ \sum_{ i_l} \left({\vphantom{\big[{k}\big]}}_{\M_{\hash,l}}\!\!\braket{\Omega_{ i_l}| 0}_{\M_{\hash,l}}\right)^*	\vphantom{\bra{\Omega_{ i_l}}}_{\M_{\hash,l}}\!\!\bra{\Omega_{ i_l}} \right]{\vphantom{\big[{k}\big]}}_{\Cl_2} \!\!\bra{\Psi_{ \bar i_k  }(t_{k,r})}\right)   \ket{t_{k,r}}_{\bar\M_{\hash,k+1} \Cl_2}   \right\} \\
	=&1- \Re\Bigggg\{\prod_{\substack{l=1  \\ l\not= k}}^{N_g} \left[\sum_{ i_l} \left({\vphantom{\big[{k}\big]}}_{\M_{\hash,l}}\!\!\braket{\Omega_{ i_l}| 0}_{\M_{\hash,l}}\right)^*	\vphantom{\bra{\Omega_{i_l}}}_{\M_{\hash,l}}\!\!\bra{\Omega_{ i_l}}\right] \\
	& \left(\sum_{q\in\mathcal{S}_d(t_{k,r}d/T_0)} \prod_{l\in{\bf W}(k,r)}\left[ \me^{\mi \Omega_{i_l}^{(l)}\int_{q-rd-t_{k,r} d/T_0}^q \dd y I_{\Cl_2,d}^{(l)}(y)} \right] \psi_\textup{nor}^{{(2)}*}\!\left(t_{k,r}d/T_0,q\right){}_{\Cl_2}\!\!\bra{\theta_q}\right)  \\
	& \prod_{\substack{l'=1  \\ l'\not\in \{k,k+1  \} }}^{N_g} \left[ \sum_{j_{l'}} {\vphantom{\big[{k}\big]}}_{\M_{\hash,l'}}\!\!\braket{\Omega_{ j_{l'}}| 0}_{\M_{\hash,l'}} \ket{\Omega_{j_{l'}}}_{\M_{\hash,l'}}\right]  \ket{\Psi_{ \bar j_k, \bar j_{k+1}   } (t_{k,r})}_{\Cl_2}   \ket{t_{k,r}}_{\M_{\hash,k+1}}    \Bigggg\} \\
	=& 1-\Re\Biggg\{ \prod_{\substack{l=1  \\ l\neq k }}^{N_g} \left[ \sum_{i_{l}} \left| {\vphantom{\big[{k}\big]}}_{\M_{\hash,l}}\!\!\braket{\Omega_{ i_{l}}| 0}_{\M_{\hash,l}}  \right|^2  \me^{-\mi r\Omega_{i_{k+1}}^{(k+1)} \delta_{l,k+1}}    \left(\sum_{q\in\mathcal{S}_d(t_{k,r}d/T_0)} \me^{\mi \bar\delta_{r,0} \Omega_{i_{k+1}}^{(k+1)}\int_{q-rd-t_{k-l} d/T_0}^q \dd y I_{\Cl_2,d}^{(k+1)}(y)}  \right)\right]\\
	&  \left|\psi_\textup{nor}^{(2)}(t_{k,r}d/T_0,q)\right|^2   \Biggg\} \label{eq:inline eq:lem C3 8}  \\
	=& 1-\Re\Biggg\{  \sum_{i_{k+1}} \left| {\vphantom{\big[{k}\big]}}_{\M_{\hash,k+1}}\!\!\braket{\Omega_{ i_{k+1}}| 0}_{\M_{\hash,k+1}}  \right|^2     \left(\sum_{q\in\mathcal{S}_d(t_{k,r}d/T_0)} \me^{\mi \bar\delta_{r,0} \Omega_{i_{k+1}}^{(k+1)}\int_{q-t_{-1} d/T_0}^q \dd y I_{\Cl_2,d}^{(k+1)}(y)}  \right)  \left|\psi_\textup{nor}^{(2)}(t_{k,r}d/T_0,q)\right|^2   \Biggg\} \label{eq:inline eq:lem C3 9}\\
	=& \bar\delta_{r,0}\,   \sum_{i_{k+1}} \left| {\vphantom{\big[{k}\big]}}_{\M_{\hash,k+1}}\!\!\braket{\Omega_{ i_{k+1}}| 0}_{\M_{\hash,k+1}}  \right|^2     \sum_{q\in\mathcal{S}_d(t_{k,r}d/T_0)} \left(\Omega_{i_{k+1}}^{(k+1)}\right)^2  \left( \int_{q+ d/N_g}^q \dd y I_{\Cl_2,d}^{(k+1)}(y)  \right)^2  \left|\psi_\textup{nor}^{(2)}(t_{k,r}d/T_0,q)\right|^2   \label{eq:inline eq:lem C3 10}\\
	\begin{split}
		\leq& \bar\delta_{r,0}   \sum_{i_{k+1}}    \left(\Omega_{i_{k+1}}^{(k+1)}\right)^2  \left| {\vphantom{\big[{k}\big]}}_{\M_{\hash,k+1}}\!\!\braket{\Omega_{ i_{k+1}}| 0}_{\M_{\hash,k+1}}  \right|^2    d \max_{q'\in[-1/2,1/2]}  \left(\frac{2\pi}{d}\right)^{\!2} \!\!\left( - \int_{d q'+ kd/N_g+rd}^{d q' +(k+1)d/N_g+rd} \dd y \bar V_0\left(\frac{2\pi}{d}y \right)\! \Bigg{|}_{x_0=x_0^{'(k+1)}}  \right)^2  \\
		& {A_\textup{nor}^{(2)}}^2 \me^{-\frac{2\pi}{\sigma_2^2} (d q')^2}  
	\end{split} \label{eq:inline eq:lem C3 11}\\
	\leq& \bar\delta_{r,0}\,       d {A_\textup{nor}^{(2)}}^2  {(2\pi)^4}\Bigggg[   \max_{q'\in[-1/8,1/8]}   \left( \int_{ q'}^{q' +1/N_g} \dd z \bar V_0\left(2\pi (z+k/N_g+r) \right)\!\Bigg{|}_{x_0=x_0^{'(k+1)}}  \right)^2     \\
	&	+\max_{q'\in[-1/2,-1/8]\cup[1/8,1/2]}   \left( \int_{ q'}^{q' +1/N_g} \dd z \bar V_0\left(2\pi (z+k/N_g+r) \right)\!\Bigg{|}_{x_0=x_0^{'(k+1)}}  \right)^2  \me^{-2\pi\frac{d^2}{\sigma_2^2} (1/8)^2}  \Bigggg{]}
	\label{eq:inline eq:lem C3 12}\\
	\leq& \bar\delta_{r,0}\,       d {A_\textup{nor}^{(2)}}^2 (2\pi)^4\Bigggg[   \max_{q'\in[-1/8,1/8]}   \left(\sum_{p=-\infty}^\infty  \int_{ q'}^{q' +1/N_g} \dd z V_B\left(2\pi n_2\left(z-\frac{1}{2N_g}+\frac12+p\right) \right) \right)^2    \label{eq:inline eq:lem C3 13}\\
	&	+4\pi^2 \, \me^{-2\pi\frac{d^2}{\sigma_2^2} (1/8)^2}  \Bigggg{]} \label{eq:inline eq:lem C3 14}\\
	\leq& \bar\delta_{r,0}\,       d {A_\textup{nor}^{(2)}}^2 (2\pi)^4\Bigggg[ \frac{(2\pi^2 n_2)^{-4N_2}}{(2N_2-1)^2}   \max_{q'\in[-1/8,1/8]}   \left(\sum_{p=-\infty}^\infty  \left[\left|q'+\frac{1}{2N_g}+\frac12+p\right|^{-2N_2+1}+   \left|q'-\frac{1}{2N_g}+\frac12+p\right|^{-2N_2+1} \right]  \right)^2 \label{eq:inline eq:lem C3 15}\\
	&   +4\pi^2 \, \me^{-2\pi\frac{d^2}{\sigma_2^2} (1/8)^2}  \Bigggg{]}   \label{eq:inline eq:lem C3 16} \\
	\leq& \bar\delta_{r,0}\,       d {A_\textup{nor}^{(2)}}^2 (2\pi)^4\Bigggg[ \frac{1}{(2N_2-1)^2\, 8^2} \left(\frac{2\pi^2 n_2}{8}\right)^{-4N_2}   \Biggg( \sum_{p=-\infty}^\infty  \left[\max_{q'\in[-1/8,1/8]}   \left|8\left(q'+\frac{1}{2N_g}+\frac12+p\right)\right|^{-2N_2+1}\right]\label{eq:inline eq:lem C3 17}\\
	&+   \left[\max_{q'\in[-1/8,1/8]}   \left|8\left(q'-\frac{1}{2N_g}+\frac12+p\right)\right|^{-2N_2+1} \right]  \Biggg)^2   +4\pi^2 \, \me^{-2\pi\frac{d^2}{\sigma_2^2} (1/8)^2}  \Bigggg{]}  \label{eq:inline eq:lem C3 19}\\
	\leq& \bar\delta_{r,0}\,       d {A_\textup{nor}^{(2)}}^2 (2\pi)^4\Biggg[ \frac{5}{(2N_2-1)^2\, 8^2} \left(\frac{\pi^2 n_2}{4}\right)^{-4N_2}     +4\pi^2 \, \me^{-2\pi\frac{d^2}{\sigma_2^2} (1/8)^2}  \Biggg{]}  \label{eq:inline eq:lem C3 22}
\end{align}
Where in line~\labelcref{eq:inline eq:lem C3 8}, we have defined $\bar\delta_{r,0}:=0$ if $r=0$, and  $\bar\delta_{r,0}:=1$ if $r=1,2,3,\ldots, L$, while we have defined $\delta_{l,k+1}:=0$ if $l\neq k+1$ and $\delta_{l,k+1}=1$ if $l = k+1$. In line~\labelcref{eq:inline eq:lem C3 10}, we have taken the real part and used the bound $-\cos{\theta}\leq \theta^2-1$, for all $\theta\in\rr$. In line~\labelcref{eq:inline eq:lem C3 11}, we  have used the fact that $q\in\mathcal{S}_d(t_{k,r}d/T_0)$ is equivalent  $q\in\{ \lceil -d/2+kd/N_g+r d\rceil, \lceil -d/2+kd/N_g+r d\rceil+1, \lceil -d/2+kd/N_g+r d\rceil+2,\ldots, \lceil -d/2+kd/N_g+r d\rceil+d-1\}$. Since $\lceil -d/2+kd/N_g+r d\rceil+d-1=\lceil -d/2+d-1+kd/N_g+r d\rceil\leq d/2+kd/N_g+r d$, we have that $q$ takes on $d$ values in the interval $q\in[-d/2+kd/N_g+r d,d/2+kd/N_g+r d]$. We have then performed the change of variable $q':=-(k/N_g+r)+q/d$. In line~\labelcref{eq:inline eq:lem C3 12}, we have recalled that $\Omega_{i_{k+1}}^{(k+1)}\in(0,2\pi]$ and we have used the change of variable $z:=y/d-(k/N_g+r)$.  In line~\labelcref{eq:inline eq:lem C3 13} we used the definition of $V_0$ and its properties to interchange the order of summation and integration while in line~\labelcref{eq:inline eq:lem C3 14} noted that, due to the properties of $\bar V_0$, $\int_a^b dx \bar V_0(2\pi x + x_0)\leq 2\pi$ when $0\leq b-a\leq 1$, $x_0\in\rr$. In line~\labelcref{eq:inline eq:lem C3 15}, we have used the upper bound $V_B(\cdot)\leq (\pi (\cdot))^{-2N_2}$ followed by performing the integrals and upper bounding the outcome. In line~\labelcref{eq:inline eq:lem C3 22},  we have used that $N_2, N_g\geq 2$. In particular, first we have noted that $q'+1/(2N_g)+1/2\in[3/8,7/8]$, and $q'-1/(2N_g)+1/2\in[1/8,5/8]$, and used this observation together with $N_2\geq 2$, to bound each value in the summand individually.  Finally, we use $\sigma_2=\sqrt{d}$ to achieve~\cref{eq:memma c3}.
\end{proof}

\begin{lemma}[Bound for $2\ndd$ bus-related term] \label{lem:2 of explit lem for clock 2}
Consider the control states  $\{\ket{t_{k,r}}_{\Cl_2}\}_{k,r}$ and Hamiltonian interaction terms described in~\Cref{sec:description of clock 2}. The following holds for all $N_2-2,N_g-2\in\nnz$,  $d-4\in\nnz$, $j=1,2,3, \ldots, N_g$ and  $r\in 0,1,2,\ldots, L$.
\begin{align}
	&	\max_{y\in[0,t_1]}   \big \| \left(I_{\M\W}^{(k)}\otimes I_{\Cl_2}^{(k)}\right) \me^{-\mi y \bar H_{\M\W\Cl_2}^{(k)}} \ket{t_{k-1,r}}_{\M \Cl_2}\ket{t_{k-1,r}}_\W \big \|_2 \\
	& \leq \bar \delta_{r,0}  	 \frac{4\pi^2}{T_0}   n_2 A_0  \!\Biggg( 	
	A_\textup{nor}^{(2)} \sqrt{d} \Bigg(3 \left(\frac{\pi^2 n_2}{\sqrt{2}}\right)^{-2N_2}   +2   \me^{-\frac{\pi}{8^2} d } \Bigg) +2  \varepsilon_v^{(2)}(t_1, d)    \Biggg)  \label{eq:lemma c4}
\end{align}
where
\begin{align}
	\bar\delta_{r,0}= 
	\begin{cases}
		0 &\mbox{ if } r=0\\
		1 &\mbox{ if } r=1,2,\ldots, L
	\end{cases}
\end{align}
and $\varepsilon_v^{(2)}(t_1, d)$ is upper bounded in~\Cref
{lem:boundic dynamic gen ep}.
\end{lemma}
\begin{proof}  Here we prove the result for $k=2,3,\ldots,N_g$. While the case $k=1$ follows analogously, it is best treated separately due to the periodic boundary conditions.
\begin{align}
	&\max_{y\in[0,t_1]}   \big \| \left(I_{\M\W}^{(k)}\otimes I_{\Cl_2}^{(k)}\right) \me^{-\mi y \bar H_{\M\W\Cl_2}^{(k)}} \ket{t_{k-1,r}}_{\M \Cl_2}\ket{t_{k-1,r}}_\W \big \|_2 \\
	&= \bar\delta_{r,0}	\max_{y\in[0,t_1]}   \big \| \left(I_\M^{(k)}\otimes I_{\Cl_2}^{(k)}\right) \me^{-\mi y \bar H_{\M\Cl_2}^{[k,1]}} \ket{t_{k-1,r}}_{\M \Cl_2} \big \|_2 \label{eq:lfinal proof line one} \\
	&=	\bar \delta_{r,0}  \max_{y\in[0,1]}   \big \| I_\M^{(k)}  \ket{t_{k-1,r}}_{ \M_{\hash,k}}   \big \|_2\,\,  \big\|  I_{\Cl_2}^{(k)} \me^{-\mi t_y \bar H_{\M\Cl_2}^{[k,1]}} \ket{t_{k-1,r}}_{\bar \M_{\hash,k} \Cl_2} \big \|_2 \label{line: 2nd 2 clock lemma 1}  \\
	& \leq 	\bar \delta_{r,0}   \big \| I_\M^{(k)}  \big \|_2 \max_{y\in[0,1]} \!\!\left( 	\!
	\big \|  I_{\Cl_2}^{(k)}   \prod_{\substack{l=1  \\ l\not= k}}^{N_g} \left[\sum_{ i_l} \,\,{\vphantom{\big[{k}\big]}}_{\M_{\hash,l}}\!\!\braket{\Omega_{ i_l}| 0}_{\M_{\hash,l}} \ket{\Omega_{ i_l}}_{\M_{\hash,l}} \right]	\ket{\Psi_{\bar i_{k}   }(t_{k-1+y,r})}_{\Cl_2}      \big \|_2 + \big \|  I_{\Cl_2}^{(k)} \ket{\epsilon}_{\Cl_2}  
	\big \|_2\!\!\right) \label{line: 3nd 2 clock lemma 1} \\
	& \leq 	\bar \delta_{r,0} \,  2\pi \max_{y\in[0,1]} \!\Biggg( 	
	\frac{d}{T_0}\sqrt{ \sum_{q\in\mathcal{S}_d(t_{k-1+y,r}d/T_0)}\left(I_{\Cl_2,d}^{(k)} (q)\right)^2 \left|\psi_\textup{nor}^{(2)}(t_{k-1+y,r} d/T_0,q)\right|^2 } \label{line: 4nd 2 clock lemma 1}\\
	&+\frac{d}{T_0} \sqrt{  \sum_{q\in0,1,\ldots, d-1} \left|I_{\Cl_2,d}^{(k)} (q) \,\,{}_{\Cl_2}\!\!\braket{\theta_q|\epsilon}_{\Cl_2} \right|^2   }
	\Biggg) \label{line: 5nd 2 clock lemma 1}\\
	& \leq 	\bar \delta_{r,0}   \,2\pi \max_{y\in[0,1]} \!\Biggg( 	
	\frac{2\pi}{T_0} \sqrt{\sum_{q\in\mathcal{S}_d(t_{k-1+y,r}d/T_0)}\left(\bar V_0 \Big( \frac{2\pi}{d} q\Big)\bigg{|}_{x_0=x_0^{\prime (k)}}\right)^2 \left|\psi_\textup{nor}^{(2)}(t_{k-1+y,r} d/T_0,q)\right|^2 } \label{line: 6nd 2 clock lemma 1}\\
	&+\frac{d}{T_0} \big\| \ket{\epsilon}_{\Cl_2}\big\|_2   \max_{q\in0,1,\ldots, d-1}  \left|I_{\Cl_2,d}^{(k)} (q) \right|  \Biggg) \label{line: 8nd 2 clock lemma 1}\\
	& \leq 	\bar \delta_{r,0} \,2  \pi  \max_{y\in[0,1]} \!\Biggg( 	
	A_\textup{nor}^{(2)} \frac{2\pi}{T_0}  \sqrt{d\max_{q\in[-d/2+[k-1+y]d/N_g+r d,\,d/2+[k-1+y]d/N_g+r d]}\left(\bar V_0 \Big( \frac{2\pi}{d} q\Big)\bigg{|}_{x_0=x_0^{\prime (k)}}\right)^2 \left|\psi_\textup{nor}^{(2)}(t_{k-1+y,r} d/T_0,q)\right|^2}  \label{line: 9nd 2 clock lemma 1}\\
	&+\frac{2\pi}{T_0} \varepsilon_v^{(2)}(t_{k-1+y,r}, d)   \max_{x\in[-\pi,\pi]}  \bar V_0(x) \Biggg) \label{line: 10nd 2 clock lemma 1}\\
	& \leq 	\bar \delta_{r,0} \,2  \pi \max_{y\in[0,1]} \!\Biggg( 	
	n_2 A_\textup{nor}^{(2)}  A_0  \frac{2\pi}{T_0}  \sqrt{d}  \max_{q'\in[-1/2,\,1/2]}\left[\,\sum_{p=-\infty}^\infty V_B \left(2\pi n_2(q'+[1/2+y-1]/N_g+1/2+p) \right) \right]\me^{-\pi \left(\frac{d}{\sigma_2}\right)^2 {q'}^2}     \label{line: 11nd 2 clock lemma 1}\\
	&+\frac{2\pi}{T_0} \varepsilon_v^{(2)}(t_{k,r}, d)  n_2 A_0   \left[ 1+ \sum_{p=0}^\infty   \left(\pi+2\pi p\right)^{-2N_2}+\sum_{p=0}^\infty \left(-\pi-2\pi p\right)^{-2N_2}     \right]\Biggg) \label{line: 12nd 2 clock lemma 1}\\
	\begin{split}
		& \leq 	\bar \delta_{r,0} \, 2  \pi \max_{y\in[0,1]} \!\Biggg( 	
		n_2 A_\textup{nor}^{(2)}  A_0  \frac{2\pi}{T_0} \sqrt{d} \Biggg(\!\! \left(\frac{2\pi^2 n_2}{\sqrt{8}}\right)^{-2N_2}  \left[\max_{q'\in[-1/8,\,1/8]}\,\sum_{p=-\infty}^\infty \left(\sqrt{8}(q'+[1/2+y-1]/N_g+1/2+p) \right)^{-2N_2} \right]\\
		&\qquad +\left[\max_{q'\in[-1/2,-1/8]\cup[1/8,1/2]}\,\sum_{p=-\infty}^\infty V_B \left(2\pi n_2(q'+[1/2+y-1]/N_g+1/2+p) \right) \right]\me^{-\pi \left(\frac{d}{\sigma_2}\right)^2 {\left(1/8\right)}^2} \Bigg)
	\end{split}    \label{line: 13nd 2 clock lemma 1}\\
	&+\frac{4\pi}{T_0} \varepsilon_v^{(2)}(t_{k,r}, d)  n_2 A_0 \Biggg) \label{line: 14nd 2 clock lemma 1} \\
	& \leq 	\bar \delta_{r,0}\,   2 \pi \max_{y\in[0,1]} \!\Biggg( 	
	n_2 A_\textup{nor}^{(2)}  A_0  \frac{2\pi}{T_0} \sqrt{d} \Biggg(3 \left(\frac{\pi^2 n_2}{\sqrt{2}}\right)^{-2N_2}   +\left(5 \left(\frac{3\pi^2 n_2}{2}\right)^{-2N_2}+ 1\right)   \me^{-\pi \left(\frac{d}{\sigma_2}\right)^2 {\left(1/8\right)}^2} \Bigg)   \label{line: 15nd 2 clock lemma 1} \\
	&+\frac{4\pi}{T_0}   n_2 A_0 \, \varepsilon_v^{(2)}(t_{k,r}, d) \Biggg) \label{line: 17nd 2 clock lemma 1}\\
	& \leq \bar \delta_{r,0}  	 \frac{4\pi^2}{T_0}   n_2 A_0  \!\Biggg( 	
	A_\textup{nor}^{(2)}  \sqrt{d} \Biggg(3 \left(\frac{\pi^2 n_2}{\sqrt{2}}\right)^{-2N_2}   +2   \me^{-\pi \left(\frac{d}{\sigma_2}\right)^2 {\left(1/8\right)}^2} \Bigg) +2   \, \varepsilon_v^{(2)}(t_{k,r}, d) \Biggg) \label{line: 19nd 2 clock lemma 1}
\end{align}
Where in line~\labelcref{eq:lfinal proof line one} we have used the definitions of $I_{\M\W}^{(k)}$ and $\bar H_{\M\W\Cl_2}^{(k)}$ together with~\cref{eq:ononon ofofoof 1st} and recalled the Hamiltonian~\cref{eq:bat h [] ham def}:
\begin{align}
	\bar H_{\M\Cl_2}^{[k,1]}= H_{\Cl_2}+ \sum_{\substack{l=1 \\ l \neq k}}^{N_g} I_{\M}^{(l)}\otimes I_{\Cl_2}^{(l)}. 
\end{align}
In line~\labelcref{line: 2nd 2 clock lemma 1}, we have defined $t_y=y t_1= y T_0/N_g$, $y\in\rr$; used $\ket{t_{k-1,r}}_{\M\Cl_2}$ $=$ $\ket{t_{k-1,r}}_{\M\backslash\{\M_{\hash,k-1}, \M_{\hash,k}\}\Cl_2}$ $ \ket{t_{k-1,r}}_{ \M_{\hash,k-1}}$ $\ket{t_{k-1,r}}_{ \M_{\hash,k}}$,  and recalled that $\bar H_{\M\Cl_2}^{(k)}$ acts trivially on $\M_{\hash,k}$ and $I_\M^{(k)}$ acts trivially on ${\bar \M_{\hash,k} \Cl_2}$. In line~\labelcref{line: 3nd 2 clock lemma 1} we have then calculated $\me^{-\mi y \bar H_{\M\Cl_2}^{(k)}}\ket{t_{k-1,r}}_{\bar \M_{\hash,k} \Cl_2}$ following the same steps as in lines~\labelcref{eq:cal of ev t1 1} to~\labelcref{eq:line 5 big product line} but for a time $y$ rather than a time $t_1$. We have also used the fact that the spectrum of $I_\M^{(l)}$ is bounded (\cref{eq:spect IM}).  In line~\labelcref{line: 4nd 2 clock lemma 1}, we have explicitly calculated the corresponding 2-norm in the basis $\{ \ket{\Omega_{i_l}}_{\M_{\hash,l}} \}_{\{i_l\}}$ on $\M$ and $\{\ket{\theta_k}_{\Cl_2}\}_k$ on $\Cl_2$.  In line~\labelcref{line: 9nd 2 clock lemma 1}, we have used the fact that $q\in\mathcal{S}_d(t_{k-1+y,r}d/T_0)$ is equivalent to  $q\in\{ \lceil -d/2+[k-1+y]d/N_g+r d\rceil, \lceil -d/2+[k-1+y]d/N_g+r d\rceil+1, \lceil -d/2+[k-1+y]d/N_g+r d\rceil+2,\ldots, \lceil -d/2+[k-1+y]d/N_g+r d\rceil+d-1\}$. Since $\lceil -d/2+[k-1+y]d/N_g+r d\rceil+d-1=\lceil -d/2+d-1+[k-1+y]d/N_g+r d\rceil\leq d/2+[k-1+y]d/N_g+r d$, we have that $q$ takes on $d$ values in the interval $q\in[-d/2+[k-1+y]d/N_g+r d,d/2+[k-1+y]d/N_g+r d]$.  In line~\labelcref{line: 10nd 2 clock lemma 1} we have used the definition in~\Cref{lem:boundic dynamic gen ep}. In line~\labelcref{line: 11nd 2 clock lemma 1}, we have performed the change of variable $q$ to $q'=q/d-(k-1+y)/N_g-r$ and substituted for the definitions of $\bar V_0(\cdot)$, $x_0^{\prime (k)}$, $\psi_\textup{nor}^{(2)}(\cdot,\cdot)$. In line~\labelcref{line: 13nd 2 clock lemma 1}, we have spilt the interval into subintervals and used the bound $V_B(\cdot)\leq (\pi (\cdot))^{2N_2}$. For the $1\stt$ term in  line~\labelcref{line: 15nd 2 clock lemma 1} we have used the conditions $N_g-2,N_2-2\in\nnp$ (stated in the Lemma) and observed that  $q'+[1/2+y-1]/N_g+1/2\in[1/8,7,8]$. For the second term, we have used the bound $V_B(x)\leq 1$ for all $x\in\rr$ in the case of $p=0,-1$ and $V_B(x)\leq (\pi x)^{-2N_2}$ for or all $x\in\rr$ in the cases $p\neq 0$. We have then used the fact that $N_g-2,N_2-2\in\nnp$ to generate the final bound.
\end{proof}

Finally we are in a stage to prove~\Cref{thm:contrl with two clocks}. Below we state a more explicit version of it, which is what we will prove. In particular, we state more explicitly the form which the functions $\SupPolyDecay(\cdot)$ take (One can readily check that they belong to said function class). For this, we introduce a new function $h$: In the following, $h( \bar\varepsilon)>0$ (for $\bar\varepsilon>0$) is independent of $E_0'$ and $L$, while $\textup{poly}\big(E_0'\big)$ is a polynomial in $E_0'$ and independent of $\bar\varepsilon$. 
Both are independent from the elements in the set $\{\tilde{d} (\m)\}_{\m\in\mathcal{G}\cup\{\zero\}}$. 

\begingroup
\renewcommand{\thetheorem}{\ref{thm:contrl with two clocks}}  

\begin{theorem}[\text{[}More explicit than main-text version.\text{]} Optimal quantum frequential computers only require a classical internal bus]
	\label{thm:contrl with two clocksapp}
	For all gate sets $\mathcal{U}_\mathcal{G}$, initial memory states $\ket{0}_{\M}\in\mathcal{C}_{\M}$ and initial logical states $\ket{0}_\lo\in\mathcal{P}(\mathcal{H}_\lo)$,  there exists $\ket{0}_\Cl$, $\ket{0}_{\Cl_2}$, $\{\ket{t_{j,l}}_\Cl,\ket{t_{j,l}}_{\M\Cl_2}\}_{j=1,2,\ldots,N_g;\, l=0,1,\ldots, L }$, $N_g$, $H_{\M_0\lo\Cl}$ parametrised by the energy $E_0'>0$ and a dimensionless parameter $\bar\varepsilon$ (where elements $\ket{t_{j,l}}_\Cl$, $\ket{t_{j,l}}_{\M\Cl_2}$ satisfy~\cref{cond:eq:cyclicity of cl state,eq:cell restriction} respectively), such that for all  $j=1,2,3, \ldots, N_g$; $l=0,1,2,\ldots, L$ and fixed $\bar\varepsilon\in(0,1/6)$, the large-$E_0'$ scaling is as follows
	\begin{align}\label{eq:thrm2 main eqapp}
		&	T\Big( \me^{-\mi t_{j,l} H_{\M\W\lo\Cl\Cl_2}} \ket{0}_\M\ket{0}_{\Cl_2}\ket{0}_\W\ket{0}_\lo\ket {0}_\Cl ,\,\, \ket{t_{j,l}}_{\M\Cl_2}\ket{t_{j,l}}_\W\ket{t_{j,l}}_\lo \ket{t_{j,l} }_\Cl \Big) \\
		& 	\leq   \left(\sum_{r=0}^l   \sum_{k=1}^j  \tilde d(\m_{r,k}) \right) h( \bar\varepsilon)  \, \textup{poly}\big(E_0'\big) \,\,  \big(E_0'\big)^{-1/\sqrt{\bar\varepsilon}}, \label{eq:summation thm2app}
	\end{align}
	where 
	$\proj{0}_{\Cl_2}$, $\tr_\M\big[\proj{t_{j,l}}_{\M\Cl_2}\big]  \in \mathcal{C}_{\Cl_2}^\textup{clas.}$,   and
	\begin{align}\label{eq:thm fixed memory 4app}
		f= \frac{1}{T_0}\left( T_0 E_0' \right)^{1-\bar\varepsilon}+ \delta f'', \qquad |\delta f''| \leq \frac{1}{T_0} +\bo\left(\textup{poly}\big(E_0'\big) \big( E_0'\big)^{-1/\sqrt{\bar\varepsilon}}\,\right) \text{ as } E_0'\to\infty.
	\end{align}
\end{theorem}
\endgroup

\begin{proof}
	Recall~\cref{eq:def: Ng in terms of ep g}: $N_g= \lfloor d^{1-\varepsilon_g} \rfloor$. Let $\varepsilon_g=\bar\varepsilon$ so that we are in the regime where~\cref{eq:first lem unitary quant eq1 version 4}  holds, and $N_g \leq d^{1-\bar\varepsilon}$. Choosing
	\begin{align}
		n_2=d^{\bar\varepsilon/4}
	\end{align}
	we find
	\begin{align}
		\frac{d}{N_g n_2^2}\geq d^{\bar\varepsilon/2}.
		\end{align}
		Therefore, so long as $N_2$ is $d$-independent, from~\Cref{lem:boundic dynamic gen ep} it follows that $\varepsilon_{v}^{(2)}(t_1,d)$ decays faster than any polynomial in $d$ for all fixed $\bar\varepsilon\in(0,1/6)$. (As a side comment, note that we would not have been able to calculate its decay rate if $N_2$ were $d$-dependent. This is because the function $C_0(N_2)$ is unknown. This will limit the rate at which $n_2^{-2 N_2}$ can decay as we will now see.) We now choose the parameter $N_2$, 
	\begin{align}
	N_2=\lceil 2/\bar\varepsilon^{1+1/2} \rceil \geq 2/\bar\varepsilon^{1+1/2}.
\end{align}
Therefore
\begin{align}
	n_2^{-2 N_2} \leq d^{-1/\sqrt{\bar\varepsilon}}
\end{align}
		for all $\bar\varepsilon\in(0,1/6)$. 

Therefore, using the above upper bounds for $\varepsilon_{v}^{(2)}(t_1,d)$ and $n_2^{-2 N_2}$, we can plug the bounds of~\cref{eq:first lem unitary quant eq1 version 4,lem:1 of explit lem for clock 2,lem:2 of explit lem for clock 2} into the r.h.s. of~\Cref{lem:bus cpu decoupling} to achieve
\begin{align}
		&	\big\|  \me^{-\mi t_{j,l} H_{\M\W\lo\Cl\Cl_2}} \ket{0}_\M\ket{0}_\W\ket{0}_\lo \ket{\Psi(0)}_{\Cl} \ket {\Psi(0)}_{\Cl_2} - \ket{t_{j,l}}_{\M\Cl_2}\ket{t_{j,l}}_\W\ket{t_{j,l}}_\lo \ket{\Psi(t_{j,l} d/T_0)}_\Cl \big \|_2 \\
	&	\leq   \left(\sum_{r=0}^l   \sum_{k=1}^j  \tilde d(\m_{r,k}) \right) h( \bar\varepsilon)  \, \textup{poly}(d) \, d^{-3/\sqrt{\bar\varepsilon}}  +    \textup{poly}'(d) \,d^{-1/\sqrt{\bar\varepsilon}}, 
\end{align}
where $\textup{poly}(d)$, $\textup{poly}'(d)$ are $\bar\varepsilon$-independent polynomials and $h(\bar\varepsilon)\geq 0$ is $d$-independent. All $\textup{poly}(d)$,$\textup{poly}'(d)$, and $h(\bar\varepsilon)$ are independent from the elements of $\{\tilde d(\m_{r,k}) \}_{l,j}$. Therefore, 
\begin{align}
	&	\big\|  \me^{-\mi t_{j,l} H_{\M\W\lo\Cl\Cl_2}} \ket{0}_\M\ket{0}_\W\ket{0}_\lo \ket{\Psi(0)}_{\Cl} \ket {\Psi(0)}_{\Cl_2} - \ket{t_{j,l}}_{\M\Cl_2}\ket{t_{j,l}}_\W\ket{t_{j,l}}_\lo \ket{\Psi(t_{j,l} d/T_0)}_\Cl \big \|_2 \\
	&	\leq   \left(\sum_{r=0}^l   \sum_{k=1}^j  \tilde d(\m_{r,k}) \right) h( \bar\varepsilon)  \, \textup{poly}''(d) \, d^{-1/\sqrt{\bar\varepsilon}} ,\label{eq:them2 bound in terms of d}
\end{align}
where $\textup{poly}''(d)$ is independent from the elements of $\{\tilde d(\m_{l,j})   \}_{l,j}$ and $\bar\varepsilon$.

We will now express $d$ in terms of the power $P'$. For this, we start by calculating the mean energy of the initial state. Defining $\rho^0_{\M\W\lo\Cl\Cl_2}:=   \ket{0}_\M\ket{0}_\W\ket{0}_\lo \ket{\Psi(0)}_{\Cl} \ket {\Psi(0)}_{\Cl_2} \big(\,{\mathstrut}_\M\!\bra{0}{\mathstrut}_\W\!\bra{0}{\mathstrut}_\lo\!\bra{0} {\mathstrut}_{\Cl}\!\bra{\Psi(0)} {\mathstrut}_{\Cl_2}\!\bra {\Psi(0)}\big)$, we find by direct calculation,
\begin{align}
	E_0':\!&=\tr[\rho^0_{\M\W\lo\Cl\Cl_2}H_{\M\W\lo\Cl\Cl_2} ]= {\mathstrut}_\Cl\!\braket{0| H_\Cl| 0}_\Cl + {\mathstrut}_\W\!\bra{0} {\mathstrut}_\Cl\!\bra{0} \left(\sum_{l=1}^{N_g} I_{\W_l}^{(l)}\otimes I_\Cl^{(l)} \right) \ket{0}_\W 
	\ket{0}_\Cl   +  {\mathstrut}_{\Cl_2}\!\!\braket{0| H_{\Cl_2}| 0}_{\Cl_2}\\
	& = \frac{2\pi}{T_0} \tilde n_0(d+1) +\delta E' +  {\mathstrut}_{\Cl_2}\!\!\braket{0| H_{\Cl_2}| 0}_{\Cl_2} ,\label{eq:inteline}\\
	& = \frac{2\pi}{T_0} (\tilde n_0+\tilde n_{0,2})(d+1) +\delta E'' ,\label{eq:E0' linear}
\end{align}
where in line~\labelcref{eq:inteline} we have used~\cref{eq:free mean energy,eq:E0 clcu} and where $|\delta E''|$ satisfies the same bound as $\delta E'$, namely~\cref{eq:uper bound dleta E}. To achieve the last line, we have used that ${\mathstrut}_{\Cl_2}\!\braket{0| H_\Cl| 0}_{\Cl_2}$ is the same as ${\mathstrut}_\Cl\!\braket{0| H_\Cl| 0}_\Cl$ after mapping the pair $\{\tilde n_0, \sigma\}$ to $\{\tilde n_{0,2}, \sigma_2\}$, as is readily verifiable from their definitions.

5Using the relation
~\cref{eq:E0' linear}, 
\begin{align}
	d= \frac{T_0 E_0'}{2\pi (\tilde n_0+\tilde n_{0,2})} +\delta d',\qquad\qquad \delta d':=1-  \frac{T_0(L+1)}{2\pi (\tilde n_0+\tilde n_{0,2})} \delta E''.
\end{align}
therefore, up to an additive vanishing term, $\delta d'-1$, we have that $d$ scales linearly with $E_0'$, thus using~\Cref{lem:trace dist} to lower bound the l.h.s. of~\cref{eq:them2 bound in terms of d} in term of trace distance, we find~\cref{eq:thrm2 main eqapp}. 

To achieve the scaling of the gate frequency $f$ with energy $E_0'$, i.e.~\cref{eq:thm fixed memory 4app}, we can proceed analogously to~\cref{eq:frequnsy case} for $\varepsilon_g=\bar\varepsilon$. This leads to
\begin{align}\label{eq:thm fixed memory 4 2}
	f= \frac{1}{T_0}\left( \frac{T_0}{2\pi( \tilde n_0+ \tilde n_{0,2})} E_0' \right)^{1-\bar\varepsilon}+ \delta f'', \qquad |\delta f''| \leq \frac{1}{T_0} +\bo\left(\textup{poly}\big(E_0'\big) \big(E_0'\big)^{-1/\sqrt{\bar\varepsilon}}\,\right) \text{ as } E_0'\to\infty.
\end{align} 
Hence~\cref{eq:thm fixed memory 4app}  is achieved by choosing $\tilde n_0= \tilde n_{0,2}=1/\pi$ [which is permitted since $\tilde n_0, \tilde n_{0,2}\in (0,1)$].

To finalise the proof, using the definition of the set of non-squeezed states $\mathcal{C}_{\Cl_2}^\textup{clas.}$ from~\Cref{sec:alternative def of squeezeing on C2}, we need to show that $\proj{0}_{\Cl_2}$, $\tr_\M\big[\proj{t_{j,l}}_{\M\Cl_2}\big] \in \mathcal{C}_{\Cl_2}^\textup{clas.}$. We show this property now. For $\proj{0}_{\Cl_2}$ this follows immediately from the fact that $\sigma_2=\sqrt{d}$ [from~\cref{eq:defs of sigma and n0 for clock 2}], the definitions from~\Cref{sec:alternative def of squeezeing on C2} and the proof that $\ket{t_j}_\Cl\in\mathcal{C}_\Cl^\textup{clas.}$ (which can be found in the paragraph after~\cref{eq:eta barepsilon classical}). From the definition~\labelcref{eq:main line of def M clock 2 state} a direct calculation of $\Delta t_{\Cl_2}\left(  \tr_\M\big[\proj{t_{j,l}}_{\M\Cl_2}\big]  \right)$ yields
\begin{align}
	\Delta t_{\Cl_2}\left(  \tr_\M\big[\proj{t_{j,l}}_{\M\Cl_2}\big]  \right)&= \sum_{k\in\mathcal{S}_d(t_{j,l} d/T_0)} {\mathstrut}_{\Cl_2}\!\!\bra{\theta_k} t_{\Cl_2}^2|\psi^{(2)}_\textup{nor}(t_{j,l} d/T_0, k)|^2 \ket{\theta_k}_{\Cl_2} -  \left({\mathstrut}_{\Cl_2}\!\!\bra{\theta_k} t_{\Cl_2}|\psi^{(2)}_\textup{nor}(t_{j,l} d/T_0, k)|^2 \ket{\theta_k}_{\Cl_2} \right)^2\\
	&=    {\mathstrut}_{\Cl_2}\!\!\bra{\Psi(t_{j,l} d/T_0)} t_{\Cl_2}^2  \ket{\Psi(t_{j,l} d/T_0)}_{\Cl_2}- \left({\mathstrut}_{\Cl_2}\!\!\bra{\Psi(t_{j,l} d/T_0)} t_{\Cl_2}  \ket{\Psi(t_{j,l} d/T_0)}_{\Cl_2}\right)^2,
\end{align}
for $j=1,2,3,\ldots, N_g$, $l=0,1,2,\ldots,L$ and where we have defined $\ket{\Psi(t_{j,l} d/T_0)}_{\Cl_2}$ analogously to $\ket{\Psi(t_{j,l} d/T_0)}_{\Cl}$  but on $\Cl_2$ rather than $\Cl$, i.e. 
\begin{align}
	\ket{\Psi(t_{j,l} d/T_0)}_{\Cl_2} :=\sum_{k\in\mathcal{S}_d(t_{j,l}d/T_0)} \psi_\textup{nor}^{(2)} \big(  t_{j,l}d/T_0,k \big) \ket{\theta_k}_{\Cl_2},
\end{align}
where $ \psi_\textup{nor}^{(2)}$ is defined in~\cref{eq:psi 2 def}. Similarly 
\begin{align}
	\Delta H_{\Cl_2}\left(  \tr_\M\big[\proj{t_{j,l}}_{\M\Cl_2}\big]  \right)&
	=   {\mathstrut}_{\Cl_2}\!\!\bra{\Psi(t_{j,l} d/T_0)} H_{\Cl_2}^2  \ket{\Psi(t_{j,l} d/T_0)}_{\Cl_2}- \left({\mathstrut}_{\Cl_2}\!\!\bra{\Psi(t_{j,l} d/T_0)} H_{\Cl_2} \ket{\Psi(t_{j,l} d/T_0)}_{\Cl_2}\right)^2.
\end{align}
for $j=1,2,3,\ldots, N_g$, $l=0,1,2,\ldots,L$. Therefore, from the definition of $\mathcal{C}_{\Cl_2}^\textup{clas.}$ in~\Cref{sec:alternative def of squeezeing on C2} it follows that $\tr_\M\big[\proj{t_{j,l}}_{\M\Cl_2}\big] \in \mathcal{C}_{\Cl_2}^\textup{clas.}$ for $j=1,2,3,\ldots, N_g$, $l=0,1,2,\ldots,L$.
\end{proof}

\section{Generic computer self-oscillator models considered in this~\doc}\label{sec:generic self-oscillator dynamics}

\subsection{Setup and definitions for generic computer self-oscillator models considered in this~\doc}\label{sec:Setup of the dynamical semigroup 2}
In this section, we detail the setup of the self-oscillator and define the important quantities associated with it. Since we are interested in the generic setup, in this section, we will model the computer as a bipartite system $\Sy\Cl$, where the $\Cl$ system is a self-oscillator controlling system $\Sy$.  For the reader interested in the special case pertinent to~\Cref{thm:heat dissipation}, they can set $\Sy=\M_0\lo\W$, and $H_{\Sy\Cl}= H'_{\M_0\lo\W\Cl}$ in this section. Moreover, such specialisation will be applied in~\Cref{sec:Steady-state and heat generation}, but not in~\Cref{Upper bounds on gate frequency for self-oscillators}  where~\Cref{thm:heat dissipation upper bounds} is proven.

\subsubsection{Dynamics and its partially-unravelled description}\label{sec:Dynamics and its partially-unravelled description}

As motivated in the main text, we model the self-oscillator as follows. A generic dynamical semigroup is of the form
\begin{align}
	\mathcal{L}{(\cdot)}= -\mi [H,(\cdot)] + \mathcal{D}(\cdot),\quad \mathcal{D}(\cdot)=\sum_j J^j (\cdot) {J^j}^\dag -\frac{1}{2} \{ {J^j}^\dag J^j ,(\cdot)  \},
\end{align}
where $H$ is self-adjoint and  the dissipative term, $\mathcal{D}(\cdot)$, is  formed by a set $\{J^j\}_j$ of arbitrary linear operators~\cite{Lindblad1976,Gorini1976}. Terms of the form $J^j (\cdot) {J^j}^\dag$ correspond to completely positive maps  while those of the form $-\frac{1}{2} \{ {J^j}^\dag J^j ,(\cdot)  \}$ correspond to the back-reaction of the environment onto the system caused by the action of $J^j(\cdot) {J^j}^\dag$. The evolution operator for a time $\tau\geq 0$ is $\me^{\tau \mathcal{L}}(\cdot)$.

In our case of a self-oscillator, $H$ is of the form $H_{\Sy\Cl}$, where $\Sy$ is the system which the control on $\Cl$ controls, while the dissipater part only acts non-trivially on the system it needs to correct|the control on $\Cl$. Furthermore, it is bi-partitioned into the form $	\mathcal{D}_{\Cl}(\cdot )$ $=$ $\mathcal{D}_{\Cl}^\textup{no re}(\cdot) +\mathcal{D}_{\Cl}^\textup{re}(\cdot)$. This gives rise to the total Lindbladian of our system
\begin{align}
	\mathcal{L}_{\Sy\Cl}= -\mi [H_{\Sy\Cl}, (\cdot)]  + \mathcal{D}_{\Sy\Cl}.  \label{eq:dynamic semigroup for computation}
\end{align}

 For simplicity, we assume that $\mathcal{H}_\Sy$ and $\mathcal{H}_\mathcal{\Cl}$ are finite-dimensional.  We make the definition 
\begin{align}
	\mathcal{D}^\textup{re}_\Cl(\cdot):= \sum_j J^j_\Cl (\cdot) {J^j_\Cl}^\dag, \label{eq:def:renawal channel generic def}
\end{align}
and impose  the constraint
\begin{align}
	\mathcal{D}_\Cl^\textup{re}(\rho_\Cl)\propto \rho_\Cl^0,	\quad \forall \,\rho_\Cl\in\mathcal{S}(\mathcal{H}_\Cl) , \label{eq:D re prop constrinat}
\end{align}
where $\rho_\Cl^0\in\mathcal{S}(\mathcal{H}_\Cl)$ is a fixed state.  W.l.o.g., we assume that $H_{\Sy\Cl}$ has ground-state energy zero (if it does not, we can subtract $\id_{\Sy\Cl} E_0$, where $E_0$ is the ground-state energy without affecting the dynamics). 
 Let us denote by $P(t+\tau_l,+1|\tau_l) \,\textup{d} t$ the probability that the $l+1$ renewal event occurs during the infinitesimal time interval $[t+\tau_l, t+\tau_l+\textup{d} t]$, given that the $l\thh$ renewal event occurred at time $\tau_l$. In the following, it is useful  to have the mental picture that this distribution is highly peaked around the time of approximately one period $t\approx T_0$, otherwise, the self-oscillator will be of bad quality since the purpose of the self-oscillator is to correct for small errors in each oscillation which would otherwise build up over time---not to significantly alter its dynamics in each cycle.  We will see that this condition is indeed met in the special case of~\Cref{thm:heat dissipation}.


In order to describe our results, we are interested in considering the state of the self-oscillator during its $l\thh$ cycle. Since the renewal of a self-oscillator is statistical in nature, we need to condition on the dynamics of the dynamical semigroup to obtain the relevant quantities. This conditioning on the dynamics of the dynamical semigroup is generically known as unravelling in the literature.
In our case, these conditional states can be defined recursively. Namely for $t\geq 0$, the state of the self-oscillator at time $t+\tau_l$, given that $l\thh$ renewal occurred at time $\tau_l$ (and no more renewals have occurred after this time) is 
\begin{align}
	\rho_{\Sy\Cl}\lb t+\tau_l|\tau_l\rb:= \me^{t \mathcal{L}_{\Sy\Cl}^\textup{no re}} \left(	\rho_{\Sy\Cl}(\tau_l|\tau_l) \right), \quad \rho_{\Sy\Cl}(\tau_l|\tau_l):=\begin{cases}
		\displaystyle	\frac{ \mathcal{D}^\text{re}_\Cl \left( \rho_{\Sy\Cl}\lb \tau_l|\tau_{l-1}\rb\right) }{\tr\left[ \mathcal{D}^\text{re}_\Cl \left( \rho_{\Sy\Cl}\lb\tau_l|\tau_{l-1}\rb\right) \right]} &\mbox{ if } l\in\nnp  	\vspace{0.1cm}\\
		\rho_{\Sy\Cl}(0|0):=\rho_{\Sy\Cl}(0):=\rho_{\Sy}(0)\otimes\rho^0_\Cl   &\mbox{ if } l=0.
	\end{cases},\label{eq:conditional states AC def}
\end{align}
where 
\begin{align}
	\mathcal{L}_{\Sy\Cl}^\textup{no re}(\cdot):=  -\mi [H_{\Sy\Cl},(\cdot)]  +  	\mathcal{D}_{\Cl}^\textup{no re}(\cdot),\qquad  	\mathcal{D}_{\Cl}^\textup{no re}(\cdot)= -\frac{1}{2}\sum_j \{ J^{j\dag}_\Cl J^{j}_\Cl, (\cdot)  \}= - \{ V_\Cl, (\cdot)\}, \label{eq:def linblad no jump 2}
\end{align}
and we define $0=:\tau_0< \tau_1< \tau_2 < \ldots < \tau_l$, and $V_\Cl:= \frac{1}{2}\sum_j  J^{j\dag}_\Cl J^{j}_\Cl$.  Note that the state $\rho_{\Sy\Cl}\lb\tau_l|\tau_{l-1}\rb$ on the r.h.s. in~\cref{eq:conditional states AC def} is defined by setting $\tau_l$ to $\tau_{l-1}$  and $t=\tau_l-\tau_{l-1}$ in the l.h.s. of~\cref{eq:conditional states AC def}. What is more, while we are only explicitly labelling the conditioning of the state $\rho_{\Sy\Cl}\lb t+\tau_l|\tau_l\rb$ on having renewed at time $\tau_l$; in fact, due to the recursive nature of the definition, it is conditioned on having renewed at times $\tau_1, \tau_2,\ldots, \tau_l$. 

Note that $\rho_{\Sy\Cl}\lb t+\tau_l|\tau_l\rb$ is not necessarily normalised for $t>0$. The trace of $\rho_{\Sy\Cl}\lb t+\tau_l|\tau_l\rb$ corresponds to the probability of this event occurring.  In general, we will use  left and right bold font brackets $\lb$ and $\rb$ (or $\lsb$ and $\rsb$) to indicate that the states (or kets and bras) are not necessarily normalised and to distinguish themselves from their normalised counterparts which will use normal brackets $($ and $)$ (or $[$ and $]$).
We denote the normalised counterpart of $\rho_{\Sy\Cl}\lb t+\tau_l|\tau_l\rb$ by
\begin{align}\label{eq:tilde def norm}
	\rho_{\Sy\Cl}(t+\tau_l|\tau_l):= \rho_{\Sy\Cl}\lb t+\tau_l|\tau_l\rb / \tr[ \rho_{\Sy\Cl}\lb t+\tau_l|\tau_l\rb].
\end{align}

The probability that the $(l+1)\thh$ renewal occurs at a time $t+\tau_{l}$, given that the $l\thh$ renewal occurred at time $\tau_{l}$, is thus
\begin{align}
	P(t+\tau_l,+1|\tau_{l}):=  \tr\left[\mathcal{D}^\text{re}_\Cl \left( \rho_{\Sy\Cl}\lb t+\tau_l|\tau_{l}\rb \right) \right]= 2 \,\tr[ V_\Cl   \rho_{\Sy\Cl}\lb t+\tau_l|\tau_{l}\rb  ], \quad l\in\nnz. \label{eq:prob plus one renewal def}
\end{align}

\subsubsection{Generation of a classical signal when each renewal process occurs}\label{sec:classical bit for renewal prcess}
The aim of this section is to provide extra justification of the previous one. In particular, a justification of the partial-unravelling of the dynamical semigroup. Those familiar with dynamical semigroups, may wish to skip this section. 

Let us construct a new Lindbladian over subsystems $\Sy\Cl\Reg$ denoted $\mathcal{L}_{\Sy\Cl\Reg}$. We do so by mapping $\{J^j_\Cl\}_j \to \{\tilde J^j_{\Cl\Reg}\}_j$, in the definition of $\mathcal{L}_{\Sy\Cl}$,  where $\tilde J^j_\Cl:= J^j_\Cl \otimes J_\Reg$,   $J_\Reg:=\ketbra{1}{0}_\Reg+ \ketbra{2}{1}_\Reg+\ldots+\ketbra{N_T}{N_T-1}_\Reg+\ketbra{0}{N_T}_\Reg$
and $\{J^j_\Cl\}_j$ are defined in~\Cref{sec:Dynamics and its partially-unravelled description}.  It can readily be seen from the dynamical semigroup that if the register on $\Reg$ is initiated to $\ket{0}_\Reg$|a product state with the rest of the system|then the dynamics leads to a probabilistic mixture over product states with the register and the rest of the system with the register in one of the states $\mathcal{C}_\Reg:=\{ \ket{l}_\Reg\}_l$. Moreover, the state of the register keeps track of the partitioning of the ensemble into the number of renewals which have occurred at time $t$. In other words, it serves as a classical counter for the number of renewal event which updates itself in real time as new events occur.

To see this, consider an arbitrary state $\rho_{\Sy\Cl}$ which is in a tensor-product state with the counter in state $\proj{k}_\Reg$ and consider an infinitesimal time step $\delta t$:
\begin{align}
	\me^{\delta t \mathcal{L}_{\Sy\Cl\Reg}} \left( \rho_{\Sy\Cl}\otimes \proj{k}_\Reg\right) =  \Big( \rho_{\Sy\Cl} +\delta t \,\mathcal{L}^\textup{no re}_{\Sy\Cl}(\rho_{\Sy\Cl})  \Big) \otimes \proj{k}_\Reg + \delta t \,\mathcal{D}^\textup{re}_{\Cl} (\rho_{\Sy\Cl}) \otimes \proj{k+1}_\Reg+ \bo(\delta t^2).\label{eq:computer w register dnamics 1}
\end{align}
Intuitively, this can be understood as saying that in every infinitesimal time step, there are two possibilities: either no renewal occurs in said time step (and the counter's value remains unchanged), or there is a renewal event which occurs (and the counter increases its value by one).

Thus tacking into account  the divisibility of dynamical semigroups, if we start out in a product state $\rho^0_{\Sy\Cl}\otimes \proj{0}_\Reg$, the state of the entire computer and counter at time $t> 0$ will be of the form 
\begin{align}
	\me^{t \mathcal{L}_{\Sy\Cl\Reg}} \left( \rho^0_{\Sy\Cl}\otimes \proj{0}_\Reg\right) =  \sum_{k=0}^\infty \rho_{\Sy\Cl}^{(k)}(t)\otimes \proj{k}_\Reg,\label{eq:computer w register dnamic 2}
\end{align}
where the (un-normalised) state $ \rho_{\Sy\Cl}^{(k)}(t)$, has passed through $k$ renewal events.

Finally, it is important to note that the introduction of the counter register, has not modified the dynamics of the computer in anyway, i.e. it is straightforward to show
\begin{align}
	\tr_\Reg\!\left[	\me^{t \mathcal{L}_{\Sy\Cl\Reg}} \big( (\cdot)\otimes \proj{0}_\Reg\big)\right] = \me^{t \mathcal{L}_{\Sy\Cl}} (\cdot),
\end{align}
for all $t \geq 0$. Intuitively, one can see this as a consequence of the counter register never acquiring coherence, and hence it can be measured via non-demolition measurements.

Since the renewal events occur at non-deterministic times, the state of the computer and counter system in~\cref{eq:computer w register dnamic 2}, at any given time, there is some possibility that the renewal process has occurred an arbitrary number of times.  Of course, in reality, we need these non-deterministic times to be very close to deterministic in order to not incur large operation errors. That such a scenario is indeed possible, is proven in~\Cref{thm:heat dissipation}.

If one were to actually realise~\cref{eq:computer w register dnamic 2}, and then watched the counter in real time, they would never see ``a mixture of all register values'' at any given time. Moreover, what they would experience is the register starting off in $\ket{0}_\Reg$, and then, after a short time period, changing to $\ket{1}_\Reg$, then a bit later to $\ket{2}_\Reg$, etc. This is because~\cref{eq:computer w register dnamic 2} represents the experience averaged over many runs.   \Cref{eq:conditional states AC def} represents the dynamics one would actually see and experience in any one run of the computer.

What is more, one can easily reproduce the conditional dynamical state $\rho_{\Sy\Cl}(t+\tau_l | \tau_{l})$ in~\cref{eq:tilde def norm}, by measuring the counter register in~\cref{eq:computer w register dnamics 1} from time $t=0$ to the present time $t+\tau_l$, and conditioning on observing the following: the counter register transition from state  $\ket{0}_\Reg$ to  $\ket{1}_\Reg$ at time $\tau_1$, 
 from $\ket{1}_\Reg$ to  $\ket{2}_\Reg$ at time $\tau_2$,  from $\ket{2}_\Reg$ to  $\ket{3}_\Reg$ at time $\tau_3$, etc, until observing  it transition from $\ket{l-1}_\Reg$ to  $\ket{l}_\Reg$ at time $\tau_l$;  and no more changes until (and including)  the present time $t+\tau_l$. On the other hand, the probability $\tr[ \rho_{\Sy\Cl}\lb t+\tau_l | \tau_{l}\rb]$ is the answer to: give that said counter register transitions occurred at said times, what is the probability that no new register transitions occur  in the time interval $[\tau_l, t+\tau_l)$.

\subsection{(Non)-isentropic time interval,  power consumption and dissipation}\label{sec:Power consumption derivation}
Here we will accomplish three things: First explain why the states in the so-called isentropic time interval are isentropic. Second, derive the states which $\mathcal{D}^\textup{re}_\Cl(\cdot)$ has to correct. Third, derive an expression for the power consumed per cycle and power dissipated.

By definition, the states in the $l\thh$ isentropic time interval, are the states on $\Sy\Cl$ during the time interval starting just after the $l\thh$ application of $\mathcal{D}^\textup{re}_\Cl(\cdot)$ at time $\tau_l$ and finishing at time $\tau_l+t_{\max,l}$, where $t_{\max,l}>0$ is the largest real value such that
\begin{align}
		\rho_{\Sy\Cl}\lb t+\tau_l|\tau_l\rb:= \me^{t \mathcal{L}_{\Sy\Cl}^\textup{no re}} \approx \me^{ -\mi t H_{\Sy\Cl}}  \rho_{\Sy\Cl}(\tau_l|\tau_l) \me^{ \mi t H_{\Sy\Cl}}\label{eq:approx isentropic regime}
\end{align}
holds for all $t\in[0,t_{\max,l}]$, where the $1\stt$ equality is simply the definition~\labelcref{eq:conditional states AC def}. This is an imprecise definition of $t_{\max,l}$ since we have not characterised what the approximation $\approx $ is. In order not to break the flow or ideas, we leave its rigorous definition to~\cref{sec:Quantities pertaining to the quality of computer}.  Clearly, the r.h.s. of~\cref{eq:approx isentropic regime} is normalised and undergoing unitary dynamics and hence the state's entropy is constant throughout this time interval and the process is reversible (at least to a good approximation).

The state at the end of the $l\thh$ cycle which the channel $\mathcal{D}^\textup{re}_\Cl(\cdot)$ acts on is
\begin{align}
	\lim_{t\,\uparrow \, \tau_{l+1}-\tau_{l}}\tr_{\Sy}[\rho_{\Sy\Cl}(t+\tau_l| \tau_l)]= \frac{\tr_{\Sy}\!\left[\me^{(\tau_{l+1}-\tau_{l}) \mathcal{L}_{\Sy\Cl}^\textup{no re}} \left(	\rho_{\Sy\Cl}(\tau_l|\tau_l) \right)\right]} { \tr\!\left[\me^{(\tau_{l+1}-\tau_{l}) \mathcal{L}_{\Sy\Cl}^\textup{no re}} \left(	\rho_{\Sy\Cl}(\tau_l|\tau_l) \right)\right]}.\label{eq:state of c which is corrected}
\end{align}
 Recalling the $2\ndd$ equation in~\cref{eq:conditional states AC def}, it follows from the property $\mathcal{D}^\textup{re}_\Cl(\cdot)\propto \rho_\Cl^0$,  that the state $\rho_{\Sy\Cl}(\tau_l|\tau_l)$ itself is a product state of the form $\sigma_\Sy^{(l)}\otimes\rho_\Cl^0$, for some system state $\sigma_\Sy$. (This is because $\mathcal{D}^\textup{re}_\Cl(\cdot)= 2 \tr_\Cl[V_\Cl \,\cdot] \, \rho^0_\Cl$, as we show later in the proof of~\Cref{lem: power per cycle}). Moreover, unless $H_{\Sy\Cl}$ does not contain interaction terms coupling the control with the $\Sy$ system|something necessary for the control to do its job, (control the application so unitary gates on $\Sy$)|the state on $\Cl$ in~\cref{eq:state of c which is corrected}, is a mixed state where the weights of the mixture depend on  $\sigma_\Sy^{(l)}$. (In~\cref{fig:power}, the  three lines correspond to the different states in the mixture as the evolve through their distinct orbits).  In the analogy with Landauer erasure, the states~\cref{eq:state of c which is corrected} generated by choosing different input states $\sigma_\Sy^{(l)}$, is analogous to the set of memory states which have to be erasable via the Landauer erasure process.
The possible states $\sigma_\Sy^{(l)}$ typically, among other things, will depend on all possible logical states and permissible memory states, e.g. states in $\lo$ and $\M_0$ respectively in the context of~\Cref{sec:thm3 main text}. Note also that the channel generating the output states are not fully deterministic since the times $\tau_l$, $l\in\nnz$ are the outcomes of random variables. This feature means that these different states also have to be erasable by $\mathcal{D}^\textup{re}_\Cl$ at multiple times in their trajectories.

We now turn to the question of the required power consumed to perform the renewal event (c.f. Landauer erasure). 

The derivation will be analogous to the standard derivation of energy transfer between system and environment in dynamical semigroups (see, e.g.~\cite{Rivas2012,Weiss2008}), with the minor difference that we will conditional on the dynamics, so that we can find the contribution associated with the application of the irreversible renewal operator, $\mathcal{D}^\textup{re}_\Cl(\cdot)$ per cycle, as opposed as the ensemble averages over all possible events which can occur at time $t$, i.e. a probabilistic ensemble of zero to infinite number of renewals occurring.  

First, we calculate the infinitesimal change in the state, given that the $l\thh$ renewal occurred in the time interval $[t+\tau_{l-1},t+\tau_{l-1}+\dd t)$:  this is
\begin{align}
	\mathcal{D}^\textup{re}_\Cl \left( \rho_{\Sy\Cl}\lb t+\tau_{l-1}| \tau_{l-1}\rb \right) \dd t,
\end{align}
where we have intentionally not normalised the state since we wish to keep the associated probability of said event occurring. One can alternatively derive this from the framework where we included the counter register in~\cref{sec:classical bit for renewal prcess}, i.e. one has
\begin{align}
	\mathcal{D}^\textup{re}_\Cl & \left( \rho_{\Sy\Cl}\lb t+\tau_{l-1}| \tau_{l-1}\rb  \right) \otimes \proj{l}_\Reg \, \delta t + \bo(\delta t^2) \\
	&=  \proj{l}_\Reg  \Big( \me^{\delta t \mathcal{L}_{\Sy\Cl\Reg}}    \rho_{\Sy\Cl}\lb t+\tau_{l-1}| \tau_{l-1}\rb \otimes\proj{l-1}_\Reg -    \rho_{\Sy\Cl}\lb t+\tau_{l-1}| \tau_{l-1}\rb \otimes\proj{l-1}_\Reg \Big)  \proj{l}_\Reg.
\end{align}
The associated change in internal energy during this infinitesimal time step is
\begin{align}
\tr\left[H_{\Sy\Cl}	\mathcal{D}^\textup{re}_\Cl \left( \rho_{\Sy\Cl}\lb t+\tau_{l-1}| \tau_{l-1}\rb \right)\right] \dd t. \label{eq:energy change}
\end{align}
Thus, the total energy exchanges associated with the implementation of the irreversible operator $\mathcal{D}^\textup{re}_\Cl$ in the $l\thh$ cycle is
\begin{align}\label{eq:def: power in general}
	P^\textup{in}_l:=\frac{\braket{E^\text{re}}}{T_0};\quad 	\braket{E^\text{re}}:=  \int_{0}^{\infty} \dd s \, \tr[H_{\Sy\Cl}\,	\mathcal{D}_{\Cl}^\text{re}( \rho_{\Sy\Cl}\lb  s+\tau_{l}|\tau_{l} \rb)].
\end{align}
(Here we have used the convention that the $l\thh$ cycle corresponds to the interval $[\tau_l, \tau_{l+1})$ as opposed to $[\tau_{l-1}, \tau_{l})$. Hence~\cref{eq:def: power in general} is a function of $\tau_l$ rather than $\tau_{l-1}$.)

Note that since by definition $H_{\Sy\Cl}\geq 0$ and $\mathcal{D}^\textup{re}_\Cl (\cdot)$ is completely positive, it follows that the integrand in~\cref{eq:def: power in general} is non-negative. Thus, it corresponds to energy flowing \emph{into} the system from the environment (if non-zero). Taking into account the form of the operator $\mathcal{D}^\textup{re}_{\Cl}$, it would be zero if the renewed state $\rho^0_\Cl$, has zero mean energy (i.e. its the ground state) since probability $\tr[\rho_{\Sy\Cl}\lb t+\tau_{l}| \tau_{l}\rb]$ is non-negative.

However, if $\rho^0_\Cl$ is in the ground state (i.e. has mean energy zero), then, since the initial state is $\rho_{\Sy\Cl}(0| 0)= \rho_{\Sy}(0)\otimes \rho^0_\Cl$, (recall~\cref{eq:conditional states AC def}), no free dynamics of the oscillator due to $H_\Cl$ would have occurred, and the only energy acquired  by the control on $\Cl$, would have need internally from $\Sy$ via the $\Sy$-$\Cl$ interaction terms in the $H_{\Sy\Cl}$ Hamiltonian.  This wound most likely result in a terrible computer with huge errors, and the energy required for renewal could come solely from $\Sy$. In the case of~\Cref{thm:heat dissipation}, the special case used in the proof which achieves the optimal scaling has $\tr[H_\Cl \rho^0_\Cl ]$ which is large (proportional ot the dimension which is itself proportional to the power consumed), while in the case of~\Cref{thm:heat dissipation upper bounds}, it is effectively eliminated by assumption (the to-be-defined  instantaneous initial-cycle-state parameter error, $\varepsilon_\textup{H}^0$ is small by assumption).

By analogous reasoning, we find that, if in the interval $(\tau_l, t+\tau_l]$ no renewal event has occurred, then the net flow of energy \emph{into} $\Sy\Cl$ from the environment is 
\begin{align}
	\int_0^{t} \dd s\,   \tr[H_{\Sy\Cl}\,	\mathcal{D}_{\Cl}^\text{no re}(\rho_{\Sy\Cl}\lb s+\tau_l|\tau_l\rb )].
\end{align}
 Thus, to find the average  flow of energy into $\Sy\Cl$ during one cycle, one needs to integrate up until the end of the cycle, at time $\tau_{l+1}$. This happens at time $t+\tau_l$ with probability $P(t+\tau_l,+1|\tau_l)$, thus taking the averaged time, we find that the energy flowing \emph{into} the system on cycle average is
\begin{align}
 \int_{0}^{\infty} \dd t  \, P( {t+\tau_l}  ,+1|\tau_l)  \int_0^{t} \dd s\,   \tr[H_{\Sy\Cl}\,	\mathcal{D}_{\Cl}^\text{no re}(\rho_{\Sy\Cl}\lb s+\tau_l|\tau_l\rb  )].
\end{align}
This is negative, since it has to balance out the positive flow of energy into the system during the renewal even itself. If not, there would be a build-up of energy in $\Sy\Cl$, yet the internal state of the computer and control is zero. 
Thus the energy flowing \emph{out} of the computer into the environment is  
\begin{align}
	  \braket{E^\text{no re}}:= -  \int_{0}^{\infty} \dd t  \, P( {t+\tau_l}  ,+1|\tau_l)  \int_0^{t} \dd s\,   \tr[H_{\Sy\Cl}\,	\mathcal{D}_{\Cl}^\text{no re}(\rho_{\Sy\Cl}\lb s+\tau_l|\tau_l\rb  )].
\end{align}
Thus we define the dissipated energy per cycle into the environment as
 \begin{align}
 		P^\textup{diss}_l:= \frac{\braket{E^\text{no re}}}{T_0}.\label{eq:diss def app}
 \end{align}

 Finally, it is worth pointing out that these definitions of $P^\textup{in}$ and $P^\textup{diss}$ are perfectly consistent with the existence of the isentropic time intervals discussed at the beginning of this section.  The existence of an isentropic interval in the $l\thh$ cycle simply implies that the integrands in~\cref{eq:def: power in general,eq:diss def app} are close to zero for a significant interval over which they are being integrated, where said interval corresponds to the isentropic interval.

\subsection{Dependency on $T_0$}\label{Dependency on T0}
This section can be skipped on a $1\stt$ read. It only derives how certain quantities scale with the cycle time $T_0$.

Recall that in the cases where we considered computational models with Hamiltonian dynamics,  the dependency on gate frequency $f$ on $T_0$ was the same in both the quantum and classical cases, and thus there was no quantum advantage deliverable by changing $T_0$. Here, we will show that the same is true in the open quantum systems case considers here, where the dynamics is given by a dynamical semigroup as opposed to a Hamiltonian.  This was discussed in a paragraph shortly after~\Cref{thm:upperboundsEnergy}.

We will show that the dependency of $P^\textup{in}$ on $T_0$ is that it is proportional to $1/T_0^2$, so that the theorems quantifying the gate frequency $f$ in terms of $P^\textup{in}$, namely~\Cref{thm:heat dissipation,thm:heat dissipation upper bounds,thm:heat dissipation upper bounds generalization}, all convey the same dependency of the gate frequency  with cycle time: $f$ is directly proportional to $1/T_0$. This is to say, its dependency is identical to the dependency of the computers which were modelled via Hamiltonian dynamics.

To start with, recall~\cref{eq:conditional states AC def,eq:def linblad no jump 2}. The initial state on $\rho_{\Sy\Cl}(0|0)$ is $T_0$-independent.  We can re-write the dynamics in terms of the dimensionless parameters, $\bar t:= t/ T_0$, and $\bar \tau_l := \tau_l / T_0$ as follows:  

\begin{align}
	\bar\rho_{\Sy\Cl}\lb\bar t+\bar\tau_l|\bar\tau_l\rb:= \me^{\bar t \bar{\mathcal{L}}_{\Sy\Cl}^\textup{no re}} \left(	\bar\rho_{\Sy\Cl}(\bar\tau_l|\bar\tau_l) \right), \quad \bar\rho_{\Sy\Cl}(\bar\tau_l|\bar\tau_l):=
	\begin{cases}
		\displaystyle	\frac{ \bar{\mathcal{D}}^\text{re}_\Cl \left( \rho_{\Sy\Cl}\lb\bar\tau_l|\bar\tau_{l-1}\rb\right) }{\tr\left[ \bar{\mathcal{D}}^\text{re}_\Cl \left( \rho_{\Sy\Cl}\lb\bar\tau_l|\bar\tau_{l-1}\rb\right)  \right]} &\mbox{ if } l\in\nnp  	\vspace{0.1cm}\\
		\bar\rho_{\Sy\Cl}(0|0):= \rho_{\Sy\Cl}(0|0):=\rho_{\Sy\Cl}(0):=\rho_{\Sy}(0)\otimes\rho^0_\Cl   &\mbox{ if } l=0.
	\end{cases},\label{eq:conditional states AC def 3}
\end{align}
where 
\begin{align}
	\bar{\mathcal{L}}_{\Sy\Cl}^\textup{no re}(\cdot):=  -\mi [\bar H_{\Sy\Cl},(\cdot)]  +  	\bar{\mathcal{D}}_{\Cl}^\textup{no re}(\cdot),\qquad  \bar{\mathcal{D}}_{\Cl}^\textup{no re}(\cdot)= - \{ \bar V_\Cl, (\cdot)\}, \label{eq:def linblad no jump 3}
\end{align}
with
\begin{align}
	\bar H_{\Sy\Cl}= H_{\Sy\Cl} T_0, \qquad  \bar V_\Cl=: V_\Cl T_0. 
\end{align}
Here, we define the dimensionless Hamiltonian and interaction terms, $\bar H_{\Sy\Cl}$ and $\bar V_\Cl$, to be  $T_0$-independent. Or equivalently, that $H_{\Sy\Cl}$ and $V_\Cl$ are directly proportional to $1/T_0$. This definition can be readily  motivated as follows: Recall that $T_0$ is the cycle time, it thus follows that the dynamics of the initial control cycle state, namely $\rho_\Cl^0$,  under the free Hamiltonian $H_\Cl$, is $T_0$-periodic. Thus  $H_\Cl$ should be directly proportional to $T_0$ such that a change in $T_0$ does not break the periodicity of the dynamics. Likewise, the full dynamics including the  interaction terms, needs to be finely tuned to the periodicity of the free dynamics of $H_\Cl$, because the same dynamics on $\Sy\Cl$ needs to be achieved, regardless of the cycle time $T_0$. Thus one needs all interaction strengths to be proportional to $1/T_0$ in order for $1/T_0$ to be a common factor in the dynamics so  that a change in $T_0$, merely re-scales the speed at which the full dynamics unfolds. Hence we arrive at 	$\bar H_{\Sy\Cl}$ and $\bar V_\Cl$ being $T_0$-independent. The fact that they are dimensionless follows from the fact that we are using natural units (i.e.``units such that $\hbar=1$''), so energy has units of inverse time. Similarly, we have that the quantity $\bar{\mathcal{D}}^\text{re}_\Cl (\cdot) T_0$ is dimensionless and $T_0$-independent. 

Finally, note that the particular choice of $H_{\Sy\Cl}$ and $V_\Cl$ which were used in the proof of~\Cref{thm:heat dissipation} are exactly directly proportional to $1/T_0$, in agreement to the above definition.

From these definitions, it thus follows that $\bar\rho_{\Sy\Cl} {\lb}(\cdot)+\bar\tau_l|\bar\tau_l {\rb}: \rrz\to\mathcal{\hat S}({\mathcal{H}_{\Sy\Cl}})$ is a $T_0$-independent function whose real input is dimensionless. By comparing~\cref{eq:conditional states AC def,eq:def linblad no jump 2} with~\cref{eq:conditional states AC def 3,eq:def linblad no jump 3}, we find the relation
\begin{align}
	\rho_{\Sy\Cl}\lb \bar t\, T_0+\tau_l|\tau_l\rb= 	\bar\rho_{\Sy\Cl}\lb \bar t+\bar\tau_l|\bar\tau_l\rb, \qquad 	\rho_{\Sy\Cl}(\tau_l|\tau_l)= 	\bar\rho_{\Sy\Cl}(\bar\tau_l|\bar\tau_l),\qquad l\in\nnz.
\end{align}

It thus follows from~\cref{eq:def: power in general}, that 
\begin{align}\label{eq:def: power in general 2}
	\braket{E^\text{re}} :=&  \int_{0}^{\infty} \dd s \, \tr[H_{\Sy\Cl}\,	\mathcal{D}_{\Cl}^\text{re}(\rho_{\Sy\Cl}\lb s+\tau_l|\tau_{l}\rb )] =  \frac{1}{T_0} \int_{0}^{\infty} \dd \bar s \, \tr[\bar H_{\Sy\Cl}\,	\mathcal{\bar D}_{\Cl}^\text{re}(\rho_{\Sy\Cl}\lb \bar s T_0+\tau_l|\tau_{l}\rb )]
	\\ =&\frac{1}{T_0} \int_{0}^{\infty} \dd \bar s \, \tr[\bar H_{\Sy\Cl}\,	\mathcal{\bar D}_{\Cl}^\text{re}(\bar\rho_{\Sy\Cl}\lb \bar s +\bar\tau_l|\bar\tau_{l}\rb )],\label{eq: E re independet of T0}
\end{align}
where we performed the change of integration variable, $\bar s := s /T_0$ after the $2\ndd$ equality. Now note that the integral in line~\labelcref{eq: E re independet of T0} is written entirely in terms of $T_0$-independent quantities. It thus follows that 
\begin{align}
	\braket{\bar E^\text{re}} :=  \braket{E^\text{re}}  T_0
\end{align}
is $T_0$-independent.  Hence
\begin{align}
	P^\textup{in}_l:=\frac{\braket{E^\text{re}}}{T_0}= \frac{\braket{\bar E^\text{re}}}{T_0^2}
\end{align}
and the dimensionless quantity 
\begin{align}
	T_0^2 P^\textup{in}_l 
\end{align}
is $T_0$-independent. Therefore, the gate frequency $f$ in~\cref{eq:thm fixed memory 5} in~\Cref{thm:heat dissipation}  is directly proportional to $1/T_0$, analogously to the dependency in~\cref{eq:SQL frequncy,eq:HL frequncy}.

For completeness, we will now show that the instantaneous initial-cycle-state parameter error, $\varepsilon_\textup{H}^0$ can be chosen to be $T_0$-independent w.l.o.g. (this is only relevant for~\Cref{sec:thm4 main text up}.)
For this, we start by noting that the $1\stt$ moment of the probability density $P(t+\tau_l,+1|\tau_l)$ can be written in the form
\begin{align}
	\mathbb{M}_l(1):=	 \int_0^\infty \dd t\,  t\, P(t+\tau_l,+1|\tau_l) = T_0   \int_0^\infty \dd \bar t\,  \bar t \, 2 \tr[ \bar V_\Cl \bar\rho_{\Sy\Cl}\lb\bar t+\bar\tau_l|\bar\tau_l\rb ] =: T_0 \,\bar{\mathbb{M}}_l(1),
	\label{def:1st moment finite 2}
\end{align}
where  $\bar{\mathbb{M}}_l(1)$ is manifestly $T_0$-independent.

From the definition of the instantaneous initial-cycle-state parameter, $\varepsilon^0_\textup{H}$ in~\cref{def:epsion 0 H}, and the  above definitions, one has that  $\|  \tr_\Cl[\tilde H_{\Sy\Cl} \,\rho^0_\Cl]  \|_F$, and $\| \tilde H_{\Sy\Cl}   \|_F$ are directly proportional to $1/T_0$. Then using~\cref{def:1st moment finite 2}, we conclude that $\varepsilon^0_\textup{H}$ is an upper bound to a $T_0$-independent quantity, and hence can itself be chosen to be $T_0$-independent.

\section{Achieving the lower bound for a quantum frequential computer in a non-equilibrium steady state}\label{sec:Steady-state and heat generation}

Here we will now specialise the description of a computer with a self-oscillator from~\Cref{sec:Setup of the dynamical semigroup 2} to the case relevant for~\Cref{thm:heat dissipation} and describe how the memory on $\M_0$ is updated in accordance with how this was previously achieved via the bus. We finish with the proof of~\Cref{thm:heat dissipation}.

\subsection{Specialization of the generator of dynamics and $\M_0$ memory updating rule}\label{sec:Setup of the dynamical semigroup}
We start by speciating the terms in the dynamical semigroup~\cref{eq:dynamic semigroup for computation} to the relative terms. Some of the parameter choices will be left to the proofs in the subsequent section.

In our case, we choose  $H_{\Sy\Cl}=H'_{\M_0\lo\W\Cl}$ where
\begin{align}
	H'_{{\M_0}\lo\W\Cl}= H_\Cl +   \sum_{l=1}^{N_g-1} I_{{\M_0}\lo\W} ^{(l)}\otimes I_\Cl^{(l)} ,\label{eq:main complete ham 3}
\end{align}
and where the terms on the r.h.s.  are as in~\Cref{sec:proof of 2 clock theomre in main text}. E.g. for $I_{{\M_0}\lo\W} ^{(l)}$ see definition in~\cref{eq:int meme sys swit}.
Recall that the dissipative part is fully determined by specifying the operators $\{J_\Cl^j\}_j$. In our case, we set
\begin{align}
	J_\Cl^j=  \sqrt{2 v_j}\, {}_\Cl\!\ketbra{0}{\theta_j}_\Cl, \qquad v_j >0,
\end{align}
$j=0,1,\ldots, d-1$ where $\ket{0}_\Cl$ is the unperturbed initial state of the oscillator and $\{ \ket{\theta_{j}}_\Cl  \}_{j=0}^{d-1}$  is the discrete Fourier transform basis of the energy eigenbasis of $H_\Cl$ defined in~\cref{eq:psi def}.

One can check that this definition readily satisfies~\cref{eq:D re prop constrinat}. This is to say, the renewal process therefore maps all input states to the state of the oscillator to its initial state: $\mathcal{D}_{\Cl}^\textup{re}(\rho_\Cl)\propto\proj{0}_\Cl$ for all $\rho_\Cl\in\mathcal{S}(\mathcal{H}_\Cl)$. Moreover, in this case the output state $\proj{0}_\Cl$ is exactly the state after each cycle when it evolves according to its free dynamics \textit{without} any perturbations  due to it controlling the implementation of gates, in other words, according to $\me^{-\mi j T_0 H_\Cl} \ket{0}_\Cl=\ket{0}_\Cl$, $j\in\nnz$.  Therefore, if the renewal process occurs periodically with period $T_0$ it will correct for the small perturbative errors incurred by the oscillator due to its implementation of the logical gates required for the computation. 

We can also evaluate the generator of dynamics under the scenario of no renewal. Recall that if a renewal operation does not occur in the infinitesimal time interval $[\tau, \tau+\dd\tau]$, then the state is mapped from $\rho_{\M_0\lo\W\Cl}(\tau)$ to $\mathcal{L}^\textup{no re}( \rho_{\M_0\lo\W\Cl}(\tau) )$  where 
\begin{align}
	\mathcal{L}_{\M_0\lo\W\Cl}^\textup{no re}(\cdot):=  -\mi [H_{\M_0\lo\W\Cl},(\cdot)]  +  	\mathcal{D}_{\Cl}^\textup{no re}(\cdot),\qquad  	\mathcal{D}_{\Cl}^\textup{no re}(\cdot)= -\frac{1}{2}\sum_{j=0}^{d-1} \{ J^{j\dag}_\Cl J^{j}_\Cl, (\cdot)  \}.  \label{eq:def linblad no jump}
\end{align}
Thus from the form of~\cref{eq:def linblad no jump} one can solve the dynamics to find a significant simplification. Namely
\begin{align}
	\me^{\tau 	\mathcal{L}_{\M_0\lo\W\Cl}^\textup{no re}}(\cdot)= \me^{-\mi \tau 	G_{\M_0\lo\W\Cl}} (\cdot)\, \me^{\mi \tau 	G_{\M_0\lo\W\Cl}^\dag} \label{eq:no re dynamics}
\end{align}
where
\begin{align}
	G_{\M_0\lo\W\Cl}:= H_{\M_0\lo\W\Cl}-\mi V_\Cl, \qquad  V_\Cl:= \sum_{j=0}^{d-1} v_j \proj{\theta_j}_\Cl>0.
\end{align} 

We choose
\begin{align}
	V_\Cl= \gamma_0\, I_\Cl^{(N_g)}+  \frac{\varepsilon_b}{2T_0}\id_\Cl,   \label{def:V c pot}
\end{align}
where $I_\Cl^{(N_g)}$ is the positive semidefinite interaction term used in the proofs of~\Cref{thm:comptuer with fixed memory,thm:contrl with two clocks} defined in~\cref{eq:def I in terms of bar V0},  $\varepsilon_b>0$ will be specified later in~\cref{eq:lvarepsion b choice} is vanishing for large $E_0$ for all fixed $\bar\varepsilon>0$. The parameter  $\gamma_0>0$ is a scale factor (specified later in~\cref{eq:initi condition ideal sols})  and $\id_\Cl$ the identity operator.

The states $\{\rho_{\M_0\lo\W\Cl}(\tau_l|\tau_l)\}_l$  have just undergone the renewal process and are thus the output of channel~\cref{eq:def:renawal channel generic def}. As a consequence, they are of the form
\begin{align}
	\rho_{\M_0\lo\W\Cl}(\tau_l | \tau_l)= \tr_\Cl[\rho_{\M_0\lo\W\Cl}(\tau_l| \tau_l)]\otimes \proj{0}_\Cl.  \label{eq:rho intermediate product}
\end{align}

We will now discuss how the implicit updating on the memory cells in $\M_0$ via the  bus is modelled. To this end, we start by noting that~\cref{eq:rho intermediate product} is a product state with memory cell $\M_{0,1}$ due to~\cref{eq:cell restriction thm3}. Moreover, it is convenient  for the following proofs to work with pure states. We will therefore purify the state on $\M_0\lo\W$ via ancillary systems which we denote $\Pu$, leading to a state
\begin{align}
	\ket{\rho(\tau_l | \tau_l)}_{\M_0\lo\W\Cl\Pu}:=\ket{\m_{l,1}}_{\M_{0,1} \Pu_1} \ket{(\tau_l|\tau_l)}_{\bar\M_{0,1} \W\lo \bar\Pu_1}  \ket{0}_\Cl  \label{eq:puried rho contitione t zero}
\end{align}
which satisfies $\tr_\Pu\big[\proj{\rho(\tau_l | \tau_l)}_{\M_0\lo\W\Cl\Pu}\big] 
=  \rho_{\M_0\lo\W\Cl}(\tau_l | \tau_l)$. 
In this~\app{} we used the convention $\bar\M_{0,k}=\M_0\backslash\M_{0,k}$\footnote{This is the convention used in all of~\Cref{sec:Steady-state and heat generation}. Note that it differs from the convention used in~\Cref{sec:proof of 2 clock theomre in main text} where $\bar\M_{0,k}=\M\backslash\M_{0,k}$.} 
and each memory cell $\M_{0,k}$ we associate with its own ancilla $\Pu_k$ and additional ancillae for $\W\lo$. Here $\bar \Pu_1$ denotes the total ancilla system after the removal of $\Pu_1$. We will use the notation $\ket{\m}_{\M_{0,k} \Pu_k} =   \ket{\m}_{\M_{0,k} }\ket{\m}_{\Pu_k}$ for all $\m\in\mathcal{G}\cup\{\zero\}$ since said states on $\M_{0,k}$ are already pure.

The state $\rho_{\M_0\lo\W\Cl}(t+\tau_l | \tau_l)$  can be calculated from the state  $\ket{\rho(\tau_l | \tau_l)}_{\M_0\lo\W\Cl}$ by evolving it according to the dynamical semigroup while conditioning on a renewal event not occurring. Had this been the only process involved, the dynamics would have been given by $	\rho_{\M_0\lo\W\Cl}(t+\tau_l | \tau_l)= \me^{t  	\mathcal{L}_{\M_0\lo\W\Cl}^\textup{no re}} (\rho_{\M_0\lo\W\Cl}(\tau_l | \tau_l) ) / P(t+\tau_l | \tau_l)$ where $P(t+\tau_l | \tau_l)= \tr\left[ \me^{t  	\mathcal{L}_{\M_0\lo\W\Cl}^\textup{no re}} (\rho_{\M_0\lo\W\Cl}(\tau_l | \tau_l) )  \right]$ in accordance with~\cref{eq:conditional states AC def 3}. However, recall that the bus is updating the memory cells  $\M_0$ in a similar fashion to the autonomous and explicit formulation in~\Cref{thm:contrl with two clocks}, although here modelled implicitly for convenience. We will thus model said memory updates as the application of unitaries at discrete times to the memory cells in $\M_0$ which the bus updates as said times.

 In particular, we have assumed that~\cref{eq:cell restriction thm3} holds and complete ignorance on the other memory cells contained in $\M_0$. We thus implement this  assumption via the unitary transformation denoted $U_{\M_{0,k}\Pu_k\M_{0,k-1}\Pu_{k-1}}$ on the states in $\M_0$ which have been purified.  Therefore, noting that the channel~\cref{eq:no re dynamics} preserves purity, the states $\{  \rho_{\M_0\lo\W\Cl}(t +\tau_l| \tau_l)\}_{l=0}^\infty$ for $t\in[t_{k-1}, t_k)$ and $k\in\nnp$ are given by~\footnote{As in the previous sections, kets and bras containing left and right bold font brackets $\lb$ and $\rb$ (or $\lsb$ and $\rsb$) indicate that the states are not necessarily normalised and to distinguish them from their normalised counterparts which will use normal brackets $($ and $)$ or $[$ and $]$.}
\begin{align}
		\rho_{\M_0\lo\W\Cl}(t+\tau_l | \tau_l)   P(t+\tau_l|\tau_l)  &=   	\rho_{\M_0\lo\W\Cl}\lb t+\tau_l | \tau_l\rb , \label{eq:line:def normalization}\\
	\rho_{\M_0\lo\W\Cl}\lb t +\tau_l | \tau_l\rb    &=   \tr_\Pu\big[\proj{\rho\lb t +\tau_l | \tau_l \rb}_{\M_0\lo\W\Cl\Pu}\big],\label{eq:line:def normalization 2}\\
	\ket{\rho\lb t+\tau_l  | \tau_l\rb}_{\M_0\lo\W\Cl\Pu}&:= \me^{-\mi (t-t_{k-1}) G_{\M_0\lo\W\Cl}} \,	U^{(l,k)}_{\M_{0,k}\Pu_k\M_{0,k-1}\Pu_{k-1}}  	\ket{\lb k,l \rb}_{\M_0\lo\W\Cl\Pu},   \label{eq:rho conditional dybamics def}
\end{align}
where\footnote{In the following we have denoted by $U^{(l,k)}_{\M_{0,k}\Pu_k\M_{0,k-1}\Pu_{k-1}}$ the unitary denoted by $U^{(l,k)}_{\M_0\Pu}$ in the main text. This is to avoid notational clutter in the main text while permitting more expression in the~\app. Furthermore, note that memory cell and ancilla $\M_{0,0}\Pu_{0}$ have not been defined. Such systems are not required, and only appear here. Since this operator is proportional to the identity, such systems are only introduced for notational convenience.}
\begin{align}
	\ket{\lb k,l \rb}_{\M_0\lo\W\Cl\Pu} &:=  \me^{-\mi t_{1} G_{\M_0\lo\W\Cl}}\, U^{( l,k-1 )}_{\M_{0,k-1}\Pu_{k-1}\M_{0,k-2}\Pu_{k-2}}   		\ket{\lb k-1,l \rb}_{\M_0\lo\W\Cl\Pu} ,  \qquad k-1\in\nnp\\
	\ket{(1,l)}_{\M_0\lo\W\Cl\Pu} &:= 	\ket{(\tau_l | \tau_l)}_{\M_0\lo\W\Cl\Pu}, \qquad U^{(l,1)}_{\M_{0,1}\Pu_k\M_{0,0}\Pu_{0}}=\id_{\M_{0,1}\Pu_1\M_{0,0}\Pu_{0}},\footnotemark[\thefootnote]
\end{align}
and the conditional probability $P(t+\tau_l|\tau_l)$ is defined by taking the trace on both sides of~\cref{eq:line:def normalization} and noting that $	\tr[\rho_{\M_0\lo\W\Cl\Pu}(t +\tau_l | \tau_l)]=1$.

The set $\big\{ U^{(l,k)}_{\M_{0,k}\Pu_k\M_{0,k-1}\Pu_{k-1}} \big \}_{k=1}^{N_g-1}$ uniquely determine the full set $\big\{ U^{(l,k)}_{\M_{0,k}\Pu_k\M_{0,k-1}\Pu_{k-1}} \big \}_{k=1}^{\infty}$ since 
\begin{align}
	U^{(l,k)}_{\M_{0,k}\Pu_k\M_{0,k-1}\Pu_{k-1}} = U^{(l,k +qN_g)}_{\M_{0,k +qN_g}\Pu_{k +qN_g}\M_{0,k-1  +qN_g}\Pu_{k-1  +qN_g }}  \label{eq:periodi unitary}
\end{align}
for $q\in\nnp$.

Finally, we now introduce the idealised states which we wish to proximate by the  actual dynamics generated by the dynamical semigroup. We denote them by $\big\{ \ket{[t_k+\tau_l|\tau_l]}_{\M_0 \lo\W\Pu} \ket{[ t_k+\tau_l| \tau_l ]}_\Cl \big \}_{k=0}^{\infty}$  and define them as follows
\begin{align}
	\ket{[t_k+\tau_l|\tau_l]}_{\M_0 \lo\W\Pu}&= U^{(l,k+1)}_{\M_{0,k+1}\Pu_{k+1}\M_{0,k}\Pu_{k}}       \U(\m_{l,k}) \ket{[t_{k-1}+\tau_l|\tau_l]}_{\M_0 \lo\W\Pu},\\
	\ket{[\tau_l|\tau_l]}_{\M_0 \lo\W\Pu}&= 	\ket{\rho(\tau_l | \tau_l)}_{\M_0\lo\W\Cl\Pu}=\ket{\m_{l,1}}_{\M_{0,1} \Pu_1} \ket{(\tau_l|\tau_l)}_{\bar\M_{0,1} \W\lo \bar\Pu_1}  \ket{0}_\Cl \label{eq:initi condition ideal sols}\\
	\ket{\lsb t_k+\tau_l| \tau_l \rsb}_\Cl  &= \delta(k) \,\me^{-  \frac{\varepsilon_b t_k}{2T_0}}  \ket{t_k}_\Cl, \qquad \delta(k)=\me^{- \gamma_0\,\beta(k) N_g}, \qquad   	\gamma_0= \bar\gamma_0 \,d^{\,\bar\varepsilon^2},   \label{def:squre braket control}   \\
		\ket{[ t_k+\tau_l| \tau_l ]}_\Cl &=  	\frac{\ket{\lsb t_k+\tau_l| \tau_l \rsb}_\Cl }{\| 	\ket{\lsb t_k+\tau_l| \tau_l \rsb}_\Cl\|_2}= \ket{t_k}_\Cl= \ket{\Psi(t_k d/T_0)}_\Cl,  \label{def:squre braket control  below} 
\end{align}
$k=1,2,3\ldots$ where $\bar\gamma_0>0$ is any positive $d$-independent parameter, $\beta{(k)}$ is the largest non-negative integer such that
\begin{align}
	k=\beta(k) \,N_g+r, \label{eq:first beta k def}
\end{align}
for some  $0\leq r < N_g$. Notice that the choice $\ket{t_k}_\Cl= \ket{\Psi(t_k d/T_0)}_\Cl$ is identical to that of the proofs to~\Cref{thm:comptuer with fixed memory,thm:contrl with two clocks} (namely~\cref{eq:qusi idea no pot def}). Recall that $\ket{\Psi(t_k d/T_0)}_\Cl$ is periodic: $\ket{\Psi(t_k d/T_0)}_\Cl=\ket{\Psi(t_{k+N_g} d/T_0)}_\Cl$, c.f.~\cref{eq:clock cyclcicity}. And where $\U(\m_{l,N_g}):=\id$ (since there is no gate being applied in this time step), and we assume periodic boundary conditions $ \U(\m_{l,k +q N_g})=\U(\m_{l,k})$, $k=1,2,\ldots, N_g$, $q\in \nnp$. This condition is only relevant in the event that the renewal process occurs late, which as 1) in~\Cref{thm:heat dissipation} shows, is extremely unlikely.\footnote{We could have alternatively imposed a hard ``cut-off'' condition where we define the unitary to be the identity for $k\geq N_g$ and a similar theorem would follow.}  
The only constraint on the unitaries $\big\{ U^{(l,k)}_{\M_{0,k}\Pu_k\M_{0,k-1}\Pu_{k-1}} \big \}_{k=2}^{N_g-1}$ is that they are such that~\cref{eq:cell restriction thm3} is satisfied. This amounts to property
\begin{align}
   \ket{\lsb t_{k}+\tau_l|\tau_l\rsb}_{\M_0 \lo\W\Pu}=   \ket{\m_{l,k+1}}_{\M_{0,{k+1}}\Pu_{k+1}}  \ket{\lsb t_{k}+\tau_l|\tau_l\rsb}_{\bar\M_{0,k+1} \lo\W\bar\Pu_{k+1}}, \label{eq:cond register idealised}
\end{align}
for $k=1,2,\ldots, N_g-2$ (for $k=0$, the state $  \ket{\lsb t_{k}+\tau_l|\tau_l\rsb}_{\M_0 \lo\W\Pu}$ already satisfies condition~\cref{eq:cond register idealised} by virtue of~\cref{eq:initi condition ideal sols}).
Notice that due to the periodicity of the 
 control states, $\{\ket{\Psi(t_k d/T_0)}_\Cl\}_k$, there are implied periodic boundary conditions on the states $\{  \ket{\m_{l,k+1}}_{\M_{0,{k+1}}\Pu_{k+1}}\}_k$ in~\cref{eq:cond register idealised}. Physically, the nature of the imposed boundary conditions is not very relevant since the probability of the dynamics not being renewed after each cycle of the control is very small as 1) in~\Cref{thm:heat dissipation} shows.

We now make a definition needed for the following lemma. For $j\in\nnp$, let $\beta'(j)$ be the largest non-negative integer such that
\begin{align}
	j=\beta'(j) N_g+r,  \label{def:beta prime}
\end{align}
for some $r\in\{1,2,\ldots ,N_g-1\}$. (Note that this differs from the definition of $\beta(k)$ in~\cref{eq:first beta k def}, due to a difference in the domain or $r$.)

\subsection{Proof of~\Cref{thm:heat dissipation}}
Before proving~\Cref{thm:heat dissipation}, we start by proving a lemma which we will use multiple times in the proof of~\Cref{thm:heat dissipation}. 
\begin{lemma}\label{lem:last lemm}
For all $l,j\in\nnz$, $\bar\varepsilon>0$ 
\begin{align}
 	&\me^{\frac{\varepsilon_b }{2T_0}t_j }  	\big \|    		\ket{\rho\lb t_j+\tau_l | \tau_l\rb}_{\M_0\lo\W\Cl\Pu} - \ket{[t_j+\tau_l|\tau_l]}_{\M_0 \lo\W\Pu} \ket{\lsb t_j+\tau_l| \tau_l \rsb}_\Cl  \big\|_2 
 	\\&\leq   \left(\sum_{q=1}^{N_g}  \tilde d(\m_{l,q})\right) \frac{1-\me^{-\gamma_0 N_g (1+\beta'(j))}}{1-\me^{-\gamma_0 N_g}} g( \bar\varepsilon) \,\textup{poly}(E_0) \, E_0^{-1/\sqrt{\bar\varepsilon}}	, \label{eq:bound on states in thm new}
\end{align}	
where the r.h.s. is $\varepsilon_b$-independent. Where as previously, $g(\bar\varepsilon)>0$ is $E_0$-independent while $\textup{poly}(E_0)$ is an $\bar\varepsilon$-independent polynomial in $E_0$. The product $g( \bar\varepsilon) \,\textup{poly}(E_0)$ is independent of $l$, $j$ and $\varepsilon_b$.
\end{lemma}

Recall that $E_0:=\tr[H_{\M_0\lo\Cl} \rho_{\M_0\lo\Cl}^0]$ [as per definition~\labelcref{eq:def of power one clock only case}]. In the present context, the relevant mean energy is $ \tr[H'_{\M_0\lo\W\Cl}\,	\rho_{\M_0\lo\W\Cl}(\tau_l|\tau_l ) ]$.  However, writing 
\begin{align}
	\tr[H'_{\M_0\lo\W\Cl}\,  \rho_{\M_0\lo\W\Cl}(\tau_l|\tau_l)  ]
	=  E_0 + 	\delta E^\text{re}_1
\end{align}
one can show that the difference $\delta E^\text{re}_1$ is superpolynomially small, namely $\delta E^\text{re}_1= \SupPolyDecay(E_0)$. This is because of two points: 1) both Hamiltonians only differ by their interaction terms and both states are close to orthogonal to these. 2) the reduced states on $\Cl$ are the same. It is not essential as a stand-alone result, so we only prove it when strictly needed at the end of the proof of~\Cref{thm:heat dissipationapp} in~\cref{eq:initi condition used 00}. 

 Ultimately, in~\Cref{thm:heat dissipation} we wish to describe our results as a function of power $P^\text{in}$. As we will show in the poof of~\Cref{thm:heat dissipation} below, $P^\text{in}$ is also proportional to $E_0$.

\begin{proof}
	The proof will follow similar steps to that of~\Cref{thm:comptuer with fixed memory}, but with important differences. We will only briefly cover the steps which follow analogously to that of~\Cref{thm:comptuer with fixed memory} for brevity.
	
	To start with, note that we can write $\ket{\rho(t_j +\tau_l| \tau_l)}_{\M_0\lo\W\Cl\Pu}$ in the form
	\begin{align}
		\	\ket{\rho\lb t_j +\tau_l| \tau_l \rb}_{\M_0\lo\W\Cl\Pu}= \Delta_j \Delta_{j-1}\ldots \Delta_1 \,	\ket{\rho(\tau_l | \tau_l)}_{\M_0\lo\W\Cl\Pu}, \qquad \Delta_j= U^{(l,j+1)}_{\M_{0,j+1}\Pu_{j+1}\M_{0,j}\Pu_{j}} \me^{-\mi t_1 G_{\M_0\lo\W\Cl}}   
	\end{align}
for $j=0,1,2, \ldots$.   Therefore, by making the association $\{ \ket{\Phi_j}= \ket{[t_j+\tau_l|\tau_l]}_{\M_0 \lo\W\Pu}    \ket{\lsb t_j+\tau_l| \tau_l \rsb}_\Cl \}_{j=0}^{\infty}$, applying~\Cref{lem:unitari errs add linararly} we find for $j=0,1,2, \ldots$.
	\begin{align}
&\me^{\frac{\varepsilon_b}{2T_0} t_j}\, 	\big \|    		\ket{\rho\lb t_j+\tau_l | \tau_l \rb}_{\M_0\lo\W\Cl\Pu} - \ket{[t_j+\tau_l|\tau_l]}_{\M_0 \lo\W\Pu} \ket{\lsb t_j+\tau_l| \tau_l \rsb}_\Cl  \big\|_2\\
	 &  \leq \me^{\frac{\varepsilon_b}{2T_0} t_j}\,  \sum_{q=1}^j  \big\|   \ket{[t_q+\tau_l|\tau_l]}_{\M_0 \lo\W\Pu}   \ket{\lsb t_q+\tau_l| \tau_l \rsb}_\Cl    -  U^{(l,q+1)}_{\M_{0,q+1}\Pu_{q+1}\M_{0,q}\Pu_{q}} \me^{-\mi t_1 G_{\M_0\lo\W\Cl}}     \ket{[t_{q-1}+\tau_l|\tau_l]}_{\M_0 \lo\W\Pu}  \ket{\lsb t_{q-1}+\tau_l| \tau_l \rsb}_\Cl   
	   \big \|_2\\
	   	 &  \leq  \sum_{q=1}^j  \big\|    \ket{[t_q+\tau_l|\tau_l]}_{\M_0 \lo\W\Pu} \delta(q) \ket{ t_q}_\Cl    -  U^{(l,q+1)}_{\M_{0,q+1}\Pu_{q+1}\M_{0,q}\Pu_{q}} \me^{-\mi t_1 G'_{\M_0\lo\W\Cl}}     \ket{[t_{q-1}+\tau_l|\tau_l]}_{\M_0 \lo\W\Pu}  \delta({q-1}) \ket{t_{q-1}}_\Cl   
	   \big \|_2\\
	  	 &  =  \sum_{q=1}^j  \big\|    U(\m_{l,q}) \ket{[t_{q-1}+\tau_l|\tau_l]}_{\M_0 \lo\W\Pu}   \delta(q) \ket{t_{q}}_\Cl      -   \me^{-\mi t_1 G'_{\M_0\lo\W\Cl}}     \ket{[t_{q-1}+\tau_l|\tau_l]}_{\M_0 \lo\W\Pu}  \delta(q-1)  \ket{ t_{q-1}}_\Cl   
	  \big \|_2, \label{eq:intermed g one}
\end{align}
where in the third line we have defined $G'_{\M_0\lo\W\Cl}:=G_{\M_0\lo\W\Cl}   -\mi {\frac{\varepsilon_b}{2T_0}}$  while in the last line we used the unitary invariance of the two-norm.

Define
\begin{align}
	G^{{(k)\prime}}_{\M_0\lo\W\Cl}:=   
	\begin{cases}
		H_\Cl +   I_{\M_0\lo\W}^{(k)} \otimes I_\Cl^{(k)}  & \mbox{ if }  k=1,2,\ldots, N_g-1 		\vspace{0.1cm} \\ 
		H_\Cl-\mi \gamma_0I_\Cl^{(N_g)}  & \mbox{ if }  k=N_g
	\end{cases}
\end{align}
Thus adding and subtracting an appropriate term in~\cref{eq:intermed g one} we find 
	\begin{align}
	&\me^{\frac{\varepsilon_b}{2T_0}t_j}\, 	\big \|    		\ket{\rho\lb t_j +\tau_l| \tau_l\rb}_{\M_0\lo\W\Cl\Pu} - \ket{[t_j+\tau_l|\tau_l]}_{\M_0 \lo\W\Pu} \ket{\lsb t_j+\tau_l| \tau_l \rsb}_\Cl  \big\|_2\\
  	 &  \leq  \sum_{q=1}^j  \big\|   \U(\m_{l,q}) \ket{[t_{q-1}+\tau_l|\tau_l]}_{\M_0 \lo\W\Pu}   \delta(q) \ket{t_{q}}_\Cl      -   \me^{-\mi t_1 G^{(q)\prime} _{\M_0\lo\W\Cl}}     \ket{[t_{q-1}+\tau_l|\tau_l]}_{\M_0 \lo\W\Pu}  \delta(q-1)  \ket{ t_{q-1}}_\Cl   
\big \|_2\label{eq:eq:l:1st sum 1}\\
&\;\;\quad +   \delta(q-1)\, \big\|    \me^{-\mi t_1 G^{(q)\prime} _{\M_0\lo\W\Cl}}     \ket{[t_{q-1}+\tau_l|\tau_l]}_{\M_0 \lo\W\Pu}   \ket{ t_{q-1}}_\Cl        -   \me^{-\mi t_1 G'_{\M_0\lo\W\Cl}}     \ket{[t_{q-1}+\tau_l|\tau_l]}_{\M_0 \lo\W\Pu}    \ket{ t_{q-1}}_\Cl   
\big \|_2.\label{eq:eq:l:1st sum 2}
\end{align}
To proceed, for line~\labelcref{eq:eq:l:1st sum 1}  we can replace $G^{(q)\prime}_{\M_0\lo\W\Cl}$ with 
\begin{align}
	G^{(q, \m_{l,q})}_{\lo\W\Cl}:= 
	\begin{cases}
		H_\Cl +   I_{\lo}^{(q, \,\m_{l,q})} \otimes I_\Cl^{(q)}  & \mbox{ if }  q=1,2,\ldots, N_g-1 	\text{ and  }	\m_{l,q} \in\mathcal{G} \vspace{0.1cm} \\ 
			H_\Cl +   I_{\W}^{(q,\, \m_{l,q})} \otimes I_\Cl^{(q)}  & \mbox{ if }  q=1,2,\ldots, N_g-1 \text{ and  }  \m_{l,q} =\zero \vspace{0.1cm} \\ 
		H_\Cl-\mi\,\gamma_0 I_\Cl^{(N_g)}  & \mbox{ if }  q=N_g
	\end{cases}
\end{align}
since $\ket{[t_{q-1}+\tau_l|\tau_l]}_{\M_0 \lo\W\Pu} $ is proportional to $\ket{\m_{l,q}}_{\M_{0,q}\Pu_q}$ due to~\cref{eq:cond register idealised}.

For line~\labelcref{eq:eq:l:1st sum 2}, we can proceed analogously to lines~\labelcref{eq:first of differene main lem} to the end of the proof. Crucially, we note that the final results still follow even when the generator is not trace-preserving nor self-adjoint (such as is the case for $G'_{\M_0\lo\W\Cl}$). Thus we find
\begin{align}
	&\me^{\frac{\varepsilon_b}{2T_0} t_j}\, 	\big \|    		\ket{\rho\lb t_j +\tau_l| \tau_l\rb}_{\M_0\lo\W\Cl\Pu} - \ket{[t_j+\tau_l|\tau_l]}_{\M_0 \lo\W\Pu} \ket{\lsb t_j+\tau_l| \tau_l \rsb}_\Cl  \big\|_2\\
	&  \leq  \sum_{q=1}^j  \bigg(\big\|   \ket{[t_{q-1}+\tau_l|\tau_l]}_{\M_0 \lo\W\Pu}   \delta(q) \ket{t_{q}}_\Cl      -   \U^\dag(\m_{l,q})\,\me^{-\mi t_1 G^{(q, \m_{l,q})}_{\lo\W\Cl}}     \ket{[t_{q-1}+\tau_l|\tau_l]}_{\M_0 \lo\W\Pu}  \delta(q-1)  \ket{ t_{q-1}}_\Cl   
	\big \|_2 \\
	& +  \delta(q-1)\, t_1\max_{x\in[0,t_1]} \big\|  \left( G^{(q)\prime} _{\M_0\lo\W\Cl}- G'_{\M_0\lo\W\Cl}  \right)   \me^{-\mi x\, G^{(q)\prime} _{\M_0\lo\W\Cl}}    \ket{[t_{q-1}+\tau_l|\tau_l]}_{\M_0 \lo\W\Pu}  \delta(q-1)  \ket{ t_{q-1}}_\Cl   	\big \|_2 \bigg)\\
		&  = \sum_{q=1}^j  \bigg(\big\|   \ket{[t_{q-1}+\tau_l|\tau_l]}_{\M_0 \lo\W\Pu}   \delta(q) \ket{t_{q}}_\Cl      -   \U^\dag(\m_{l,q})\,\me^{-\mi t_1 G^{(q, \m_{l,q})}_{\lo\W\Cl}}     \ket{[t_{q-1}+\tau_l|\tau_l]}_{\M_0 \lo\W\Pu}   \delta(q-1) \ket{ t_{q-1}}_\Cl   
	\big \|_2 \label{eq:last line diff 1}\\
	& +  \delta(q-1) \,t_1 \max_{x\in[0,t_1]} \big\|  \left( G^{(q, \m_{l,q})} _{\lo\W\Cl}- G'_{\M_0\lo\W\Cl}  \right)   \me^{-\mi x\, G^{(q,\m_{l,q})} _{\lo\W\Cl}}    \ket{[t_{q-1}+\tau_l|\tau_l]}_{\M_0 \lo\W\Pu}   \ket{ t_{q-1}}_\Cl   	\big \|_2 \bigg), \label{eq:last line diff}
\end{align}	
where in the last equality we have noted that we can exchange $G^{(q)\prime} _{\M_0\lo\W\Cl}$ with $G^{(q, \m_{l,q})}_{\lo\W\Cl}$ for the same reasons as this exchange was possible before.  Recall
\begin{align}
	 \U(\m_{l,k})=\begin{cases}
	 \me^{\mi I_\lo^{(k,\,\m_{l,k})}}  &\mbox{ if } \m_{l,k}\in\mathcal{G}  \text{ and } k\neq N_g\\
	 \me^{\mi I_\W^{(k,\,\m_{l,k})}}  &\mbox{ if } \m_{l,k}=\zero \text{ and } k\neq N_g\\
	 	 \id &\mbox{ if } k=N_g.
	 \end{cases}
\end{align}
Therefore, by expanding $ \ket{[t_{q-1}+\tau_l|\tau_l]}_{\M_0 \lo\W\Pu}  $ in the eigenbasis of $I_\lo^{(q,\,\m_{l,q})}\otimes I_\W^{(q,\,\m_{l,q})}$ (and any basis for $\M_0\Pu$) we can bound line~\labelcref{eq:last line diff} analogously to as in the proof of~\Cref{lem:1st secondrary lem for unitary bound}. Crucially, for this to work in the case $j=N_g$ it is important to note that Theorem IX.1 (\textit{Moving the clock through finite time with a potential}) from~\cite{WoodsAut} which is used in the proof to~\Cref{lem:1st secondrary lem for unitary bound}, applies when the potential function is from $\rr$ to $\rr\cup \hh^{-}$, where  $\hh^{-}:= \{a_0+\mi b_0 \;|\; a_0\in\rr, b_0<0\}$, i.e. not only to potential functions of the form $\rr$ to $\rr$ as in the cases we have considered thus far. For the $j=N_g$ case the potential function is $\rr$ to $\mi \,\rr_{\leq 0}$. This small change modifies  the $\varepsilon_v(t_1,d)$ function slightly, for this $j=N_g$ case. As such we can upper bound line~\labelcref{eq:last line diff} by lines~\labelcref{eq:1 line inequality ln1st tech lem,eq:2 line inequality ln1st tech lem} up to a modification in $\varepsilon_v(t_1,d)$ [defined in line~\labelcref{eq:var e v def}] when  $q=N_g$  which we will detail below starting in the paragraph above~\cref{eq:new epsilon v}.


As for upper bounding line~\labelcref{eq:last line diff 1}, for $q\in\{ 1,2,3,\ldots \} \backslash \{ N_g, 2N_g, 3 N_g, \ldots\}$, we have that $\delta(q)=\delta(q-1)$ and can be taken outside of the two-norm as a common multiplicative factor. Then,  by expanding $ \ket{[t_{q-1}+\tau_l|\tau_l]}_{\M_0 \lo\W\Pu}  $ in the eigenbasis of $I_\lo^{(q,\,\m_{l,q})}\otimes I_\W^{(q,\,\m_{l,q})}$ (and any basis for $\M_0\Pu$) it follows identically to the proof of~\Cref{lem:2st secondrary lem for unitary bound}  and as such, line~\labelcref{eq:last line diff 1} is upper bounded by line~\labelcref{eq:line: 1 of upper bound of tech lmee 2}. For $q=m N_g$, $m\in\nnp$ there are some modifications due to the difference in the potential function mentioned above. As such, we calculate it here for completeness. Noting $\delta(q)=\me^{-\gamma_0}\delta(q-1)$,  for $q=m N_g$, the square of line~\labelcref{eq:last line diff 1} reduces to
\begin{align}
  & \delta^2(q-1) \, \big\|   \ket{[t_{N_g-1}+\tau_l|\tau_l]}_{\M_0 \lo\W\Pu}  \me^{-\gamma_0} \ket{(t_{N_g}d/T_0)}_\Cl      -  \me^{-\mi t_1 (H_\Cl-\mi \gamma_0 I_\Cl^{(N_g)})}     \ket{[t_{N_g-1}+\tau_l|\tau_l]}_{\M_0 \lo\W\Pu}   \ket{(t_{N_g-1}d/T_0)}_\Cl   
	\big \|_2 ^2\\
	& \leq  \delta^2(q-1) \left(\me^{-2 \gamma_0}+ \big\| \me^{-\mi t_1 (H_\Cl-\mi \gamma_0 I_\Cl^{(N_g)})}     \ket{ (t_{N_g-1}d/T_0)}_\Cl  \big \|_2^2 -2 \me^{-\gamma_0} \,\Re\left[ {}_\Cl\!\bra{ (t_{N_g}d/T_0)}  \me^{-\mi t_1 (H_\Cl-\mi \gamma_0 I_\Cl^{(N_g)})}     \ket{(t_{N_g-1}d/T_0)}_\Cl  \right] \right)  \\
	& \leq \delta^2(q-1) \left(\me^{-2 \gamma_0}+   3 \| \ket{\varepsilon(t_1, d)}_\Cl \|_2+   \| \ket{\varepsilon(t_1, d)}_\Cl \|_2^2 + \sum_{k\in\mathcal{S}_d(d) } \me^{-2\gamma_0 \int_{k-d/N_g}^k dy I_{\Cl,d}^{(N_g)}(y) } | \psi_\textup{nor}(d;k)|^2 \right)  \label{eq:line:D41 last }  \\
		& \leq \delta^2(q-1) \Bigg(\me^{-2 \gamma_0}+   3 \| \ket{\varepsilon(t_1, d)}_\Cl \|_2+   \| \ket{\varepsilon(t_1, d)}_\Cl \|_2^2\\
		&\quad + d \max_{k'\in[-1/2, 1/2]} \, \me^{-2\gamma_0 \frac{2\pi}{d}\int_{k'd-d/(2N_g)}^{k' d+d/(2 N_g) } dy' \bar V_0\left(  \frac{2\pi}{d} y' +x_0^{(N_g)} \right) } | \psi_\textup{nor}(d;dk'+d)|^2 \Bigg) \label{eq:line:D41 last 3}  
		\\		& \leq  \delta^2(q-1)  \Bigg(\me^{-2 \gamma_0}+   3 \| \ket{\varepsilon(t_1, d)}_\Cl \|_2+   \| \ket{\varepsilon(t_1, d)}_\Cl \|_2^2
		\\&\quad +  A_\textup{nor}^2  \bigg(2\, \me^{-\frac{\pi}{8 }\frac{d^2 }{\sigma^2 N_g^2} } + d \max_{q\in[-1, 1]} \, \me^{-2\gamma_0 \frac{2\pi}{d}\int_{-1/2+q/4}^{1/2-q/4 } dy'' \bar V_0\left(  \frac{2\pi}{N_g} y'' +x_0^{(N_g)} \right) } \bigg)\Bigg) \label{eq:line:D41 last 4}  
		\\		& \leq  \delta^2(q-1)  \Bigg(\me^{-2 \gamma_0}+   3 \| \ket{\varepsilon(t_1, d)}_\Cl \|_2+   \| \ket{\varepsilon(t_1, d)}_\Cl \|_2^2
		\\&\quad +   A_\textup{nor}^2 \left(2\, \me^{-\frac{\pi}{8 }\frac{d^2 }{\sigma^2 N_g^2} } + d \,    \me^{-2\gamma_0-2\gamma_0  \left( 4\pi n A_0  \right) \left(   \left(\frac{2N_g}{\pi^2\,n} \right)^{2N} + \frac{\pi^2}{3}\left(\frac{1}{2\pi^2 n}\right)^{2N}  \right) } \right)\Bigg)  \label{eq:line:D41 last 5}
		\\		& \leq  \delta^2(q-1)  \left(\me^{-2 \gamma_0}+   3 \| \ket{\varepsilon(t_1, d)}_\Cl \|_2+   \| \ket{\varepsilon(t_1, d)}_\Cl \|_2^2 +   A_\textup{nor}^2 \left(2\, \me^{-\frac{\pi}{8 }\frac{d^2 }{\sigma^2 N_g^2} } + d \,  \me^{-2\gamma_0 \mathstrut \vspace{0.3cm} } \right)\right),\label{eq:up to here expalicit}
\end{align}
where in line~\labelcref{eq:line:D41 last } we have employed Theorem IX.1 (\textit{Moving the clock through finite time with a potential}) from~\cite{WoodsAut} now with the pure imaginary potential function.  In line~\labelcref{eq:line:D41 last 3}  we have noted that $k\in\mathcal{S}_d(d)$ is the same as $k\in \{ \lceil -d/2+d \rceil, \lceil -d/2+d\rceil+1, \ldots,\lceil -d/2+d \rceil+d-1  \}$. Therefore, since $\lceil -d/2+d \rceil+d-1 = \lceil -d/2+d +d-1 \rceil \leq d/2+d$ we can generate an upper bound by replacing $\sum_{k\in\mathcal{S}_d(d) }(\cdot)$ with $d \max_{k\in[-d/2+d, d/2+d]}(\cdot)$. We have then performed a change of variable for $k$ and $y$ in the integral, followed by substituting for the definition of $I_{\Cl,d}^{(N_g)}(\cdot)$. In line~\labelcref{eq:line:D41 last 4}, we have split the interval $[-1/2,1/2]$ over which we are maximising into three partitions, $[-1/2,1/2]= [-1/2,-1/(4 N_g)]\cup \big(-1/(4 N_g),1/(4 N_g)\big)\cup [1/(4 N_g),1/2]$ and maximised over them individually. Over the $1\stt$ and last interval, we have upper bounded the exponentiated integral by one, while we have bounded the modulus-squared wave function by one in the middle interval. Finally, we have performed a change of variable for $k'\in\big(-1/(4 N_g),1/(4 N_g) \big)$ to $q\in(-1,1) \subset [-1,1]$ followed by a change of integration variable.  In to achieve line~\labelcref{eq:line:D41 last 5}, we have written the last exponential in line~\labelcref{eq:line:D41 last 4} in the form $2 \gamma_0 (-1+[1-\beta_\text{Int} ])$  with $\beta_\text{Int}:=\frac{2\pi}{d}\int_{-1/2+q/4}^{1/2-q/4 } dy'' \bar V_0\left(  \frac{2\pi}{N_g} y'' +x_0^{(N_g)} \right)$ followed by noting that we have already upper bounded $[1-\beta_\text{Int}]^2$ between lines~\labelcref{eq:line:1 of many,eq:q change line 7}. Since it follows from the normalization of $\bar V_0(\cdot)$ that $ \beta_\text{Int}\leq 1$, $[1-\beta_\text{Int}]=\sqrt{[1-\beta_\text{Int}]^2}$ and hence we have used said bound to generate line~\cref{eq:line:D41 last 5}.

We now calculate $\| \ket{\varepsilon(t_1, d)}_\Cl \|_2$ noting that it takes on a slightly different expression to that of previous proofs due to the relevant potential function having the $\gamma_0$ prefactor. Our potential function $\bar V_0(\cdot)$ satisfies
\begin{align}
	\sup_{k\in\nnp} \left(2 \max_{x\in[0,2\pi]} \left| \frac{d^{(k-1)}}{dx^{(k-1)} } \bar V_0(x)  \right|  \right)^{1/k} \leq 2n C_0,
\end{align}
and $b$ was any upper bound to the l.h.s. Before, we thus set $b=2n C_0$ 
[See Eqs. (F217) and (F218) pg 63 in and text between them in~\cite{WoodsPRXQ} and recall that we have set $\delta=1$ in this~\doc]. 
Furthermore, independently of how $b$ is set, we have the relations
\begin{align}
	d^{\epsilon_5}&=\frac{d}{\bar v \sigma},\qquad  \epsilon_5>0,\\
	\bar{v}&=\frac{\pi \alpha_0 \kappa}{\ln \left(\pi \alpha_0 \sigma^2\right)} b,  \\
	\sigma&=d^{\epsilon_6},\\
	\epsilon_5& < \epsilon_6,\\
\varepsilon_\nu(t, d)&=|t| \frac{d}{T_0}\left(\mathcal{O}\left(\frac{\sigma^3}{\bar{v} \sigma^2 / d+1}\right)^{1 / 2}+\mathcal{O}\left(\frac{d^2}{\sigma^2}\right)\right) \exp \left(-\frac{\pi}{4} \frac{\alpha_0^2}{\left(\frac{d}{\sigma^2}+\bar{v}\right)^2}\left(\frac{d}{\sigma}\right)^2\right)+\mathcal{O}\left(|t| \frac{d^2}{\sigma^2}+1\right) \mathrm{e}^{-\frac{\pi}{4} \frac{d^2}{\sigma^2}}+\mathcal{O}\left(\mathrm{e}^{-\frac{\pi}{2} \sigma^2}\right),
\end{align}
where in our notation $\|	\ket{\varepsilon(t,d)}_\Cl\|_2=\varepsilon_\nu(t, d)$   (see eqs. (F219), (F33),  (F214), (F213)  and (F38) respectively in~\cite{WoodsPRXQ}). These give rise to~\cref{eq:var e v def}.

We now have that the potential function of interest is $\gamma_0 \bar V_0(\cdot)$, and hence the value for $	\varepsilon_\nu(t, d)$ will be modified. In particular, we have 
\begin{align}
	\sup_{k\in\nnp} \left(2 \max_{x\in[0,2\pi]} \left| \frac{d^{(k-1)}}{dx^{(k-1)} } \gamma_0\bar V_0(x)  \right|\,  \right)^{1/k} \leq 2n \gamma_0 C_0,
\end{align}
leading to  $b=2 n \gamma_0 C_0$.
We recall that $\gamma$ was chosen in~\cref{def:squre braket control} with the parametrization
\begin{align}
	\gamma_0= \bar\gamma_0 d^{\epsilon_p},  \label{eq:pot strengh}
\end{align} 
with $\epsilon_p=\bar\varepsilon^2$, and where we assume $\bar\gamma_0$ to be $d$-independent.  let us now justify this parametrization (i.e. we will now write $\epsilon_p$, and show why setting $\epsilon_p=\bar\varepsilon^2$ is a reasonable choice). From above we find that
\begin{align}
	\varepsilon_\nu(t, d)&=|t| \, \textup{poly}(d) \exp{\left(-\frac{\pi}{4} \frac{\alpha_0^2}{(\bar\gamma_0+d^{\epsilon_5}/\sigma)^2} d^{2(\epsilon_5-\epsilon_p)}\right)}.\label{eq:new epsilon v}
\end{align}
This reduces to~\cref{eq:var e v def} in the limit $\bar\gamma_0\to1$, $\epsilon_p\to0$ as expected. From~\cref{eq:new epsilon v} we observe that we need to choose $\epsilon_p$ such that $\epsilon_5-\epsilon_p>0$ and parametrize $\varepsilon$ in terms of $\bar\varepsilon$.  To do so, recall $\epsilon_5=\eta \bar \varepsilon$ (\cref{last epsion eqs}) and that in the case under consideration, where the optimal logical frequency $f$ is asymptotically achievable, we have using~\cref{eq:eta as func of ep}
\begin{align}
	\epsilon_5=\eta\bar\varepsilon = \frac{2}{2\bar\varepsilon+2\sqrt{\bar\varepsilon}+1} \,\bar\varepsilon^2
\end{align}
where we set $\epsilon_p = \bar\varepsilon^2$ such that 
\begin{align}
	\epsilon_5-\epsilon_p = \left(\frac{2}{2\bar\varepsilon+2\sqrt{\bar\varepsilon}+1} -1\right) \bar\varepsilon^2 
\end{align}
for $\bar\varepsilon\in(1-\sqrt{3}/2, 0)$ so that the r.h.s. is positive.  Therefore, in this parameter range the quantities $\| \ket{\varepsilon(t_1, d)}_\Cl \|_2$ in~\cref{eq:up to here expalicit} decay exponentially in $d^{2(\epsilon_5-\epsilon_p)}$ for $q=N_g$  (for $q=1,2,\ldots, N_g-1$ they decay exponentially in $d^{2\epsilon_5}$ since they are given by~\cref{eq:var e v def}). Furthermore, from~\cref{eq:pot strengh} and~\cref{eq:bound of d sigma Ng decay term} we have that the other terms also decay exponentially in $d^x$ for some $x>0$ solely determined by $\bar\varepsilon$

We thus have that $	\big \|    		\ket{\rho(t_j +\tau_l| \tau_l)}_{\M_0\lo\W\Cl\Pu} - \ket{[t_j+\tau_l|\tau_l]}_{\M_0 \lo\W\Pu} \ket{\lsb t_j+\tau_l| \tau_l \rsb}_\Cl  \big\|_2$ is upper bounded by the r.h.s. of~\cref{eq:thm fixed memory 1app} up to the replacement of $\tilde d(\m_k)$ with  $\tilde d(\m_{l,k})$ since the gate being applied is distinct in the case at hand. Putting everything together we thus obtain for $j\in\nnz$ 
\begin{align}
	\big \|    	&	\ket{\rho\lb t_j+\tau_l | \tau_l\rb}_{\M_0\lo\W\Cl\Pu} - \ket{[t_j+\tau_l|\tau_l]}_{\M_0 \lo\W\Pu} \ket{\lsb t_j+\tau_l| \tau_l \rsb}_\Cl  \big\|_2 \leq \me^{-\frac{\varepsilon_b }{2T_0}t_j }   \left(\sum_{k=1}^{j}  \tilde d(\m_{l,k}) \, \delta(k-1)  \right) g( \bar\varepsilon) \,\textup{poly}(E_0) \, E_0^{-1/\sqrt{\bar\varepsilon}}
	\\& \leq 	\me^{-\frac{\varepsilon_b }{2T_0}t_j }   \left(\sum_{q=1}^{N_g}\sum_{k=0}^{\beta'(j)}  \tilde d(\m_{l,k N_g+q}) \, \delta(k N_g+q-1)  \right) g( \bar\varepsilon) \,\textup{poly}(E_0) \, E_0^{-1/\sqrt{\bar\varepsilon}}
\\& \leq 	\me^{-\frac{\varepsilon_b }{2T_0}t_j }   \left(\sum_{q=1}^{N_g}  \tilde d(\m_{l,q})\right) \left(\sum_{k=0}^{\beta'(j)}  \me^{-\gamma_0 N_g k} \right) g( \bar\varepsilon) \,\textup{poly}(E_0) \, E_0^{-1/\sqrt{\bar\varepsilon}}	,	\label{eq:bound on states in thm new 2}
\end{align}	
where in line~\labelcref{eq:bound on states in thm new 2}, we have taken into account the definition of $\delta(k)$ in~\cref{def:squre braket control}. 
\end{proof}

Finally we are in a stage to prove~\Cref{thm:heat dissipation}. Below we state a more explicit version of it, which is what we will prove. In particular, we state more explicitly the form which the functions $\SupPolyDecay(\cdot)$ take (One can readily check that they belong to said function class). Like before, $g(\bar\varepsilon)>0$ (for $\bar\varepsilon>0$) is a $P^\text{in}$-independent function of $\bar\varepsilon>0$, while $\textup{poly}(P^\text{in})$ is an $\bar\varepsilon$-independent polynomial in $P^\text{in}$. They are both $l$-independent.
\begingroup
\renewcommand{\thetheorem}{\ref{thm:heat dissipation}}  

\begin{theorem}[\text{[}More explicit than main-text version.\text{]} Nonequilibrium steady-state quantum frequential computers exist]
	\label{thm:heat dissipationapp}
	For all gate sets $\mathcal{U}_\mathcal{G}$, initial gate sequences $(\m_{l,k})_{l,k}$ with elements in $\mathcal{G}$, and initial logical states $\ket{0}_\lo\in\mathcal{P}(\mathcal{H}_\lo)$,  there exists $\ket{0}_\Cl$, $\{\ket{t_{j}+\tau_l|\tau_l}_\Cl \}_{j=1,2,\ldots,N_g; l\in\nnz }$, $N_g$, $\mathcal{L}_{\M_0\lo\W\Cl}$ parametrised by the power $P^\textup{in}>0$ and a dimensionless parameter $\bar\varepsilon$ (where elements $\ket{[t_j+\tau_l|\tau_l]}_\Cl$, satisfy~\cref{eq:idealised states condition thm3}), such that for all  $j=1,2,3, \ldots, N_g$; $l\in\nnz$ and fixed $\bar\varepsilon\in(0,1/6)$, the following large-$P^\textup{in}$ scaling hold simultaneously
	\item [1)]    Given that $l\in\nnz$ renewals occurred in the time interval $[0,\tau_l]$, the probability that the next renewal occurs in the interval $[\tau_l+T_0-t_1, \tau_l+T_0]$ is:
	\begin{align}
		\int_{T_0-t_1}^{T_0} \dd t\,  {P}(t+\tau_l,+1|\tau_l) = 1 - \varepsilon_r,\qquad 0<\varepsilon_r \leq  \left(\sum_{k=1}^{N_g}  \tilde d(\m_{l,k})\right) g( \bar\varepsilon) \,\textup{poly}({P^\textup{in}}) \, {P^\textup{in}}^{-1/(2\sqrt{\bar\varepsilon})}, \label{eq:thm3 prob dist explicitexplicit}
	\end{align}
	\item [2)]  The deviations in the state between renewals are small:  For $j=1,2,\ldots,N_g$, 
	\begin{align}
		T\Big( \rho_{\M_0\lo\W\Cl}(t_{j}+\tau_l|\tau_l) ,\,\, \rho_{\M_0 \lo\W\Pu}{[t_j+\tau_l|\tau_l]} \otimes \rho_\Cl{[ t_j+\tau_l| \tau_l ]} \Big)\leq    \left(\sum_{k=1}^{j}  \tilde d(\m_{l,k})\right) g( \bar\varepsilon) \,\textup{poly}({P^\textup{in}}) \, {P^\textup{in}}^{-1/\sqrt{\bar\varepsilon}},   \label{eq:dist condition thmor explicitexplicit}
	\end{align}
	\item [3)]    The gate frequency has the asymptomatically optimal scaling in terms of power:
	\begin{align}\label{eq:thm fixed memory 5 explicitexplicit}
		f= \frac{1}{T_0}\left( {T_0^2}\, {P^\textup{in}} \right)^{1-\bar\varepsilon}+ \delta f', \qquad |\delta f'| \leq \frac{1}{T_0} +\bo\left(\textup{poly}({P^\textup{in}}) {P^\textup{in}}^{-1/(2\sqrt{\bar\varepsilon})}\right) \text{ as } {P^\textup{in}}\to\infty.
	\end{align} 
\end{theorem}

\endgroup

\begin{proof}
We start with the proof of 1). This requires the calculation of the probability that the $(l+1)\thh$ renewal occurs in the interval $(\tau_l+t_{N_g-1}, \tau_l+t_{N_g})$, namely
	\begin{align}
	 \varepsilon_r=	1- \lim_{\epsilon\to 0} 		\int_{t_{N_g-1}+\epsilon}^{t_{N_g}-\epsilon} \dd t\,  P(t+\tau_l,+1|\tau_l)   \label{eq:thm3 prob dist 2}
	\end{align}
The probability, $P(t+\tau_l,+1|\tau_l)$, can be evaluate by conditioning on not renewing in the interval $(0,\tau_l)$ given that the $l^\text{th}$ renewal occurred at time $\tau_l$, followed by the $(l+1)^\text{th}$ renewal occurring at time $t$. This can be calculated by applying the renewal generator $\mathcal{D}_\Cl^\textup{re}(\cdot)$ to $\ket{\rho\lb t+\tau_l | \tau_l\rb}_{\M_0\lo\W\Cl\Pu}$ followed by taking the trace. Thus from~\cref{eq:def:renawal channel generic def,eq:rho conditional dybamics def} we find that $P(t+\tau_l,+1|\tau_l)= \tr[ 2 V_\Cl \proj{\rho\lb t +\tau_l| \tau_l \rb}_{\M_0\lo\W\Cl\Pu}]$. For $\epsilon>0$, in the interval $t\in(t_{N_g-1}, t_{N_g-1}+\epsilon)$,   $\ket{\rho\lb t+\tau_l | \tau_l \rb}_{\M_0\lo\W\Cl\Pu}= \me^{-\mi (t-t_{N_g-1}) G_{\M_0\lo\W\Cl}}\ket{\rho\lb t_{N_g-1}+\tau_l | \tau_l \rb}_{\M_0\lo\W\Cl\Pu}$ and in the interval $t\in (t_{N_g}-\epsilon, t_{N_g})$,   $\ket{\rho\lb t +\tau_l| \tau_l \rb}_{\M_0\lo\W\Cl\Pu}= \me^{-\mi (t-t_{N_g-1}) G_{\M_0\lo\W\Cl}}\ket{\rho\lb t_{N_g-1} +\tau_l| \tau_l \rb}_{\M_0\lo\W\Cl\Pu}$. Therefore, we can solve the integral analytically to find
\begin{align}
		 \varepsilon_r=\,&1-		\lim_{\epsilon\to 0^+} 		\int_{t_{N_g-1}+\epsilon}^{t_{N_g}-\epsilon} \dd t  P(t+\tau_l,+1|\tau_l)\\
	=\,&1-  	\lim_{\epsilon\to 0^+} 	\Bigg(   \tr\big[ \proj{\rho\lb t_{N_g-1}+\varepsilon+\tau_l | \tau_l \rb}_{\M_0\lo\W\Cl\Pu}\big] \\
	&-\tr\big[ \proj{\rho\lb t_{N_g}-\varepsilon +\tau_l| \tau_l \rb}_{\M_0\lo\W\Cl\Pu}\big] \Bigg)\\
	=\,& 1  +   \| \ket{\rho\lb t_{N_g} +\tau_l| \tau_l \rb}_{\M_0\lo\W\Cl\Pu}\|^2_2- \|\ket{\rho\lb t_{N_g-1}+\tau_l | \tau_l \rb}_{\M_0\lo\W\Cl\Pu}\|^2_2.
\end{align}
	We can now use~\Cref{lem:last lemm} to exchange $\ket{\rho\lb t_{N_g} +\tau_l| \tau_l \rb}_{\M_0\lo\W\Cl\Pu}$ and $\ket{\rho\lb t_{N_g-1} +\tau_l| \tau_l \rb}_{\M_0\lo\W\Cl\Pu}$ for \\ $\ket{[t_{N_g}+\tau_l|\tau_l]}_{\M_0 \lo\W\Pu} \ket{\lsb t_{N_g}+\tau_l| \tau_l \rsb}_\Cl$  and $\ket{[t_{N_g-1}+\tau_l|\tau_l]}_{\M_0 \lo\W\Pu} \ket{\lsb t_{N_g-1}+\tau_l| \tau_l \rsb}_\Cl $ respectively up to the small errors dictated by the r.h.s. of~\cref{eq:bound on states in thm new}. Recall that $\|  \ket{[t_{N_g}+\tau_l|\tau_l]}_{\M_0 \lo\W\Pu} \ket{\lsb t_{N_g}+\tau_l| \tau_l \rsb}_\Cl \|_2$ is exponentially small in $d^{\bar\varepsilon^2}$,  while $\| \ket{[t_{N_g-1}+\tau_l|\tau_l]}_{\M_0 \lo\W\Pu} \ket{\lsb t_{N_g-1}+\tau_l| \tau_l \rsb}_\Cl \|_2^2 = \me^{-\varepsilon_b}=1-\varepsilon_b+\bo(\varepsilon_b)^2$. This gives us
		\begin{align}
		\int_{T_0-t_1}^{T_0} \dd t\, {P}(t+\tau_l,+1|\tau_l) = 1 - \varepsilon_r,\qquad 0<\varepsilon_r \leq \varepsilon_b+\bo(\varepsilon_b)^2 + \left(\sum_{k=1}^{N_g}  \tilde d(\m_{l,k})\right) g( \bar\varepsilon) \,\textup{poly}({E_0}) \, {E_0}^{-1/\sqrt{\bar\varepsilon}}. \label{eq:thm3 prob dist app}
	\end{align}
We now set the free parameter $\varepsilon_b$. We choose
\begin{align}
	\varepsilon_b= \left(\sum_{q=1}^{N_g}  \tilde d(\m_{l,q})\right)      E_0^{-1/(2\sqrt{\bar\varepsilon})}. \label{eq:lvarepsion b choice}
\end{align}
[The reason for this choice will become apparent later when bounding $|\delta E_1^\text{no re}|$. In particular, this choice of scaling will be relevant in line~\labelcref{eq:line:communte14.1}].
Thus plugging in to~\cref{eq:thm3 prob dist app} we find
		\begin{align}
	\int_{T_0-t_1}^{T_0} \dd t\, {P}(t+\tau_l,+1|\tau_l) = 1 - \varepsilon_r,\qquad 0<\varepsilon_r \leq  \left(\sum_{k=1}^{N_g}  \tilde d(\m_{l,k})\right) g( \bar\varepsilon) \,\textup{poly}({E_0}) \, {E_0}^{-1/(2\sqrt{\bar\varepsilon})}. \label{eq:thm3 prob dist app 2}
\end{align}

We now calculate $\braket{E^\text{re}}$ as a function of $E_0$. Since from~\cref{eq:def: enery re} we have $P^\text{in}= \braket{E^\text{re}}/T_0$, this in conjunction with~\cref{eq:thm3 prob dist app 2} will allow us to prove item 1) in~\Cref{thm:heat dissipation}. 
	\begin{align}
		 \braket{E^\text{re}}&= \int_{0}^{\infty} \dd s\,  \tr[H_{\M_0\lo\W\Cl}'\,	\mathcal{D}_{\M_0\lo\W\Cl}^\text{re}(\rho_{\M_0\lo\W\Cl}(s+\tau_l|\tau_l) )] 
		\\& =2  \int_{0}^{\infty} \dd s \, \tr[H_{\Cl}\,	\proj{0}_\Cl \otimes\sum_{j=1}^{d-1} v_j  {}_\Cl\!\bra{\theta_j}\rho_{\M_0\lo\W\Cl}(s+\tau_l|\tau_l) \ket{\theta_j}_\Cl ]  
		\\& + 2 \sum_{l=1}^{N_g-1}  \int_{0}^{\infty} \dd s \, \tr\left[\left(I_{\M_0\lo\W}^{(l)} \otimes I_\Cl^{(l)}\right)	\left(\proj{0}_\Cl \otimes\sum_{j=1}^{d-1} v_j  {}_\Cl\!\bra{\theta_j}\rho_{\M_0\lo\W\Cl}(s+\tau_l|\tau_l) \ket{\theta_j}_\Cl\right) \right] 
		\\& =  {}_\Cl\!\braket{0|H_\Cl|0}_\Cl   \int_{0}^{\infty} \dd s \, P(t+\tau_l,+1|\tau_l)
		\\& + 2 \sum_{l=1}^{N_g-1} \tr[ I_\Cl^{(l)}\proj{0}_\Cl ] \int_{0}^{\infty} \dd s \, \tr[I_{\M_0\lo\W}^{(l)}\otimes V_\Cl\, \rho_{\M_0\lo\W\Cl}(s+\tau_l|\tau_l)  ] 
		\\& =  {}_\Cl\!\braket{0|H_\Cl|0}_\Cl  
	\\& - \sum_{l=1}^{N_g-1} \tr[ I_\Cl^{(l)}\proj{0}_\Cl ] \int_{0}^{\infty} \dd s \, \tr[I_{\M_0\lo\W}^{(l)} \frac{d}{ds}\rho_{\M_0\lo\W\Cl}(s+\tau_l|\tau_l)  ] 
			\\& =  {}_\Cl\!\braket{0|H_\Cl|0}_\Cl + \sum_{l=1}^{N_g-1} \tr[ I_\Cl^{(l)}\proj{0}_\Cl ] \, \tr[I_{\M_0\lo\W}^{(l)} \,\rho_{\M_0\lo\W\Cl}(\tau_l|\tau_l)  ]   
			\\& =  \tr[H'_{\M_0\lo\W\Cl}\,  \rho_{\M_0\lo\W\Cl}(\tau_l|\tau_l)  ]   \label{eq:initi condition used}
			\\& =  E_0 + 	\delta E^\text{re}_1 \label{eq:initi condition used 2}
	\end{align}
	Where in~\cref{eq:initi condition used}, we have used~\cref{eq:rho intermediate product} and where
\begin{align}
	\delta E^\text{re}_1&:=   \tr[H'_{\M_0\lo\W\Cl}\,	\rho_{\M_0\lo\W\Cl}(\tau_l|\tau_l ) ] - \tr[H_{\M_0\lo\Cl}\,	\rho_{\M_0\lo\Cl}^0 ]  \label{eq:initi condition used 00}
	\\& = \sum_{q=1}^{N_g -1}  \tr\big[ I_{\M_0\lo\W\Cl}^{(q)} \proj{\m_{l,1}}_{\M_{0,1} \Pu_1} \proj{(\tau_l|\tau_l)}_{\bar\M_{0,1} \W\lo \bar\Pu_1}\big]    \tr\big[ I_{\Cl}^{(q)} \proj{\Psi(0)}_\Cl\big]  
	\\& - \sum_{q=1}^{N_g}  \tr\big[ I_{\M_0\lo\Cl}^{(q)} \proj{\m_{l,1}}_{\M_{0,1} } \rho^0_{\M_0\lo}\big]    \tr\big[ I_{\Cl}^{(q)} \proj{\Psi(0)}_\Cl\big] ,  \label{eq:line almost energy 1111}
	\\ |\delta E^\text{re}_1| &\leq    g(\bar\varepsilon)\text{poly}(E_0)  E_0^{-1/\sqrt{\bar\varepsilon}} . \label{eq:line:communte52}
\end{align}	
In inequalities~\labelcref{eq:line:communte52,eq:line almost energy 1111}, we have recalled definition~\labelcref{eq:def of power one clock only case} and noted that the interaction terms have a vanishing contribution to the energy as proven in the proof of~\Cref{thm:contrl with two clocks} from~\cref{eq:sum interaction temts small} onwards. 

Thus 
\begin{align}
	P^\textup{in}= \braket{E^\text{re}}/T_0=  E_0  + 	\delta E^\text{re}_1/T_0. \label{Eq: lower bound P''}
\end{align}
Therefore, there exists $x_0(\bar\varepsilon)>0$, $P_0(\bar\varepsilon)>0$ such that $P^\textup{in}\geq x_0(\bar\varepsilon)  E_0/ T_0$ for all $E_0\geq E_{00}(\bar\varepsilon)$, $\bar\varepsilon\in(0,1/6)$.	This concludes the proof of item 1),  [namely~\cref{eq:thm3 prob dist explicitexplicit} in~\Cref{thm:heat dissipation}].

For item 2) [namely~\cref{eq:dist condition thmor explicitexplicit} in~\Cref{thm:heat dissipation}], start by noting that~\Cref{lem:last lemm} can easily be repeated for the normalised version of the states in~\cref{eq:bound on states in thm new}. Then recalling that the two-norm distance between two normalised states is an upper bound to the trance distance between said states (recall~\Cref{lem:trace dist}), we obtain  item 2) up to a replacement of $P^\textup{in}$ with $E_0$ and up to a trace over the purifying system $\Pu$. Using~\cref{eq:initi condition used 2}, we can then convert from $E_0$ to $P^\textup{in}$, and we can trace out the purification by noting that the trace distance is contractive under partial trace.
	
We now prove item 3) [namely~\cref{eq:thm fixed memory 5 explicitexplicit} in~\Cref{thm:heat dissipation}].  Since we are using the same parametrizations as in~\Cref{thm:comptuer with fixed memory}, for case b), we have that~\cref{eq:thm fixed memory 3 app} holds, namely
\begin{align}\label{eq:thm fixed memory 3 thm3}
	f= \frac{1}{T_0}\left( \frac{T_0}{2\pi \tilde n_0} E_0 \right)^{1-\bar\varepsilon}+ \delta f', \qquad |\delta f'| \leq \frac{1}{T_0} +\bo\left(\textup{poly}(E_0) E_0^{-1/\sqrt{\bar\varepsilon}}\right) \text{ as } E_0\to\infty.
\end{align}
Recall that $\tilde n_0$ is a free parameter of the model in the interval $(0,1)$ [see text below~\cref{eq:bound of d sigma Ng decay term}]. We can therefore choose $\tilde{n}_0=1/(2\pi \kappa_1)$. Finally, to convert from $E_0$ to $P^\textup{in}$, we substitute for $E_0$ using using~\cref{Eq: lower bound P'',eq:line:communte52} followed by Tailor expanding about the point $\delta E_1^\textup{re}=0$ analogously to the expansion in $\delta d$ in~\cref{eq:frequnsy case}. This concludes the derivation of item 3) \, [namely~\cref{eq:thm fixed memory 5 explicitexplicit} in~\Cref{thm:heat dissipation}].
\end{proof}

The following lemma demonstrates that, up to vanishing errors, that the amount of dissipated power $P^\textup{diss}$,  is upper bounded  by the consumed power $P^\textup{in}$, up to a vanishing error term.
\begin{lemma}\label{lem:upper bound on P diss}
	For the quantum frequential computer of~\Cref{thm:heat dissipation}, the dissipation is upper bounded by 
	\begin{align}
	   P^\textup{diss} \leq P^\textup{in}+ 		\delta P^{(+)}_{\bar\varepsilon}(P^\textup{in}),\label{eq:P diss bounds}
	\end{align}
where, for all fixed $\bar\varepsilon\in(0,1/6)$,  we have $	\delta P^{(+)}_{\bar\varepsilon}(P^\textup{in})\to 0$  as $P^\textup{in}\to \infty$.
\end{lemma}
\begin{proof}
 Since $ P^\textup{diss}:= 		
 	\braket{E^\text{no re}} /T_0$, we need to upper bound $	
 		\braket{E^\text{no re}}$. We will establish said bound in terms of $E_0$ and then convert this to abound in terms of $P^\textup{in}$.
	
	Observe that
	\begin{align}
		\frac{d}{dt} \tr[H'_{\M_0\lo\W\Cl} \,\rho_{\M_0\lo\W\Cl}\lb t+\tau_l |\tau_l\rb]= \, \tr[H'_{\M_0\lo\W\Cl}\mathcal{D}^\textup{no re}_{\Cl} \left(\rho_{\M_0\lo\W\Cl}\lb t+\tau_l |\tau_l \rb\right)],
	\end{align}
	in the intervals where $\tr[H'_{\M_0\lo\W\Cl}\, \rho_{\M_0\lo\W\Cl}\lb t+\tau_l |\tau_l\rb]$ is differentiable. Therefore, defining $N(t)$ as the largest integer $N$ s.t. $t_{N+1}<t$
	Recalling~\cref{eq:S cycle new} we have 
	\begin{align}
			\braket{E^\text{no re}} \label{eq:S cycle 2}
		=& - \int_{0}^{\infty} \dd t\, P(t+\tau_l ,+1|\tau_l)  \int_0^t \tr[{H}'_{\M_0\lo\W\Cl}\,	\mathcal{D}^\textup{no re}_{\Cl}\left(\rho_{\M_0\lo\W\Cl}\lb s+\tau_l |\tau_l\rb \right) ] \dd s.
		\\=& -  \int_{0}^{\infty} \dd t\, P(t+\tau_l ,+1|\tau_l)  \times
		\\& \Bigg[\left(\sum_{j=0}^{N(t)} \lim_{\epsilon\to 0^+}   \int_{t_j+\epsilon}^{t_{j+1}-\epsilon} \!\!\!\! \tr[H'_{\M_0\lo\W\Cl}\,	\mathcal{D}^\textup{no re}_{\M_0\lo\W\Cl}\left(\rho_{\M_0\lo\W\Cl}\lb s+\tau_l |\tau_l \rb\right) ] \dd s \right) 
		\\&+   \lim_{\epsilon\to 0^+} \int_{t_{N+1}+\epsilon}^{t} \!\!\!\!\! \tr[H'_{\M_0\lo\W\Cl}\,	\mathcal{D}^\textup{no re}_{\M_0\lo\W\Cl}\left(\rho_{\M_0\lo\W\Cl}\lb s+\tau_l |\tau_l\rb \right) \dd s\Bigg]
		\\=& -  \int_{0}^{\infty} \dd t\, P(t+\tau_l ,+1|\tau_l) 
		\\& \left(\sum_{j=0}^{N(t)} \lim_{\epsilon\to 0^+}   \left[    \tr[H'_{\M_0\lo\W\Cl} \, \rho_{\M_0\lo\W\Cl}\lb t_{j+1}-\epsilon+\tau_l |\tau_l \rb ]-  \tr[H'_{\M_0\lo\W\Cl}\, \rho_{\M_0\lo\W\Cl}\lb t_j+\epsilon+\tau_l |\tau_l \rb ] \right]  \right) 
		\\&\quad+   \lim_{\epsilon\to 0^+}  \bigg(   \tr[H'_{\M_0\lo\W\Cl}\,	\rho_{\M_0\lo\W\Cl}\lb t+\tau_l |\tau_l \rb ]- \tr[H'_{\M_0\lo\W\Cl}\,	\rho_{\M_0\lo\W\Cl}\lb t_{N+1}+\epsilon +\tau_l |\tau_l \rb ] \bigg)
		\\=&  \int_{0}^{\infty} \dd t\, P(t+\tau_l ,+1|\tau_l)  \tr[H'_{\M_0\lo\W\Cl} \, \rho_{\M_0\lo\W\Cl}\lb \tau_l |\tau_l \rb ] +\delta E^\text{no re}_1 \label{eq:line:upper bound line negleact-1}
		\\&\quad-  \int_{0}^{\infty} \dd t\, P(t+\tau_l ,+1|\tau_l) \, \tr[H'_{\M_0\lo\W\Cl}\,	\rho_{\M_0\lo\W\Cl}\lb t+\tau_l |\tau_l \rb ] \label{eq:line:upper bound line negleact}
		\\=&  	E_0+\delta E^\text{re}_1 +\delta E^\text{no re}_1 -  \int_{0}^{\infty} \dd t\, P(t+\tau_l ,+1|\tau_l) \, \tr[H'_{\M_0\lo\W\Cl}\,	\rho_{\M_0\lo\W\Cl}\lb t+\tau_l |\tau_l \rb ], \label{eq:line:upper bound line negleact lasr equality}
		\\\leq&   	E_0     +\delta E^\text{re}_1 +\delta E^\text{no re}_1. \label{eq:upper bound on P''}
	\end{align}	
	
In line~\labelcref{eq:upper bound on P''} we have lower  bounded the integral by zero since $\tr[H'_{\M_0\lo\W\Cl}\,	\rho_{\M_0\lo\W\Cl}\lb t+\tau_l |\tau_l \rb ]\geq 0$ because $H'_{\M_0\lo\W\Cl}\geq 0$ by definition and thus the integral is non-negative. The quantity $\delta E^\text{no re}_1$ appearing in line~\labelcref{eq:line:upper bound line negleact-1}, is defined as the difference between this line and the previous line. As such, we have that $|\delta E^\text{no re}_1|$ satisfies the bound
	\begin{align}
		| \delta E^\text{no re}_1| &\leq    \int_{0}^{\infty} \dd t\, P(t+\tau_l ,+1|\tau_l) \times
		\\&  \sum_{j=0}^{\infty} \Bigg|\lim_{\epsilon\to 0^+}   \Big[    \tr[H'_{\M_0\lo\W\Cl} \, \rho_{\M_0\lo\W\Cl}\lb t_{j+1}-\epsilon+\tau_l |\tau_l \rb ]-  \tr[H'_{\M_0\lo\W\Cl}\, \rho_{\M_0\lo\W\Cl}\lb t_{j+1}+\epsilon+\tau_l |\tau_l\rb ] \Big] 
		\\& -   \lim_{\epsilon\to 0^+} \tr[H'_{\M_0\lo\W\Cl}\,	\rho_{\M_0\lo\W\Cl}\lb t_{N+1}+\epsilon+\tau_l  |\tau_l\rb ] \Bigg|
		\\& \quad+ \sum_{j=0}^{\infty} \Biggg|       \tr\Big[H'_{\M_0\lo\W\Cl}  \Big(\proj{\lb t_{j+1}+\tau_l |\tau_l \rb}_{\M_0\lo\W\Cl\Pu}
		\\& \quad- U^{(l,j+2)}_{\M_{0,j+2}\Pu_{j+2}\M_{0,j+1}\Pu_{j+1}}\proj{\lb t_{j+1}+\tau_l |\tau_l \rb}_{\M_0\lo\W\Cl\Pu}  U^{(l,j+2)\dag}_{\M_{0,j+2}\Pu_{j+2}\M_{0,j+1}\Pu_{j+1}}\Big) \Big]  \Biggg| 
		\\&\leq    \sum_{j=0}^{\infty} \Biggg|       \tr\Big[    \left(I^{(j+1)}_{\M_0\lo\W}\otimes I^{(j+1)}_\Cl\bar\delta_{j+1,N_g}+ I^{(j+2)}_{\M_0\lo\W}\otimes I^{(j+2)}_\Cl\bar\delta_{j+2,N_g} \right)  \Big(\proj{\lb t_{j+1}+\tau_l |\tau_l\rb}_{\M_0\lo\W\Cl\Pu}  \label{eq:line:communte}
		\\&\quad - U^{(l,j+2)}_{\M_{0,j+2}\Pu_{j+2}\M_{0,j+1}\Pu_{j+1}}\proj{\lb t_{j+1}+\tau_l |\tau_l \rb}_{\M_0\lo\W\Cl\Pu}  U^{(l,j+2)\dag}_{\M_{0,j+2}\Pu_{j+2}\M_{0,j+1}\Pu_{j+1}}\Big) \Big]  \Biggg|   \label{eq:line:communte2}
		\\&\leq    \sum_{j=0}^{\infty}   \Big\|       \left(I^{(j+1)}_{\M_0\lo\W}\otimes I^{(j+1)}_\Cl \bar\delta_{j+1,N_g}  \right)^{1/2}\ket{\lb t_{j+1}+\tau_l |\tau_l\rb}_{\M_0\lo\W\Cl\Pu} \bigg \|_2  
		\\& + \left\|   \left(I^{(j+2)}_{\M_0\lo\W}\otimes I^{(j+2)}_\Cl \bar\delta_{j+2,N_g} \right)^{1/2}\ket{\lb t_{j+1}+\tau_l |\tau_l\rb}_{\M_0\lo\W\Cl\Pu} \right \|_2   \label{eq:line:communte4}
		\\& \quad+  \left\|       \left(I^{(j+1)}_{\M_0\lo\W}\otimes I^{(j+1)}_\Cl \bar\delta_{j+1,N_g} \right)^{1/2}  U^{(l,j+2)}_{\M_{0,j+2}\Pu_{j+2}\M_{0,j+1}\Pu_{j+1}} \ket{\lb t_{j+1}+\tau_l |\tau_l\rb}_{\M_0\lo\W\Cl\Pu} \right\|_2  
		\\&\quad + \left\|   \left(I^{(j+2)}_{\M_0\lo\W}\otimes I^{(j+2)}_\Cl  \bar\delta_{j+1,N_g} \right)^{1/2}  U^{(l,j+2)}_{\M_{0,j+2}\Pu_{j+2}\M_{0,j+1}\Pu_{j+1}}  \ket{\lb t_{j+1}+\tau_l |\tau_l\rb}_{\M_0\lo\W\Cl\Pu} \right \|_2   \label{eq:line:communte5}
		\\&\leq    \sum_{j=0}^{\infty} \, \bar\delta_{j+1,N_g} \left\|       \left(I^{(j+1)}_{\M_0\lo\W}\otimes I^{(j+1)}_\Cl\right)^{1/2} \ket{[t_{j+1}+\tau_l |\tau_l]}_{\M_0 \lo\W\Pu} \ket{\lsb t_{j+1}+\tau_l | \tau_l \rsb}_\Cl \right\|_2  
		\\& \quad +\bar\delta_{j+2,N_g}\left\|   \left(I^{(j+2)}_{\M_0\lo\W}\otimes I^{(j+2)}_\Cl \right)^{1/2} \ket{[t_{j+1}+\tau_l |\tau_l]}_{\M_0 \lo\W\Pu} \ket{\lsb t_{j+1}+\tau_l | \tau_l \rsb}_\Cl \right \|_2   \label{eq:line:communte6}
		\\& \quad+ \bar\delta_{j+1,N_g} \left\|       \left(I^{(j+1)}_{\M_0\lo\W}\otimes I^{(j+1)}_\Cl\right)^{1/2}  U^{(l,j+2)}_{\M_{0,j+2}\Pu_{j+2}\M_{0,j+1}\Pu_{j+1}} \ket{[t_{j+1}+\tau_l |\tau_l]}_{\M_0 \lo\W\Pu} \ket{\lsb t_{j+1}+\tau_l | \tau_l \rsb}_\Cl \right\|_2  
		\\&\quad + \bar\delta_{j+2,N_g} \left\|   \left(I^{(j+2)}_{\M_0\lo\W}\otimes I^{(j+2)}_\Cl \right)^{1/2}  U^{(l,j+2)}_{\M_{0,j+2}\Pu_{j+2}\M_{0,j+1}\Pu_{j+1}} \ket{[t_{j+1}+\tau_l |\tau_l]}_{\M_0 \lo\W\Pu} \ket{\lsb t_{j+1}+\tau_l | \tau_l \rsb}_\Cl  \right\|_2   \label{eq:line:communte7}
		\\&\quad+  \bar\delta_{j+1,N_g}  \left\|       \left(I^{(j+1)}_{\M_0\lo\W}\otimes I^{(j+1)}_\Cl\right)^{1/2}\ket{\epsilon(j+1)} \right\|_2   + \bar\delta_{j+2,N_g}\left\|   \left(I^{(j+2)}_{\M_0\lo\W}\otimes I^{(j+2)}_\Cl \right)^{1/2}\ket{\epsilon(j+1)}  \right\|_2   \label{eq:line:communte6 1}
		\\& \quad+ \bar\delta_{j+1,N_g} \left\|       \left(I^{(j+1)}_{\M_0\lo\W}\otimes I^{(j+1)}_\Cl\right)^{1/2}  U^{(l,j+2)}_{\M_{0,j+2}\Pu_{j+2}\M_{0,j+1}\Pu_{j+1}} \ket{\epsilon(j+1)} \right\|_2  
		\\&\quad +\bar\delta_{j+2,N_g} \left\|   \left(I^{(j+2)}_{\M_0\lo\W}\otimes I^{(j+2)}_\Cl \right)^{1/2}  U^{(l,j+2)}_{\M_{0,j+2}\Pu_{j+2}\M_{0,j+1}\Pu_{j+1}}  \ket{\epsilon(j+1)}_{\M_0\lo\W\Cl\Pu}  \right\|_2   \label{eq:line:communte7 1}
		\\ &\leq \sum_{j=0}^\infty \, 2 \sqrt{4\pi} \me^{-\varepsilon_b t_{j+1}/(2T_0)} \delta_{(j+1)} \left( \bar\delta_{j+1,N_g}\left\|   \left({I_\Cl^{(j+1)}}\right)^{1/2}  \ket{t_{j+1}}_\Cl \right\|_2  + \bar\delta_{j+2,N_g} \left\|   \left({I_\Cl^{(j+2)}}\right)^{1/2}  \ket{t_{j+1}}_\Cl \right\|_2  \,\right) \label{eq:line:communte8.1}
		\\ &\quad +  2 \sqrt{4\pi}  \left(  \bar\delta_{j+1,N_g}  \left\|   \left({I_\Cl^{(j+1)}}\right)^{1/2} \right\|_2  + \bar\delta_{j+2,N_g}  \left\|   \left({I_\Cl^{(j+2)}}\right)^{1/2}  \right\|_2  \, \right)   \big\|   \ket{\epsilon(j+1)} \big\|_2   \label{eq:line:communte8.2}
		\\&\leq   g(\bar\epsilon) \text{poly}(E_0) \me^{-1/\sqrt{\bar\epsilon}} \left( \sum_{j=0}^\infty \,  \me^{-\varepsilon_b t_{j+1}/(2T_0)} \delta_{(j+1)} \right) \label{eq:line:communte9.1}
		\\ &\quad +   \text{poly}(E_0)   \left(\sum_{j=0}^\infty  \big\|   \ket{\epsilon(j+1)}_\Cl \big\|_2 \right)  \label{eq:line:communte9.2}
		\\&\leq    g(\bar\epsilon) \text{poly}(E_0) \me^{-1/\sqrt{\bar\epsilon}} \left( N_g \sum_{k=0}^\infty \,  \me^{-\varepsilon_b k/2} \me^{-\gamma_0 N_g k} \right) \label{eq:line:communte10.1}
		\\ &\quad +   g(\bar\varepsilon)\text{poly}(E_0)  E_0^{-1/\sqrt{\bar\varepsilon}}    	 \left(\sum_{q=1}^{N_g}  \tilde d(\m_{l,q})\right)    \left(  \sum_{j=1}^\infty \me^{-\frac{\varepsilon_b }{2T_0}t_j }    \frac{1-\me^{-\gamma_0 N_g (1+\beta'(j))}}{1-\me^{-\gamma_0 N_g}} \right)  \label{eq:line:communte10.2}
		\\&\leq     g(\bar\epsilon) \text{poly}(E_0) \,\me^{-1/\sqrt{\bar\epsilon}} \, \frac{N_g}{1-\me^{-(\gamma_0 N_g+\varepsilon_b/2)}}  \label{eq:line:communte11.1}
		\\ &\quad +   g(\bar\varepsilon)\text{poly}(E_0)  E_0^{-1/\sqrt{\bar\varepsilon}}    	 \left(\sum_{q=1}^{N_g}  \tilde d(\m_{l,q})\right)    \left( (N_g+1)  \sum_{j=0}^\infty \me^{-\frac{\varepsilon_b j }{2} }    \frac{1-\me^{-\gamma_0 N_g (1+j)}}{1-\me^{-\gamma_0 N_g}} \right)  \label{eq:line:communte11.2}
		\\&\leq    g(\bar\epsilon) \text{poly}(E_0) \,\me^{-1/\sqrt{\bar\epsilon}}  \label{eq:line:communte12.1}
		\\ &\quad +   g(\bar\varepsilon)\text{poly}(E_0)  E_0^{-1/\sqrt{\bar\varepsilon}}    	 \left(\sum_{q=1}^{N_g}  \tilde d(\m_{l,q})\right)    \left(   \frac{1}{1-\me^{-\varepsilon_b/2}} \right).  \label{eq:line:communte12.2}
		\\ &=   g(\bar\varepsilon)\text{poly}(E_0)  E_0^{-1/\sqrt{\bar\varepsilon}}    	 \left(\sum_{q=1}^{N_g}  \tilde d(\m_{l,q})\right)    \left( \frac{2}{\varepsilon_b+\bo({\varepsilon_b})^2} \right). \label{eq:line:communte13.1}
		\\&  \leq  g(\bar\varepsilon)\text{poly}(E_0)  E_0^{-1/(2\sqrt{\bar\varepsilon})}  \label{eq:line:communte14.1}
	\end{align} 
	In lines~\labelcref{eq:line:communte,eq:line:communte2} we have used the fact that $I^{(j)}_{\M_0\lo\W}$ has support on ${\M_{0,j}\lo\W}$ only, and the unitary invariance of the trace. We have also  defined $\bar\delta_{j+1,N_g}$ as
	\begin{align}
		\bar\delta_{q,r}= 
		\begin{cases}
			0 &\mbox{ if } q=r\\
			1 &\mbox{ otherwise.} 
		\end{cases}
	\end{align}
	In line~\labelcref{eq:line:communte6 1} we have  defined $\ket{\epsilon(j+1)}= \ket{\lb t_{j+1}+\tau_l |\tau_l\rb}_{\M_0\lo\W\Cl\Pu}- \ket{[t_{j+1}+\tau_l |\tau_l]}_{\M_0 \lo\W\Pu} \ket{\lsb t_{j+1}+\tau_l | \tau_l \rsb}_\Cl$. In line~\labelcref{eq:line:communte8.1} we have used definitions~\cref{eq:int meme sys swit,eq:interaction terms with subscripts,def:squre braket control} together with the definition of the two norm.
	
	In line~\labelcref{eq:line:communte9.1} we have used $\Big( \bar\delta_{j+1,N_g}\Big\|   \left({I_\Cl^{(j+1)}}\right)^{1/2}  \ket{t_{j+1}}_\Cl \Big\|_2  + \bar\delta_{j+2,N_g} \Big\|   \left({I_\Cl^{(j+2)}}\right)^{1/2}  \ket{t_{j+1}}_\Cl \Big\|_2  \,\Big)$ \\ $\leq  g(\bar\epsilon) \text{poly}(E_0) \me^{-1/\sqrt{\bar\epsilon}}$ (where the r.h.s. is $j$-independent) which follows from the proof on~\Cref{lem:1st secondrary lem for unitary bound} and the parametrization of $E_0$ in terms of $d$ from the proof of~\Cref{thm:comptuer with fixed memory}. Similarly, in line~\labelcref{eq:line:communte9.2} we have used that  $\Big(  \bar\delta_{j+1,N_g}  \Big\|   \left({I_\Cl^{(j+1)}}\right)^{1/2} \Big\|_2  + \bar\delta_{j+2,N_g}  \Big\|   \left({I_\Cl^{(j+2)}}\right)^{1/2}  \Big\|_2  \, \Big) \leq   \text{poly} (E_0)$ (where the r.h.s. is $j$-independent) which follows from~\cref{eq:def I in terms of bar V0,eq:def I in terms of bar V0} and the parametrization of $E_0$ in terms of $d$ in the proof of~\Cref{thm:comptuer with fixed memory}.  In line~\labelcref{eq:line:communte10.2} we have applied~\Cref{lem:last lemm}. In lines~\labelcref{eq:line:communte12.1,eq:line:communte12.2} we have used that $N_g=\text{poly}(E_0)$. In line~\labelcref{eq:line:communte14.1} we used definition~\labelcref{eq:lvarepsion b choice}.
	
	Thus, recalling~\cref{Eq: lower bound P''}, from~\cref{eq:line:communte14.1} we achieve the upper bound on $ P^\textup{diss}$ in~\cref{eq:P diss bounds}.
\end{proof}

\section{Upper bounds on gate frequency for self-oscillators}\label{Upper bounds on gate frequency for self-oscillators}

In the following, we assume the description laid-out in~\Cref{sec:generic self-oscillator dynamics}. Thus, as in said section, here  $H_{\Sy\Cl}$ and $H_\Cl$ are arbitrary Hamiltonians on finite dimensional Hilbert spaces  $\mathcal{H}_{\Sy\Cl}$ and $\mathcal{H}_\Cl$ respectively. (Finite dimensions are only assumed for simplicity, to avoid the usual complications which can arise in infinite dimensions. Since the actual dimensions do not appear in any of the results, generalizations to the infinitive-dimensional case should be possible via conventional techniques.)  In this section, we also assume that $\Sy$ contains a logical system $\lo$ as a subsystem, and denote the remaining system after removing $\lo$ from $\Sy$ by $\Sym$. (Note the consistency with the special case assumed in the context of~\Cref{thm:heat dissipation}, where $\Sy=\M_0\lo\W$.)

\subsection{Quantities pertaining to the quality of computer's self-oscillator and some lemmas about them}\label{sec:Quantities pertaining to the quality of computer}
Since $V_\Cl$ is positive semi-definite, we can write it in the form $V_\Cl=\sum_j v_j \proj{v_j}$, with $v_j\geq 0$ and $\{\ket{v_j}\}_j$ an orthonormal basis. This in turn allows one to write $	P(t+\tau_l,+1|\tau_{l})$ in terms of a weighted 1-norm where the weights are time dependent:
\begin{align}
	P(t+\tau_l,+1|\tau_{l})&= 2 \,\|  \vec V_\Cl \|_{w,1}(t,l), \quad \vec V_\Cl:=(v_1, v_2, v_3,\ldots), \label{def:weighted q norm0}
	\\   \|   \vec h \|_{w,q}(t,l)&:= {\left( \sum_j \left|h_j \right|^q p_j(t,l) \!\right)\!}^{1/q}, \quad p_j(t,l):= \braket{ v_j | \rho_{\Sy\Cl}\lb t+\tau_l|\tau_{l}\rb | v_j }. \label{def:weighted q norm}
\end{align}
As we have discussed in the main text, self-oscillators have a notion of an isentropic time interval. We will soon use these above identities to define when the isentropic interval of the $l\thh$ cycle ends. The isentropic time interval is the time interval in each cycle in which the self-oscillator is not being corrected and can implement gates on the logical space. The influence of the restoring force  on the oscillator does not dynamically turned off and on abruptly, but gradually ramped up and down over the cycle. As such, the time at which the isentropic time interval of the $l\thh$ cycle ends is only approximately defined. We define it here, by imposing a hard cut-off, beyond which the effects of the restoring force on the dynamics are considered to be not-neglectable and before which they are neglectable. With this in mind,  we denote the isentropic time interval of the $l\thh$ cycle as
\begin{align}
	[\tau_l,t_{\max,l}+\tau_l).
\end{align}
The value of $t_{\max,l}$ depends on 2 factors: 1) the quality of the self-oscillators ability to only apply the restoring force towards the end of the cycle while still being strong enough to correct for errors, and  2) our willingness to tolerate small deviations from unitary dynamics over $\mathcal{H}_{\Sy\Cl}$ during the time-interval $[\tau_l,t_{\max,l}+\tau_l)$.  We now define these two quantities followed by quantifying their values.

The $1\stt$ can be characterised by integrating the weighted 2-norm of $\vec V_\Cl$, over interval $[0,\tau]$ for some $\tau\in\rr$,  namely
\begin{align}
	\varepsilon^\textup{quality}_l(\tau):=  \int_0^{\tau}\dd t \,  \|  \vec V_\Cl \|_{w,2}(t,l).      \label{eq:quality factor}
\end{align}
Note that if $V_\Cl$ and $\rho_{\Sy\Cl}(t+\tau_l|\tau_l)$ are orthogonal for $t$ in some interval, then the integrand on the r.h.s. of~\cref{eq:quality factor} is also zero over said interval. It is also a non-decreasing function of its argument $\tau$. Thus~\cref{eq:quality factor} is a measure of isentropicity in its own right.
For the $2\ndd$, we can consider how much gate implementation error we are willing to tolerate. For that, we can devise a figure of merit based on the performance in the special case in which, in the $l\thh$ cycle, we wish to implement a sequence of gates, which map the logical state to a sequence of orthonormal states, denoted
\begin{align}
	\ket{1,l}_\lo, \ket{2,l}_\lo,  \ket{3,l}_\lo, \ldots  \label{eq:sequnce orthgonal states}
\end{align}
with equal time $t_1$ between them (i.e., the $j\thh$ gate maps state $\ket{j,l}_\lo$ to  $\ket{j+1,l}_\lo$). For simplicity of notation, we assume that the sequence spans $\mathcal{H}_\Sy$. We will implement as many as possible within the time interval $[\tau_l,t_{\max,l}+\tau_l)$, and the rest in later cycles: specifically the first $h_l$ gates in the sequence, where $h_l= \max \{ j\in\zzp\, | \,t_{\max,l} -j t_1 > 0 \}$. As with the upper bounds in~\Cref{sec:Classical and quantum upper limits}, for simplicity, we assume that the gate frequency of our computer is the same for all gates it can implement.  We quantify our tolerance to error in terms of trace distance $T$  be requiring that the following holds
\begin{align}
	T\big(\rho_\lo(t_j+\tau_l|\tau_l),  \ket{j,l}_\lo\big) & \leq \varepsilon^\textup{gate}\in[0,1/2)
	 ,\label{eq:impelmentation quality}
\end{align}
where $t_j=j t_1\in(0,t_{\max,l})$, $j\in\zzp$ and 
\begin{align}
	 \rho_{\lo}(t+\tau_l|\tau_l)=\tr_{\Sym\Cl}[\rho_{\Sy\Cl}(t+\tau_l|\tau_l)],
\end{align}
(where $ \rho_{\Sy\Cl}(t+\tau_l|\tau_l)$ defined in~\cref{eq:tilde def norm}), are the normalised states conditioned on the cycle in question.  The $1\stt$ condition in~\cref{eq:impelmentation quality} quantifies how close the actual dynamics is to the desired gate, while the $2\ndd$ condition quantifies how close this is to orthogonality with the next gate.  For a good implementation, we would demand that the parameter $\varepsilon^\textup{gate}$ is small.

Together, the parameters $\varepsilon^\textup{quality}_l$, $\varepsilon^\textup{gate}$  allow us to quantitatively define length of the $l\thh$ isentropic time interval $[\tau_l,t_{\max,l}+\tau_l)$: For a given $\varepsilon^\textup{gate}$, the quantity  $t_{\max,l}\geq 0$ is the largest constant such that
\begin{align}
\frac{1}{4(1+\lambda)^2} = \varepsilon^\textup{gate}(2 -{\varepsilon^\textup{gate}})   + 2  \varepsilon^\textup{quality}_l(t_{\max,l})  \sqrt{1-\varepsilon^\textup{gate} },
	\label{eq:condition for t max}
\end{align}
where $\lambda=4.64$. Note that since $\varepsilon^\textup{quality}_l$ increases with $t_{\max,l}$, the  $l\thh$  isentropic time interval will be largest when $\varepsilon^\textup{gate}$ is small.

Before continuing with the relevant definitions. We make a small side-step with the following lemma.
\begin{restatable}[Bounds for $\varepsilon^\textup{quality}_l(\cdot)$]{lemma}{LemupperAndLowerBoundingEpQuality}\label{lem:upperAndLowerBoundingEpQuality1} For all $\tau \geq 0$, $l\in\nnz$ \textup{:}
\begin{align}
\frac{1}{2}\int_0^\tau \dd t \, P(t+\tau_l, +1 | \tau_l ) =   \int_0^{\tau}\dd t \,  \|  \vec V_\Cl \|_{w,1}(t,l)&	\leq \varepsilon^\textup{quality}_l(\tau), \\
\varepsilon^\textup{quality}_l(\tau) \leq  \big(\| V_\Cl \| \,\tau\big)^{3/4} & \left(\int_0^{\tau}\dd t \,  \|  \vec V_\Cl \|_{w,1}(t,l)\right)^{1/4} =   \frac{\left(\| V_\Cl \| \,\tau\right)^{3/4}} {2^{1/4}}   \left(\int_0^\tau \dd t \, P(t+\tau_l, +1 | \tau_l )\right)^{1/4},\label{eq:lem:inter over tau}
\end{align}
where $ \| \cdot\| $ is the operator norm.
\end{restatable}
Proof is in~\Cref{sec:proof of lemma upperAndLowerBoundingEpQuality1}. Due to monotonicity with respect to $\tau$ of all terms in~\cref{eq:lem:inter over tau},  \Cref{lem:upperAndLowerBoundingEpQuality1} allows one to generate lower and upper bounds on $t_{\max,l}$ when the conditional probability distribution $P(t+\tau_l, +1 | \tau_l )$ is known. What is more, it also shows that if $P(t+\tau_l, +1 | \tau_l )$ is sufficiently small when integrated over $t\in[0,t^*]$, then $t_{\max,l}\geq t^*$, while if said probability becomes sufficiently large when integrated over $t\in[t^*,t^{**}]$, then $t_{\max,l}\leq t^{**}$. This makes sense, because $P(t+\tau_l, +1 | \tau_l )$ is the probability that the $(l+1)\thh$ renewal event occurs, and can only become non-zero when the system $\Sy\Cl$ starts to interact with the environment and subsequently takes it out of the isentropic time interval. In~\Cref{sec:Calculation of the range of the self-oscillator isentropic regime  and} we apply~\Cref{lem:upperAndLowerBoundingEpQuality1} to find a bound for $t_{\max,l}$ in the context of~\Cref{thm:heat dissipation}.

We aim to bound the gate frequency
\begin{align}
	f =1/ t_1
\end{align}
with which the gates during the $l\thh$ isentropic time interval, $[\tau_l, t_{\max,l}+\tau_l)$ are applied 
up to error $\varepsilon^\textup{gate}$.

\begin{itemize}
	\item 1) [Universally bounded $1\stt$ moments]: there exists $\mathbb{M}(1)<\infty$ such that
	\begin{align}
		\mathbb{M}_l(1):=	 \int_0^\infty \dd t\,  t\, P(t+\tau_l,+1|\tau_l) \leq \mathbb{M}(1)\label{def:1st moment finite}
	\end{align}
	for all $l\in\nnz$. This is to say, we assume that the $1\stt$ moment of the conditional distributions $\{P(t+\tau_l,+1|\tau_l)\}_l$ are  universally upper bounded. As we show in~\cref{Dependency on T0}, $\mathbb{M}_l(1)$ is always directly proportional to $T_0$. Moreover, when the oscillator is of good quality, we expect $\mathbb{M}_l(1)/T_0	\approx 1$ because in such cases the probability of renewing the oscillator should be peaked around the cycle time.\\
	
	\item	2) [Instantaneous Initial-cycle-state control of  $\Sy$]:
	The interval of the $l\thh$ isentropic time interval initiates immediately after the oscillator state on $\Cl$ is corrected to $\rho_\Cl^0$. We want the control on $\Sy$ to start gradually after the state has been renewed to $\rho_\Cl^0$. For this, we need the state of the control on $\Cl$ to initially be non-interacting with the system $\Sy$, otherwise, among other things, the interactions would interfere with the renewal process.  As such we want $\rho^0_\Cl$ to be orthogonal to all terms in $H_{\Sy\Cl}$ which act non-trivially on both $\Sy$ and $\Cl$, while allowing it to only act non-trivially on the terms which solely act on $\Cl$. This way the state of the oscillator can only start to interact with $\Sy$ after dynamically evolving due to its free dynamics.  For generality, rather than demanding that $\rho^0_\Cl$ is \emph{exactly} orthogonal to said terms, we allow for a small overlap of $\rho^0_\Cl$ with the non-free terms.  Explicitly,  for any  Hamiltonian $H_\Cl$,  $\varepsilon_\textup{H}^0$  is defined as any quantity to satisfy the following bound 
	\begin{align}
		{	\left	\|  \tr_\Cl[\tilde H_{\Sy\Cl} \,\rho^0_\Cl] \right \|_F} {	\|  \tilde H_{\Sy\Cl}   \|_F} 2 \mathbb{M}(1) T_0 \leq \varepsilon_\textup{H}^0, \quad \tilde H_{\Sy\Cl}:= H_{\Sy\Cl}- H_\Cl, \label{def:epsion 0 H}
	\end{align}
	where $\|  \cdot \|_F $ is the Frobenius norm. We will be interested in the cases in which  $\varepsilon_\textup{H}^0$ is small relative to the power $T_0^2 P^\textup{in}_l$.
	
	  We call $\varepsilon_\textup{H}^0$ the \emph{instantaneous initial-cycle-state parameter}. In order for the l.h.s. of~\cref{def:epsion 0 H} to be small, there must be no free terms of the form $H_\Sy$ in $H_{\Sy\Cl}$, otherwise, it would appear to not be possible to satisfy~\cref{def:epsion 0 H}. Therefore, all the control and dynamics on $\Sy$ is generated via interactions with the oscillator. This is in accordance with all previously-defined hamiltonians in this~\doc. In addition to $\varepsilon_\textup{H}^0$ being dimensionless, it is $T_0$-independent, as we show in~\cref{Dependency on T0}.
\end{itemize}

\subsection{Proof of~\Cref{thm:heat dissipation upper bounds}: upper bounds for gate frequency with  self-oscillators as control}\label{sec:secnical proofs for self-oscillators}
We now state the main lemma for the upper bounds on orthogonalization time for self-oscillators. From this result,~\Cref{thm:heat dissipation upper bounds} readily follows by using simplifying assumptions.   
\begin{lemma}[Upper bounds on gate frequency in the nonequilibrium steady-state regime]\label{eq:upper bounds orthgonalization isentropic}
	For all $\varepsilon^\textup{gate}\in[0,1)$, the frequency $f$ at which the gate sequence~\cref{eq:sequnce orthgonal states} is implemented up to error $\varepsilon^\textup{gate}$
	during the $l\thh$ isentropic time interval $[\tau_l,t_{\max,l}+\tau_l)$, is bounded as follows: 
	
		\vspace{0.25cm}
		\noindent Case 1):
		\begin{align}
			f \leq \frac{ \sqrt{T_0\tr[H_{\Sy\Cl} \,\rho_{\Sy\Cl}(\tau_l|\tau_l)]} }{c_0 T_0}, \label{eq:lemma statement main bound linear E eq 1}
		\end{align}
		if  $\rho_{\Sy\Cl}(\tau_l|\tau_l)\in \mathcal{C}^{\text{clas.,\,}l}_{\Sy\Cl}$.
		
		\vspace{0.25cm}
		\noindent Case 2):
		\begin{align}
			f \leq \frac{ \tr[H_{\Sy\Cl} \,\rho_{\Sy\Cl}(\tau_l|\tau_l)] }{\kappa}, \label{eq:lemma statement main bound linear E eq 2}
		\end{align}
		otherwise.\vspace{0.2cm}
		
		\noindent Here $\mathcal{C}^{\text{clas.,\,}l}_{\Sy\Cl}$ is the set of classical states defined in~\cref{def:semi calssical states self-oscillator} and $\kappa>0, c_0>0$ are numerical constants.
	\end{lemma}
	\begin{proof}
		The proof will follows similar lines of reasoning to those of~\cref{app:upper bounds proofs}, but with additional complication due to the dynamics being give by a dynamical semigroup rather than a Hamiltonian, and the states during the computation not being exactly mutually orthogonal but merely close. 
	
	We start by setting up the metrology problem: Consider an observable
	\begin{align}
		M=\sum_k 
		 t_\textup{est}(\xi_k) \proj{k,l}_{\lo}\otimes  \id_{\Sym\Cl} \label{eq:lower bound last proof 09}
	\end{align}
	where our signal estimates are chosen to be $t_\textup{est}(\xi_j)= t_j+\Delta\, t_1=(j +\Delta)t_1$ and $t_\textup{est}(\xi_k)= (j-\Delta+1) t_1$ for $k\neq j$. Note that for additional clarify, and in contrast to the rest of this~\doc, we have refrained from omitting the tensor product with the identity in~\cref{eq:lower bound last proof 09}. 
	We will be interested in bounding the probabilities associated with getting outcome $\xi_j$ given that the time is $t_k$, namely $P(\xi_j|t_k)$ when the state $\rho_{\Sy\Cl}(\tau_l|\tau_l)$ is evolved according to the unitary part of the dynamics:
	\begin{align}
		P(\xi_j|t_k):= \tr\big[\proj{j,l}_\lo\otimes\id_{\Sym\Cl} \,\me^{-\mi t_k H_{\Sy\Cl}} \rho_{\Sy\Cl}(\tau_l|\tau_l)  \me^{\mi t_k H_{\Sy\Cl}} \big]= \big\|  \proj{j,l}_\lo\otimes\id_{\Sym\Cl\Pu}    \ket{\rho^\textup{H}(t_k)}_{\Sy\Cl\Pu} \big\|_2^2 \label{eq:lower bound last proof 1}
	\end{align} 	
		where we have used the fact that $\proj{j,l}_\lo\otimes\id_{\Sym\Cl}$ is a projector, and defined $\ket{\rho^\textup{H}(t)}_{\Sy\Cl\Pu}:= \me^{-\mi t H_{\Sy\Cl}} \ket{(\tau_l | \tau_l)}_{\Sy\Cl\Pu}$, where $\ket{(\tau_l | \tau_l)}_{\Sy\Cl\Pu}$ is a purification of $\rho_{\Sy\Cl}(\tau_l|\tau_l)$.
		
		Let us start with the case $P(\xi_j|t_{j+1})$. Employing the triangle inequality and~\cref{eq:impelmentation quality}, we have 
		\begin{align}
			1&= T(\ket{j,l}_\lo, \ket{j+1,l}_\lo) \leq T(\ket{j,l}_\lo, \tilde \rho_\lo(t_{j+1}+\tau_l|\tau_l)) +	T\big(\tilde \rho_\lo(t_{j+1}+\tau_l|\tau_l),  \ket{j+1,l}_\lo\big)\\
			& \leq T(\ket{j,l}_\lo, \tilde \rho_\lo(t_{j+1}+\tau_l|\tau_l)) + \varepsilon^\textup{gate}.
		\end{align}
		Thus using the inequality $T(\rho_1,\rho_2)\leq \sqrt{1-F(\rho_1,\rho_2)}$ where $F$ is the quantum Fidelity, we find
		\begin{align}
			&1-(1-\varepsilon^\textup{gate})^2 \geq F(\ket{j,l}_\lo, \rho_\lo(t_{j+1}+\tau_l|\tau_l)) = \tr[ \proj{j,l}_\lo  \rho_\lo(t_{j+1}+\tau_l|\tau_l)]
			\\& =  \tr[ (\proj{j,l}_\lo  \otimes \id_{\Sym\Cl\Pu} ) \proj{\rho(t_{j+1})}_{\Sy\Cl\Pu}]= \big\| (\proj{j,l}_\lo  \otimes \id_{\Sym\Cl\Pu} ) \ket{\rho(t_{j+1})}_{\Sy\Cl\Pu} \big\|_2^2
			\\&= \frac{\big\| (\proj{j,l}_\lo  \otimes \id_{\Sym\Cl\Pu} ) \ket{ \rho\lb t_{j+1}\rb }_{\Sy\Cl\Pu} \big\|_2^2}{\big\| \ket{ \rho\lb t_{j+1}\rb }_{\Sy\Cl\Pu} \big\|_2^2} \label{eq:lower bound last proof 3}
		\end{align}
		where $\ket{\rho(t_{j+1})}_{\Sy\Cl\Pu}:=\ket{\rho\lb t_{j+1}\rb}_{\Sy\Cl\Pu}/\| \ket{\rho\lb t_{j+1}\rb }_{\Sy\Cl\Pu}\|_2$, and is a purification of $\rho_{\Sy\Cl}(t_{j+1}+\tau_l| \tau_l):= \rho_{\Sy\Cl}\lb t_{j+1}+\tau_l| \tau_l\rb /\tr[\rho_{\Sy\Cl}\lb t_{j+1}+\tau_l| \tau_l\rb ]$. We now solve~\cref{eq:conditional states AC def}, to find
		\begin{align}
			\rho_{\Sy\Cl}\lb t+\tau_l | \tau_l\rb  &= \me^{-\mi t G} \rho_{\Sy\Cl}(\tau_l | \tau_l)  \me^{\mi t G^\dag} , \quad G:= H_{\Sy\Cl} -\mi V_\Cl,  
		\end{align}
		Therefore,  it follows that
		\begin{align}
			\ket{\rho\lb t\rb }_{\Sy\Cl\Pu}= \me^{-\mi t G} \ket{\rho(0)}_{\Sy\Cl\Pu}= \ket{\rho^\textup{H}(t)}_{\Sy\Cl\Pu}+ \ket{\delta \rho\lb t\rb}_{\Sy\Cl\Pu},
		\end{align}
		where we have defined
		\begin{align}
			\ket{\rho^\textup{H}(t)}_{\Sy\Cl\Pu}:= \me^{-\mi t H_{\Sy\Cl}} \ket{(\tau_l | \tau_l)}_{\Sy\Cl\Pu}, \quad \ket{\delta\rho\lb t\rb }_{\Sy\Cl\Pu}:=\ket{\rho\lb t\rb }_{\Sy\Cl\Pu} - \ket{\rho^\textup{H}(t)}_{\Sy\Cl\Pu}= (\me^{-\mi t G}-\me^{-\mi t H_{\Sy\Cl}} )\ket{(\tau_l | \tau_l)}_{\Sy\Cl\Pu}.\label{eq: rho H def}
		\end{align}
		By direct calculation, 
		\begin{align}
			\frac{d}{dt} \|	\ket{\rho\lb t\rb}_{\Sy\Cl\Pu} \|_2^2=  - 2 \,	{}_{\Sy\Cl\Pu}\!\bra{\rho\lb t\rb } V_\Cl  \ket{\rho\lb t\rb }_{\Sy\Cl\Pu} \leq 0, 
		\end{align}
		hence since $\rho_{\Sy\Cl}(\tau_l | \tau_l)$ is normalised, it follows that 
		\begin{align}
			 \|	\ket{\rho\lb t\rb}_{\Sy\Cl\Pu} \|_2\leq 1 \label{eq:lower bound last proof 5}
		\end{align} 
		for $t\geq 0$.	Similarly, since $\ket{\rho^\textup{H}(t)}_{\Sy\Cl\Pu}$ evolves unitarily,
				\begin{align}
			\|	\ket{\rho^\textup{H}(t)}_{\Sy\Cl\Pu} \|_2 = 1. \label{eq:nor of unitary part}
		\end{align} 
	We can now derive a lower bound on norms of the form $\| \hat p 	\ket{\rho\lb t\rb }_{\Sy\Cl\Pu} \|_2$ where $\hat p$ is a projector. Via~\cref{eq:nor of unitary part}, and the reverse triangle inequality, it follows
	\begin{align}
	\big\|	\hat p \ket{\rho\lb t\rb }_{\Sy\Cl\Pu}\big\|_2&= 	\big\|\hat p  \ket{\rho^\textup{H}(t)}_{\Sy\Cl\Pu} + \hat p \ket{\delta \rho\lb t\rb}_{\Sy\Cl\Pu} \big\|_2 \geq  \big\|\hat p  \ket{\rho^\textup{H}(t)}_{\Sy\Cl\Pu} \big\|_2 - 	\big\| \hat p \ket{\delta \rho\lb t\rb }_{\Sy\Cl\Pu}\big\|_2\\
	&\geq  \big\|\hat p  \ket{\rho^\textup{H}(t)}_{\Sy\Cl\Pu} \big\|_2 - 	\big\|  \ket{\delta \rho\lb t\rb }_{\Sy\Cl\Pu}\big\|_2. \label{eq:reverse triangle 1}
	\end{align}	
	Thus from~\cref{eq:lower bound last proof 3} and~\cref{eq:lower bound last proof 5} it follows
	\begin{align}
		\sqrt{1-(1-\varepsilon^\textup{gate})^2}& \geq  \big\| \proj{j,l}_\lo  \otimes \id_{\Sym\Cl\Pu}  \ket{ \rho\lb t_{j+1}\rb }_{\Sy\Cl\Pu} \big\|_2 \geq  \big\| \proj{j,l}_\lo  \otimes \id_{\Sym\Cl\Pu}  \ket{ \rho^\textup{H}(t_{j+1})}_{\Sy\Cl\Pu} \big\|_2 - \big\|  \ket{ \delta\rho\lb t_{j+1}\rb }_{\Sy\Cl\Pu} \big\|_2
		\\& =\sqrt{P(\xi_j|t_{j+1})} - \big\|  \ket{ \delta\rho\lb t_{j+1}\rb }_{\Sy\Cl\Pu} \big\|_2,
	\end{align}
	where we used definition~\labelcref{eq:lower bound last proof 1} in the last line. It thus follows that 
	\begin{align}
	   \sqrt{P(\xi_j|t_{j+1})} & \leq    	\sqrt{ \varepsilon^\textup{gate}(2 -{\varepsilon^\textup{gate}})} +\big\|  \ket{ \delta\rho\lb t_{j+1}\rb }_{\Sy\Cl\Pu} \big\|_2. \label{eq:lower bound last proof 10}
	\end{align}
	We now proceed to bound $P(\xi_j|t_{j})$. Since one of the states in~\cref{eq:impelmentation quality} is pure, we can use the inequality $1-F(\cdot,\cdot)\leq T(\cdot,\cdot)$ to get
	\begin{align}
	1-\varepsilon^\textup{gate} &\leq F\big(\tilde \rho_\lo(t_j+\tau_l|\tau_l),  \ket{j,l}_\lo\big) = {}_\lo\!\braket{j,l | \rho_\lo(t_j+\tau_l|\tau_l) | j,l}_\lo
	\\&= \frac{\big\| (\proj{j,l}_\lo  \otimes \id_{\Sym\Cl\Pu} ) \ket{ \rho\lb t_{j}\rb}_{\Sy\Cl\Pu} \big\|_2^2}{\big\| \ket{ \rho\lb t_{j}\rb }_{\Sy\Cl\Pu} \big\|_2^2} \label{eq:lower bound last proof 6}
	\end{align}
where we proceeded analogously to in~\cref{eq:lower bound last proof 3}. To continue, we now employ~\cref{eq:reverse triangle 1} with $\hat p$ the identity operator on $\Sy\Cl\Pu$ to bound the denominator, while we apply the triangle inequality to bound the numerator:
	\begin{align}
	\sqrt{1-\varepsilon^\textup{gate} }&\leq \frac{\big\| (\proj{j,l}_\lo  \otimes \id_{\Sym\Cl\Pu} ) \ket{ \rho^\textup{H}(t_{j})}_{\Sy\Cl\Pu} \big\|_2+ \big\| (\proj{j,l}_\Sy  \otimes \id_{[\Sy\backslash\lo]\Cl\Pu} ) \ket{ \delta\rho\lb t_{j}\rb }_{\Sy\Cl\Pu} \big\|_2}{1-\big\| \ket{ \delta\rho\lb t_{j}\rb }_{\Sy\Cl\Pu} \big\|_2}
	\\&\leq  \frac{\sqrt{P(\xi_j|t_j)}+ \big\|   \ket{ \delta\rho\lb t_{j}\rb}_{\Sy\Cl\Pu} \big\|_2}{1-\big\| \ket{ \delta\rho\lb t_{j}\rb}_{\Sy\Cl\Pu} \big\|_2}. \label{eq:lower bound last proof 16}
\end{align}
	Hence 
	\begin{align}
	\sqrt{1-P(\xi_j|t_j)} \leq  \sqrt{1
	-\left(\sqrt{1-\varepsilon^\textup{gate} } -(1+ 		\sqrt{1-\varepsilon^\textup{gate} } )  \big\| \ket{ \delta\rho\lb t_{j}\rb}_{\Sy\Cl\Pu} \big\|_2\right)^2} \label{eq:lower bound last proof 17 2}
\end{align}
for $\big\| \ket{ \delta\rho\lb t_{j}\rb}_{\Sy\Cl\Pu} \big\|_2\leq 1$ (since this guarantees that the argument of the square-root is non-negative). 
	We now proceed to bound  $\|   \ket{ \delta\rho\lb t\rb }_{\Sy\Cl\Pu}\|_2$ using~\Cref{lem:unitari errs add linararly}.  Using the notation in~\Cref{lem:unitari errs add linararly}, we set $\ket{\Phi_0}= \ket{\rho(\tau_l|\tau_l)}_{\Sy\Cl\Pu}$,  $\ket{\Phi_m}= \me^{-\mi (t/n) G} \ket{\Phi_{m-1}}$,  $\ket{\Phi_m}= \left(  \me^{-\mi (t/n) G} \right)^m\ket{\Phi_0}=\ket{\rho\lb t_m\rb}_{\Sy\Cl\Pu}$, $t_m=m (t/n)$,  $\Delta_k=\me^{-\mi (t/n) H_{\Sy\Cl}}$. Therefore,
	\begin{align}
		&\|   \ket{ \delta\rho\lb t\rb }_{\Sy\Cl\Pu}\|_2 = \|   \ket{ \rho\lb t\rb }_{\Sy\Cl\Pu}-     \ket{ \rho^\textup{H}(t) }_{\Sy\Cl\Pu}\|_2 =  \|  \ket{\Phi_n} - \Delta_n \Delta_{n-1} \ldots \Delta_1\ket{\Phi_0} \|_2 
		\\ & \leq  \sum_{m=1}^{n} \|   \ket{\Phi_m}  - \Delta_m\ket{\Phi_{m-1}} \|_2 =  \sum_{m=1}^{n}  \|  \left[\id_{\Sy\Cl} -\mi (t/n) (H_{\Sy\Cl}-\mi V_\Cl) - \id_{\Sy\Cl} + \mi (t/n) H_{\Sy\Cl} + \bo(t/n)^2 \right] \ket{\rho\lb t_m\rb }_{\Sy\Cl\Pu}  \|_2
		\\ & \leq  \sum_{m=1}^{n} \left( (t/n) \|  V_\Cl \ket{\rho\lb t_m\rb}_{\Sy\Cl\Pu}  \|_2 + \bo(t/n)^2  \right) = \left( \sum_{m=1}^{n} (t/n)  \sqrt{\tr[ V_\Cl^2 \, \rho_{\Sy\Cl}\lb t_m+\tau_l| \tau_l\rb] } \right)  + \bo(t^2/n) .
	\end{align}
	Therefore, taking the limit $n\to \infty$ on both sites of the above we obtain a Riemann integral on the r.h.s.: 
	\begin{align}
		&\|   \ket{ \delta\rho\lb t\rb }_{\Sy\Cl\Pu}\|_2 \leq   \int_{0 }^{t} \dd x \, \sqrt{\tr[ V_\Cl^2 \, \rho_{\Sy\Cl}\lb x+\tau_l| \tau_l\rb] }=  \int_{0 }^{t} \dd x \, \| \vec V_\Cl \|_{w,2}(x,l) =	\varepsilon^\textup{quality}_l(t),
	\end{align}
	where in the last line we have written the integrand in terms of the weighted 2-norm defined in~\cref{def:weighted q norm}. Thus, for
	\begin{align}
		 \varepsilon^\textup{quality}_l(t_{j+1})\in [0,1]
	\end{align}
	from~\cref{eq:lower bound last proof 10,eq:lower bound last proof 17 2} we have the bounds
	\begin{align}
			   \sqrt{P(\xi_j|t_{j+1})} & \leq    	\sqrt{ \varepsilon^\textup{gate}(2 -{\varepsilon^\textup{gate}})} +			 \varepsilon^\textup{quality}_l(t_{j+1}),
			   \\	\sqrt{1-P(\xi_j|t_j)} &\leq  \sqrt{1
			   	-\left(\sqrt{1-\varepsilon^\textup{gate} } -(1+ 		\sqrt{1-\varepsilon^\textup{gate} } )  				 \varepsilon^\textup{quality}_l(t_{j+1})\right)^2}.  \label{eq:lower bound last proof 13}
	\end{align}

Since the dynamics of $\ket{\rho^\textup{H}(t)}_{\Sy\Cl\Pu}$ is unitary time evolution via a time-independent Hamiltonian, we can take advantage of the formalism developed in~\cite{Giovannetti2012}|similarly to how we did in~\cref{app:upper bounds proofs}. However, we will have to set up the measurement strategy differently this time.  In particular, we will use the fact that our signal estimate $t_\textup{est}(\xi_j)$ in~\cref{eq:lower bound last proof 09}, is degenerate. We start by evaluating the mean squared error of our signal.  
\begin{align}
	\braket{(t_\text{est}-t)^2}&= \sum_k
	  P(\xi_k|t) (t_\text{est}(\xi_k)-t)^2
	\\ =&P(\xi_j|t) (t_j+\Delta\,  t_1-t)^2+
	 P(\xi_k|t) ((j-\Delta+1)t_1-t)^2
	 \\=&  \left( P(\xi_j|t) \left[(j+\Delta -\frac{t}{t_1})^2 -(j-\Delta +1 -\frac{t}{t_1})^2\right]  +(j-\Delta +1 -\frac{t}{t_1})^2   \right) t_1 
\end{align}
where in the penultimate line, we have used $\sum_k
P(\xi_k|t)=1$.   Therefore,
\begin{align}
 \sqrt{	\braket{(t_\text{est}-t_j)^2}} &= t_1\sqrt{(1-P(\xi_j|t_{j}) )  [ (1-\Delta)^2-\Delta^2] +\Delta^2 } , 	\label{eq:deltas1}
 \\ \sqrt{	\braket{(t_\text{est}-t_{j+1})^2}} &=  	t_1\sqrt{P(\xi_j|t_{j+1})  [(1-\Delta)^2 -\Delta^2] +\Delta^2 }.	\label{eq:deltas2}
\end{align}
Analogously to how we did in~\cref{app:upper bounds proofs}, we now evaluate the criteria~\cref{def: delta X form other paper} for times $t_j$ and $t_{j+1}$. Namely, we want to find solutions to the parameters such that 
\begin{align}
\begin{split}
		1)& \qquad   \sqrt{	\braket{(t_\text{est}-t_j)^2}} >0 , \qquad \sqrt{	\braket{(t_\text{est}-t_{j+1})^2}} >0\\
	2)&\qquad |t_j-t_{j+1}|=(\lambda+1) \left(\sqrt{	\braket{(t_\text{est}-t_j)^2}}  + \sqrt{	\braket{(t_\text{est}-t_{j+1})^2}}  \right) 
\end{split} \label{def: delta X form other paper 3}
\end{align} 
where $\lambda=4.64$, holds. This will yield the criterion for being in the isentropic time interval.  We start with criterion 2), which, using~\cref{eq:deltas1,eq:deltas2}, holds iff: 
	\begin{align}
     \frac{1}{(\lambda+1)}= \sqrt{(1-P(\xi_j|t_{j}) )  [ (1-\Delta)^2-\Delta^2] +\Delta^2 }+ 	\sqrt{P(\xi_j|t_{j+1})  [(1-\Delta)^2 -\Delta^2] +\Delta^2 }. \label{def: delta X form other paper 4}
	\end{align}
 Now observer that this is continuous in $\Delta\in[0,1/2]$ for all $P(\xi_j|t_{j+1}), P(\xi_j|t_{j})\in[0,1]$. Furthermore, at $\Delta=1/2$, we have the the r.h.s. of~\cref{def: delta X form other paper 4} is 1 for all $P(\xi_j|t_{j+1}), P(\xi_j|t_{j})\in[0,1/2]$.  This is greater than the l.h.s. of~\cref{def: delta X form other paper 4}, since $1/(\lambda+1)\approx hhh$.  At $\Delta=0$, the r.h.s. of~\cref{def: delta X form other paper 4} is 
\begin{align}
	\sqrt{1-P(\xi_j|t_{j})}+ 	\sqrt{P(\xi_j|t_{j+1})}. \label{def: delta X form other paper 5}
\end{align}
	Therefore, due to continuity in $\Delta\in(0,1/2]$ for all $P(\xi_j|t_{j+1}), P(\xi_j|t_{j})\in[0,1]$, if 
\begin{align}
	     \frac{1}{\lambda+1} > \sqrt{1-P(\xi_j|t_{j})}+ 	\sqrt{P(\xi_j|t_{j+1})}, \label{def: delta X form other paper 6}
\end{align}
there exists $\Delta\in(0,1/2)$ such that~\cref{def: delta X form other paper 4} holds. Hence~\cref{def: delta X form other paper 6} gives a sufficient condition for 2) in~\cref{def: delta X form other paper 3} to hold.  Moreover, from~\cref{eq:deltas1,eq:deltas2} we see that condition 1) also holds for all  $P(\xi_j|t_{j+1}), P(\xi_j|t_{j})\in[0,1]$ if $\Delta\in(0,1/2)$.
Consider the condition
	\begin{align}
 	     \frac{1}{\lambda+1} >	\sqrt{ \varepsilon^\textup{gate}(2 -{\varepsilon^\textup{gate}})} +		 		 \varepsilon^\textup{quality}_l(t_{j+1})	+  \sqrt{1
		-\left(\sqrt{1-\varepsilon^\textup{gate} } -(1+ 		\sqrt{1-\varepsilon^\textup{gate} } )  		 		 \varepsilon^\textup{quality}_l(t_{j+1})  \right)^2}.  \label{eq:lower bound last proof 33}
\end{align}
Due to~\cref{eq:lower bound last proof 13},  condition~\cref{eq:lower bound last proof 33} implies~\cref{def: delta X form other paper 6} and hence it also implies conditions 1) and 2) in~\cref{def: delta X form other paper 3}.  
Moreover, since 
\begin{align}
	&\sqrt{ \varepsilon^\textup{gate}(2 -{\varepsilon^\textup{gate}})} +		 		 \varepsilon^\textup{quality}_l(t_{j+1})	+  \sqrt{1
		-\left(\sqrt{1-\varepsilon^\textup{gate} } -(1+ 		\sqrt{1-\varepsilon^\textup{gate} } )  		 		 \varepsilon^\textup{quality}_l(t_{j+1})  \right)^2} 
	\\&\leq
2 \sqrt{ \varepsilon^\textup{gate}(2 -{\varepsilon^\textup{gate}})   + 2  \varepsilon^\textup{quality}_l(t_{j+1})  \sqrt{1-\varepsilon^\textup{gate} } },
\end{align}
we arrive at condition~\cref{eq:condition for t max} by noting that the r.h.s. is a non-increasing function of $t_{j+1}$. 
The conditions~\cref{def: delta X form other paper} in our case are conditions~\cref{def: delta X form other paper 3}, [which we have shown to be satisfied for $t_{j+1}\in [0,t_{\max,l})$], it thus follows from~\cite{Giovannetti2012} that 
\begin{align}
\sqrt{	\braket{(t_\text{est}-t_{j+1})^2}}  \geq \frac{\kappa}{\tr[H_{\Sy\Cl}  \proj{ \rho^\textup{H}(t_{j+1})}_{\Sy\Cl\Pu} ]}, \label{eq:energy contraint e1}
\end{align}
for $t_{j+1}\in [0,t_{\max,l})$, where $\kappa$ is a numerical constant $\kappa\approx 0.091$. (Recall that $H_{\Sy\Cl}$ has ground state energy equal to zero by definition, and hence the ground state does not appear on the r.h.s. of~\cref{eq:energy contraint e1}.) Using~\cref{eq:deltas2} and noting that $P(\xi_j|t_{j+1})\in[0,1]$, $\Delta\in(0,1/2]$, we have
\begin{align}
	 \\ \sqrt{	\braket{(t_\text{est}-t_{j+1})^2}} \leq 	t_1.\label{eq:new one}
\end{align}
Thus
\begin{align}
t_1 \geq 	\sqrt{	\braket{(t_\text{est}-t_{j+1})^2}}  \geq \frac{\kappa}{\tr[H_{\Sy\Cl}  \proj{ \rho^\textup{H}(t_{j+1})}_{\Sy\Cl\Pu} ]}=\frac{\kappa}{\tr[H_{\Sy\Cl} { \rho_{\Sy\Cl}(\tau_l| \tau_l)} ]},  \label{eq:energy contraint e2}
\end{align}
from which~\cref{eq:lemma statement main bound linear E eq 2} follows.

We now proceed to prove~\cref{eq:lemma statement main bound linear E eq 1}. Recall that we have defined  $c_0$ be the numerical constant defined in~\cite{Maccone2020squeezingmetrology}.  Similarly to how we proceeded in the proof in~\cref{app:upper bounds proofs}, one can repeat the above calculations for the dimensionless version of the time operator $M$ in~\cref{eq:lower bound last proof 09}, namely $M/T_0$, which has a root mean squared error which differs by a factor of $1/T_0$.  Defining the dimensionless Hamiltonian via $\bar H_{\Sy\Cl}:=T_0 H_{\Sy\Cl}$, it thus follows from~\cite{Maccone2020squeezingmetrology} that 
\begin{align}
\frac{\sqrt{	\braket{(t_\text{est}-t_{j+1})^2}}}{T_0}  \geq \frac{c_0}{\sqrt{\tr\big[\bar H_{\Sy\Cl} \proj{\rho^\textup{H}(t_{j+1})}_{\Sy\Cl\Pu} \big]}}, \label{eq:t lower bound 233}
\end{align}
for $t_{j+1}\in [0,t_{\max,l})$ and the states $\ket{\rho^\textup{H}(t_j)}_{\Sy\Cl\Pu}$ and $\ket{\rho^\textup{H}(t_{j+1})}_{\Sy\Cl\Pu}$ are not squeezed. (Note that using  dimensionless time and Hamiltonian operators does not modify the definition of $\ket{\rho^\textup{H}(t_{j+1})}_{\Sy\Cl\Pu}$ in the above equation  since the factors of $T_0$ in the dimensionless time $t/T_0$ and dimensionless Hamiltonian $\bar H_{\Sy\Cl}$ cancel each other out in its definition, namely~\cref{eq: rho H def}.)  Here, a state $\ket{\rho^\textup{H}(t_k)}_{\Sy\Cl\Pu}$, \,\,$k\in{j,j+1}$ being not squeezed means it is a eigenstate of $\left(\bar M \lambda +\mi \bar H_{\Sy\Cl}\right)\otimes \id_\Pu$, with $|\lambda|=1$ and where $M$ is given by~\cref{eq:lower bound last proof 09}. Now, observe that 
\begin{align}
	 (\lambda \bar M^{(j,l)}_{\lo} +\mi\, \bar H_{\Sy\Cl}) -  (\lambda \bar M  +\mi \bar H_{\Sy\Cl} )= t_\text{est}(\xi_j) \id_{\Sy\Cl}.
\end{align}
As such, the set of eigenvectors of  $(\lambda \bar M^{(j,l)}_{\lo} +\mi\, \bar H_{\Sy\Cl})\otimes \id_\Pu$ and $\left(\lambda\bar M  +\mi \bar H_{\Sy\Cl}\right)\otimes \id_\Pu$ coincide. Therefore, by definition of $\mathcal{C}(\cdot,\cdot)$ in~\cref{def:set semi classical states} and since the proof holds for any purification $\Pu$,  it follows that all purifications $\ket{\rho^\textup{H}(t_k)}_{\Sy\Cl\Pu}$, are in $\mathcal{C}(\bar M_\lo^{(k,l)},\bar H_{\Sy\Cl})$ for $k\in{j,j+1}$. 
 Finally, note that $\ket{\rho^\textup{H}(t_{k})}_{\Sy\Cl\Pu}$ is any purification of $\rho_{\Sy\Cl}(\tau_l+t_k|\tau_l)$. Therefore, we have that~\cref{eq:t lower bound 233} holds if $ \rho_{\Sy\Cl}(\tau_l+t_k|\tau_l)\in\mathcal{C}(\bar M_\lo^{(k,l)},\bar H_{\Sy\Cl})$ for $k\in{j,j+1}$.

Finally, using~\cref{eq:deltas2,eq:new one}, but now for the normalised time operator,  it follows from~\cref{eq:t lower bound 233} that
\begin{align}
	\frac{t_1}{T_0}\geq \frac{\sqrt{	\braket{(t_\text{est}-t_{j+1})^2}}}{T_0}  \geq \frac{c_0}{\sqrt{\tr\big[\bar H_{\Sy\Cl} \proj{\rho^\textup{H}(t_{j+1})}_{\Sy\Cl\Pu} \big]}}= \frac{c_0}{\sqrt{T_0\tr\big[ H_{\Sy\Cl} \rho(\tau_l|\tau_l)_{\Sy\Cl} \big]}}, \label{eq:t lower bound 235}
\end{align}
if $\rho_{\Sy\Cl}(\tau_l| \tau_l) \in \mathcal{C}^{\text{clas.,\,}l}_{\Sy\Cl}$. Rearranging~\cref{eq:t lower bound 235} gives us~\cref{eq:lemma statement main bound linear E eq 1}.
\end{proof}

\begin{lemma}[Mean-energy power-and-dissipation relationships]\label{lem: power per cycle}
		The power and dissipation per cycle of a self-oscillator, $P^\text{re}$ and $P^\text{diss}$ respectively,  are give by
		\begin{align}
			P^\text{in}_l := \frac{\braket{E^\text{re}}}{T_0}= \frac{\tr[H_{\Sy\Cl} \,\rho_{\Sy\Cl}(\tau_l|\tau_l)]}{T_0}+ \frac{\delta E_l^\text{re}}{T_0}, \quad 	P^\text{diss}_l := \frac{\braket{E^\text{no re}}}{T_0}= 
			  \frac{\tr[H_{\Sy\Cl} \,\rho_{\Sy\Cl}(\tau_l|\tau_l)]}{T_0} -
			   \frac{\delta E_l^\text{no re}}{T_0},\label{eq:lemma E re and E after}
		\end{align}
		where
		\begin{align}
			\delta E_l^\text{re}:= \mi \!\int_0^\infty \!\!\textup{d} s\,  \tr\!\!\left[  [\bar H_\Sy , H_{\Sy\Cl} ] \,\rho_{\Sy\Cl}\lb s+\tau_{l}|\tau_{l}\rb \right]\in\rr, \quad  \delta E_l^\text{no re}\!:=   \int_0^\infty\!\! \textup{d}s \, P(s+\tau_l,+1|\tau_l)\, \tr[H_{\Sy\Cl}\, \rho_{\Sy\Cl}\lb s+\tau_l|\tau_l\rb] \geq 0,
		\end{align}
		with
		\begin{align}
			\bar H_\Sy  := \tr_\Cl[ H_{\Sy\Cl}\,  \rho_\Cl^0],
		\end{align} 
	and
		\begin{align}
			\left| \delta E^\text{re}_l \right|  \leq  \frac{2}{T_0} \,\varepsilon_\textup{H}^0. 
			\label{eq:bound on delt E}
		\end{align}
	\end{lemma}
	
	\begin{proof}
		We start by finding a representation for the map $D^\text{re}_\Cl(\cdot)$ which can be assumed w.l.o.g.  To start with, all completely positive and trace non-increasing maps which have a single output $\rho^0_\Cl$, can be written as $D^\text{re}_\Cl(\cdot)=\tr[M \cdot] \rho^0_\Cl$, where $M$ is positive semi-definite matrix on $\mathcal{H}_\Cl$. (To see this, note that $D^\text{re}_\Cl(\cdot)$ can be written in the form $p(\cdot)\rho_\Cl^0$ with $p(\cdot)\in[0,1]$ by definition. Now equate to the Kraus form of any completely positive and trace non-increasing map, followed by taking the trace on both sides of the equality.)  Thus using~\cref{eq:def:renawal channel generic def}, we require
		\begin{align}
			\sum_k	J_\Cl^k \rho_{\Cl} {J_\Cl^k}^\dag =  \tr[M \rho_{\Cl} ] \rho^0_\Cl.
		\end{align}
		Hence by taking the trace on both sides
		\begin{align}
			\tr\left[ \left( \sum_k	 {J_\Cl^k}^\dag J_\Cl^k  \right)\rho_{\Cl} \right]=  \tr[M \rho_{\Cl}],
		\end{align}
		for all $\rho_{\Cl}\in\mathcal{S}(\mathcal{H}_{\Cl})$. Taking into account that both $M$ and $\sum_k	 {J_\Cl^k}^\dag J_\Cl^k $ are Hermitian, and using the fact that it must hold for all $\rho_{\Cl}\in\mathcal{S}(\mathcal{H}_{\Cl})$, we conclude that $\sum_k	 {J_\Cl^k}^\dag J_\Cl^k =M$.   (To see this, note that the trace forms an inner product for the real vector space of Hermitian matrices. As such, we can expand $M$ and $\sum_k	 {J_\Cl^k}^\dag J_\Cl^k$  in an orthonormal basis in said space, and choose $\rho_{\Cl}$ as the basis elements. The orthogonality then implies that the real expansion coefficients for $M$ and $\sum_k{J_\Cl^k}^\dag J_\Cl^k$  must be the same.) We therefore have
		\begin{align}
			\mathcal{D}^\text{re}_\Cl(\cdot)= 2\, \tr_\Cl[ V_\Cl\,  \cdot] \, \rho^0_\Cl.  \label{eq:reneal map in term of V}
		\end{align}
		
		Finally, it is necessary to check that there does indeed exist a set of operators $\{J_C^k\}_k$ which realise~\cref{eq:reneal map in term of V}. (We will not need to prove uniqueness of the representation, since we will not use it in the proofs).  For this task, let one write $\rho^0_\Cl$ and $V_\Cl$ in diagonal form: $\rho^0_\Cl=\sum_q p_q^0 \proj{\rho_q^0}_\Cl$, $p^0_q\geq 0$, $\{ \ket{\rho_q^0}_\Cl \}_k$ an orthonormal basis, and similarly, $V_\Cl  =  \sum_l c_l \proj{l}_\Cl$, $c_l\geq 0$, $\{ \ket{l}_\Cl \}_k$ orthonormal basis. We can then define
		\begin{align}
			J^{q,l}_\Cl= \sqrt{p_q^0 \, c_l} \ketbra{\rho_q^0}{l}_\Cl,
		\end{align}
		where we have multi-indexed $k$ via $q,l$ (i.e. summing over $k$ is equivalent to summing over both $q$ and $l$ independently). One can readily check that $\sum_{q,l} J^{q,l}_\Cl (\cdot) {J^{q,l}_\Cl }^\dag$ is equal to~\cref{eq:reneal map in term of V}. We now proceed with the proof.
		
		The first equation in~\cref{eq:conditional states AC def}, is defined via the solution to
		\begin{align}
			\frac{d}{dt} \rho_{\Sy\Cl} \lb t+\tau_{l}|\tau_l\rb= 	\mathcal{L}_{\Sy\Cl}^\textup{no re} \left(  \rho_{\Sy\Cl} \lb t+\tau_l|\tau_l\rb \right)=-\mi [H_{\Sy\Cl},  \rho_{\Sy\Cl} \lb t+\tau_l|\tau_l\rb]- \{ V_\Cl,  \rho_{\Sy\Cl} \lb t+\tau_l|\tau_l\rb \}.\label{eq:dev rho AC}
		\end{align}
		Therefore, $\frac{d}{dt} \tr[ \rho_{\Sy\Cl} \lb t+\tau_l|\tau_l\rb]= 	- 2 \tr[ V_\Cl  \rho_{\Sy\Cl} \lb t+\tau_l|\tau_l\rb]$. Hence, from~\cref{eq:prob plus one renewal def,eq:reneal map in term of V}, it follows
		\begin{align}
			P(t+\tau_l,+1|\tau_l)=  2 \tr[ V_\Cl   \rho_{\Sy\Cl}\lb t+\tau_l|\tau_{l}\rb  ] = -  \frac{d}{dt} \tr[ \rho_{\Sy\Cl} \lb t+\tau_l|\tau_l\rb]. \label{eq:prob as a div}
		\end{align}
		For $	P(t+\tau_l,+1|\tau_l)$ to be a well defined probability distribution, it must be that the probability of the $l\thh$ renewal occurring at some point in time is one, thus,
		\begin{align}
			1= \int_0^\infty \textup{d} t \,   P(t+\tau_l,+1|\tau_l)=  \tr[ \rho_{\Sy\Cl} (\tau_l|\tau_l)] -  \lim_{t\to\infty}\tr[ \rho_{\Sy\Cl} \lb t+\tau_l|\tau_l\rb]= 1-   \lim_{t\to\infty}\tr[ \rho_{\Sy\Cl} \lb t+\tau_l|\tau_l\rb].\label{eq:prob norm def}
		\end{align}
		Therefore, $ \lim_{t\to\infty}\tr[ \rho_{\Sy\Cl} \lb t+\tau_l|\tau_l\rb]=0$. Now, let us write $\rho_{\Sy\Cl} \lb t+\tau_l|\tau_l\rb$ in diagonalised form: $\rho_{\Sy\Cl} \lb t+\tau_l|\tau_l\rb=\sum_j P_j(t,\tau_l) \proj{(t,\tau_l)}_{\Sy\Cl}$.   We thus have $\lim_{t\to\infty} P_j(t,\tau_l) =0$ for all $j$,  $l\in\nnz$. Hence it follows that for all linear operators $O_{\Sy\Cl}$ on $\mathcal{H}_{\Sy\Cl}$, 
		\begin{align}
			\lim_{t\to\infty}\tr[O_{\Sy\Cl}\, \rho_{\Sy\Cl} \lb t+\tau_l|\tau_l\rb] =0,\quad l\in\nnz. \label{eq:rho AC large t lim}
		\end{align}
		
		\begin{align}
			\braket{E^\text{re}}&=  \int_{0}^{\infty} \dd s \, \tr\!\!\left[ H_{\Sy\Cl}\,	2\, \tr_\Cl[V_\Cl \rho_{\Sy\Cl}\lb  s+\tau_{l} |\tau_{l}\rb ]\otimes \rho^0_\Cl \right] \label{eq:expec E Re line 1}\\
			&=  - \int_{0}^{\infty} \dd s \, \frac{d}{ds} \tr\!\!\left[ H_{\Sy\Cl} \tr_\Cl[ \rho_{\Sy\Cl}\lb  s+\tau_{l} |\tau_{l}\rb ]\otimes \rho^0_\Cl \right]   - \mi \int_{0}^{\infty} \dd s \,  \tr\!\left[ H_{\Sy\Cl}\,  \tr_\Cl\big[  [H_{\Sy\Cl} ,\rho_{\Sy\Cl}\lb  s+\tau_{l} |\tau_{l}\rb]  \big]\otimes \rho^0_\Cl \right] \label{eq:expec E Re line 2}\\
			&= \tr\!\left[ H_{\Sy\Cl}\,  \rho_{\Sy\Cl}\lb  s+\tau_{l} |\tau_{l}\rb \right]     - \mi \int_{0}^{\infty} \dd s \,  \tr\!\left[ \bar H_{\Sy}\,    [H_{\Sy\Cl} ,\rho_{\Sy\Cl}\lb  s+\tau_{l} |\tau_{l}\rb]    \right],  \label{eq:expec E Re line 3}
		\end{align}
		where in line~\labelcref{eq:expec E Re line 1} we have used~\cref{eq:reneal map in term of V}, while in line~\labelcref{eq:expec E Re line 2}, we have used~\cref{eq:dev rho AC}. In line~\labelcref{eq:expec E Re line 3}, we have used that $\rho_{\Sy\Cl}(\tau_l|\tau_l)=\rho_{\Sy}(\tau_l|\tau_l)\otimes\rho_\Cl^0$ and~\cref{eq:rho AC large t lim}. Finally, to conclude the proof of the expression for $\braket{E^\text{re}}$ in~\cref{eq:lemma E re and E after}, now follows by noting that commutator in line~\labelcref{eq:expec E Re line 3} can be moved to the first two terms due to the cyclicity of the trace. We now derive the expression for $\braket{E^\text{no re}}$. 
		\begin{align}
			\braket{E^\text{no re}}&=-  \int_{0}^{\infty} \dd t  \, P( {t+\tau_l}  ,+1|\tau_l)  \int_0^{t} \dd s\,   \tr[H_{\Sy\Cl}\,	\mathcal{D}_{\Cl}^\text{no re}(\rho_{\Sy\Cl}\lb  s+\tau_{l} |\tau_{l}\rb)]   \label{eq:S cycle line 1} 
			\\&= - \int_{0}^{\infty} \dd t  \, P( {t+\tau_l}  ,+1|\tau_l)  \int_0^{t} \dd s\,  \frac{d}{ds}  \tr[ H_{\Sy\Cl} \rho_{\Sy\Cl}\lb  s+\tau_{l} |\tau_{l}\rb ]  \label{eq:S cycle line 2} 
			\\&= 
			 \tr[ H_{\Sy\Cl} \rho_{\Sy\Cl}(\tau_l|\tau_l) ]  
			 - \int_{0}^{\infty} \dd t  \, P( {t+\tau_l}  ,+1|\tau_l)\,  \tr[ H_{\Sy\Cl} \rho_{\Sy\Cl}\lb t+\tau_l|\tau_l\rb ] , \label{eq:S cycle line 3} 
		\end{align}
		where in line~\labelcref{eq:S cycle line 2}, we have used~\cref{eq:dev rho AC} and the fact that $\tr\big[ H_{\Sy\Cl} [H_{\Sy\Cl},  \rho_{\Sy\Cl} \lb t+\tau_l|\tau_l\rb]   \big]= 0$. And finally, in line~\labelcref{eq:S cycle line 3}, we have used~\cref{eq:prob norm def}.
		We now calculate the bound on $\delta E_l^\text{re}$. First note that for any (normalised or sub-normalised) density matrix ${\brho}_{\Sy\Cl}$, we have that 
		\begin{align}
			\left|\tr\!\left[[\bar H_\Sy, H_{\Sy\Cl}] \brho_{\Sy\Cl} \right]  \right|  &=   	\left|\tr\!\left[[\tilde H_\Sy, \tilde H_{\Sy\Cl}] \rho_{\Sy\Cl} \right]  \right| \tr[\brho_{\Sy\Cl}]    \leq 2 \sqrt{  \tr[\tilde H_\Sy ^2 \rho_{\Sy\Cl}]  \tr[\tilde H_{\Sy\Cl} ^2  \rho_{\Sy\Cl}]  } \,  \tr[\brho_{\Sy\Cl}]  	\label{eq: delta E re bound line 1}
			\\&  \leq 2 \sqrt{  \| \tilde H_\Sy ^2 \|_1  \| \tilde H_{\Sy\Cl} ^2 \|_1  } \,  \tr[\rho_{\Sy\Cl}]  =2 \sqrt{  \| \tilde H_\Sy ^2 \|_1  \| \tilde H_{\Sy\Cl} ^2 \|_1  } \,  \tr[\rho_{\Sy\Cl}]  		\label{eq: delta E re bound line 2} \\
			&= 2  \| \tilde H_\Sy  \|_F \| \tilde H_{\Sy\Cl}  \|_F   \,  \tr[\rho_{\Sy\Cl}] 
			\label{eq: delta E re bound line 3}
		\end{align}
		where in line~\labelcref{eq: delta E re bound line 1} we defined $\tilde H_\Sy  := \tr_\Cl[\tilde H_{\Sy\Cl} \,\rho^0_\Cl]$, (where $\tilde H_{\Sy\Cl}$ is defined in~\cref{def:epsion 0 H}); defined the normalised  state $\rho_{\Sy\Cl}:=  \brho_{\Sy\Cl}/ \tr[ \brho_{\Sy\Cl}]$, and  used Cauchy-Schwartz inequality twice. In lines~\labelcref{eq: delta E re bound line 2,eq: delta E re bound line 3} we have used the definitions of the one-norm and Frobenius norm, followed by definition~\labelcref{def:epsion 0 H}. 
		Now observer that from~\cref{eq:prob as a div,eq:prob norm def} we have that
		\begin{align}
			\tr[\rho_{\Sy\Cl}\lb x+\tau_l|\tau_l\rb]= \int_x^\infty \dd t P(t+\tau_l,+1|\tau_l).
		\end{align}
		Thus
		\begin{align}
			\int_0^\infty \dd x \,	\tr[\rho_{\Sy\Cl}\lb x+\tau_l|\tau_l\rb]  =  \left(  \lim_{x\to\infty}   x \int_x^\infty \dd t P(t+\tau_l,+1|\tau_l)  \right) +  \int_0^\infty \dd t\,  t\, P(t+\tau_l,+1|\tau_l),
		\end{align}
		where we have performed integration by parts. From~\cref{def:1st moment finite} it follows that $\lim_{x\to\infty} F(1,x)\in \rr$, where $F(1,x)$ is such that $d F(1,x) /dx =  x P(x+\tau_l,+1|\tau_l)$. We thus have,
		\begin{align}
			0\leq 	 \lim_{x\to\infty}   x \int_x^\infty \dd t P(t+\tau_l,+1|\tau_l) \leq   \lim_{x\to\infty}    \int_x^\infty \dd t\, t\, P(t+\tau_l,+1|\tau_l) =   \lim_{x\to\infty}    \left(  \left[\lim_{y\to\infty } F(1,y) \right] - F(1,x) \right)=  0.
		\end{align}
		Hence 
		\begin{align}
			\int_0^\infty \dd x \,	\tr[\rho_{\Sy\Cl}\lb x+\tau_l|\tau_l\rb]  =   \int_0^\infty \dd t\,  t\, P(t+\tau_l,+1|\tau_l)= \mathbb{M}_l(1).\label{eq:first mo eqivalence}
		\end{align}
		We thus conclude~\cref{eq:bound on delt E} from~\cref{eq: delta E re bound line 3,eq:first mo eqivalence} and definition~\labelcref{def:epsion 0 H}.
	\end{proof}

We now state and prove~\Cref{thm:heat dissipation upper bounds}, which  is straightforward to prove given the lemmas developed in this section:

	\thmStadyStateUpperbounds*
		\begin{proof}
			Follows straightforwardly by combining~\Cref{lem: power per cycle,eq:upper bounds orthgonalization isentropic}.
		\end{proof}
\subsection{Generalisation of~\Cref{thm:heat dissipation upper bounds} to more general semi-classical states}\label{sec:generalisation to more general semi-classical states}
We now generalise~\Cref{thm:heat dissipation upper bounds} to the case where $\rho_{\Sy\Cl}(\tau_l|\tau_l)\notin \mathcal{C}_{\Sy\Cl}^{clas.,l}$, but $\rho_{\Sy\Cl}(\tau_l|\tau_l)$ is a probabilistic mixture of semi-classical states. Namely to the case where for some $l\in\nnz$, there exists, 
\begin{align}
	\Big\{ \  \rho_{\Sy\Cl}^{(m)}(\tau_l|\tau_l)\in \mathcal{C}_{\Sy\Cl}^{clas.,l}   \Big\}_m \quad\text{and}\quad  \Big\{  \, p_m^{{l}}  \, \Big| \,  p_m^{{l}} >0, \, \sum_m  p_m^{{l}}=1  \Big\}
\end{align}
 such that 
 \begin{align}
 	\sum_m  p_m^{{l}}  \rho_{\Sy\Cl}^{(m)}(\tau_l|\tau_l) =  \rho_{\Sy\Cl}(\tau_l|\tau_l).\label{eq:mixed semi-calssical states def up bounds}
 \end{align}

Recall that the states $\{ \rho_{\Sy\Cl}(t+\tau_l|\tau_l) \}_l$ are provided by the dynamics in~\cref{eq:conditional states AC def,eq:tilde def norm}. It thus follows that 
 \begin{align}\label{eq:mixture time dependent}
	\sum_m  p_m^{{l}}(t) \rho_{\Sy\Cl}^{(m)}(t+\tau_l|\tau_l) =  \rho_{\Sy\Cl}(t+\tau_l|\tau_l), 
\end{align} 
where
 \begin{align}\label{eq:dynamics of m dependent states}
  \rho_{\Sy\Cl}^{(m)}(t+\tau_l|\tau_l) := \frac{ \rho_{\Sy\Cl}^{(m)}\lb t+\tau_l|\tau_l\rb}{\tr[\rho_{\Sy\Cl}^{(m)}\lb t+\tau_l|\tau_l\rb]},\qquad \rho_{\Sy\Cl}^{(m)}\lb t+\tau_l|\tau_l\rb:= \me^{t \mathcal{L}_{\Sy\Cl}^\textup{no re}} \left(	\rho_{\Sy\Cl}^{(m)}(\tau_l|\tau_l) \right).
\end{align} 
and
\begin{align}
	 p_m^{{l}}(t):= \frac{\tr[ \rho_{\Sy\Cl}^{(m)}\lb t+\tau_l|\tau_l\rb]}{\sum_n  p_n^{{l}} \tr[\rho_{\Sy\Cl}^{(n)}\lb t+\tau_l|\tau_l\rb]} p_m^{{l}}.
\end{align}
Before stating the generalization of~\Cref{thm:heat dissipation upper bounds}, we need to introduced a few definitions. These are essentially previous definitions but specialised to  the states $\tilde \rho_{\Sy\Cl}^{(m)}\lb t+\tau_l|\tau_l\rb$ in the mixture. These are:\\
Power consumption from $\tilde \rho_{\Sy\Cl}^{(m)}\lb t+\tau_l|\tau_l\rb$. We evaluate~\cref{eq:def: power in general} on $ \rho_{\Sy\Cl}^{(m)}\lb t+\tau_l|\tau_l\rb$ rather than $\rho_{\Sy\Cl}\lb t+\tau_l|\tau_l\rb$
\begin{align}\label{eq:def: enery re 2}
	P^{\textup{in}{(m)}}_l :=\frac{1}{T_0} \int_{0}^{\infty} \dd s \, \tr[H_{\Sy\Cl}\,	\mathcal{D}_{\Cl}^\text{re}(\rho_{\Sy\Cl}^{(m)}\lb  s+\tau_l|\tau_{l}\rb )].
\end{align}
We also need to specialise the definition of $l\thh$ isentropic time interval, $[\tau_l, t_{\max,l}+\tau_l)$, since now $t_{\max,l}$ will depend on $m$, since said interval depends on how the states $\rho^{(m)}_{\Sy\Cl}(\tau_l| \tau_l)$ evolve. 	
In particular, it is defined as follows: for a given $\varepsilon^\textup{gate}$,  the quantity  $t_{\max,l}^{(m)}\geq 0$ is the largest constant such that
\begin{align}
	\frac{1}{4(1+\lambda)^2} = \varepsilon^\textup{gate}(2 -{\varepsilon^\textup{gate}})   + 2  \varepsilon^{\textup{quality},(m)}_l(t_{\max,l}^{(m)}) \, \sqrt{1-\varepsilon^\textup{gate} }
	\label{eq:condition for t max 2}
\end{align}
where $\varepsilon^{\textup{quality},(m)}_l(\cdot)$ 
 is defined by
\begin{align}
	\varepsilon^{\textup{quality},(m)}_l(\tau)&:=  \int_0^{\tau}\dd t \,  \|  \vec V_\Cl \|_{w,2}^{(m)}(t,l), \\
	\\   \|   \vec V_\Cl \|_{w,2}^{(m)}(t,l)&:= {\left( \sum_j v_j^2 p_j^{(m)}(t,l) \!\right)\!}^{1/2}, \qquad p_j^{(m)}(t,l):= \braket{ v_j | \rho_{\Sy\Cl}^{(m)}\lb t+\tau_l|\tau_{l}\rb | v_j }. 
	 \label{eq:quality factor 2}
\end{align}

Since we are considering the dynamics of states during the $l$ isentropic time interval, $[\tau_l, t_{\max,l}^{(m)}+\tau_l)$,  instead of $[\tau_l, t_{\max,l}+\tau_l)$, different frequencies may be achievable. As such, we must define the gate frequency to depend on $m$. To do so, we still demand that~\cref{eq:sequnce orthgonal states} holds, but now $h_l$ will be $m$ dependent:  $h_l^{(m)}:= \max \{ j\in\zzp\, | \,t_{\max,l}^{(m)} -j t_1^{(m)} > 0 \}$. As before, we assume that the gate frequency of our computer is the same for all gates it can implement.  Furthermore, we quantify our tolerance to error in terms of trace distance $T$  from $\rho_\lo^{(m)}(t_j^{(m)}+\tau_l|\tau_l)$  (instead of $\rho_\lo(t_j+\tau_l|\tau_l)$ as before), namely by demanding
\begin{align}
	T\big(\rho_\lo^{(m)}(t_j^{(m)}+\tau_l|\tau_l),  \ket{j,l}_\lo\big) & \leq \varepsilon^\textup{gate}\in[0,1/2)
	,\label{eq:impelmentation quality m}
\end{align}
where $t_j^{(m)}=j t_1^{(m)}\in(0,t_{\max,l}^{(m)})$, $j\in\zzp$ and 
\begin{align}
	 \rho_{\lo}^{(m)}(t+\tau_l|\tau_l)=\tr_{\Sym\Cl}[ \rho_{\Sy\Cl}^{(m)}(t+\tau_l|\tau_l)].
\end{align}

Note that we have also made the time interval between gates $m$-dependent, (denoted $t_1^{(m)}$). This is necessary because that times for which~\cref{eq:impelmentation quality m} can be satisfied may differ now that the state itself is $m$-dependent. Crucially though, this $m$-dependency of the state dynamics enters through the state at the beginning of each cycle, as opposed to  the generator of the conditional dynamics, namely $\mathcal{L}^\textup{no re}_{\Sy\Cl}$. (These properties are due to~\cref{eq:dynamics of m dependent states}.)

Let us thus denote by $f^{(m)}$ the gate frequency at which gates are applied during the $l\thh$ isentropic time interval $[\tau_l, t_{\max,l}+\tau_l]$. It is defined as 
\begin{align}
	f^{(m)}:=1/t_1^{(m)}.
\end{align}

We also have to define the $m$ dependent counterpart to $\varepsilon^0_\textup{H}$. For this, we define the $m$ dependent upper bound on the 1st moment in analogy to~\cref{def:1st moment finite m}. We assume there exists $\mathbb{M}^{(m)}(1)<\infty$ such that
\begin{align}
 \int_0^\infty \dd t\,  t\, P^{(m)}(t+\tau_l,+1|\tau_l) \leq \mathbb{M}^{(m)}(1)\label{def:1st moment finite m}
\end{align}
for all $l\in\nnz$, where  (in analogy with~\cref{eq:prob plus one renewal def})
\begin{align}
	P^{(m)}(t+\tau_l,+1|\tau_{l}):=  \tr\left[\mathcal{D}^\text{re}_\Cl \left( \rho_{\Sy\Cl}^{(m)}\lb t+\tau_l|\tau_{l}\rb\right) \right].\label{eq:prob plus one renewal def m}
\end{align}
With this quantity, we can define $\varepsilon^{0,(m)}_\textup{H}$ analogously to~\cref{def:epsion 0 H}:  $\varepsilon^{0,(m)}_\textup{H}$ is any real number which upper bounds the l.h.s. of 
\begin{align}
	{	\left	\|  \tr_\Cl[\tilde H_{\Sy\Cl} \,\rho^0_\Cl] \right \|_F} {	\|  \tilde H_{\Sy\Cl}   \|_F} 2 \mathbb{M}(1) T_0 \leq \varepsilon_\textup{H}^{0,(m)}, \quad \tilde H_{\Sy\Cl}:= H_{\Sy\Cl}- H_\Cl, \label{def:epsion 0 H m}
\end{align}
where $\|  \cdot \|_F $ is the Frobenius norm. 

Finally, note that we could also let the upper bound in~\cref{eq:impelmentation quality m}, namely $\varepsilon^\textup{gate}$, be $m$ dependent, and the following theorem would hold equally well. We refrain from doing so simply to avoid extra complication. 

The following theorem shows that the results of~\Cref{thm:heat dissipation upper bounds} apply individually to each state $ \rho_{\Sy\Cl}^{(m)}(\tau_l|\tau_l) $ in the mixture of states comprising  $\rho_{\Sy\Cl}(\tau_l|\tau_l)$. Since Case 2) in~\Cref{thm:heat dissipation upper bounds} already applies to all states $\rho_{\Sy\Cl}(\tau_l|\tau_l)$, we will only state the generalization for case 1) in~\Cref{thm:heat dissipation upper bounds}.

\begin{theorem}[Generalization of~\Cref{thm:heat dissipation upper bounds} to more general semi-classical states]\label{thm:heat dissipation upper bounds generalization} 
	Consider a state $\rho_{\Sy\Cl}(\tau_l| \tau_l)$ in the $l\thh$ isentropic time interval which is a mixture of semi-classical states, namely satisfies~\cref{eq:mixed semi-calssical states def up bounds}, \textup{(}$l\in\nnz$\textup{)}. Consider the dynamics of any state $\rho_{\Sy\Cl}^{(m)}(t+\tau_l|\tau_l)$ during its $l\thh$ isentropic time interval, namely $[\tau_l, t_{\max, l}+\tau_l)$. The gate frequency $f^{(m)}$ during said regime is upper bounded for all $\varepsilon^\textup{gate}\in[0,1/2]$ as follows:
	\begin{align}
		f^{(m)} \leq 	\frac{ \sqrt{T_0^2 P^{\textup{in}{(m)}}_l +\varepsilon^{0,(m)}_\textup{H}}}{(\lambda+1)c_0 T_0}.
	\end{align} 
	Here $\lambda>0$, $c_0>0$ are numerical constants.
\end{theorem}
\begin{proof}
 It follows identically to that of~\Cref{thm:heat dissipation upper bounds}. One can check this by going though the proof step-by-step, adding the label $m$ to the states etc where appropriate. At a higher level, one can appreciate why the proof still works, by noting that the gate frequency, for a given $\varepsilon^\textup{gate}$, is fully determined by applying the dynamical map $\me^{t \mathcal{L}_{\Sy\Cl}^\textup{no re}}$ to a normalised quantum state on $\Sy\Cl$. In the case of~\Cref{thm:heat dissipation upper bounds}, this state is $\rho_{\Sy\Cl}(\tau_l | \tau_l)$. In the case of~\Cref{thm:heat dissipation upper bounds generalization}, it is  $\rho_{\Sy\Cl}^{(m)}(\tau_l | \tau_l)$. The quantities defined in~\Cref{sec:generalisation to more general semi-classical states}, are analogous to those required to make the theorem statement on~\Cref{thm:heat dissipation upper bounds} well-defined, but with the additional label $m$. Thus, the proof follows analogously.
\end{proof}

Let us now comment on how the gate frequency and power of the total dynamics relate to $\{f^{(m)}\}_m$ and  $\{P_l^{\textup{in},(m)} \}_m$ respectively in the $l\thh$ cycle.

Observe from~\cref{eq:def: power in general} that the total power, $P_l^\textup{in}$, is linear in the state. 
Thus from~\cref{eq:mixed semi-calssical states def up bounds,eq:def: enery re 2}, it follows that the total power per cycle is a probabilistic mixture of $P_l^{\textup{in}}$, namely
\begin{align}
	P_l^\textup{in} = \sum_m  p_m^l 	P_l^{\textup{in},{(m)}}.
\end{align}

Since the state is a probabilistic mixture of states all with (potentially) different gate frequencies $f^{(m)}$, the gate frequency associated with the dynamics of the state $\rho_{\Sy\Cl}(\tau_l | \tau_l)$ in this case is not well-defined. However, we can study upper bound the average gate frequency of the mixture when the conditions of~\Cref{thm:heat dissipation upper bounds generalization} apply, we find 
\begin{align}
	f^{\textup{ave}} :=\sum_m p^l_m  f^{(m)}    \leq \sum_m p^l_m   \frac{ \sqrt{T_0^2 P^{\textup{in}{(m)}}_l +\varepsilon^{0,(m)}_\textup{H}}}{(\lambda+1)c_0 T_0} \leq \frac{\sqrt{\sum_m  p_m^l   \left(T_0^2 P^{\textup{in}{(m)}}_l +\varepsilon^{0,(m)}_\textup{H} \right)}}{(\lambda+1)c_0 T_0} \leq \frac{\sqrt{ T_0^2 P^{\textup{in}}_l +\bar{\varepsilon_\textup{H}^0}}}{(\lambda+1)c_0 T_0}, \label{eq:average frequancy}
\end{align}
where in the penultimate line we used Jensen’s inequality, and we have defined the average instantaneous initial-cycle-state parameter error,
\begin{align}
	\bar{\varepsilon_\textup{H}^0} := \sum_m  p_m^l \varepsilon^{0,(m)}_\textup{H}.
\end{align}
Therefore, from~\cref{eq:average frequancy}, we see that when the state at the begging of the $l\thh$ cycle is a mixture of semi-classical states in $\mathcal{C}_{\Sy\Cl}^{clas.,l}$, that the average gate frequency achieved by the mixture still scales as the square-root of the power consumed in said cycle.

\subsection{Proof of~\Cref{lem:upperAndLowerBoundingEpQuality1}}\label{sec:proof of lemma upperAndLowerBoundingEpQuality1}
\LemupperAndLowerBoundingEpQuality*
\begin{proof}
The first and last equalities follow trivially from the definition~\labelcref{def:weighted q norm0}. The $1\stt$ inequality follows from Jensen’s inequality for the concave function $\sqrt{(\cdot)}$. For the second inequality, we note that the weighted one-norm and the weighted 2-norm can both be written in the from
\begin{align}
	\| \vec V_\Cl \|_{w,1}(t,l)&= \tr[V_\Cl \rho_{\Sy\Cl}\lb t+\tau_l| \tau_l\rb], 
	\\ 	\| \vec V_\Cl \|_{w,2}(t,l)&= \left( \tr[V_\Cl^2 \rho_{\Sy\Cl}\lb t+\tau_l | \tau_l\rb]\right)^{1/2}=  \left( \tr[\rho^{1/2}_{\Sy\Cl}\lb t+\tau_l | \tau_l\rb V_\Cl^{3/2} V_\Cl^{1/2} \rho^{1/2}_{\Sy\Cl}\lb t+\tau_l | \tau_l\rb]\right)^{1/2},\label{eq:scharts ready in proof eq}
\end{align}
where we have used the positive semi-definiteness of $V_\Cl$ and $\rho^{1/2}_{\Sy\Cl}\lb t+\tau_l | \tau_l\rb$. The last inequality now follows from two steps: first  apply the Cauchy–Schwarz inequality to the r.h.s. of~\cref{eq:scharts ready in proof eq}, and followed by noting the bound $\tr[V_\Cl^3 \brho] \leq \| V_\Cl^3\| \leq \| V_\Cl\|^3$ holds for all normalised (or sub-normalised) positive semi-definite operators $\brho$. Second, apply Jensen's inequality to the integral, noting that $d\mu:=dt/\tau$ is a probability density (the uniform one) for a probability measure on $[0,\tau]$.
\end{proof}

\subsection{Bounds for the range of the self-oscillator's isentropic time interval  and instantaneous initial-cycle-state parameter for~\Cref{thm:heat dissipation}}\label{sec:Calculation of the range of the self-oscillator isentropic regime  and}
Here we calculate the length $t_{\max,l}$, of the  $l\thh$ isentropic time interval, $[\tau_l, \tau_l+t_{\max,l})$, and $\varepsilon_\textup{H}^0$ for~\Cref{thm:heat dissipation}.

In this section, we will calculate the parameters $t_{\max,l}$ and $\varepsilon_\textup{H}^0$ for the quantum frequential computer which achieves the bounds in~\Cref{thm:heat dissipation}. By doing so, we show how  the requirements of~\Cref{thm:heat dissipation upper bounds}  can be readily met.
\boundingtmax*
The proof can be found at the end of this section, (and is reachable via this~\phantomsection
~\emph{\hyperref[proof:lemma-my-lemma]{\!\!hyperlink}}).  See below~\Cref{lem:bounding tmax for the quantum frequantic comp in the proof} for technical comments. On a more technical note, while the  interval~\cref{eq:isentropic interval in lemma} readily does not depend on the value of $\varepsilon^\textup{gate}\in\big[0,\,1-\sqrt{1-1/(4(1+\lambda)^2)}\,\big)$, its actual value within said interval will depend on $\varepsilon^\textup{gate}$ (this necessity can be seen from~\cref{eq:condition for t max}). 

We will now bound the instantaneous initial-cycle-state parameter, $\varepsilon_\textup{H}^0$, introduced in~\cref{def:epsion 0 H}.
\boundingInstantaneous*
Recall that $\SupPolyDecay(\cdot)$ is defined in~\cref{eq:def:SupPolyDecay}.
\begin{proof}
	From the defining relation for $\varepsilon_\textup{H}^0$, namely~\cref{def:epsion 0 H}, we observe that we need to bound the quantities $\left\|  \tr_\Cl[\tilde H_{\Sy\Cl} \,\rho^0_\Cl] \right \|_F$,  $\|  \tilde H_{\Sy\Cl}   \|_F$,  $\mathbb{M}(1)$, where $\tilde H_{\Sy\Cl}:= H_{\Sy\Cl}- H_\Cl$. 
	Recalling~\cref{eq:main complete ham 3},  we find that in the case of~\Cref{thm:heat dissipation}, the Hamiltonian $\tilde H_{\Sy\Cl}$  is given by
\begin{align}
	\tilde{H'}_{{\M_0}\lo\W\Cl}:=H'_{{\M_0}\lo\W\Cl}- H_\Cl= \sum_{l=1}^{N_g-1} I_{{\M_0}\lo\W} ^{(l)}\otimes I_\Cl^{(l)}.\label{eq:main complete ham tilde}
\end{align}
Furthermore, recall that the output state of the channel $\mathcal{D}_\Cl(\cdot)$, namely $\rho_\Cl^0$, in the case of~\Cref{thm:heat dissipation}, is given by $\rho_\Cl^0=\proj{0}_\Cl=\proj{\Psi(0)}_\Cl$.
Thus
\begin{align}
	\left\|	\tr_\Cl\!\!\left[\tilde{H'}_{{\M_0}\lo\W\Cl} \,\rho_\Cl^0\right] \right\|_F &=   \sqrt{\tr\left[\left(\tr_\Cl\!\!\left[\tilde{H'}_{{\M_0}\lo\W\Cl} \,\rho_\Cl^0\right]\right)^2\right] } \leq \sqrt{  \sum_{l,k=1}^{N_g-1} \left|\tr\left[I_{{\M_0}\lo\W} ^{(l)}  I_{{\M_0}\lo\W} ^{(k)} \right]  \tr[I_\Cl^{(l)} \rho_\Cl^0] \tr[I_\Cl^{(k)} \rho_\Cl^0]\right|  }
	\\ &\leq \sqrt{  \sum_{l,k=1}^{N_g-1}\sqrt{ \tr\!\left[{I_{{\M_0}\lo\W}^{(l)^{ \textstyle{2}} } } \right]  \tr\!\left[{I_{{\M_0}\lo\W}^{(k)^{ \textstyle{2}} }  } \right] } \left|\tr[I_\Cl^{(l)} \rho_\Cl^0] \tr[I_\Cl^{(k)} \rho_\Cl^0]\right|  }\label{eq:line:intermediatewsqrt}
	\\ &=  \sum_{l,k=1}^{N_g-1}\sqrt{ \tr\!\left[{I_{{\M_0}\lo\W}^{(l)^{ \textstyle{2}} } } \right] }\, \tr[I_\Cl^{(l)} \rho_\Cl^0] \leq N_g \max_{l\in\{1,\ldots,N_g-1  \} }  \left\{\sqrt{ \tr\!\left[{I_{{\M_0}\lo\W}^{(l)^{ \textstyle{2}} } } \right] } \left| \braket{\Psi(0)|I_\Cl^{(l)}|\Phi(0)}_\Cl \right|  \right\}
	\\&	\leq   	f'(\eta, \bar\varepsilon) \, \textup{poly}'(d) \, d^{-3/\sqrt{\bar\varepsilon}}  \sqrt{  \max_{l\in\{1,\ldots,N_g-1  \} }  \left\{ \tr\!\left[{I_{{\M_0}\lo\W}^{(l)^{ \textstyle{2}} } } \right]    \right\} },   \label{eqfoast line of1st eq fo last proof}
\end{align}
for all $\eta\in(0,1]$, $\bar\varepsilon\in(0,1/6)$, $d\in\nnp$. In line~\labelcref{eq:line:intermediatewsqrt}, we have used the Cauchy–Schwarz inequality,  and~\cref{eq:sum interaction temts small} 
 in line~\labelcref{eqfoast line of1st eq fo last proof}. Recall that the terms $\{ I_{{\M_0}\lo\W}^{(l)}\}_l$ are finite-dimensional, have bound spectrum, and are $d$-independent. As such, the factor 
\begin{align}\label{eq:argument why indepdendedt}
 \max_{l\in\{1,\ldots,N_g-1  \} }  \left\{ \tr\!\left[{I_{{\M_0}\lo\W}^{(l)^{ \textstyle{2}} } } \right]    \right\} ,
\end{align}
from the last line of the above inequality, has an upper bound which is $d$-independent even though $N_g$ appearing in the set which $l$ belongs to, is $d$ dependent. 
Thus recalling the linear relationship between $d$ and $P^\textup{in}$, we find 
\begin{align}
		\left\|	\tr_\Cl\!\!\left[\tilde{H'}_{{\M_0}\lo\W\Cl} \,\rho_\Cl^0\right] \right\|_F \leq f'(\eta, \bar\varepsilon) \, \textup{poly}'(P^\textup{in}) \, {P^\textup{in}}^{-3/\sqrt{\bar\varepsilon}}.\label{eq:finel bound on Frobenus specail case small}
\end{align}
Thus to prove the lemma, the only remaining challenge is to prove that the terms $\|  \tilde H_{\Sy\Cl}   \|_F$,  $\mathbb{M}(1)$ grow at most polynomially in $d$. (Since recall that $d$ is itself directly proportional to $P^\textup{in}$).

Similarly to above, we have that 
\begin{align}
	\left\|	\tilde{H'}_{{\M_0}\lo\W\Cl}  \right\|_F &=   \sqrt{\tr\left[\tilde{H'}_{{\M_0}\lo\W\Cl}^2\right] } \leq \sqrt{  \sum_{l,k=1}^{N_g-1} \left|\tr\left[I_{{\M_0}\lo\W} ^{(l)}  I_{{\M_0}\lo\W} ^{(k)} \right]  \tr[I_\Cl^{(l)} I_\Cl^{(k)} ] \right|  } \leq \sum_{l=1}^{N_g-1} \sqrt{\tr\!\left[I_{{\M_0}\lo\W} ^{(l)^{ \textstyle{2}}}  \right]  \tr\!\left[I_\Cl^{(l)^{ \textstyle{2}}}\right]  }
	\\& \leq N_g \sqrt{\max_{l=1,\ldots, N_g-1} \left\{\tr\!\left[I_{{\M_0}\lo\W} ^{(l)^{ \textstyle{2}}}  \right] \right\}    \max_{l=1,\ldots, N_g-1}  \left\{\tr\!\left[I_\Cl^{(l)^{ \textstyle{2}}}\right]  \right\}}
		\\& \leq  N_g \sqrt{d\max_{l=1,\ldots, N_g-1} \left\{\tr\!\left[I_{{\M_0}\lo\W} ^{(l)^{ \textstyle{2}}}  \right] \right\}   \max_{l=1,\ldots, N_g-1} \left\{ \left( {E.V.}_{\max}\,   I_\Cl^{(l)} \right)^2\right\}},\label{eq:middel product term last eq}
\end{align}
where ${E.V.}_{\max}\,   I_\Cl^{(l)}$ is the maximum eigenvalue of the positive semi-definite operator $ I_\Cl^{(l)}$. In~\cref{eq:jth ev of lth int term} an upper bound for all $l\in{1,\ldots, N_g}$ and for all eigenvalues of $I_\Cl^{(l)}$ was found which  is independent of $l$. It was shown to grow at most polynomially in $d$ (and thus in $P^\textup{in}$) for all $\bar\varepsilon \in(0,1)$. The prefactor 
\begin{align}
	\max_{l=1,\ldots, N_g-1} \left\{\tr\!\left[I_{{\M_0}\lo\W} ^{(l)^{ \textstyle{2}}}  \right] \right\}  
\end{align}
is finite and $d$-independent, as reasoned around~\cref{eq:argument why indepdendedt}.

All that is left, is to show that $\mathbb{M}(1)$, defined in~\cref{def:1st moment finite}, grows sufficiently slowly such that the product of all three terms is still decaying with $P^\textup{in}$. By definition of $\mathbb{M}(1)$, we need to obtain an upper bound on the 1st moment of $P(t+\tau_l| \tau_l)$ which holds for all $l$. We proceed as follows
\begin{align}
	 \mathbb{M}_l(1)& =   \int_0^\infty \dd t\,  t\, P(t+\tau_l,+1|\tau_l) = \int_0^\infty \dd t \,	\tr[\rho_{\M_0\lo\W\Cl}\lb t+\tau_l|\tau_l\rb] \label{eq:line first moment bounding 1}
	 \\ & =  \int_0^\infty \dd t \,	\tr[\me^{\mi t H_{ \M_0\lo\W\Cl} -t V_\Cl} \me^{-\mi t H_{ \M_0\lo\W\Cl} -t V_\Cl} \rho_{\M_0\lo\W\Cl}(\tau_l|\tau_l)]   \label{eq:line first moment bounding 2}
	 \\ & \leq  \int_0^\infty \dd t \,   \left\|  \me^{\mi t H_{ \M_0\lo\W\Cl} -t V_\Cl} \me^{-\mi t H_{ \M_0\lo\W\Cl} -t V_\Cl} \right\|_\infty     \leq  \int_0^\infty \dd t \,   \left\|  \me^{\mi t H_{ \M_0\lo\W\Cl} -t V_\Cl} \right\|_\infty     \left\|  \me^{-\mi t H_{ \M_0\lo\W\Cl} -t V_\Cl} \right\|_\infty      \label{eq:line first moment bounding 3}
	  \\ & \leq \int_0^\infty \dd t \,   \left\|  \me^{\mi t H_{ \M_0\lo\W\Cl}} \right\|_\infty     \left\|  \me^{-\mi t H_{ \M_0\lo\W\Cl}} \right\|_\infty   \left\|  \me^{-t V_\Cl} \right\|_\infty^2 = \int_0^\infty \dd t \,  \me^{-2t E.V._{\min} V_\Cl}    \label{eq:line first moment bounding 4}
	 \\ & \leq  \int_0^\infty \dd t \,  \me^{-t \varepsilon_b/T_0}  = \frac{T_0}{\varepsilon_b},  \label{eq:line first moment bounding 5}
\end{align}
where in line~\labelcref{eq:line first moment bounding 1} we have used~\cref{eq:first mo eqivalence} but now specialise the system to the case setup in~\Cref{thm:heat dissipation}, namely $\Sy=\M_0\lo\W$. In line~\labelcref{eq:line first moment bounding 2}, we have used~\cref{eq:no re dynamics} . In line~\labelcref{eq:line first moment bounding 3}, we used the Schatten $p=\infty$ norm. In line~\labelcref{eq:line first moment bounding 4}, we have used the Golden-Thompson inequality for the operator norm~\cite{Bhatia1997} and defined $E.V._{\min} V_\Cl$ to be the minimum eigenvalue of the positive semi-definite operator $V_\Cl$.  In line~\labelcref{eq:line first moment bounding 5}, we have used~\cref{def:V c pot} and the fact that $\gamma_0  I_\Cl^{N_g} \geq 0$.

The upper bound on $\mathbb{M}_l(1)$ is $l$-independent and holds for all $l$. It thus constitutes an upper bound on $\mathbb{M}(1)$ too, by definition of $\mathbb{M}(1)$. Thus
\begin{align}
	\mathbb{M}(1) \leq  \frac{T_0}{\varepsilon_b}.
\end{align}
Now, using~\cref{eq:lvarepsion b choice,Eq: lower bound P''} we thus find
\begin{align}
	\mathbb{M}(1) \leq  \frac{T_0}{\left(\sum_{q=1}^{N_g}  \tilde d(\m_{l,q})\right)  }   \, E_0^{1/(2\sqrt{\bar\varepsilon})} \leq  T_0   \, E_0^{1/(2\sqrt{\bar\varepsilon})}= T_0 \left(	P^\textup{in} -\delta E^\text{re}_1/T_0 \right)^{1/(2\sqrt{\bar\varepsilon})}. \label{eq:lvarepsion b choice new}
\end{align}
Thus recalling that $|\delta E^\text{re}_1|$ is a decaying function for all $\bar\varepsilon\in(0,1)$ (see~\cref{eq:line:communte52}).

This from~\cref{eq:finel bound on Frobenus specail case small,eq:middel product term last eq,eq:lvarepsion b choice new} we conclude~\cref{eq:main lemma eq for initantaneous para calc}.
\end{proof}

We now prove~\Cref{lem:bounding tmax for the quantum frequantic comp in the proof}:
\begin{proof}
	  \phantomsection
	\label{proof:lemma-my-lemma}
	Due to the data processing inequality (with the partial trace as the data processing channel), it follows from~\cref{eq:dist condition thmor} in~\Cref{thm:heat dissipation}, that the gate error $\varepsilon^\textup{gate}$ (defined in~\cref{eq:impelmentation quality}) in the case of~\Cref{thm:heat dissipation} converges to zero as $P^\textup{in}\to\infty$.  
	
	Clearly, $\varepsilon^\textup{gate}\to 0$ as $P^\textup{in}\to\infty$ implies $\varepsilon^\textup{gate}\in\big[0,\,1-\sqrt{1-1/(4(1+\lambda)^2)}\,\big)$ as $P^\textup{in}\to\infty$ (since recall that by definition  $\varepsilon^\textup{gate} \geq 0$ and  $\lambda>0$ is a dimensionless constant). Therefore, in the following proof, we will show that $\varepsilon^0_\textup{H}$ satisfies the lemma statement (i.e.~\cref{eq:isentropic interval in lemma}) for all $\varepsilon^\textup{gate}\in\big[0,\,1-\sqrt{1-1/(4(1+\lambda)^2)}\,\big)$.
	
	Note that the upper limit of the interval for $\varepsilon^\textup{gate}$, namely $1-\sqrt{1-1/(4(1+\lambda)^2)}$, is the solution to~\cref{eq:condition for t max} when $\varepsilon_l^\textup{quality}(t_{\max,l})=0$. Furthermore, recall that $\varepsilon_l^\textup{quality}(t_{\max,l})\geq 0$ is  a monotonically increasing function of $t_{\max,l}$. Therefore, note that a value of $\varepsilon_l^\textup{quality}(t_{\max,l})$ which tends to zero (as $P^\textup{in}\to\infty$) is a solution to~\cref{eq:condition for t max} (for sufficiently large $P^\textup{in}$ and a value of $\varepsilon^\textup{gate}$ which tends to $ 1-\sqrt{1-1/(4(1+\lambda)^2)}$ as $P^\textup{in}\to\infty$). Thus, for every $\varepsilon^\textup{gate}\in\big[0,\,1-\sqrt{1-1/(4(1+\lambda)^2)}\,\big)$, a lower bound, denoted $\tau^*$, on $t_{\max,l}$ in the large $P^\textup{in}$ limit, can be found by finding a value of $\tau^*$ such that $\varepsilon_l^\textup{quality}(\tau^*)\to0$ as $P^\textup{in}\to\infty$.
	
	Likewise, an upper bound, denoted $\tau^{**}$, on  $t_{\max,l}$ can be found for all $\varepsilon^\textup{gate}\in\big[0,\,1-\sqrt{1-1/(4(1+\lambda)^2)}\,\big)$ in the large $P^\textup{in}$ limit,  by finding a value $\tau^{**}$ such that $\varepsilon_l^\textup{quality}(t_{\max,l})\geq 1$ as $P^\textup{in}\to\infty$. This follows from the same continuity and monotonicity arguments as above in addition to noting that~\cref{eq:condition for t max} has no solution for $\varepsilon_l^\textup{quality}(t_{\max,l})\geq 1/4$.

	We start by proving the upper bound. For $\tau^{**}=T_0$, we have, using~\Cref{lem:upperAndLowerBoundingEpQuality1,eq:thm3 prob dist}
	\begin{align}
		\varepsilon_l^\textup{quality}(T_0) \geq 	\frac{1}{2}\int_0^{T_0} \dd t \, P(t+\tau_l, +1 | \tau_l )\geq 	\frac{1}{2}\int_{T_0-t_1}^{T_0} \dd t \, P(t+\tau_l, +1 | \tau_l ) =\frac{1-\varepsilon_r}{2},
	\end{align}
	where $\varepsilon_r$ is given by~\cref{eq:thm3 prob dist} in~\Cref{thm:heat dissipation}, and decays super-polynomially [i.e. it belongs to the previously-defined in~\cref{eq:def:SupPolyDecay} class of functions $\SupPolyDecay(P^{\textup{in}})$]. This thus proves the upper bound in~\cref{eq:isentropic interval in lemma}.
	
	For the lower bound, we choose $\tau^{*}=T_0-t_1$ and use~\Cref{lem:upperAndLowerBoundingEpQuality1,eq:thm3 prob dist} to show
	\begin{align}
		\varepsilon^\textup{quality}_l(T_0-t_1) &\leq     \frac{\left(\| V_\Cl \| \,(T_0-t_1)\right)^{3/4}} {2^{1/4}}   \left(\int_0^{T_0-t_1} \dd t \, P(t+\tau_l, +1 | \tau_l )\right)^{1/4} 
		\\ &\leq    \frac{\left(\| V_\Cl \| \,T_0 \right)^{3/4}} {2^{1/4}}   \left(1-\int_{T_0-t_1}^\infty \dd t \, P(t+\tau_l, +1 | \tau_l )\right)^{1/4} \leq    \frac{\left(\| V_\Cl \| \,T_0 \right)^{3/4}} {2^{1/4}}   \left(1-\int_{T_0-t_1}^{T_0} \dd t \, P(t+\tau_l, +1 | \tau_l )\right)^{1/4}  \label{eq:line used prob property}
		\\ &=  \frac{\left(\| V_\Cl \| \,T_0 \right)^{3/4}} {2^{1/4}}   \left(\varepsilon_r\right)^{1/4} \label{eq:last line varepsilon quanlity eq}
	\end{align}
where in line~\labelcref{eq:line used prob property} we have used the fact that the integration of $P(t+\tau_l, +1 | \tau_l )$ over $\rrz$ is one due to normalization of the probability distribution.

We now need to upper bound $\| V_\Cl\|$ to show that it does not grow too quickly with $P^\textup{in}$. Recall that $P^\textup{in}$ is proportional to $d$ in the case of the quantum frequential computer in the proof of~\Cref{thm:heat dissipation}. Therefore, it suffices to upper bound how it scales with $d$ for large $d$. The operator norm of a positive semi-definite matrices is give by the largest eigenvalue of said matrix. We now  find an upper bound which holds for all the eigenvalues of $V_\Cl$ and hence also for its largest. From~\cref{eq:thm3 prob dist} and definitions~\cref{eq:rho conditional dybamics def,eq:rho conditional dybamics def,def:squre braket control}, we find that the $j\thh$ eigenvalue of $V_\Cl\geq 0$ is given by
\begin{align}
	\gamma_0 \left[  \frac{2\pi}{T_0} \bar V_0 \left(\frac{2\pi}{d} j\right)\Big{|}_{x_0x_0^{(l)}}  \right] + \frac{\varepsilon_b}{2 T_0} ,\label{eq: intert for now}
\end{align}
$j\in\mathcal{S}_d(k_0)$, for $l=N_g$.
Let us first bound the term in square brackets for any $l$ since it will be of independent use later: Thus for all  $j\in\mathcal{S}_d(k_0)$, we have
\begin{align}\label{eq:jth ev of lth int term}
	  \frac{2\pi}{T_0} \bar V_0 \left(\frac{2\pi}{d} j\right)\Big{|}_{x_0=x_0^{(l)}}  \leq &  \frac{2\pi }{T_0} \left[ \max_{x\in[-\pi,\pi]}   \bar V_0 \left(x\right)\Big{|}_{x_0=x_0^{(l)}}   \right]    
		= \frac{2\pi }{T_0} \left[ \max_{x\in[-\pi,\pi]}   \bar V_0 \left(x\right)  \right]
		\\ \leq & \frac{2\pi  n A_0}{T_0} \left[\left( \max_{x\in[-\pi,\pi]}  V_B \left(n x\right)\right) + \left( \max_{x\in[-\pi,\pi]} \sum_{\substack{p=-\infty  \\ p\neq 0}}^\infty V_B \left(n (x+2\pi p)\right)\right)\right]
		\\ \leq & \frac{2\pi  n A_0}{T_0} \left[1 + \left( \max_{x\in[-\pi,\pi]} \sum_{\substack{p=-\infty  \\ p\neq 0}}^\infty \frac{1}{\left(\pi n (x+2\pi p)\right)^{2N}}\right)\right]
		\\ \leq & \frac{2\pi n A_0}{T_0} \left[1 +  \max_{x\in[-\pi,\pi]} \sum_{p=1}^\infty \frac{1}{\left(\pi n (x+2\pi p)\right)^{2N}}+\frac{1}{\left(\pi n (-x+2\pi p)\right)^{2N}}\right]
		\\ \leq & \frac{2\pi  n A_0}{T_0} \left[1 +  2\sum_{p=1}^\infty \frac{1}{\left(\pi n (-\pi+2\pi p)\right)^{2N}}\right]
		 \leq  \frac{2\pi n A_0}{T_0} \left[1 +   \frac{1}{\left(\pi^2 n\right)^{2N}} \frac{\pi^2}{4}\right]
		\\ \leq & \frac{2\pi  n A_0}{T_0} \left[1 +   \frac{1}{4 n^{2}} \right],
\end{align}
where we have used basic properties of the functions involved, e.g. the $2\pi$-periodicity of $\bar V_0$, the boundedness of $V_B(\cdot)\leq 1$ and the fact that $N\in\nnp$. Recall that $A_0$ is $d$-independent (since $N$ is also $d$-independent by design). Recall that $n$ on the other hand, is $d$ dependent. As stated in the proof of~\Cref{thm:heat dissipation}, its dependency is the same as that used in the proof of~\Cref{thm:comptuer with fixed memoryapp} and hence given by~\cref{eq:n equation}. Thus since $\epsilon_5\in(0,1)$ and $\sigma =d^{\epsilon_6}$ with $\epsilon_6\in(\epsilon_5, 1)$, and all other quantities appearing on the r.h.s. of~\cref{eq:n equation} are independent of $d$ and $\bar\varepsilon$, we conclude that $n\leq d^2$ for sufficiently large $d$ and independently of $\bar\varepsilon\in(0,1/6)$.

Therefore, from~\cref{eq: intert for now}, we find that the maximum eigenvalue of $V_\Cl\geq 0$ is bounded by 
\begin{align}
\max_{j\in\mathcal{S}_d(k_0)}\left\{   	\gamma_0 \left[  \frac{2\pi}{T_0} \bar V_0 \left(\frac{2\pi}{d} j\right)\Big{|}_{x_0=x_0^{(l)}}  \right]  + \frac{\varepsilon_b}{2 T_0}  \right\}  \leq \,   \frac{2\pi \gamma_0 n A_0}{T_0} \left[1 +   \frac{1}{4 n^{2}} \right] + \frac{\varepsilon_b}{2 T_0}.
\end{align}
The $\gamma_0$ quantity is given by~\cref{def:squre braket control}, where  $\bar\gamma_0$ is independent of $d$ and $\bar\varepsilon$. It scales with $d$ as $\gamma_0\sim d^{{\bar\varepsilon}^2}$. Finally, recall that $\varepsilon_b$ is given by~\cref{eq:lvarepsion b choice} and hence tends to zero as $E_0$ tends to infinity for all $\bar\varepsilon\in(0,1/6)$. Recall that $E_0$ is directly proportional to $P^\textup{in}$ for all $\bar\varepsilon\in(0,1/6)$, and therefore $\varepsilon_b$ tends to zero, for all $\bar\varepsilon\in(0,1/6)$, as $P^\textup{in}$ tends to infinity.

Therefore, we conclude that the largest eigenvalue $v_\Cl$ is upper bounded by an $\bar\varepsilon$-independent polynomial in $P^\textup{in}$. And thus $\| V_\Cl \|$ is also upper bounded by an $\bar\varepsilon$-independent polynomial in $P^\textup{in}$.  Hence since $\varepsilon_r$ belongs to the function class $\SupPolyDecay(P^{\textup{in}})$, (introduced in~\cref{eq:def:SupPolyDecay}), it follows from~\cref{eq:last line varepsilon quanlity eq} that $\varepsilon^\textup{quality}_l(T_0-t_1)\to 0$ for all fixed $\bar\varepsilon\in(0,1/6)$ as $P^\textup{in}\to\infty$. From our above argument, this thus establishes the lower bound in~\cref{eq:isentropic interval in lemma}.
\end{proof}

\bibliographystyle{alphaurl_mod}
\bibliography{myrefs}
\end{document}